\documentclass[12pt]{article}
\usepackage{mheck}

\usepackage{longtable}

\usepackage{blkarray}

\usepackage{tikz} 
\usetikzlibrary{matrix}

\theoremstyle{theorem}
\newtheorem{thm}{Theorem}
\newtheorem*{thms}{Theorem}
\newtheorem*{fact}{Fact}

\newtheorem{corl}{Corollary}[subsection]
\newtheorem{pro}{Proposition}[subsection]

\newtheorem{lem}{Lemma}[subsection]
\newtheorem*{cla}{Claim}

\newtheorem*{ass}{Mild assumption}

\newtheorem{prob}{Problem}
\newtheorem*{prin}{Physical principle}
\newtheorem*{stra}{Strategy}

\theoremstyle{definition}
\newtheorem{defn}{Definition}
\newtheorem{rem}{Remark}[subsection]

\newtheorem{exe}{Example}

\def\Dsl{\,\raise.15ex\hbox{/}\mkern-13.5mu D}

\def\BK{\mathbb{K}}
\def\BQ{\mathbb{Q}}
\def\fO{\mathfrak{O}}
\def\fo{\mathfrak{o}}
\def\bn{\breve{n}}

\date{January, 2018}

\title{Geometric classification of 4d $\cn=2$ SCFTs
}

\authors{Matteo Caorsi\footnote{e-mail: {\tt matteocao@gmail.com}} and Sergio Cecotti\footnote{e-mail: {\tt cecotti@sissa.it}}\vskip 9pt

\centerline{SISSA, via Bonomea 265, I-34100 Trieste, ITALY}}

\abstract{The classification of 4d $\cn=2$ SCFTs boils down to the classification of conical special geometries with closed Reeb orbits (CSG).  Under mild assumptions, one shows that the underlying complex space of a CSG is (birational to) an affine cone over a simply-connected $\BQ$-factorial \emph{log}-Fano variety with Hodge numbers $h^{p,q}=\delta_{p,q}$.
With some plausible restrictions, this means that the Coulomb branch chiral ring $\mathscr{R}$ is a graded polynomial ring
generated by global holomorphic functions $u_i$ of dimension $\Delta_i$. The coarse-grained classification of the CSG consists in listing the (finitely many) dimension $k$-tuples $\{\Delta_1,\Delta_2,\cdots,\Delta_k\}$ which are realized as Coulomb branch dimensions of some rank-$k$ CSG: this is the problem we address in this paper.
Our sheaf-theoretical analysis leads to an Universal Dimension Formula for the possible $\{\Delta_1,\cdots,\Delta_k\}$'s. For Lagrangian SCFTs the Universal Formula reduces to the fundamental theorem of Springer Theory. 

The number $\boldsymbol{N}(k)$ of dimensions allowed in rank $k$ is given by a certain sum of the Erd\"os-Bateman Number-Theoretic function (sequence A070243 in OEIS) so that for large $k$
$$
\boldsymbol{N}(k)=\frac{2\,\zeta(2)\,\zeta(3)}{\zeta(6)}\,k^2+o(k^2).
$$
In the special case $k=2$ our dimension formula reproduces a recent result by Argyres \emph{et al.}

Class Field Theory implies a subtlety: certain dimension $k$-tuples $\{\Delta_1,\cdots,\Delta_k\}$ are consistent only if supplemented by additional selection rules on the electro-magnetic charges, that is, for a SCFT with these Coulomb dimensions not all charges/fluxes consistent with Dirac quantization are permitted.

Since the arguments tend to be abstract, we illustrate the various aspects with several concrete examples and perform a number of explicit checks. We include detailed tables of dimensions for the first few $k$'s.}

\begin{document}

\maketitle

\newpage

\tableofcontents

\newpage

\section{Introduction and Overview}

Following the seminal papers by Seiberg and Witten \cite{SW0,SW1}, in the last years a rich landscape of four-dimensional $\cn=2$ superconformal field theories (SCFT) had emerged, mostly without a weakly-coupled Lagrangian formulation \cite{Eguchi:1996ds,Minahan:1996cj,KKV,GVW,Gaiotto:2009we,Gaiotto:2009hg,Cecotti:2011gu,Gaiotto:2012rg,Cecotti:2012jx,Cecotti:2013lda,DelZotto:2015rca,Wang:2015mra,Xie:2015rpa,Chen:2016bzh,Wang:2016yha,Chen:2017wkw}. It is natural to ask for a map of this vast territory, that is, for a classification of unitary 4d $\cn=2$ SCFTs. The work in this direction follows roughly two approaches: the first one aims to partial classifications of $\cn=2$ SCFTs having some specific construction \cite{KKV,GVW,Gaiotto:2009we,DelZotto:2015rca,Wang:2016yha} or particular property \cite{Eguchi:1996ds,Cecotti:2011rv}. The second approach, advocated in particular by the authors of refs.\!\cite{ACSW,ACS2,Argyres:2015ffa,Argyres:2015gha,Argyres:2016xua,Argyres:2016xmc}, relates  the classification of $\cn=2$ SCFTs to  the geometric problem of classifying the \emph{conic special geometries} (CSG) which describe their IR physics along the Coulomb branch $M$ \emph{\'a la} Seiberg and Witten \cite{SW0,SW1}. The present paper belongs to this second line of thought: it is meant to be a contribution to the geometric classification of CSG with applications to $\cn=2$ SCFT.
\smallskip

Comparison with other classification problems in complex geometry suggests that, while describing all CSG up to isomorphism may be doable when $M$ has very small dimension,\footnote{\ The complex dimension $k$ of the Coulomb branch $M$ is also known as the \emph{rank} of the SCFT. We shall use the two terms interchangeably.}  it  becomes rapidly intractable as we increase the rank $k$.
A more plausible program would be a \emph{coarse-grained} classification of the CSG not up to isomorphism but rather up to some kind of ``birational'' equivalence, that is, neglecting the details of the geometry along ``exceptional'' loci of positive codimension. This is the point of view we adopt in our analysis.
\smallskip

The classification problem is in a much better shape than one may expect. Indeed, while \emph{a priori} the Coulomb branch $M$ is just a complex analytic space, it follows from the local properties of special K\"ahler geometry that a CSG must be a complex cone over a base $K$ which is a normal projective variety. $K$ is (birational to) an algebraic variety of a very special kind: a simply-connected \emph{log-Fano} \cite{fano} with Picard number one and trivial Hodge diamond (a special instance of a  \emph{Mori dream space} \cite{moridream}).
\smallskip

In practice, in the coarse-grained classification we limit our ambition to the description of the allowed rings $\mathscr{R}$ of holomorphic functions on $M$ (the Coulomb branch chiral rings). Several distinct (deformation-types of) SCFTs have the same $\mathscr{R}$ but differ in other respects as their flavor symmetry. The simplest example of a pair of distinct SCFT with the same $\mathscr{R}$ is given by $SU(N)$ SQCD with $2N$ fundamentals and $\cn=2^*$ with the same gauge group; see refs.\!\cite{Argyres:2015ffa,Argyres:2015gha,Argyres:2016xua,Argyres:2016xmc} for additional examples in rank 1.
  The chiral ring $\mathscr{R}$ of a SCFT is graded by the value of the $U(1)_R$ charge (equal to the dimension $\Delta$ for a chiral operator). The general expectation\footnote{\ See however \cite{Argyres:2017tmj} for a discussion of the phenomena which would appear if this is not the case. We shall briefly elaborate on this topic in the third \emph{caveat} of \S.\,\ref{kkkaz1204}.} is that the Coulomb branch chiral ring is a graded free  polynomial ring, 
\be\label{wwwharing}
\mathscr{R}=\C[u_1,u_2,\cdots,u_k].
\ee
This is equivalent\footnote{\ Since the dimensions $\Delta$ are positive rationals.} to saying that the log-Fano $K$ is a \emph{weighted projective space} (WPS). The last statement is only slightly stronger than the one in the previous paragraph: all WPS are simply-connected log-Fano with Picard number one and trivial Hodge diamond  \cite{dolga}. Conversely, a \emph{toric} log-Fano with these properties is necessarily a (fake\footnote{\ All \emph{fake} WPS are quotients of WPS by finite Abelian groups.}) weighted projective space \cite{toric}. 
Then it appears that the log-Fano varieties  which carry all the structures implied by special geometry form to a class of manifolds only slightly more general than the WPS.
This explains why ``most'' $\cn=2$ models have free chiral rings.

Assuming \eqref{wwwharing}, the information encoded in the ring $\mathscr{R}$ is just the $k$-tuple $\{\Delta_1,\Delta_2,\cdots,\Delta_k\}$ of the $U(1)_R$ charges of its free generators
$u_i$. Even if the ring is non-free
the spectrum of dimensions of $\mathscr{R}$ is a basic invariant of the SCFT.
The coarse-grained classification of CSG then aims to list the allowed Coulomb branch dimension $k$-tuples $\{\Delta_1,\cdots,\Delta_k\}$ for each rank $k\in\mathbb{N}$. An even less ambitious program is to list the finitely-many real numbers $\Delta$ which may be the dimension of a Coulomb branch generator in a $\cn=2$ SCFT of rank at most $k$. In a unitary theory, all $\Delta$'s are rational numbers $\geq1$.
If the chiral ring is non-free, but with a finite free covering, we easily reduce to the above case.

\subparagraph{Note added {\rm (\textit{non-free chiral rings})}.} After the submission of the present paper, the article \cite{Bourget:2018ond} appeared in the arXiv where  examples of $\cn=2$ theories with non-free chiral rings are constructed. Those examples are in line with our geometric discussion being related to the free ring case by a finite quotient (gauging). Most of the discussion of the present paper applies to these more general situation as well (many arguments are formulated modulo finite quotients). The only point where we assume that the ring $\mathscr{R}$ is free\footnote{\ A part for the explicit examples of course.}, in order to simplify the analysis, is in showing that the Coulomb branch contains ``many'' normal rays. Our argument in the present form requires just \emph{one} normal ray, so the request of ``many'' of them is rather an overkill. It is relatively straightforward to extend the details of our analysis to finite quotients.   
\medskip

In the rest of this \textbf{Introduction}
we present a non-technical survey of our main results for both the list of allowed dimensions
and dimension $k$-tuples. In particular, \S.\,\ref{introoooo2} contains a heuristic derivation of the Universal Dimension Formula.

\subsection{Coulomb branch dimensions $\Delta$} 

In section 4 we present a very simple recursive algorithm to produce the list of  (putative) dimensions $\Delta$ allowed in a rank-$k$ CSG for all $k\in\mathbb{N}$. 
The dimension lists for ranks up to 13
are presented in 
 the tables of section 6. After the completion of this paper, ref.\!\cite{Argyres:2018zay} appeared on the arXiv  where the list of dimensions for $k=2$ is also computed.
Our results are in perfect agreement with theirs.\smallskip

Let us describe some general property of the set of dimensions in given rank $k$. The number of allowed $\Delta$'s is not greater than a certain Number-Theoretical function $\boldsymbol{N}(k)$ of the rank $k$ 
\be\label{kkkazqw12n}
\#\left\{\Delta\in \mathbb{Q}_{\geq1} \ \left|\begin{array}{l}\text{$\Delta\equiv$ dimension of a Coulomb branch}\\ \text{free generator in a CSG of rank}\leq k\end{array}\right.\right\}\leq \boldsymbol{N}(k).
\ee
There is evidence that $\leq$ may be actually replaced by an equality sign.

$\boldsymbol{N}(k)$ is a rather peculiar function: it is stationary for ``most'' $k\in\mathbb{N}$,
$\boldsymbol{N}(k)=\boldsymbol{N}(k-1)$
and, while the ratio 
\be
\nu(2k)=\big(\boldsymbol{N}(k)-\boldsymbol{N}(k-1)\big)/2k
\ee
 is ``typically'' a small integer, it takes all  integral value $\geq2$ infinitely many times. The first few values of $\boldsymbol{N}(k)$ are listed in table \ref{firstfew}. $\boldsymbol{N}(k)$ may be written as a Stieltjes integral of the function\footnote{\ For the precise definition of $\boldsymbol{N}(x)_\text{int}$ for \emph{real} positive $x$, see \S.\,\ref{anex}.} $\boldsymbol{N}(k)_\text{int}$ which counts the number of \emph{integral} dimensions $\Delta$ at rank $k$
 \be\label{veryfancy}
 \boldsymbol{N}(k)=2\int\limits_0^{k+\epsilon} x\;d\boldsymbol{N}\big(x\big)_\text{int}.
 \ee
 
 \begin{table}
\begin{tabular}{c|cccccccccccc}\hline\hline
$k$ & 1 &2 &3 &4 &5 &6 & 7& 8&9&10&11&12\\
$\boldsymbol{N}(k)$ & 8 & 24& 48& 88& 108& 180&180& 276&348&448&492&732\\\hline
$k$ &13&14&15&16&17&18&19&20&21&22&23&24\\
$\boldsymbol{N}(k)$&732&788&848&1072&1072&1360&1360&1720&1888&2020&2112&2640
\\\hline\hline
\end{tabular}
\caption{Values of the function $\boldsymbol{N}(k)$ for ranks $k\leq 25$ ($\boldsymbol{N}(25)=\boldsymbol{N}(24)$).}\label{firstfew}
\end{table}

From this expression we may easily read the large rank asymptotics\footnote{\ $\zeta(s)$ is the Riemann zeta-function.
} of $\boldsymbol{N}(k)$ 
\be
\boldsymbol{N}(k)= \frac{2\,\zeta(2)\,\zeta(3)}{\zeta(6)}\,k^2+o(k^2)\quad\text{as }k\to\infty.
\ee        
As mentioned above, we expect this bound to be optimal, that is, the actual number of Coulomb dimensions to have the above behavior for large $k$. 

The number $\boldsymbol{N}(k)$ is vastly smaller than the number of isoclasses of CSG of rank $\leq k$, showing that the coarse-grained classification is dramatically simpler than the fine one. The counting \eqref{kkkazqw12n} should be compared with the corresponding one for Lagrangian $\cn=2$ SCFTs
\be
\#\left\{\Delta\in \mathbb{Q}_{\geq1} \ \left|\begin{array}{l}\text{$\Delta\equiv$ dimension of a Coulomb branch free}\\ \text{generator in a \textit{Lagrangian model} of rank}\leq k\end{array}\right.\right\}= \left\lfloor \frac{3\,k}{2}\right\rfloor,\quad k\geq 15,
\ee
which confirms the idea that the Lagrangian dimensions have ``density zero'' in the set of all $\cn=2$ Coulomb dimensions. The Lagrangian dimensions are necessarily integers; the number of allowed \emph{integral} dimension at rank $k$ is (not greater than)
\be\label{asssyym}
\boldsymbol{N}(k)_\text{int}=\frac{2\,\zeta(2)\,\zeta(3)}{\zeta(6)}\,k+o(k)\quad\text{as }k\to\infty
\ee
so, for large $k$, roughly $38.5\%$ of all allowed \emph{integral} dimensions may be realized by a Lagrangian SCFT. Remarkably, for $k\geq 15$ the ratio 
\be
\boldsymbol{\varrho}(k)=\frac{\#\big\{\text{Lagrangian dimensions in rank $k$}\big\}}{\#\big\{\text{integral dimensions in rank $k$}\big\}}
\ee
is roughly independent of $k$ up to a few percent modulation, see e.g.\! table \ref{firstfew2}.

\subsection{Dimension $k$-tuples and Dirac quantization of charge}

The classification of the
dimension $k$-tuples $\{\Delta_1,\cdots,\Delta_k\}$ allowed in a rank-$k$ CSG contains much more information than the list of the individual dimensions $\Delta$. Indeed, the values of the dimensions of the various operators in a given SCFT  are strongly correlated. The list of dimension $k$-tuples may also be explicitly determined recursively in $k$ using our Universal Formula.

The problem may be addressed at two levels: there is a simple algorithm which produces, for a given $k$, a finite list of would-be dimension $k$-tuples. However there are subtle Number Theoretical aspects, and some of these $k$-tuples are consistent only under special circumstances.
The tricky point is as follows: a special geometry is, in particular, an analytic family of \emph{polarized} Abelian varieties. The polarization corresponds physically to the Dirac electro-magnetic pairing $\Omega$, which is an integral, non-degenerate, skew-symmetric form on the charge lattice. Usually one assumes this polarization to be \emph{principal}, that is, that all charges which are consistent with Dirac quantization are permitted. But physics allows $\Omega$ to be non-principal \cite{Donagi:1995cf} at the cost of introducing additional selection rules on the values of the electro-magnetic charges and fluxes (see \S.\,\ref{kkkaz1204} for details). The deep arguments of ref.\!\cite{Banks:2010zn} suggest that $\Omega$ should be principal for a $\cn=2$ QFT which emerges from a consistent quantum theory of gravity in some decoupling limit. 
It turns out that only a subset of the dimension $k$-tuples produced by the simple algorithm are consistent with a principal polarization; the others may be realized only in generalized special geometries endowed with suitable \emph{non-}principal polarization i.e.\! to be consistent they require additional selection rules on the electro-magnetic charges. Therefore one expects that such Coulomb dimensions would not appear in $\cn=2$ SCFT having a stringy construction. On the other hand, the Jacobian of a genus $g$ curve carries a canonical principal polarization; thus the special geometry of a SCFT with such ``non-principal'' Coulomb dimensions cannot be described by a Seiberg-Witten \emph{curve.}  
\smallskip

\begin{table}
\begin{center}
\begin{tabular}{c|cccccccccc}\hline\hline
$k$ & $\boldsymbol{1}$ &15 &16 &17 &18 &19 & 20& 21&22&23\\
$\boldsymbol{\varrho}(k)$ & $\boldsymbol{0.4}$ & 0.3859& 0.375& 0.3906& 0.375& 0.3888&0.3703& 0.3647&0.375&0.3777\\\hline
$k$ &24&25&26&27&28&29&30&31&32&$\boldsymbol{\infty}$\\
$\boldsymbol{\varrho}(k)$&0.3564&0.3663&0.3786&0.3809&0.3888&0.3909&0.3781&0.3865&0.3779&$\boldsymbol{0.3858}$
\\\hline\hline
\end{tabular}
\caption{Values of the ratio $\boldsymbol{\varrho}(k)$ (up to four digits)
for various values of the rank $k\in\mathbb{N}$, including $k=1$ and the asymptotic value for $k=\infty$. }\label{firstfew2}
\end{center}
\end{table}

To determine the dimension $k$-tuples which are compatible with a principal polarization is a subtle problem in Number Theory. For instance, the putative dimension list in rank 2 contains the two pairs\footnote{\ Note that all three numbers $12$, $8$, and $6$ are allowed as \emph{single} dimensions in rank 2 even if $\Omega$ is principal.} $\{12,6\}$
and $\{12,8\}$ (resp.\!\! the two pairs $\{10/7,8/7\}$ and $\{12/7,8/7\}$) but only the first one is consistent with a principal polarization. The pair  $\{12,6\}$ corresponds to rank $2$
Minahan-Nemeshansky (MN) of type $E_8$
\cite{ganor,Benini:2009gi} (resp.\! $\{10/7,8/7\}$ to Argyres-Douglas (AD) model of type $A_4$);
since this model has a stringy construction,
the Number Theoretic subtlety is consistent with the physical arguments of \cite{Banks:2010zn}.

 The reason why four of the putative rank-2 pairs $\{\Delta_1,\Delta_2\}$ are not consistent with a principal $\Omega$ looks rather exoteric at first sight: while the ideal class group of the number field $\mathbb{Q}[\zeta]$ 
($\zeta$ a primitive 12-th root of unity) is trivial, the \emph{narrow} ideal class group of its totally real subfield $\BQ[\sqrt{3}]$ is $\Z_2$, and the narrow class group is an obstruction to the consistency of such dimension pairs in presence of a principal polarization (a hint of why this group enters in the game will be given momentarily in \S.\,\ref{introoooo2}). To see which one of the two pairs $\{12,6\}$ or $\{12,8\}$ survives, we need to understand the action of the narrow ideal class group; it turns out that Class Field Theory properly selects the physically expected dimensions $\{12,6\}$. We regard this fact as a non-trivial check of our methods.

\begin{rem}\label{cft} Let us give a rough physical motivation for the role of Class Field Theory in our problem. It follows from  the subtle  interplay between the dynamical breaking of the SCFT  $U(1)_R$  symmetry\footnote{\ Properly speaking, what we call ``$U(1)_R$'' is the quotient group $U(1)_R/\langle (-1)^F\rangle$ acting effectively on the bosons.} in the supersymmetric vacua and the Dirac quantization of charge. Along the Coulomb branch $M$, the  $U(1)_R$ symmetry should be spontaneously broken.
But there are special holomorphic subspaces $M_n\subset M$ which parametrize  \textsc{susy} vacua where a discrete subgroup $\Z_n\subset U(1)_R$ remains unbroken. 
Assuming eqn.\eqref{wwwharing}, the locus
\be
\big\{u_i=0\ \text{for }i\neq i_0\big\}\subset M
\ee
 is such a subspace $M_n$ with $n$ the order of $1/\Delta_{i_0}$ in $\BQ/\Z$. 
To the locus $M_n$ one associates the rational group-algebra $\BQ[e^{2\pi iR/\Delta_{i_0}}]$ of the unbroken $R$-symmetry $\Z_n$. The chiral ring $\mathscr{R}$ is then a module of this group-algebra (of non-countable dimension). Replace $\mathscr{R}$ by the much simpler subring
$\mathscr{I}\subset\mathscr{R}$ of chiral operators of integral $U(1)_R$ charge; $\mathscr{I}=\mathscr{R}^\mathbb{M}$ where $\mathbb{M}$ is the quantum monodromy of the SCFT \cite{Cecotti:2010fi,Cecotti:2014zga}. Since $\mathsf{Proj}\,\mathscr{I}\cong \mathsf{Proj}\,\mathscr{R}$ \cite{gro} there is no essential loss of information in the process.  $\mathscr{I}$ is a $\C$-algebra; Dirac quantization is the statement that $\mathscr{I}$ is obtained from a $\BQ$-algebra $\mathscr{I}_\BQ$ by extension of scalars, $\mathscr{I}\cong\mathscr{I}_\BQ\otimes_\BQ \C$. An element of $\mathscr{I}_\BQ$ is simply a holomorphic function which locally restricts to an element of $\BQ[a^i,b_j]$, where $(a^i,b_j)$ are the periods of special geometry well-defined modulo $Sp(2k,\Z)$. 
E.g.\! if our model is a Lagrangian SCFT with gauge group $G$ of rank $k$, $\mathscr{I}_\BQ=\BQ[a^i,b_j]^{\text{Weyl}(G)}$.
$\mathscr{I}_\BQ$ is a module of $\BQ[e^{2\pi iR/\Delta_{i_0}}]$ of just countable dimension.
By Maschke theorem \cite{maschke} $\mathscr{I}_\BQ$ is a countable sum of Abelian number fields 
\be\label{rrrqas}
\mathscr{I}_\BQ= \bigoplus_\alpha \mathbb{F}_\alpha.
\ee
Being Abelian, the fields $\mathbb{F}_\alpha$ are best studied by the methods of Class Field Theory.
On the other hand, the Coulomb dimensions $\Delta(\phi)$ are just the characters of $U(1)_R$ appearing in $\mathscr{R}$
\be
\chi_\phi\colon e^{2\pi i t R}\mapsto e^{2\pi i t \Delta(\phi)}\in \C^\times,\qquad \text{for }\phi\in\mathscr{R}\ \text{of definite dimension}.
\ee
Focusing on  the subspace  $M_n\subset M$,
the characters  $\{\chi_\phi\}$ induce characters of the unbroken subgroup $\Z_n$.
 Hence the Coulomb dimensions $\Delta(\phi)$ may be read from the decomposition of $\mathscr{R}$ into characters of $\Z_n$. If all Coulomb dimensions are integral, $\mathscr{R}=\mathscr{I}_\BQ\otimes_\BQ\C$, and the last decomposition is obtained from the one in
 \eqref{rrrqas} by tensoring with $\C$, so that we may read $\Delta(\phi)$ directly from the Number Theoretic properties of the $\mathbb{F}_\alpha$. The same holds in the general case, \textit{mutatis mutandis.} 
In the main body of the paper we shall deduce the list of allowed Coulomb dimensions by a detailed geometric analysis, but the final answer is already given by eqn.\eqref{rrrqas} when supplemented with  the obvious relation between the rank $k$ of the SCFT and the degrees of the number fields $\mathbb{F}_\alpha$.
\end{rem}

\subsection{Rank 1 and natural guesses for $k\geq 2$}\label{introoooo2}

The case of rank one is well known \cite{ACSW}. The allowed Coulomb dimensions are
\be\label{rrran11}
\Delta= 1,\ 2,\ 3,\ 3/2,\ 4,\ 4/3,\ 6,\ 6/5,
\ee
$\Delta=1$ corresponds to the free (Lagrangian) theory, $\Delta=2$ to interacting Lagrangian models (i.e.\! $SU(2)$ gauge theories), and all other dimensions to strongly interacting SCFTs.
A crucial observation is that the list of  dimensions \eqref{rrran11} is organized into orbits of an Abelian group
$H_\R$. For $k=1$ the group is simply $H_\R\cong\Z_2$ generated by the involution $\iota$
\be\label{hhhrrqb}
\iota\colon \Delta\mapsto \Delta^\prime= \frac{1}{\langle 1-\Delta^{-1}\rangle},\qquad
 \left|\;\begin{minipage}{170pt}
\begin{footnotesize}
where, for $x\in\R$, $\langle x\rangle$  denotes the real\\ number equal $x$ mod 1 with $0<x\leq 1$.
\end{footnotesize}
\end{minipage}\right.
\ee
\be\label{kzaa994}
\text{The Lagrangian models correspond to the fixed points of $H_\R$, $\Delta=1,2$.}
\ee
 
There are dozens of ways to prove that eqn.\eqref{rrran11}  is the correct set of dimensions for $k=1$ $\cn=2$ SCFT; each argument leads to its own interpretation\footnote{\ To mention just a few: the set of $\Z_2$-orbits in eqn.\eqref{rrran11} is  one-to-one correspondence with:  \textit{(i)} Coxeter labels of the unique node of valency $>2$ in an affine Dynkin graph which is also a star; \textit{(ii)} Coxeter numbers of semi-simple rank-2 Lie algebras;
\textit{(iii)} degrees of elliptic curves written as complete intersections in WPS, \textit{(iv)} and so on.}  of this remarkable list of rational numbers and of the group $H_\R$. Each interpretation suggests a possible strategy to generalize the list \eqref{rrran11} to higher $k$.
We resist the temptation to focus on the most elegant viewpoints, and stick ourselves to the most obvious  interpretation of the set 
\eqref{rrran11}:
\begin{fact} The allowed values of the Coulomb dimension $\Delta$ for rank $1$ $\cn=2$ SCFTs, eqn.\eqref{rrran11},  are in one-to-one correspondence with the elliptic conjugacy classes in the rank-one duality-frame group, $Sp(2,\Z)\equiv SL(2,\Z)$. Lagrangian models correspond to central elements (which coincide with their class). The group $H_\R\cong GL(2,\Z)/SL(2,\Z)$ permutes the distinct $SL(2,\Z)$-conjugacy classes which are conjugate in the bigger group $GL(2,\Z)$.
\end{fact}
By an \emph{elliptic} conjugacy class we mean a conjugacy class whose elements have finite order. There are several ways to check that the above  \textbf{Fact} is true. The standard method is comparison with the Kodaira classification of exceptional fibers in elliptic surfaces \cite{koda}. Through the homological invariant \cite{koda}, Kodaira sets the (multiplicity 1) exceptional fibers in one-to-one correspondence with the quasi-unipotent conjugacy classes of $SL(2,\Z)$. In dimension 1 the homological invariant of a CSG must be semi-simple.
Since quasi-unipotency and  semi-simplicity together imply finite order, \textbf{Fact} follows. The trivial conjugacy class of 1 corresponds to the free SCFT, the class of the central element $-1$ to $SU(2)$ gauge theories, and the regular elliptic classes to strongly-coupled models with no Lagrangian formulation. The map between (conjugacy classes of) elliptic elements of $SL(2,\Z)$ and Coulomb branch dimensions $\Delta$ is through their modular factor $(c\tau+d)$ evaluated at their fixed point\footnote{\ The locus of fixed points in $\mathfrak{H}$ of an elliptic element of the modular group $SL(2,\Z)$ is not empty and connected, see \textbf{Lemma \ref{kkkazzz1v}}. Note that $(c\tau+b)$, being a root of unity, is independent of the chosen $\tau$ in the fixed locus.} $\tau$ in the upper half-plane $\mathfrak{h}$. Explicitly:
\be\label{kkkaqw}
\begin{bmatrix} a & b\\ c& d\end{bmatrix}\in SL(2,\Z)\ \text{elliptic} \longrightarrow 
\Delta=\frac{2\pi i}{\boldsymbol{\log}(c\tau+d)},\quad 
\left|\begin{array}{c}\text{\begin{footnotesize}where $\tau$ is a solution to\end{footnotesize}}\hfill\\
a\tau+b=\tau(c\tau+d),\end{array}\right.
\ee 
and $\boldsymbol{\log}\, z$ is the branch of the logarithm such that $\boldsymbol{\log}(e^{2\pi i x})=2\pi i \langle x\rangle$ for $x\in\R$ (cfr.\!
eqn.\eqref{hhhrrqb} for the notation).
The action of $\iota\in H_\R$ (eqn.\eqref{hhhrrqb})
is equivalent to\footnote{\ $\bar \tau$ is in the lower half-plane; to write everything in the canonical form, one should conjugate it to a point in the upper half-plane by acting with the proper orientation-reversing element of $GL(2,\Z)$.} 
$\tau\leftrightarrow \bar\tau$, i.e.\!
$\boldsymbol{\log}(c\tau+d)/2\pi i\leftrightarrow \langle 1 -\boldsymbol{\log}(c\tau+d)/2\pi i\rangle$.\medskip

The basic goal of the coarse-grained classification of $\cn=2$ SCFT is to provide the correct generalization of the above \textbf{Fact} to arbitrary rank $k$. The natural guess is to replace the rank-one duality group $SL(2,\Z)$ by its rank-$k$ counterpart, i.e.\! the Siegel modular group $Sp(2k,\Z)$, and consider its finite-order conjugacy classes. Now the fixed-point modular factor, $C\boldsymbol{\tau}+D$, is a $k\times k$ unitary matrix with eigenvalues $\lambda_i$, ($i=1,\dots,k$ and $|\lambda_i|=1$), to which we may tentatively associate the $k$-tuple $\{\Delta_i\}_{i=1}^k$ of  would-be Coulomb dimensions
\be\label{kkkazqw2}
\Delta_i = 1+\frac{2\pi i-\boldsymbol{\log}\,\lambda_i}{\boldsymbol{\log}\,\lambda_1}\in\BQ_{\geq1} \quad\text{for }i=1,2,\dots, k,
\ee
giving a putative 1-to-$k$ correspondence\footnote{\ Since we have a $k$-fold choice of which eigenvalue we wish to call $\lambda_1$.} between $Sp(2k,\R)$-conjugacy classes of elliptic elements in the Siegel group $Sp(2k,\Z)$ and would-be dimension $k$-tuples.
The candidate correspondence \eqref{kkkazqw2} reduces to the Kodaira one for $k=1$, and is consistent with the physical intuition of \textbf{Remark \ref{cft}}. 

It turns out that for $k\geq 2$ the guess \eqref{kkkazqw2}  is morally correct, but there are many new phenomena and subtleties with no counterpart in rank 1, so the statement of the correspondence should be taken with \emph{a grain of salt,} and supplemented with the appropriate limitations and specifications, as we shall do in the main body of the present paper. In particular, the same $k$-tuple $\{\Delta_i\}_{i=1}^k$ is produced by a number $\leq k$ of distinct
conjugacy classes; in facts, the geometrically consistent $k$-tuples are those which appear precisely $k$ times (properly counted).    

In rank $k\geq 2$ the notion of ``duality-frame group'' is subtle. The Siegel modular group $Sp(2k,\Z)$ is
the arithmetic group preserving the
principal polarization. If $\Omega$ is not principal, $Sp(2k,\Z)$ should be replaced by the arithmetic group $S(\Omega)_\Z$ which preserves it
\be\label{nonprrir}
S(\Omega)_\Z=\Big\{ m\in GL(2k,\Z) \colon m^t \Omega m =\Omega\Big\}.
\ee
 $Sp(2k,\Z)$ and $S(\Omega)_\Z$ are commensurable arithmetic subgroups of $Sp(2k,\BQ)$ \cite{J3}.
If a SCFT has a non-principal polarization $\Omega$, 
its Coulomb dimensions are related to the elliptic conjugacy classes in $S(\Omega)_\Z$ which (in general) lead to different eigenvalues $\lambda_i$ and Coulomb dimensions $\Delta_i$.

As in the $k=1$ case, the dimension $k$-tuples $\{\Delta_i\}_{i=1}^k$ form orbits under a group. The most naive guess is that this is the ``automorphism group'' of eqn.\eqref{kkkazqw2},
$\Z_k\rtimes \Z_2^k$, where the first factor cyclically permutes the $\lambda_i$ while $\Z_2^k$ is the  straightforward generalization of
$\iota$ for $k=1$ : 
\be\label{kkkxcczaq}
\iota_j\colon \boldsymbol{\log}\,\lambda_i /2\pi i\longmapsto \begin{cases}
\langle 1-\boldsymbol{\log}\,\lambda_j /2\pi i\rangle&\text{for }i=j\\
\boldsymbol{\log}\,\lambda_i/2\pi i &\text{otherwise,}
\end{cases}\qquad j=1,2,\dots,k.
\ee
However, in general, this action would not map classes in $Sp(2k,\Z)$ to classes in the same group but rather in some other arithmetic group $S(\Omega)_\Z$. The proper generalization of the $k=1$ case requires to replace\footnote{\ $\Z_2^k$ is the group which permutes the $Sp(2k,\R)$-conjugacy classes of elliptic elements of the \emph{real} group  $Sp(2k,\R)$ which are conjugate in $GL(2k,\C)$. However some real conjugacy class may have no integral element, and only a subgroup survives over $\Z$. This implies that $H_\R$ is indeed a subgroup of $\Z_2^k$.} $\Z_2^k$ by the Abelian group $H_\R$ which permutes the $Sp(2k,\R)$-conjugacy classes of elliptic elements of the Siegel modular group $Sp(2k,\Z)$ which are conjugate in $GL(2k,\C)$. 
$H_\R$ is a subgroup of the group \eqref{kkkxcczaq}, i.e.\! we have an exact sequence
\be
1\to H_\R\to \Z_2^k\to C\to1
\ee
for some 2-group $C$. In the simple case when the splitting field $\BK$ of the elliptic element has class number 1, $C$ is just the narrow class group $C^\text{nar}_\Bbbk$ of its maximal totally real subfield $\Bbbk\subset\BK$. 
For instance, the dimension pair  $\{12,6\}$ (the $k=2$ $E_8$ MN model mentioned before) is reproduced by eqn.\eqref{kkkazqw2}
for $\boldsymbol{\log}\,\lambda_1=2\pi i/12$ and $\boldsymbol{\log}\,\lambda_2=14\pi i/12$; applying naively eqn.\eqref{kkkxcczaq} we would get the dimension pair $\iota_2\{12,6\}=\{12,8\}$. However in this case $\iota_2\not\in H_\R$, 
and the $H_\R$-orbit of $\{12,6\}$ does not contain $\{12,8\}$ which then is not a valid dimension pair for $k=2$ for the duality-frame group $Sp(2k,\Z)$ but it is admissible if the duality-frame group is $S(\Omega)_\Z$ with $\det\Omega=2$ or larger.

\begin{rem} The generalization to higher $k$ of the $k=1$ criterion \eqref{kzaa994} for the Coulomb dimensions to be consistent with a weakly-coupled Lagrangian description is as follows. Let $\iota\in H_\R$ be given by
\be
\iota=\iota_1\,\iota_2\,\cdots\,\iota_k.
\ee
If a dimension $k$-tuple $\{\Delta_1,\cdots,\Delta_k\}$ may be realized by a Lagrangian SCFT then it is left invariant by $\iota$ up to a permutation of the $\Delta_i$. The inverse implication is probably false.
\end{rem}

\subsection{Springer Theory of reflection groups}\label{STST}

 The proposed dimension formula \eqref{kkkazqw2} may look puzzling at first.
Being purportedly universal, it should, in particular, reproduce the correct dimensions for a weakly-coupled Lagrangian SCFT with gauge group an arbitrary semi-simple Lie group $G$. By the non-renormalization and Harrish-Chandra theorems \cite{HC}, in the Lagrangian case 
the dimension $k$-tuple $\{\Delta_i\}$
is just the set $\{d_i\}$  of the degrees of the Casimirs of $G$
(its exponents $+1$). Thus eqn.\eqref{kkkazqw2}, if correct, implies a strange universal formula for the degrees of a Lie algebra which looks rather counter-intuitive from the Lie theory viewpoint.

The statement that \eqref{kkkazqw2} is the correct degree formula (not just for Weyl groups of Lie algebras, but for all finite  reflection groups) is the main theorem in the Springer Theory of reflection groups 
\cite{springer,cohen,springertheory},
We shall see in \S.\,\ref{spirpri} that 
the correspondence between our geometric analysis of the CSG and Springer Theory is more detailed than just giving the right dimensions.
In particular, Springer Theory together with weak-coupling QFT force us to use in eqn.\eqref{kkkazqw2} the 
universal determination of the logarithm we call
$\boldsymbol{\log}$ (see after eqn.\eqref{kkkaqw}), which is therefore implied by conventional Lagrangian QFT.

In other words, the proposed Universal  Dimension Formula \eqref{kkkazqw2} may also be obtained using the following
\begin{stra}
Write the usual dimension formula valid for \underline{all} weakly-coupled Lagrangian SCFTs in a clever way, so that it continues to make sense even for non-Lagrangian SCFT, i.e.\! using only intrinsic physical data such as the breaking pattern of $U(1)_R$. This leads you to eqn.\eqref{kkkazqw2}. Then claim the formula to have general validity. 
\end{stra}

This is the third heuristic derivation of \eqref{kkkazqw2} after the ones in \S.1.2 and 1.3. The sheaf-theoretic arguments of \S.\,\ref{steintube} will make happy the pedantic reader (at least we hope). It will also supplement \eqref{kkkazqw2} all the required  specifications and limitations.

\begin{rem} 
Inverting the argument, we may say that
our analysis of the CSG yields
a (simpler) transcendental proof of the classical Springer results. 
\end{rem}

\subsection{Organization of the paper}

The rest of the paper is organized as follows.
Section 2 contains a review of special geometry, structures on Riemannian cones, and all the basic geometric tools we need.
The only new materials in this section are the implications for special geometry  of the sphere theorems of comparison geometry  and the relation of CSG with the theory of log-Fano varieties. In section 3 we discuss the Coulomb chiral ring and the class of CSG with constant period map. Section 4 is the core of the paper, where we deduce the dimension formulae. Section 5 is quite technical: here we discuss fine points about  the elliptic conjugacy classes in Siegel modular groups and their non-principal counterparts $S(\Omega)_\Z$. Section 6 contains a few sample dimension tables. Some technical material is confined in the appendix. 

\begin{rem}
The various sections of the paper are written with quite different standards of mathematical rigour. The core of the paper --- \S.\,\ref{steintube}
where the dimension formula is deduced --- is (as far as we can see) totally rigorous once we take for granted that the chiral ring $\mathscr{R}$ is a (graded) free polynomial ring. The dimension formula then follow as a simple application of the Oka principle.
\end{rem}

\section{Special cones and \emph{log}-Fano varieties}

In this section we review special geometry and related topics to set up the scene. The first three subsections contain fairly standard material; our suggestion to the experts is to skip them (except for the \emph{disclaimer} in \S.\ref{kkkaz1204}).
Later subsections describe basic properties of conic special geometries (CSG) which were not previously discussed in the literature: we aim to extablish that a CSG is an affine (complex) cone over a special kind of normal projective manifold: a simply-connected log-Fano with minimal Hodge numbers.

\subsection{Special geometric structures}

In this paper by a ``special geometry''
we mean a holomorphic integrable system with a Seiberg-Witten (SW) meromorphic differential \cite{Donagi:1995cf,Donagi:1997sr,Freed:1997dp}:

\begin{defn}\label{jzcxv} By a \textit{special geometry}
we mean the following data:
\begin{description}
\item[D1:] A holomorphic map $\pi\colon X\to M$  between two normal complex analytic manifolds, $X$ and $M$, of complex dimension $2k$ and $k$, respectively,
whose generic fiber is (analytically isomorphic to) a principally polarized Abelian variety. $\pi$ is required to have a zero-section.
The closed analytic set $\cd\subset M$ at which the fiber degenerates is called the \emph{discriminant}. The dense open set $M^\sharp\equiv M\setminus \cd$ is called the \emph{regular locus};
\item[D2:] A meromorphic $1$-form (1-current) $\lambda$ on $X$  (the Seiberg-Witten (SW) differential) such that $d\lambda$ is
a holomorphic symplectic form on $X$,
with respect to which the fibers of $\pi$ are Lagrangian.  
\end{description}
\end{defn}

\subsubsection{Three crucial \textit{caveats} on the definition}\label{kkkaz1204}

 The one given above is the definition which is natural from a geometric perspective. However in the physical applications one also considers slightly more general situations which may easily be reduced to the previous one.
This aspect should be kept in mind when making comparison of our findings with existing results in the physics literature. We stress three aspects:

\subparagraph{Multivalued symplectic forms.}
In \textbf{Definition \ref{jzcxv}}  $X$ is \emph{globally} a holomorphic integrable system with a well-defined holomorphic symplectic form $d\lambda$. Since the overall phase of the SW differential $\lambda$ is not observable, in the physical applications sometimes one also admits   geometries in which $\lambda$ is well-defined only up to (a locally constant) phase, see \cite{ACSW} for discussion and examples.  
Let $C$ be the Coulomb branch of such a generalized special geometry; there is an unbranched cover of the regular locus, $M^\sharp\to C^\sharp$, on which $\lambda$ is univalued. $M^\sharp$ is the regular locus of a special geometry in the sense of \textbf{Definition \ref{jzcxv}}. The cover branches only over the discriminant $\cd$. Dually, we have an embedding of chiral rings $\mathscr{R}_C\hookrightarrow \mathscr{R}_M$. Away from the discriminant, there is little difference between the two descriptions: working in $C$ we identify vacua having the same physics, while in $M$ we declare them to be distinct states (with the same physical properties). In the first picture we consider non-observable the chiral operators which distinguish the physically equivalent Coulomb vacua, that is, $\mathscr{R}_C\cong \mathscr{R}_M^G$, where $G$ is the (finite) deck group of the covering. $\mathscr{R}_C$ is still free iff $G$ is a reflection group
\cite{ST1,ST2} acting homogeneously; in this case, the dimensions of its generators are multiples of the ones for $\mathscr{R}_M$. 

The two special geometries $C$ and $M$ may lead to different ways of resolving the singularities along $\cd$, and hence they may correspond to physical inequivalent theories in the ``same'' coarse class. It may happen that we may attach physical sense only to the chiral sub-ring $\mathscr{R}_C$. 
To compare our results with those of papers which allow multivalued $\lambda$, one should first pull-back their geometries to a cover on which the holomorphic symplectic form is univalued. 

\subparagraph{Non-principal polarizations.}
 
In \textbf{Definition \ref{jzcxv}} the generic fiber of $X\to M$ is taken to be a \emph{principally} polarized Abelian variety.
As already stressed in the Introduction,
 we may consider non principal-polarization. 
This means that not all electric/magnetic charges and fluxes consistent with Dirac quantization
are present in the system
\cite{Donagi:1995cf}. This is believed not to be possible in theories arising as limits of consistent quantum theories containing gravity \cite{Banks:2010zn}.
Every non-principally polarized Abelian variety has an isogenous principally polarized one \cite{GH,complexabelian}. 

We see the polarization of the regular fiber $X_u$ as a primitive,\footnote{\ That is, the matrix of the form $\Omega_{ij}\in \Z(2k)$ satisfies $\gcd_{i,j}\{\Omega_{ij}\}=1$.} integral, non-degenerate, pairing  \cite{igusa,complexabelian}
\be
\langle-,-\rangle\colon H_1(X_u,\Z)\times H_1(X_u,\Z)\to\Z
\ee
which has the physical interpretation of the Dirac electro-magnetic. We may find generators $\gamma_i$ of the electro-magnetic charge lattice $H_1(X_u,\Z)$ so that the matrix $\Omega_{ij}\equiv \langle \gamma_i,\gamma_j\rangle$ takes the (unique) canonical form \cite{newman} 
\be\label{chargemult}
\Omega=\begin{bmatrix} 0 & e_1\\
-e_1 & 0\end{bmatrix}\bigoplus 
\begin{bmatrix} 0 & e_2\\
-e_2 & 0\end{bmatrix}\bigoplus\cdots\bigoplus \begin{bmatrix} 0 & e_k\\
-e_k & 0\end{bmatrix},\quad\  e_i\in\mathbb{N},\ e_i\mid e_{i+1},\ \ e_1\equiv 1.
\ee
The polarization is principal iff $e_i=1$, i.e.\!
$\det\Omega=1$. Physically, the integers $e_i$ are \emph{charge multipliers}: (in a suitable duality frame) the allowed values of the $i$-th electric charge are integral multiples of $e_i$.

If $\Omega$ is principal, the duality-frame group is the Siegel modular group $Sp(2k,\Z)$, while in general it is the commensurable arithmetic group $S(\Omega)_\Z$, eqn.\eqref{nonprrir}.  Since, as mentioned in the Introduction, the Coulomb dimensions $\{\Delta_i\}$ are related to the possible elliptic subgroups of the duality-frame group, there is a correlation between the 
set of charge multipliers $\{e_i\}$ and the set of dimensions $\{\Delta_i\}$. The simplest instance of this state of affairs appears in rank 2: the set of dimensions $\{12,8\}$ is not allowed for $\Omega$ principal, but it is permitted when the charge multiplier $e_2$ is (e.g.)  2 or $3$.

\subparagraph{Non-normal Coulomb branches and ``non-free'' chiral rings.} In \textbf{Definition \ref{jzcxv}} the Coulomb branch $M$ is taken to be \emph{normal} as an analytic space, that is,
we see the Coulomb branch as a ringed space
$(M,\co_M)$ where $M$ is a Hausdorf topological space and $\co_M$ is the structure sheaf whose local sections are the local holomorphic functions. Being \emph{normal} means that the stalks
$\co_{M,\,x}$ at all points $x\in M$ are domains which are integrally closed in the  stalk $\cm_x$ of the sheaf of germs of meromorphic functions \cite{nor1,nor2}. Geometrically this is the convenient and natural definition; indeed, there is no essential loss of generality since we may always replace a non-normal analytic space $M_0$ by its normalization $M$: just replace the structure sheaf $\co_{M_0}$ with its integral closure $\co_M$ and the topological space $M_0$ by the analytic spectrum $M$ of $\co_M$ \cite{nor2}. Roughly speaking, passing to the normalization just enlarges the ring of the holomorphic functions from global sections of $\co_{M_0}$ to global sections of $\co_M$. In facts, the normalization corresponds to the maximal extension of the ring of local holomorphic functions compatible with $\co_{M_0}$-coherence.
Thus, geometrically, a non-normal Coulomb branch just amounts  to ``forget'' some (local) holomorphic function. The simplest example of a non-normal analytic space is the plane cuspidal cubic
whose ring of regular functions is
$\C[u_1,u_2]/(u_1^2-u_2^3)$. Its normalization is the affine line with ring $\C[t]$, corresponding to the parametrization $u_1=t^3$, $u_2=t^2$. In this example the normalization ring $\Gamma(M,\co_M)$ has the (topological) basis $1,t,t^2,t^3,\cdots$
while the basis of the non-normal version, $\Gamma(M_0,\co_{M_0})$, is
$1,t^2,t^3,\cdots$ where one ``forgets'' the function $t$. 

From the physical side the situation is subtler. We define the \textit{(geometric) chiral ring} $\mathscr{R}$ to be the Frech\'et ring of the global holomorphic functions
$\mathscr{R}\equiv \Gamma(M,\co_M)$. 
This geometric ring may or may not coincide with the \emph{physical} chiral ring $\mathscr{R}_\text{ph}$, defined as the ring of holomorphic functions on $M^\sharp\subset M$ which may be realized as vacuum expectation values of a chiral operator. 
Clearly $\mathscr{R}_\text{ph}\subset \mathscr{R}$, and we get the physical ring by ``forgetting'' some holomorphic function. Then $\mathscr{R}_\text{ph}=\Gamma(M_\text{ph},\co_\text{ph})$ where the stalks of $\co_\text{ph}$ are domains\footnote{\ Because the Coulomb branch is reduced.} which may or may not be integrally closed.
In the second case the physics endows the Coulomb branch with the structure of a \textit{non-normal analytic space} $(M_\text{ph},\co_\text{ph})$. Geometrically it is natural to replace it with its normalization $(M,\co_M)$ while proclaiming that only a subring $\mathscr{R}_\text{ph}$ of the chiral ring $\mathscr{R}$ is a ring of physical operators. Notice that the full geometric ring $\mathscr{R}$ may be a free polynomial ring, $\C[u_1,\cdots,u_k]$, while the physical ring
$\mathscr{R}_\text{ph}$ is a non-free finitely-generated ring, as the example of the cuspidal cubic shows. 

The putative ``non-free'' Coulomb branch geometries of ref.\!\!\cite{Argyres:2017tmj} arise this way: they are non-normal analytic spaces whose normalization has a free polynomial ring of regular functions, $\C[u_1,\cdots,u_k]$. That is, the ``non-free'' chiral rings are obtained from free geometric rings by forgetting some holomorphic functions of $\mathscr{R}$. The physical rationale for ``forgetting'' functions is the unitarity bound.
In a CSG $\mathscr{R}$ is graded by the conformal dimension $\Delta$, and unitarity requires that a non-constant \emph{physical} holomorphic function has $\Delta\geq1$. Hence one is naturally led to the proposal
\be
\mathscr{R}_\text{ph}=\mathscr{R}_{\Delta\geq1}\equiv \C\cdot\boldsymbol{1}\oplus\Big\{\phi\in\mathscr{R}\ :\ \Delta(\phi)\geq 1\Big\}\subset \mathscr{R}.
\ee 
If $0<\Delta(\phi)<1$ for some $\phi\in\mathscr{R}$,  $\mathscr{R}_\text{ph}$ defines a non-normal structure sheaf $\co_\text{ph}$ and the physical ring is non-free.
Is this fancy possibility actually realized?

The equations determining the dimensions $\Delta_i$ for the normalization $M$ of a CSG (satisfying our regularity conditions), deduced in \S.\,\ref{steintube} below, always have a (unique) solution such that $\Delta(\phi)\geq 1$ for all $\phi\in\mathscr{R}$, $\phi\neq\boldsymbol{1}$, with equality precisely when $\phi$ is a free field. 
Indeed, with the $\boldsymbol{\log}$ determination of the logarithm, the formula \eqref{kkkazqw2} expresses $\Delta_i$ as $1$ plus a manifestly non-negative 	quantity.
To produce $\mathscr{R}_\text{ph}\neq \mathscr{R}$, we may try to replace $\boldsymbol{\log}$ by some  bizzarre branch of the logarithm, with the effect that $\Delta(\phi)\to \Delta(\phi)_\text{new}= 2-\Delta(\phi)$, so that an element with $1<\Delta(\phi)<2$ would be reinterpreted as having the dimension $0<\Delta(\phi)_\text{new} <1$. However this is extremely unnatural and gruesome since it will spoil the universality of the prescription to compute the dimension $\Delta$ that better be the same one for all SCFT and all chiral operators (the correct prescription should be the unique one which reproduces the correct results for Lagrangian QFT, see \S.\ref{STST}. 

Assuming universality, the fancy possibilities of ref.\!\!\cite{Argyres:2017tmj} cannot be realized, and we shall neglect them for the rest of this paper. If the reader is aware of  physical motivations for their existence and wants to study them, he needs only to perform the non-universal analytic continuation of the relevant formulae. 
\bigskip

For most of the paper we focus on special geometries in the sense of \textbf{Definition \ref{jzcxv}}, with $\mathscr{R}_\text{ph}=\mathscr{R}$ and
 principally polarized fibers. Occasionally we comment on the modifications required for non-principal $\Omega$.

\subsubsection{Review of implied structures}

The data \textbf{D1}, \textbf{D2} imply the existence of several canonical geometric structures. We recall just the very basic ones (many others may be obtained by the construction in \textbf{S5}):
\begin{description}
\item[S1:] \textit{(polarized local system)}
A local constant sheaf $\Gamma$ on $M^\sharp$ with stalk $\cong\Z^{2k}$ equipped with a skew-symmetric form $\langle-,-\rangle\colon\Gamma\times \Gamma\to \Z$ under which $\Gamma\simeq \Gamma^\vee$. $\Gamma$ is given by the holomogy of the fiber $\Gamma_u= H_1(\pi^{-1}(u),\Z)$ with  the intersection form given by the principal polarization;
\item[S2:] \textit{(flat Gauss-Mannin connection)} On the holomorphic bundle $\ce=\Gamma^\vee\otimes \co_{M}$ over $M^\sharp$, we have the flat holomorphic connection $\nabla^\text{GM}$ defined by the condition that the local sections of $\Gamma^\vee$ are holomorphic;
\item[S3:] \textit{(monodromy representation)} $m\colon \pi_1(M^\sharp)\to Sp(2k,\Z)$; 
\item[S4:] \textit{(Hodge bundle)} $\cv\to M^\sharp$: it is the holomorphic sub-bundle
of $\ce$ whose fibers are (1,0) cohomology classes, i.e.\! $\cv_u=H^0(\pi^{-1}(u),\Omega^1)$. The flat connection
$\nabla^\text{GM}$ of $\ce$ induces the sub-bundle (holomorphic) connection
$\nabla^H$ on $\cv$ \cite{grif,perbook}. Note that $\Gamma^\vee$ acts by translation on $\cv$ and that $\cv/\Gamma^\vee\cong X^\sharp\equiv \pi^{-1}(M^\sharp)$;
\item[S5:] \textit{(period map and the family of homogeneous bundles)} the period matrix $\boldsymbol{\tau}_{ij}$ of the Abelian variety $\pi^{-1}(x)$ is a complex symmetric matrix with positive imaginary part well defined up to $Sp(2k,\Z)$ equivalence; hence $\boldsymbol{\tau}$ defines the holomorphic map:
\be\label{mumu}
\boldsymbol{\tau}\colon M^\sharp\to Sp(2k,\Z)\big\backslash Sp(2k,\R)\big/U(k).
\ee
The period map $\boldsymbol{\tau}$ yields a universal construction of  many other canonical geometrical objects on $M^\sharp$. We limit ourselves to a special class of holomorphic ones. The Griffiths period domain $Sp(2k,\R)/U(k)$ is an open domain in its complex Griffiths compact dual \cite{grif,grifpaper,perbook}
\be\label{uuuuuszx}
Sp(2k,\R)/U(k)\subset Sp(2k,\C)/P(k),\qquad \begin{minipage}{130pt}
\begin{footnotesize}
where $P(k)\subset Sp(2k,\C)$ is the Siegel parabolic subgroup.
\end{footnotesize}
\end{minipage}
\ee
By general theory, to every $P(k)$-module (in particular to all $U(k)$-modules) we associate a holomorphic vector bundle over the compact dual equipped with a unique metric, complex structure, and connection having an explicit Lie theoretic construction \cite{perbook}. These bundles, metrics, and connections may be restricted to the period domain and then pulled back to $M^\sharp$ \emph{via} $\boldsymbol{\tau}$ to get God-given bundles, metrics, and connections on $M^\sharp$.
All the quantities of ``special geometry'' (including the K\"ahler metric)
arise in this way from Lie theoretic gadgets. For instance, the Hodge bundle $\cv$ (resp.\! the flat bundle $\ce$) is the pull back of a homogenous bundle, and the connections $\nabla^H$ and $\nabla^\text{GM}$ are the pull-back of the corresponding canonical connections on the symmetric space \eqref{uuuuuszx};  
\item[S6:] \textit{(periods and local special coordinates)} Let $U\subset M^\sharp$ be simply connected. We trivialize $\Gamma$ in $U$ choosing local sections making a canonical symplectic basis $(A^i, B_j)$, $i,j=1,\dots, k$, 
\be
\langle A^i, A^j\rangle=\langle B_i,B_j\rangle=0,\quad \langle A^i, B_j\rangle= {\delta^i}_j.
\ee
The local special (holomorphic) coordinates $a^i$ and their duals $b^i$ are (in $U$)
\be
a^i=\langle A^i, \lambda\rangle,\qquad
b_i= \langle B_i,\lambda\rangle.
\ee
Writing 
\be
(\cv|_U)_\text{smooth}\equiv (\Gamma^\vee\otimes \R)|_U= U\times\Big(A^i x_i+B_jy^j\Big),\qquad (y^j,x_i)\in\R^{2k},
\ee
the holomorphic symplectic form becomes
\be
\sigma=d\lambda= da^i\wedge dx_i + db_i\wedge dy^i=da^i\wedge\left(dx_i+\frac{\partial b_j}{\partial a^i}\,dy^j\right).\ee
Since the holomorphic coordinates along the fiber are $z_i=x_i+\tau_{ij}y^j$
we get
\be
\boldsymbol{\tau}_{ij}= \frac{\partial b_i}{\partial a^j}.
\ee
Since $\boldsymbol{\tau}_{ij}$ is symmetric, locally there exists a prepotential (holomorphic) function $\cf(a^j)$ such that
\be
b_i=\frac{\partial\cf(a^j)}{\partial a^i},\qquad \boldsymbol{\tau}_{ij}(a)=\frac{\partial^2\cf}{\partial a^i\partial a^j}.
\ee 
\item[S7:] \textit{(the dual bundle $\cv^\vee\simeq \ce/\cv$)} This is yet another bundle whose metric and connection is given by the general construction in \textbf{S5}. It coincides with $TM^\sharp$, so it yields the geometry of the base.  On the intersection of two special coordinate charts $U$, $U^\prime$ we have:
\be
\begin{bmatrix} b^\prime \\ a^\prime
\end{bmatrix}=\begin{bmatrix}A  & B\\
C & D\end{bmatrix}\begin{bmatrix} b \\ a
\end{bmatrix},\quad
\begin{bmatrix} y^\prime \\ x^\prime
\end{bmatrix}=\begin{bmatrix}A  & B\\
C & D\end{bmatrix}^{-t}\begin{bmatrix} y \\ x
\end{bmatrix},\quad \begin{bmatrix}A  & B\\
C & D\end{bmatrix}\in Sp(2k,\Z),
\ee
so that the periods $(a^i,b_i)$ are flat sections of $\ce^\vee\simeq \ce$. The holomorphic tangent bundle $TM^\sharp$
is then identified with the quotient bundle $\ce/\cv$ (again the pull-back of a homogeneous bundle). In particular, the flat connection $\nabla^\text{GM}$ induces canonically a quotient bundle connection $\nabla^Q$ on $\ce/\cv$, that is, on $TM^\sharp$. Taking the differential and using $db_i=\boldsymbol{\tau}_{ij}\, da^i$ we get
the modular transformation of the $k\times k$ period matrix $\boldsymbol{\tau}$
\be
\boldsymbol{\tau}=(A\boldsymbol{\tau}+B)(C\boldsymbol{\tau}+D)^{-1};
\ee
\item[S8:] \textit{(the hyperK\"ahker structures on $X^\sharp$ and $\cv$)} On the total space of $\cv$, equivalently of the flat real bundle $\Gamma\otimes \R$,
there is a hyperK\"ahler structure $(I^a,g)$ invariant under translation by local sections of  $\Gamma\otimes \R$; then $(I^a,g)$ descends to a hyperK\"ahler structure on the total space of $\ch^\sharp$ \cite{Cecotti:1988qn}. The complex structure of $\cv$ is the $\zeta=0$ one in hyperK\"ahler $\mathbb{P}^1$-family of complex structures. We give the hyperK\"ahler structure by presenting the explicit $\mathbb{P}^1$-family of local holomorphic Darboux coordinates $X^a(\zeta)=(q^i(\zeta),p_i(\zeta))$
satisfying the reality condition \cite{Gaiotto:2008cd}
\be
X^a(\zeta)=-\overline{X^a(-1/\bar \zeta\,)},
\ee
such that the holomorphic symplectic form in complex structure $\zeta\in\mathbb{P}^1$ is 
\be
\Omega(\zeta)=-\frac{i}{2\,\zeta} \,\omega_++\omega_3-\frac{i}{2}\,\zeta\,\omega_-=dp_i(\zeta)\wedge dq^i(\zeta),\quad\text{where }\omega_\pm =\omega_1\pm\omega_3,
\ee
 and $\omega_\alpha$ ($\alpha=1,2,3$) are the three K\"ahler forms. We have
\be
q^i(\zeta)= \frac{1}{\zeta}\,a^i +iy^i +\zeta\,\overline{a}^i,\quad 
p_i(\zeta)= \frac{1}{\zeta}\,b_i +ix_i +\zeta\,\overline{b}_i,
\ee
hence
\be
\omega_+= 2\Big(da^i\wedge dx_i-db_i\wedge dy^i\Big), \qquad \omega_3=db_i\wedge d\bar a_i+d\bar b_i\wedge da^i.
\ee
\item[S9:] \textit{(fiber metric and Chern connection on $\cv$)} Restricting the hyperK\"ahler metric along the fibers, we get a Hermitian metric and associated Chern connection on the holomorphic bundle $\cv$. By uniqueness of the homogeneous connection, is coincides with the sub-bundle connection $\nabla^H$. The Hermitian metric is simply
$\|z\|^2= y^{ij}z_i\bar z_j$ where $y^{ij}$ is the inverse matrix of $y_{ij}=2\,\mathrm{Im}\,\boldsymbol{\tau}_{ij}$;
\item[S10:]\textit{(Special K\"ahler metric on $M^\sharp$ and its global K\"ahler potential)} In the same way, restricting the hyperK\"ahler metric on $\cv$ to the zero section (which is a holomorphic subspace in $\zeta=0$ complex structure) we get a K\"ahler metric on $M^\sharp$ whose K\"ahler form is the restriction of $\omega_3$. 
The restriction to the zero-section yields a K\"ahler metric on $M^\sharp$ with K\"ahler form
\be\label{spekeke}
\omega_3|_{M^\sharp}=\Big(\boldsymbol{\tau}_{ij}(a)-\overline{\boldsymbol{\tau}}_{ij}(\bar a)\Big)da^i\wedge d\bar a^j,\qquad \boldsymbol{\tau}_{ij}(a)\equiv \partial_{a_i}\partial_{a_j}\cf,\qquad \mathrm{Im}\,\boldsymbol{\tau}_{ij}(a)>0.
\ee 
We note that the assumption of the existence of a SW differential implies the existence of a globally defined K\"ahler potential on $M^\sharp$:
\be\label{kakapo}
\Phi =-i\big( b_i \,\bar{a}^i-a^i\,\bar{b}_i\big). 
\ee
Again, by uniqueness of the homogeneous connection, the Levi-Civita connection $\nabla^\text{LC}$ of this K\"ahler metric is the quotient-bundle connection $\nabla^Q$. In other words, all the relevant connections are just projections of the flat one.
\item[S11:] \textit{(The cubic symmetric  form of the infinitesimal Hodge deformation)} This is a symmetric holomorphic cubic form of type (3,0)
\be
\odot^3 TM^\sharp \to \co_M,
\ee
describing the infinitesimal deformation of Hodge structure (of the Abelian fiber) in the sense of Griffiths \cite{grif}.
Locally in special coordinates it is given just by
\be
T_{ijk}= \frac{\partial^3\cf}{\partial a^i\partial a^j\partial a^k}.
\ee
\end{description}

\begin{rem} For $k=1$ a special structure is, in particular, a surface fibered over a curve whose general section is an elliptic curve. Hence the possible local behaviors
(i.e.\! degenerations of fibers) are  described by the classical Kodaira papers \cite{koda}.
In his terminology, \textbf{S3} is called the \emph{homological invariant} and \textbf{S5} the \emph{analytic invariant}.
\end{rem}

\begin{exe}[$k=1$ locally flat special structures\footnote{\ We shall return several times to this \textbf{Example} in the paper. The present discussion is meant as a mere appetizer.}]\label{jza183} In this paper we are interested in conic special structures. Since in real dimension 2 all metric cones are locally flat, for $k=1$ we are reduced to study flat special geometries whose discriminant is a single  point. The hyperK\"ahler manifold $X^\sharp$ then is locally isometric to $\R^4\cong \mathbb{H}$, and 
the singular hyperK\"ahler geometry should be of the form $\C^2/G$ with $G$ a finite subgroup of $SU(2)$. One checks that conformal invariance requires the group $G$ to correspond \emph{via} the McKay correspondence to an affine Dynkin graph which is a star, that is, $D_4$, $E_6$, $E_7$, or $E_8$. Before resolving the singularity, the spaces $X_\text{sing}$ are the well-known Du Val singular hypersurfaces\footnote{\  In eqn.\eqref{kkkks} $h_3(x,u)$ stands for a  homogeneous cubic polynomial in $x$, $u$.} in $\C^3$ \cite{singularities}
\be\label{kkkks}
\begin{aligned}
&D_4\colon y^2-h_3(x,u)=0,
&&\phantom{mm}&&E_6\colon y^2-4\,x^3+u^4=0,\\
&E_7\colon y^2-4\,x^3+u^3x=0,
&&&&E_8\colon y^2-4\,x^3+u^5=0.
\end{aligned}
\ee
That $u\colon X_\text{sing}\to \C$ is an elliptic fibration (with section) is obvious by reinterpreting Du Val singularities as the Weierstrass model of a family of elliptic curve parametrized by $u$. 
The crepant resolution $X$ of $X_\text{sing}$ is given by the corresponding ALE space.
For each of the four special geometries we have \emph{a priori} two distinct special structures. Indeed, we have two dual choices for the SW differential $\lambda$: I) a holomorphic section of $\cv$ with no zero in $M^\sharp\equiv \{u\neq 0\}$ which vanish in the limit $u\to0$ to order at most 1, or II) a holomorphic section of $\cv^\vee$ with the same properties. Note that these properties fix $\lambda$ uniquely up to an irrelevant overall constant.  In terms of the Weierstrass model the two dual choices read:
\be
{\rm I)}\ \lambda=u\, \frac{dx}{y},\qquad\qquad\quad {\rm II)}\ \lambda=\frac{x\,dx}{y}.
\ee
The corresponding Coulomb branch dimensions are
\be
\begin{tabular}{c|cccc}\hline\hline
& $D_4$ & $E_6$ & $E_7$ & $E_8$\\\hline
I) & 2 & 3 & 4 & 6\\\hline
II) & 2 & $\tfrac{3}{2}$ &
$\tfrac{4}{3}$ & $\tfrac{6}{5}$\\\hline\hline
\end{tabular}
\ee
which is the correct list of (non-free) $\Delta$'s for $k=1$. 
The periods can be easily computed using Weierstrass elliptic functions\footnote{\ Notations as in DLMF \S.23 \cite{DLMF}.}. As an example, we write them for $E_8$:
$$
{\rm I)}\ \begin{bmatrix}a\\ b\end{bmatrix}\equiv
u \!\begin{bmatrix}e^{i\pi/3}\omega_1\\ \omega_1\end{bmatrix}=\frac{\Gamma(\tfrac{1}{3})^3}{4\pi}\!\begin{bmatrix}e^{i\pi/3}\\1
\end{bmatrix}u^{1/6},\qquad
{\rm II)}\ \begin{bmatrix}b\\ a\end{bmatrix}\equiv
\begin{bmatrix}e^{i\pi/3}\eta_3\\ \eta_3\end{bmatrix}=\frac{2\pi^2 e^{-\pi i/3}}{\sqrt{3}\,\Gamma(\tfrac{1}{3})^3}\!\begin{bmatrix}e^{i\pi/3}\\1
\end{bmatrix}u^{5/6}
$$
from which it is obvious that the dimension of $u$ is $6$ and respectively $6/5$.
In the dual choice the role of $a$ and $b$ get interchanged, since the non-trivial element of $H_\R$ inverts the sign of the polarization. Of course, the resolutions of the singularity at $u=0$ are different in the two cases, the exceptional locus being Kodaira exceptional fiber of type $II^*$ and $II$, respectively.
The periods of $dx/y$ and $x\, dx/y$ scale with opposite power of $u$ by the Legendre relation. 
\end{exe}

\subsubsection{Rigidity principle and  reducibility}
A basic trick of the trade is that global properties in special geometry fix everything. This principle is known as ``the Power of Holomorphy'' \cite{power};
mathematicians call it \emph{rigidity}.

\begin{pro}[Rigidity principle \cite{perbook}] \label{rigidity} Two special geometries with the same \emph{compact} base $M$, isomorphic monodromy representations, and isomorphic fibers over \emph{one} point, are equivalent.
\end{pro}

Thus the monodromy representation \textbf{S3} essentially determines the special structure. In particular, if the monodromy representation splits $m=m_1\oplus m_2$ (over\footnote{\ If the splitting is over $\BQ$, the geometry is a product up to an isogeny in the fiber.} $\Z$) then the special geometry is a  product.

\subsubsection{Curvature properties of special geometry}\label{uuzaqf}

Let $\cw$ be a holomorphic Hermitian vector bundle with Chern connection $\nabla$. We consider a holomorphic sub-bundle $\cs\subset\cw$ and  the  quotient bundle $\cq=\cw/\cs$ equipped with the sub-bundle and quotient connections $\nabla^\cs$ and $\nabla^\cq$, respectively.
The curvature of $\nabla^\cs$ (resp.\! $\nabla^\cq$) is bounded above (resp.\! below) by the one of $\nabla$, see \cite{GH} page 79 or \cite{perbook}.
Applying this principle to $\ce$, $\cv$ and
$\ce/\cv$ we get:
\begin{pro}\label{jjjas} The curvature of the bundle $\cv$ is \emph{non-positive,} while the curvature of the K\"ahler metric on $M^\sharp$ is \emph{non-negative} (in facts, \emph{positive}). In particular, the Ricci curvature of $M^\sharp$ is non-negative, $R_{i\bar\jmath}\geq 0$ and it vanishes iff $M^\sharp$ is locally flat.
\end{pro}

Let us give an alternative proof of the last statement.

\begin{proof}
In a K\"ahler manifold the Ricci form is\footnote{\ Cfr. \cite{Besse} eqn.(2.98).}
$\rho=-i\partial\bar \partial\log\det g$. Thus from \eqref{spekeke}
\be
\rho=-i\partial\bar\partial\log\det\mathrm{Im}\,\boldsymbol{\tau}_{kl}(a)\equiv \boldsymbol{\tau}^* \Omega
\ee
where $\Omega$ is the (positive) K\"ahler form on the locally Hermitian space
in eqn.\eqref{mumu}. Note that $R_{i\bar \jmath}=0$ only at critical points of the period map $\boldsymbol{\tau}$. By Sard theorem, the set of periods $\boldsymbol{\tau}_{ij}(a)$ at which $R_{i\bar \jmath}=0$ has zero measure. 
In particular $R_{i\bar \jmath}\equiv0$ means
$\boldsymbol{\tau}=$ (a constant map), so $M^\sharp$ is locally flat.  
\end{proof}

\begin{rem} This result may also be understood as follows. The total space of the holomorphic integrable system, $X$, is hyperK\"ahler, so carries a Ricci-flat metric. $M$ is a complex subspace, and the Ricci curvature of its induced metric is minus the curvature of the determinant of the normal bundle whose Hermitian metric is $(\det
\mathrm{Im}\,\boldsymbol{\tau})^{-1}$.
\end{rem}

\paragraph{Sectional and isotropic curvatures.}
From the above \textbf{Proposition} it  is pretty obvious that all sectional curvatures of a special K\"ahler metric are non-negative. A stronger property is that all its \emph{isotropic curvatures} are non-negative. Indeed, we claim an even stronger statement, that is, that the \emph{curvature operators} are non-negative at all points $p$.

\begin{defn} Let $X$ be a Riemannian $n$-fold with tangent space $T_pX$ at $p\in X$. The \textit{curvature operator} at $p$
is the self-adjoint linear operator
\be
\boldsymbol{R} \colon \wedge^2 T_pX\to \wedge^2 T_pX,
\ee  
given by the Riemann tensor. We say that $X$ has  
 \textit{positive} (resp.\! \textit{weakly positive}) \emph{curvature operators} iff the eigenvalues of $\boldsymbol{R}$ are positive (resp.\! non-negative) at all $p\in X$.
\end{defn}

The claim follows from the explicit form of the Riemann tensor
\be\label{cllae}
R_{i\bar j k \bar l}= T_{ik m}\,\bar T_{\bar j\bar l\bar n}\,g^{m\bar n},
\ee
where $T$ is the cubic symmetric form of the infinitesimal Hodge deformation (structure \textbf{S11}). 

\paragraph{Sphere theorems.} The positivity of the curvature operators has dramatic implications for the topology of $X$. We collect here some results which we shall use later in the paper:

\begin{thm}[Meyers \cite{meyers}]\label{meymey} Let $X$ be a complete Riemannian manifold of metric $g$ whose Ricci curvature satisfies $R\geq \lambda^2 g$ with $\lambda>0$ a constant. Then $X$ is compact with diameter $d(X)\leq \pi/\lambda$.
Applying the result to the Riemannian universal cover $\widetilde{X}$ of $X$,
we conclude that $\pi_1(X)$ is finite. 
\end{thm}

\begin{rem}
There is a version  of Meyers theorem which applies to orbifolds, see \textbf{Corollary 21} in \cite{mey1} or \textbf{Corollary 2.3.4} in \cite{mey2}. In case of Riemannian orbifolds \emph{complete} should be understood as  \emph{complete as a metric space.} The version in \cite{mey2} states that a metrically complete Riemannian orbifold $X$, whose Ricci curvature satisfies $R\geq \lambda^2 g$, is compact with a diameter $d(X)\leq \pi/\lambda$. 
\end{rem}

\begin{thm}[Synge \cite{synge}]\label{zxaqw} An even dimensional compact orientable manifold with positive sectional curvature is simply-connected.
\end{thm}

\begin{rem} Again the result extends to Riemannian orbifolds, see \textbf{Corollary 2.3.6} of \cite{mey2}, so that an even dimensional orientable complete Riemannian orbifold with positive sectional curvature is simply connected.
\end{rem}

\begin{thm}[Micallef-Moore \cite{MM},  B\"ohm-Wilking \cite{BW}]\label{jkkkkmzq162} Let $X$ be a compact $n$-dimensional Riemannian \emph{orbifold}. If $X$ has positive curvature operators it is diffeomorphic to a space form $S^n/G$, $S^n$ being the sphere and $G$ a finite subgroup of $SO(n+1)$.
\end{thm}

The special K\"ahler manifolds $M$ have just \emph{weakly} positive curvature operators (and are typically non-compact).
However, taking \textbf{Theorems 1}, \textbf{3} together, one gets the rough feeling that the non-flat special K\"ahler manifolds are ``close'' to being locally spheres. The statement will become precise under the assumption that $M$ is also a cone.

\subsubsection{Behavior along the discriminant}\label{jjzqwe}

We need to understand the behavior near the discriminant locus
$\cd\subset M$ where the fiber degenerates, that is, some periods $(a_i,b^j)$ vanish. Physically this means that along the discriminant locus some additional light degrees of freedom appear, so that the IR description in terms of the massless fields parametrizing $M$ becomes incomplete and breaks down.

The singular behavior is best understood in terms of properties of the period map $\boldsymbol{\tau}$. We see the discriminant $\cd$ as an effective divisor $\cd=\sum_i n_i S_i$, where $S_i$ are the irreducible components and $M^\sharp=M\setminus\mathrm{Supp}\,\cd$. The behavior of the period map as we approach a generic point $s$ of an irreducible component $S_i$ is described by three fundamental results: the \textbf{strong monodromy Theorem} \cite{grif,grifpaper,perbook},
the \textbf{$SL_2$-orbit Theorem} \cite{schmid}, and the \textbf{invariant cycle Theorem} \cite{schmid}.
In a neighborhood $U$ of $s\in S_i$,
we may find complex coordinates
$z_1,\cdots, z_k$ so that, locally in $U$, 
$S_i$ is given by $z_1=0$. Then
we have $U\cap M^\sharp\cong \Delta^*\times \Delta^{k-1}$
where $\Delta$ (resp.\! $\Delta^*$) stands for the unit disk (resp.\! the punctured unit disk). 
We write $p$ for the period map
$\boldsymbol{\tau}$ restricted to
$\Delta^*\times (z_2,\cdots,z_k)\subset \Delta^*\times \Delta^{k-1}$,
and $\mathfrak{h}$ for the upper half-plane seen as the universal cover of $\Delta^*$ via the map $\tau\mapsto q(\tau)\equiv e^{2\pi i\tau}$.
 We have the commutative diagram
\be\label{kkkazqw}
\begin{gathered}
\xymatrix{\mathfrak{h}\ar[d]_q\ar[rrr]^{\tilde p} &&& Sp(2k,\R)/U(k)\ar[d]^{\text{can}}\\
\Delta^*\ar[rrr]^p &&& Sp(2k,\Z)\backslash Sp(2k,\R)/U(k)}
\end{gathered}
\ee
where $\tilde p$ is the lift of the (restricted) period map. Let $\gamma$ be the generator of $\pi_1(\Delta^*\times \Delta^{k-1})\cong\Z$ and $m\equiv m(\gamma)$ the corresponding monodromy element (cfr.\! \textbf{S3}).
Then
\be\label{iiixq2}
\tilde p(\tau+1)=m\cdot \tilde p(\tau).
\ee
Let $d(\cdot,\cdot)$ be the distance function defined by the standard invariant metric on the symmetric space $Sp(2k,\R)/U(k)$ and $d_P(\cdot,\cdot)$ the distance with respect to the usual Poincar\'e metric in $\mathfrak{h}$; the inequalities on the curvatures together with the Schwarz lemma imply
\begin{pro}[Strong monodromy theorem \cite{grif,grifpaper,perbook}]\label{monthm} The lifted period map $\tilde p$ is distance-decreasing
\be
d(\tilde p(x), \tilde p(y))\leq d_P(x,y).
\ee
Then the monodromy $m$ is \emph{quasi-unipotent,} i.e.\! there are minimal integers $r\geq 1$, $0\leq s\leq k$ such that
\be
(m^r-1)^{s+1}=0.
\ee
Equivalently (by Kronecker theorem) $m$ has spectral radius $1$ (Mahler measure 1). 
\end{pro}

All eigenvalues of $m$ are $r$-th roots of unit. Since $m\in Sp(2k,\Z)$, its minimal polynomial $M(z)$ is a product of cyclotomic polynomials $\Phi_d(z)$
\be
M(z)=\prod_{d\mid r} \Phi_d(z)^{s_d},\qquad s_d\in\mathbb{N}.
\ee  
The monodromy $m$ is \emph{semi-simple} iff $s=0$, that is, if $s_d\in\{0,1\}$ for all $d$. We say that $m$ is \emph{regular} iff all its eigenvalues are distinct, i.e.\! iff $s_d\in\{0,1\}$ and
$\sum_{d\mid r} s_d=2k$.

\paragraph{The case of $m$ semi-simple.} 
Semi-simplicity of $m$ has the following consequence: 

\begin{pro}[\!\!\cite{grif,grifpaper,perbook}]\label{exeeeat}
The period map $p\colon \Delta^*\to Sp(2k,\Z)\backslash Sp(2k,\R)/U(k)$ may be extended holomorphically to the origin if and only if $s=0$.
\end{pro}
\noindent In other words, along an irreducible component $S_i$ of $\cd$ whose monodromy element $m$ is semi-simple the period matrix $\boldsymbol{\tau}_{ij}$ is defined and regular even if the Abelian fiber itself degenerates.
The K\"ahler metric $d^2s= 2\,\mathrm{Im}\,\boldsymbol{\tau}_{ij}\,da^i\otimes d\bar a^j$
 is singular along $S_i$ since the periods $a^i$ are not valid local coordinates at this locus. The singularity is of the mildest possible kind: just a cyclic orbifold singularity. We illustrate the situation along a semi-simple component $S_i$ of the discriminant $\cd$ in a typical example.

\begin{exe}[Non-Lagrangian $k=1$ SCFT\footnote{\ This is a special case of \textbf{Example \ref{jza183}}.}]\label{reqasx}
In these $k=1$ models the period of the elliptic fiber $\tau$ is frozen in a orbifold (elliptic) point of the modular fundamental domain $\mathfrak{h}/SL(2,\Z)$, i.e.\!
either $\tau=e^{2\pi i/3}$  or $\tau=i$ depending on the model. Thus the K\"ahler metric is flat 
and $M$ is locally isometric to $\R^2$
\be
ds^2=2\,\mathrm{Im}\,\tau\,da\,d\bar a \xrightarrow{\ \text{\emph{local} isometry}\ } dr^2+r^2\,d\theta^2,\qquad r^2=2\,\mathrm{Im}\,\tau\,|a|^2.
\ee
The coordinate $r$ is globally defined, since $r^2$ is the momentum map of the $U(1)$ action given by $R$-symmetry. On the contrary, the period of the canonically conjugate angle, $\theta$ needs not to be $2\pi$ (which corresponds to the free SCFT). The period of the angle $\theta$ is related to the Coulomb dimension $\Delta$
by the identification  $\theta\sim \theta+2\pi/\Delta$. Hence, if the theory is not free, $\Delta\neq1$, at the tip of the cone we have a cyclic orbifold singularity.
We note that the unitary bound $\Delta\geq 1$ (with equality iff the SCFT is free) becomes $\text{(period of $\theta$)}\leq  2\pi$. Thus unitarity requires the curvature at the tip to be \emph{non-negative} and we may smooth out the geometry by cutting away the region $r\leq \epsilon$ and gluing back a \emph{positively} curved disk. This is consistent with our discussion of the curvature in special geometry in \S.\ref{uuzaqf}. This example shows that the 
curvature inequalities apply also to the $\delta$-function curvature concentrated at orbifold points and their relation to physical unitarity. 
\end{exe}

We state this as a
\begin{prin} The unitarity bounds  guarantee that the $\delta$-function curvatures associated to the angular deficits at orbifold points are consistent with the positivity of curvatures required by special K\"ahler geometry. 
\end{prin}

We quote another useful  result:

\begin{pro}[\!\!\cite{perbook}]\label{llxmz} Suppose that the period map $\boldsymbol{\tau}$ factors trough a quasi-projective variety $K$
\be\label{aazq0}
\begin{gathered}\xymatrix{M\ar[dr]\ar[rrr]^{\boldsymbol{\tau}} &&& Sp(2k,\Z)\backslash Sp(2k,\R)/U(k)\\
& K\ar@/_1.45pc/[urr]^p}\end{gathered}
\ee
and that the discriminant of $p$
is a \textsc{snc}\footnote{\ \textsc{snc} $=$ simple normal crossing.} divisor with semi-simple monodromies.
If $p$ is \underline{not} the constant map, 
$p$ is \emph{proper}.
Its image is a closed analytic subvariety containing $p(K^\sharp)$ as the complement of an analytic set.
\end{pro}

\paragraph{$m$ non semi-simple.} We now turn to the case in which the monodromy $m$ is \emph{not} semi-simple. Again we consider the neighborhood $U\cong\Delta^*\times \Delta^{k-1}$ considered around eqn.\eqref{kkkazqw} and
pull-back all structures to its universal cover $U_\text{uni}\cong\mathfrak{h}\times \Delta^{k-1}$. By the strong monodromy theorem there exist minimal integers $r\geq 1$, $s\geq 0$ such that $(m^r-1)^{s+1}=0$. In the non semi-simple case $s\geq1$. Then,
\begin{align}\label{kkaqwer2}
&\tilde p(\tau+r)= (1+T)\cdot \tilde p(\tau)\quad \text{with }T=m^r-1\ \text{and }T^{s+1}=0,\\ 
&\textrm{so that}
\ \ N\equiv \log(1+T)=\sum_{n=1}^s \frac{(-1)^{n-1}}{n}\,T^n\ \ \text{is well defined}.
\end{align}
$N^s\neq0$ and $N^{s+1}=0$.
The nilpotent operator $N$ defines the weight filtration of a mixed Hodge structure in the sense of Deligne \cite{deligne1} to which we shall return momentarily; more elementarily, by the Jacobson-Morozov theorem \cite{J1,J2,J3} the rational matrix $N$ defines a polynomial homomorphism
$\phi\colon SL(2,\BQ)\to Sp(2k,\BQ)$
such that $N$ is the image of the raising operator of the $\mathfrak{sl}(2,\BQ)$ Lie algebra. $\phi$ induces a period map
$\mathring{p}\colon\mathfrak{h}\to Sp(2k,\R)/U(k)$ which is the simplest solution to the functional equation \eqref{kkaqwer2}:
\be\label{kkkaq12}
\mathring{p}(\tau)= e^{\tau N/r}\cdot p_0\qquad p_0\ \text{globally defined in $U_r$},
\ee
where $\sigma\colon U_r\to U$ is the local $r$-fold cover  
\be
\sigma\colon U_r\equiv\mathfrak{h}/(\tau\sim\tau+r)\times \Delta^{k-1}\longrightarrow \mathfrak{h}/(\tau\sim\tau+1)\times \Delta^{k-1}\cong U.
\ee
The $SL_2$-orbit theorem \cite{schmid} states that the actual period map $p$ differs from the Lie-theoretic  map $\mathring{p}(\tau)$ by exponentially small terms $O(q^{1/r})$ as $\tau\to i\infty$ ($q\equiv e^{2\pi i\tau}$). A physicists studying the corresponding (2,2) supersymmetric $\sigma$-model states the theorem saying that $\mathring{p}(\tau)$ is the perturbative solution, valid asymptotically as the coupling $4\pi/\mathrm{Im}\,\tau\to 0$, and this perturbative solution receives corrections only by instantons which are suppressed by the exponentially small (fractional) instanton counting parameter $q^{1/r}$. To physicists working in 4d $\cn=2$ QFT , the $SL_2$-orbit theorem is familiar as a fundamental result by Seiberg \cite{seib}. 
\smallskip

Let $0\neq x\in U_r$; $\phi$ decomposes $\Gamma_x\otimes \BQ$ into irreducible representations of $SL_2(\BQ)$; the highest weight $\BQ$-cycles $\psi$ are defined by the condition $N\psi=0$; all other $\BQ$-cycles are obtained from these ones by acting on them with the $SL_2$ lowering operator.
Since $\boldsymbol{\tau}$ is the period map of a degenerating weight 1 Hodge structure, it follows from the Deligne weight filtration (or by the Clemens-Schmid sequence, see \textbf{Corollary 2} in \cite{morrison}), that only spin 0 and spin 1/2 representations are presents, that is, $N^2=0$.
More precisely, we have a weight filtration of $\BQ$-spaces
\be
W_0\subset W_1\subset W_2\equiv \Gamma_x\otimes \BQ,\qquad W_0=\mathrm{im}\,N,\quad W_1=\ker N,
\ee
such that $N\colon W_2/W_1\to W_0$ is an isomorphism,
and the polarization $\langle -,-\rangle$ induces a perfect pairing between $W_2/W_1$ and $W_0$
as well as of $W_1/W_0$ with itself \cite{schmid} (of course, these statements are just the usual selection rules for angular momentum).

We pull-back the local family of Abelian varieties $X|_U$ to the $r$-fold cover $U_r$; we get the family
$\pi\colon \sigma^* X|_U\to U_r$.
By construction, the monodromy of the pulled back family is $m_\sigma\equiv m^r=e^N$, so the monodromy invariant $\BQ$-cycles are precisely the ones in $W_1$. 
The invariant cycle theorem guarantees that all 1-cycle $\gamma_x\in \Gamma_x$ invariant under the monodromy  there is a homologous 1-cycle $\hat\gamma$ in $\sigma^* X|_U$
in the total space of the (local) family and all 1-cycles in the total space are of this form. Then
$W_2/W_1$ consists of vanishing cycles, so that the corresponding periods vanish as $q^{1/r}\to 0$
i.e.\! $a_\text{van}\propto q^{1/r}$
and then eqn.\eqref{kkkaq12} says that the periods along the ``spin-0'' cycles $W_1/W_0$ are regular as $q\to0$ while the ones
in $W_0$ (which are dual to the vanishing ones under the Dirac pairing) go as
\be\label{kkkaqznn}
a_\text{van}^D\propto a_\text{van}\,\log a_\text{van}.
\ee 
The conclusion we got is totally trivial from the physical side.
The special geometry along the Coulomb branch is the IR description obtained integrating out the massive degrees of freedom; the singularities arise because at certain loci in $M$ some additional degree of freedom becomes massless.
One gets the leading singularity by computing the correction to the low energy coupling by loops of light fields, see the discussion in \S.\,5.4 of the original paper by Seiberg and Witten \cite{SW1}. The mixed Hodge variation formula \eqref{kkkaqznn}  is just their eqn.(5.10).

Thus, a part for the need to go to the local $r$-fold cover $U_r$, at a generic point of an irreducible component $S_i$ of the discriminant $\cd$ we do not get singularities worse than physically expected. The singularity in eqn.\eqref{kkkaqznn} is mild:\begin{itemize}
\item[\bf R1]
the squared-norm of the SW differential
\be
\begin{split}
\Phi(q) \overset{\text{def}}{=}\; & i \!\left(\int_{A^i}\lambda\int_{B_i}\bar\lambda-\int_{A^i}\bar\lambda\int_{B_i}\lambda\right) =\\
&=\mathrm{const.}\, |q|^{2/r}(-\log|q|)+\text{regular as }q\to0,
\end{split}
\ee
while \emph{not} smooth along  the discriminant, extends \emph{continuously} to $\cd$;
\item[\bf R2]
Its differential $d\Phi$, while singular in $U$, becomes continuous (non-smooth) when pulled back to the local $r$-fold cover $U_r$;
\item[\bf R3] the points on the discriminant, $q=0$, are at a finite distance from smooth points. Indeed, on the local cover $U_r$ the metric is modeled on $ds^2=(-\log|q|) |dq^{1/r}|^2$ which is length decreasing with respect the flat metric $|dz|^2$, $z=(-\log|q|)^{1/2} q^{1r}$. On $U$ the metric is asymptotically conical. In particular, $M$ remain complete as a metric space;
\item[\bf R4] the integral of the Ricci curvature on the $r$-fold  \emph{covering} disk $|q^{1/r}|<\epsilon$ vanishes as $\epsilon\to0$, i.e.\! there is no $\delta$-like curvature concentrated on the discriminant, \emph{except} for the obvious $\Z_r$-orbifold singularity implied by the covering quotient $U_r\to U_r/\Z_r\equiv U$. Thus all arguments based on curvature bounds work as in the semi-simple case. We have already remarked that orbifold singularities do not spoil the curvature bounds (cfr.\! \textbf{Physical principle}). 
\end{itemize}

All the above statements hold at  all points of the discriminant (and not just at generic points along a smooth component) when $\cd$ is a \textsc{snc} divisor \cite{schmid}.
While this is generically the situation, the special geometry describing a particular SCFT with the mass deformations switched off may well be non generic. If the SCFT admits ``enough'' mass/relevant deformations, we can make $\cd$ to be \textsc{snc} by an arbitrarily small perturbation which cannot change the qualitative aspects of the physics. Even in SCFTs without (enough) deformations, it is very likely that --- while the singularities may be more severe than the \textsc{snc} ones --- the four regularity conditions \textbf{R1}-\textbf{R4} still hold. Indeed, \textbf{R3} has been advocated by Gukov, Vafa and Witten as a necessary condition for a sound SCFT \cite{GVW}. In the rest of the paper we shall make the

\begin{ass} Our special geometry satisfies {\bf R1}-{\bf R4}.
\end{ass}

\begin{rem} We may look at the singularities also from the point of view of the hyperK\"ahler geometry of the total space $X$.
Since hyperK\"ahler manifolds  are in particular Calabi-Yau, the discussions of refs.\cite{GVW,TY} directly apply with similar conclusions. Note that the statements hold also for hyperK\"ahler \emph{orbifolds.}
\end{rem}

\subsection{Some facts about complex orbifolds}

In the last subsection we found that the analytic space $M$ typically has cyclic orbifold singularities. 
Here we collect some well known facts about complex orbifolds that we shall need below.

\begin{pro}[See e.g.\! \cite{sas}] The locally ringed space $(Z,\co_Z)$ associated to a complex orbifold has the following properties:
\begin{itemize}
\item[\textit{i)}] $(Z,\co_Z)$ is a reduced normal analytic space;
\item[\textit{ii)}] the singular locus $\Sigma(Z)$ is a closed reduced complex subspace of $Z$ and has complex codimension at least 2 in $Z$;
\item[\textit{iii)}] the smooth locus $Z_\text{reg}$ is a complex manifold and a dense open subset of $Z$;
\item[\textit{iv)}] $Z$ is $\mathbb{Q}$-factorial.\footnote{\ An analytic space is $\mathbb{Q}$-factorial if all Weil divisor has a multiple which is a Cartier divisor. }
\end{itemize}
\end{pro}

In particular, under our mild assumption, the Coulomb branch $M$ is a $\mathbb{Q}$-factorial reduced normal analytic space.

We stress that the singular set in the orbifold sense of $Z$, $S(Z)$, may be actually larger than the singular locus of the underlying analytic space, $\Sigma(Z)$, see the discussion in ref.\cite{sas}. The case of maximal discrepancy between the two sets is given by the following: 

\begin{pro} Let $G$ be a Shephard-Todd group ($\equiv$ a finite complex reflexion group {\rm\cite{cohen,ST1,ST2}}). Then the analytic space underlying the orbifold $\C^n/G$ is smooth, in fact isomorphic to $\mathbb{A}^n$.
\end{pro}

We shall also need the orbifold version of
the Kodaira embedding theorem:

\begin{thm}[Kodaira-Baily \cite{baily}]\label{kkkxza} Let $Z$ be a compact complex orbifold and suppose $Z$ has a positive orbi-bundle $\cl$. Then $Z$ is a projective algebraic variety.
\end{thm}

\subsection{Structures on cones}

We review the geometry of metric cones in a language suited for our purposes.

\subsubsection{Riemannian cones}\label{riecones}

A \textit{metric (Riemannian) cone} over the (connected, Riemannian) base $B$ is the \emph{warped product}
\be
\R_{>0}\times_{r^2}B\equiv C(B),\ee 
that is, $C(B)$ is the product space $\R_{>0}\times B$ equipped with the metric
\be\label{conical}
ds^2=dr^2+r^2\, \gamma_{ab}(y)\,dy^a dy^b
\ee
where $ds^2_y=\gamma_{ab}(y)\,dy^a dy^b$ is a metric on $B$ ($y^a$ being local coordinates in $B$). We shall write $\overline{C(B)}$ for the singular space obtained by adding the tip of the cone $r=0$ to $\c(B)$, endowed with the obvious topology. 
Note that the radial coordinate $r$ is a globally defined continuous real function on $\overline{C(B)}$ taking all non-negative  values. A cone $\C(B)$ possesses the following canonical (global) structures: the plurisubharmonic function $r^2$ and the  concurrent vector field $E= r\partial_r$ (\textit{Euler field}) which satisfy the following properties
\be\label{cococ}
\pounds_E r^2=2r^2,\qquad \pounds_E g=2g,\qquad
E(dr^2)=2r^2,\qquad 2 E_i=\nabla_i r^2.
\ee
that is
\be\label{ide1}
\nabla_i\nabla_j r^2= \nabla_iE_j+\nabla_jE_i= 2 g_{ij}.
\ee
\begin{pro}[\!\!\cite{yano}]\label{hhasqw}
 Conversely, if the Riemannian manifold $(C,g)$ has a vector field $E$ whose dual form is closed and $\pounds_E g=2g$, there exist coordinates such that the metric takes the conical form \eqref{conical}.
\end{pro}

\begin{corl} Let $C_1$, $C_2$ be metric cones. Then $C_1\times C_2$ is a metric cone with Euler vector $E_1+E_2$.
\end{corl}

\begin{defn}\label{googg} By a \emph{good} cone we mean a cone $C(B)=\R_{>0}\times_{r^2} B$ with $B$ smooth and complete. For a good cone, the only possibly singular point is the tip of the cone $r=0$.
Note that a non-trivial product of metric cones is never good unless one of the factors is $\R^k$ with the flat metric.
\end{defn}

\begin{rem} On a smooth Riemannian manifold the two notions of geodesic completeness and metric-space completeness coincide (Hopf-Rinow theorem \cite{Besse}). This is not longer true in presence of singularities. \textbf{Example \ref{reqasx}} illustrates the point: the Minahan-Nemeshanski geometry is complete in the metric space sense, but certainly not in the geodesic one.
The singular Riemannian spaces which are ``physically acceptable'' better be complete as metric space. This is part of regularity assumption \textbf{R3}.  
\end{rem}

For later use, we give the well-known formulae relating the curvatures of $C(B)$ and its base $B$. We write $R_{ijkl}$ (resp.\! $R_{ij}$) for the Riemann (Ricci) tensor of $C(B)$ and $B_{abcd}$ (resp.\! $B_{ab}$) for the Riemann (Ricci) tensor of $B$.
\begin{lem}\label{l1}
One has
\begin{align}
R_{abcr}&=R_{arbr}=0\\
{R_{abc}}^d&={B_{abc}}^d- \gamma_{ac}\,\delta_b^d+\gamma_{bc}\,\delta_a^d\\
R_{ab}&=B_{ab}-(\dim B-1)\gamma_{ab}.
\end{align}
\end{lem}

\subsubsection{Singular K\"ahler cones: the Stein property} Now suppose the Riemannian manifold $M$ is both K\"ahler (with complex structure $I$)
and a (metric) cone, $M\cong \overline{C(B)}$. Eqn.\eqref{ide1} implies that $\Phi=r^2$ is a \emph{globally defined} K\"ahler potential assuming all values $0\leq \Phi<\infty$.

In the applications to $\cn=2$ SCFT we have in mind, the K\"ahler metric on the cone $M$ is singular even away from the tip $r=0$.
We specify the class of geometries we are interested in.

\begin{defn}
By a \emph{singular K\"ahler cone} $M$ we mean the following: \textbf{1)} $M$ is an analytic space (which we may assume to be normal\footnote{\ An analytic space $(M,\co_M)$ is \emph{normal} iff the stalks of its structure sheaf $\co_M$ are integrally closed, i.e.\! valuation rings. If $M$ is not normal, replace it with its normalization.}) with an open everywhere dense smooth complex submanifold $M^\sharp=M\setminus\cd$. \textbf{2)} On $M^\sharp$ there is a smooth conical K\"ahler metric
(in particular, $M^\sharp$ is preserved by the $\C^\times$ action generated by the holomorphic Euler vector $\ce$, see eqn.\eqref{kkzxbb362}). \textbf{3)} The global K\"ahler potential $\Phi\equiv r^2$ on $M^\sharp$
extends as a continuous function to all $M$ (cfr.\! regularity condition \textbf{R1}). Then the K\"ahler form 
$i\partial\overline{\partial} r^2$ extends to $M$ as a \textit{positive (1,1) current.} The continuous function $r^2$ is then plurisubharmonic in the sense of ref.\!\cite{lelo}.
\end{defn}

\begin{pro}\label{kkkz129c} A singular K\"ahler cone with a compact base $B$ is a \emph{Stein analytic space.}
\end{pro}

Indeed, $r^2$ is a continuous plurisubharmonic function which is an exhaustion for $M$. The statement is then the Narasimhan singular version of Oka theorem (see e.g.\! page 48 of \cite{lelo}).

Then the singular cone $M$ is Stein and Cartan's \textbf{Theorem A} and \textbf{Theorem B} apply \cite{GH,stein1,stein2}. Below we shall exploit this fact in several ways.

 In the physical applications we have in 
 mind, the Fr\'echet ring 
\cite{ammon} of global holomorphic functions, $\mathscr{R}=\Gamma(M,\co_M)$, is the Coulomb branch chiral ring, our main object of interest. $M$ being Stein implies that $\mathscr{R}$ contains ``many'' functions: around all points of $M$ we may find local coordinate systems given by global holomorphic functions, and $\mathscr{R}$ separates points, i.e.\! given two distinct points we may find a global holomorphic function which takes on these two points any two pre-assigned complex values. 
Affine varieties over $\C$ are in particular Stein \cite{stein2}.
The converse is not true in general, but it holds under some mild additional conditions \cite{addstein}. We shall see that that the $M$'s which are Coulomb branches of $\cn=2$ SCFTs are always affine.

\subsubsection{K\"ahler cones: local geometry at smooth points,  Sasaki manifolds} 
Specializing the results of \S.\ref{riecones}, in the (open everywhere dense) smooth locus $M^\sharp$  we have\footnote{\ Factor 2 mismatches arise from different conventions in the real vs. complex case.}
\be
g_{i\bar j}= \partial_i\partial_{\bar j}\Phi,\qquad
\nabla_i\partial_j \Phi=0,\qquad \Phi=r^2.
\ee
In particular, the real vector $R= IE$, or in components
\be\label{mmoamapw}
R^i=i g^{i\bar j}\partial_{\bar j}\Phi,\qquad R^{\bar i}=-i g^{\bar i j}\partial_j\Phi, 
\ee
is a \textit{Killing vector}. The physical interpretation of this geometric result is as follows: a K\"ahlerian cone may be used as a target space of a (classical) 3d supersymmetric $\sigma$-model. The fact that it is K\"ahler means the model is $\cn=2$ supersymmetric, while the fact that it is a cone means that is \emph{classically} conformally invariant \cite{Cecotti:2010dg}; the two statements together imply that the model has classically $\cn=2$ superconformal symmetry hence a  $U(1)_R$ $R$-symmetry which is part of the algebra. The action of $U(1)_R$ on the scalars is given by a (holomorphic) Killing vector which is $R$. We note that 
\be
[E,R]=0,\qquad R(dr^2)=0
\ee
so that $R=R^a(y)\partial_{y_a}$ is in facts a Killing vector for the metric $ds^2_y$ on the base $B$
whose norm is $1$, i.e.\! $R^aR_a=1$.
For a holomorphic function $h$ on a conic K\"ahler manifold the actions of $E$ and $R$ (in physical language: their dimension and $R$-charge) are related by 
\be\label{kkzxbb362}
\pounds_R h=i \pounds_E h\quad\Longleftrightarrow\quad \pounds_{\bar\ce} h=0,\quad\text{where }\ce=(E-iR)/2,
\ee
which physically says that these two quantum numbers should be equal for a chiral superconformal operator. We refer to $\ce$ as the holomorphic Euler vector.

\begin{rem}
By definition, a cone $\R_{>0}\times_{r^2} B$ is K\"ahlerian if and only if its base $B$ is \emph{Sasaki} \cite{sas}. The base $B$ is in particular a $K$-contact manifold whose \emph{Reeb vector} is $R$. 
\end{rem}

\subsubsection{Geometric ``$F$-maximization''}

We pause a second to digress on a different topic, namely $F$-maximization in 3d \cite{Jafferis:2010un}. A problem one encounters in studying SCFT is the exact determination of the $R$-charge which enters in the superconformal algebra. For (classical) $\sigma$-models with conic K\"ahler target spaces \cite{Cecotti:2010dg}, this is the problem of identifying the Reeb Killing vector $R$ between the family of Killing vectors with the appropriate action on the supercharges $\pounds_V Q= \pm \tfrac{1}{2} Q$.
The general such Killing vector has the form $V=R+F$ with $F$ a `flavor' Killing symmetry. One has the following:

\begin{cla} Let $M$ be a K\"ahler cone and $V$ a Killing vector on $M$ which acts on supercharges as $\pounds_V Q= \pm \tfrac{1}{2}Q$. Then (point-wise)
\be
\|V\|^2\geq \Phi\equiv r^2,
\ee
 with equality iff $V$ is the Reeb vector $R$. That is, the true superconformal $R$-charge extremizes the square-norm. 
\end{cla}

\begin{rem} The reader may easily check that this purely geometric fact \emph{is really} $F$-maximization for the partition function on $S^3$ of the corresponding 3d $\sigma$-model in the classical limit $\hbar\to 0$. By considering the low-energy effective theory on the moduli space of \textsc{susy} vacua of a (quantum) 3d $\cn=2$ SCFT, one reduces the general case to this statement.  
\end{rem}

\subsubsection{Quasi-regular Sasaki manifolds}

In the physical applications we are mainly interested in K\"ahler cones with \textit{quasi-regular} Sasaki bases $B$, that is, with compact Reeb vector orbits, so that $R$ generates a (locally free) \emph{compact} group of isometries $U(1)_R$
which we identify with the $R$-symmetry group which must be compact on physical grounds. 
$B$ is \emph{regular} if in addition, the $U(1)_R$ isometries act freely. We collect here some useful results.

\begin{pro}[see e.g.\! \cite{sas}]\label{symquua} $B$ is a \emph{quasi-regular} Sasakian manifold. Then
\begin{itemize}
\item[i)]  The Reeb leaves are geodesic;
\item[ii)] $B$ is a principal $U(1)$ orbi-bundle 
$B\to K$;
\item[iv)] the base $K$ is K\"ahler orbifold;
\item[v)] If the flow is \emph{regular} $K$ is a K\"ahler manifold and $B$ a principal $S^1$-bundle.
\end{itemize}
\end{pro}
\noindent Indeed, $K$ is just the symplectic quotient of $M$ with respect to the Hamiltonian  
$U(1)_R$ flow, $r^2$ being its momentum map, as eqn.\eqref{mmoamapw} shows.

\subsection{Conical special geometries}

We may introduce distinct notions of ``conical special geometry''. The weakest one corresponds to a holomorphic integrable system $X\to M$ whose K\"ahlerian base $M$
(with K\"ahler form \eqref{spekeke}) 
happens to be a metric cone. A slightly stronger notion requires $M$ to be a cone and the full set of geometric structures \textbf{S1}-\textbf{S11} to be equivariant with respect to the Euler action $\pounds_E$ (or $\pounds_\ce$). An even stronger notion requires in addition that the base of $B$ of $M$ is a \emph{quasi-regular} Sasaki manifold.\footnote{\ We do not know if the special cones in the slightly stronger sense are automatically special cones in the strongest sense.}
By a \textit{conical special geometry} (CSG) we shall mean
the strongest notion together with the regularity conditions \textbf{R1}-\textbf{R4}.

\subsubsection{Weak special cones} Locally in the good locus $M^\sharp\subset M$ we may write the special structure in terms of special complex coordinates and holomorphic prepotential $\cf(a)$. By the assumption of the existence of a SW differential, we know that there is a globally defined K\"ahler potential 
\be\label{gloK}
\Phi=i\big(a^i\bar b_i-\bar a^i b_i\big).
\ee
If $M$ is also a cone, so $r^2$ is also a globally defined K\"ahler potential, and
\be
r^2=\Phi+i h-i\bar h
\ee
for some local holomorphic function $h$ with global real part. The metric is conic iff
the vector dual to the $(0,1)$ form $\bar\partial(i\Phi+\bar h)$ is holomorphic. Locally this happens if there are constants $(c^i, d_j)$ such that
\be
b_i+d_i= \tau_{ij}(a)\big(a^j+c^j)\quad
\text{and}\quad 
r^2=i (a^i+c^i)(\bar b_i+\bar d_i)-
i(\bar a^i+\bar c^i)(b_i+d_i).
\ee
The slighter stronger notion of special cone corresponds to the case $(c^i, d_j)=0$ (so $r^2\equiv \Phi$); in other words, to get a slightly stronger special cone out of a weak one we simply absorb the constants in the definition of the periods $a^i$, $b_i$ by a shift
$\delta\lambda= c^i\, dx_i+ d_i \,dy^i$ of the SW differential. Then
\be\label{jas}
\pounds_E\!\begin{bmatrix}b \\a\end{bmatrix}= \begin{bmatrix}b \\a\end{bmatrix},
\ee
which means that locally we can find a prepotential $\cf(a)$ which is homogeneous of degree 2 in the special coordinates $a^i$. Indeed, in the slightly strong conic case, from eqn.\eqref{jas}  the Euler (Reeb) vector have the local expression
\begin{align}\label{ert}
E&=a^i \partial_{a_i}+\bar a^i\partial_{\bar a^i}, &\pounds_E\cf&=2\cf,\\
R&=i a^i \partial_{a_i}-i\bar a^i\partial_{\bar a^i}, &\pounds_R\cf&=2i\cf.
\end{align}
The vector $E$ agrees in the overlaps between two special coordinate patches if the condition $\pounds_E\cf =2\cf$
holds in one of the two patches (and then also in the other, up to a constant).
Indeed, the second eqn.\eqref{ert} implies
\be
b_i=\tau_{ij}(a)\,a^j.
\ee
Now, on the overlap
\be\label{dudu}
\begin{bmatrix}b^\prime\\
a^\prime\end{bmatrix}=\begin{bmatrix}A & B\\
C &D
\end{bmatrix}\begin{bmatrix}a\\
b\end{bmatrix},\ee
so that
\be\label{jacjac}
a^{\prime\,i}= \big(C^{ik}\tau_{kj}+{D^i}_j\big)a^j,\qquad \frac{\partial a^{\prime\,i}}{\partial a^j}= C^{ik}\tau_{k j}+D^i_{\ j},
\ee
then
\be
\ce^\prime= a^{\prime\,i}\,\partial_{a^{\prime\,i}}=a^i\,\partial_{a^i}= \ce,
\ee
and the holomorphic Euler vector $\ce$ is globally defined.
The fact that the periods $a$ are (locally defined)
holomorphic functions, fixes their transformation under the Killing-Reeb vector $R$
\be
Ra=ia,\qquad Rb=ib.
\ee 
Let $\Phi$ be the global K\"ahler potential \eqref{gloK}. 

\begin{lem} For a (slightly strong) special cone $M$ we have:
\begin{itemize}
\item[1)] The function $\Phi$ in eqn.\eqref{gloK} is the squared-norm  of the Euler (and Reeb) vector
\be
r^2=\Phi=\|E\|^2=\|R\|^2;
\ee
\item[2)] the conic relations \eqref{cococ} take the form
\be\label{kkkaswz}
\partial_{\bar \imath}\Phi = g_{\bar \imath j}\,E^{j}.
\ee
\end{itemize}
\end{lem}
\begin{rem}
Eqn.\eqref{kkkaswz} together with \textbf{R2} imply that the real (complex) analytic vector fields $E$, $R$ (resp.\! $\ce$, $\overline{\ce}$) on $M^\sharp$ extend to \emph{continuous} fields in $M$ and their $\C^\times$-action makes sense in $M$.  
\end{rem}

\begin{rem}
Let $M$ be a special cone. Then the symplectic quotient $K$ (cfr.\! {\rm \textbf{Proposition \ref{symquua}}}) is a \emph{wrong sign} projective special K\"ahler manifold \cite{Freed:1997dp}.
By definition, projective special K\"ahler manifolds are the geometries appearing  in $\cn=2$ supergravity (as contrasted to $\cn=2$ gauge theory); ``wrong sign'' means that $K$ corresponds to supergravity with an unphysical sign for the Newton constant.\end{rem}

\subsubsection{Properties of the Reeb flow/foliation: the Reeb period}
As already anticipated, the special cones $M$ which arise as Coulomb branches of $\cn=2$ SCFTs  have the property that the flow of the Reeb vector field $R$ yields a $U(1)_R$ action on $M$, i.e.\! the Reeb leaves \cite{sas} are compact in $M$. This must be so because the exponential map\footnote{\ By $\exp(2\pi t R)$ we always  mean the finite isometry of $M$ or $B$ generated by the Reeb Killing field of parameter $2\pi t$.} $t\mapsto \exp(2\pi t R)$ should implement the superconformal $U(1)_R$ symmetry which is compact in a regular SCFT. The action is automatically locally free, since the Reeb vector has constant norm 1 and hence does not vanish anywhere.\footnote{\ Of course, this geometric statement also follows from unitarity of the SCFT.} The statement that $R$ generates a locally-free $U(1)_R$ action is equivalent to the statement that its basis $B$ is a \emph{quasi-regular} Sasaki orbifold.

We now give our final definition:

\begin{defn}\label{kkkxz}
By a \textit{conical special geometry} (CSG) we mean a complex analytic integrable system $X\to M$ with SW differential $\lambda$ such that the base $M$ (with the K\"ahler metric \textbf{S10}) is a singular K\"ahler cone (satisfying \textbf{R1}-\textbf{R4}) such that the restriction to $M^\sharp$ of its Euler vector $E$ satisfies
\be
\pounds_E \lambda-\lambda= d\varrho,\qquad\quad \varrho\ \text{meromorphic,}
\ee
while its base $B$ is a quasi-regular Sasaki orbifold. 
\end{defn}

Note that this means (on $M^\sharp$)
\be
\lambda= a^i\,dx_i+b_i\,dy^i+d\varrho^\prime.
\ee

\paragraph{Exponential action of the Reeb field $R$.} Quasi-regularity requires the existence of a minimal positive real number $\alpha>0$ such that the Reeb exponential map
\be\label{llzax}
\exp(2\pi \alpha R)\colon M\to M
\ee
is the identity diffeomorphism. If the base $B$ is compact, the map $\exp(2\pi \alpha t R)$ may have fixed points in $M\setminus\{r=0\}$ only for a finite set of rational values $0 < t<1$.
\begin{defn}\label{kkzmx}
The real number $\alpha>0$ is called the 
\emph{Reeb period} of the CSG.
\end{defn}
The Reeb period is a basic invariant for a 4d $\cn=2$ SCFT.

\paragraph{Reeb period and Coulomb dimensions.} 
The chiral ring $\mathscr{R}=\Gamma(M,\co_M)$ of a CSG is
(the Fr\'echet closure of) a graded ring, the grading of a (homogeneous) global holomorphic function $h$ being given by its dimension $\Delta(h)\in \R$
\be
\pounds_\ce h=\Delta(h)\,h.
\ee
In the chiral ring of a unitary SCFT,
except for the constant function 1 which has dimension zero, 
all other dimensions $\Delta(h)$ should be strictly positive for $h$ to be regular at the tip.  
The Reeb exponential map then yields
\be
\exp(2\pi t R)\cdot h= e^{2\pi i t \Delta(h)}\,h\qquad t\in\R
\ee
and the definition \eqref{llzax}  implies
\be
\Delta(h)\in \frac{1}{\alpha}\,\mathbb{N},
\ee
so $\mathscr{R}$ is graded by the semigroup $\mathbb{N}/\alpha$.
We claim that $\alpha$ is a positive rational number $\leq 1$
(so that the dimensions of all chiral operators, but the identity, are rational numbers $\geq 1$). Indeed, the map
$e^{2\pi \alpha R}$ acts on the periods as $(a^i, b_j)\mapsto e^{2\pi i \alpha}(a^i, b_j)$. From \eqref{llzax} we deduce that the initial and final periods are
equivalent up to the action of an element $m$ of the monodromy group. Then $e^{2\pi i\alpha}\equiv \lambda$ is an eigenvalue of a monodromy, hence a root of unity. Therefore $\alpha\in\BQ_{>0}$ and $\alpha= \log(\lambda)/2\pi i$. The requirement that the curvature at the tip of the cone is non-negative forces us to use the $\boldsymbol{\log}$ determination of the logarithm, see below eqn.\eqref{kkkaqw}.

The order $r$ of $1/\alpha$ in $\BQ/\Z$ coincides with the order of the quantum monodromy $\mathbb{M}$ which is well-defined in virtue of the Kontsevitch-Soibelman wall-crossing formula, see discussion in \cite{Cecotti:2010fi,Cecotti:2014zga}.

\subsubsection{Local geometry on $K^\sharp$}

In a special cone, the discriminant locus is preserved by the $\C^\times$ action generated by the vector fields $E$ and $R$. Hence the singular locus on its base $B$ is preserved by the $R$-flow, and so is its smooth locus $B^\sharp$. By our definition, $B^\sharp$ must be Sasaki quasi-regular. Then the Hamiltonian quotient yields a K\"ahler orbifold (manifold if $B^\sharp$ is regular) $K^\sharp\equiv B^\sharp/U(1)$
of dimension $(k-1)$
(cfr. \textbf{Proposition \ref{symquua}}).
The K\"ahler potential of $K^\sharp$ is
\be
\log\Phi = \log \|R\|^2=\log \|\lambda\|^2,
\ee
where the SW differential $\lambda$ is seen as a section 
of a holomorphic line sub-bundle $\cl\subset \cv$.
The period matrix $\boldsymbol{\tau}_{ij}(a)$ is homogeneous of degree zero,
$\pounds_R\boldsymbol{\tau}_{ij}(a)=0$ so the (restriction of the) period map $\boldsymbol{\tau}|_{M^\sharp}$ factors trough $K^\sharp$. We write $p\colon K^\sharp\to Sp(2k,\Z)\backslash \mathfrak{H}_k$ for the period map so defined.

\begin{lem}\label{kkkaqwe}
The Ricci form on $K^\sharp$ is
\be\label{hhhasz}
\rho= k\,\omega+p^*\Omega\geq k\,\omega,
\ee
where $\omega=i\partial\overline{\partial}\log\Phi$ is the K\"ahler form on $K^\sharp$ and $\Omega$ the K\"ahler form of the Siegel upper half-space $\mathfrak{H}_k$ \emph{(compare {\bf Proposition \ref{jjjas}}). }
\end{lem}

\subsubsection{Special cones with smooth bases} 

The bases $B$ of the CSG interesting for the physical applications are seldom smooth (as Riemannian spaces).
However, we first consider the case of $B$ a regular smooth Sasaki manifold and then discuss what may change in our conclusion in presence of singularities. 
\smallskip

Combing \textbf{Lemma \ref{l1}} and \textbf{Proposition \ref{jjjas}} we get the inequality
\be\label{ineine}
B_{ab}\geq 2(k-1)\, \gamma_{ab}.
\ee
In view of Meyer theorem (\textbf{Theorem \ref{meymey}}) we conclude 
that for $k\equiv\dim_\C M>1$ the base
$B$ of the cone is compact, $\pi_1(B)$ is finite, and the diameter of $B$ is bounded
\be\label{jjaqwz}
d(B)\leq \pi/\sqrt{2(k-1)}.
\ee
In other words: if $M$ is a special cone of dimension $k>1$ over a \emph{smooth} base $B$,  
\be
M=\widetilde{M}/G
\ee
with $\widetilde{M}$ a \emph{simply-connected} K\"ahler cone with compact base of diameter $\leq \pi/\sqrt{2(k-1)}$ and $G$ is a \emph{finite} group acting freely.

\begin{corl}  $M$ a special K\"ahler cone over a smooth base $B$ $\Rightarrow$ $M$ is Stein.
\end{corl}
Indeed $B$ is compact, and then
\textbf{Proposition \ref{kkkz129c}} applies.

\begin{pro}\label{jkaweq} Let $X$ be a special cone whose Sasaki base $B$ is smooth and regular. Then its Hamiltonian reduction $K$ (see {\rm\textbf{Proposition \ref{symquua}}}) is a smooth, compact, and \emph{simply-connected} K\"ahler manifold.
\end{pro}

Indeed, 
$K$ is compact, smooth, and, being complex, oriented of even real dimension.
Its sectional curvatures are positive. Then the last statement follows from Synge theorem (\textbf{Theorem \ref{zxaqw}}).
In facts, since $B$ is compact by Meyer theorem, \textbf{Theorem \ref{jkkkkmzq162}}
yields an even stronger statement:

\begin{pro}\label{mmmzaqrt} A smooth base $B$ of a CSG is \emph{diffeomorphic} to $S^{2k-1}/G$ for some freely acting finite subgroup $G\subset U(k)$.
\end{pro}

\subsubsection{Relation to Fano manifolds}
Recall that a Fano manifold $X$ is a smooth projective variety whose anticanonical line bundle $-K_X$ is ample. We have: 

\begin{pro} Under the assumptions of  {\rm \textbf{Proposition \ref{jkaweq}}}, the K\"ahler manifold $K$ is a \emph{Fano projective manifold.}
\end{pro}

\begin{proof} $K$ is smooth by assumption and compact by Meyers theorem. 
By \textbf{Lemma \ref{kkkaqwe}} the Ricci form is $\geq (\dim K+1) \omega$, $\omega$ being the (positive) K\"ahler form. Thus the anti-canonical is ample, hence $K$ is projective (by Kodaira embedding theorem \cite{GH}) and Fano.
\end{proof}

\begin{pro} Under the assumptions of  {\rm \textbf{Proposition \ref{jkaweq}}} and $k>1$:
\begin{itemize}
\item[i)] the universal cover $\widetilde{B}$ of the base manifold $B$ is \emph{homeomorphic} to $S^{2k-1}$;
\item[ii)] the rank of the Picard group (\emph{Picard number}) of the Fano manifold $K$ is $1$:
\be
\varrho(K)\equiv \mathrm{rank}\, \mathrm{Pic}(K)=1;
\ee
\item[iii)] the Hodge diamond of $K$ is
$h^{p,q}(K)=\delta^{p,q}$.
\end{itemize}
\end{pro}
\begin{proof}
$\widetilde{B}$ is compact and simply connected by Meyers' theorem. From \textbf{Lemma \ref{l1}} and  eqn.\eqref{cllae} we see that the eigenvalues of the curvature operators are bounded below by 1. Using \textbf{Theorem \ref{jkkkkmzq162}} we get \textit{i)}. Then $B\equiv \widetilde{B}/\pi_1(B)$ has real cohomology
\be
H^q(B,\R)=\begin{cases}\R & \text{for }q=0, 2k-1\\
0 &\text{otherwise.}
\end{cases}
\ee 
Since $B$ is Sasaki-regular, it is a principal $U(1)$-bundle over $K$
\be
\begin{gathered}
\xymatrix{S^1 \ar[r] &B\ar[d]\\
& K}
\end{gathered}\qquad \text{and $K$ is simply-connected (\textbf{Proposition \ref{jkaweq}}}).
\ee
To this fibration we apply the Leray spectral sequence of de Rham cohomology.\footnote{\ The computation is world-for-world identical to \textsc{Example} 14.22 in \cite{bott}.} We get
\be
H^q(K)= \begin{cases} \R & q=0,2,\cdots, 2(k-1)\\
0 & \text{otherwise.}
\end{cases}
\ee 
This shows \textit{iii)}. To get \textit{ii)}
note that the exponential exact sequence yields the implication 
\be
H^1(K)=0\quad\Longrightarrow\quad\mathrm{Pic}(K)\cong H^2(K,\Z),
\ee
 and then \textit{ii)} follows from \textit{iii)}.
\end{proof}

We have a much stronger statement:
 \begin{pro}\label{ssster} Under the assumptions of  {\rm \textbf{Proposition \ref{jkaweq}}}, the Fano manifold $K\cong \mathbb{P}^{k-1}$ and the period map $p$ is constant.
\end{pro}

Before proving this \textbf{Proposition}
we define the \emph{index} $\iota(F)$ of a Fano manifold $F$ whose canonical divisor we write $K_F$. The index $\iota(F)$ is the largest positive integer such that $-K_F/\iota(F)$ is a (ample) divisor \cite{fano}.
Under the assumptions of \textbf{Proposition \ref{jkaweq}} we have a line bundle $\cl\to K$ such that $\omega$ is its Chern class (up to normalization).
Thus from \textbf{Lemma \ref{kkkaqwe}} we get:
\begin{lem}
Under the assumptions of  {\rm \textbf{Proposition \ref{jkaweq}}} we have
$\iota(K)\geq \dim K+1$ with equality iff the period map $\boldsymbol{\tau}$ is constant.
\end{lem}
\begin{proof}
The Picard number of $K$ is 1, and we have
\be
[p^*\Omega]= \delta [\omega]\quad \text{for some }\delta\in\mathbb{Q}, \ \text{with }
\delta\geq 0,\ \text{and }=\ \text{only if }p
\ \text{is constant},
\ee
so
\be
\iota(K)= \dim K+1+\delta\geq \dim K+1
\ee
with equality iff $p=$constant.
\end{proof}

Then \textbf{Proposition \ref{ssster}} follows from this \textbf{Lemma} and a basic fact from the theory of Fano varieties: the index of a Fano manifold $F$
cannot exceed $\dim F+1$, and if $\iota(F)=\dim F+1$ then $F\cong \mathbb{P}^{\dim F}$ \cite{fano}.

\begin{corl} Suppose we have a special cone of dimension $k>1$. If its base $B$ is a smooth, complete, regular Sasaki manifold, the period map $\boldsymbol{\tau}$ is constant and $M=\C^k$ with the a K\"ahler metric. {\rm That is: if the geometry is regular except for the singularity at the vertex of the cone, then the SCFT is free, as expected.}
\end{corl}
\begin{rem} The result does not hold for $k=1$. All $k=1$ Minahan-Nemeshanski geometries satisfy the other assumptions, yet the SCFT is not free.
The essential point is that for $k>1$ the assumptions imply regularity in codimension 1, whereas the tip of the cone is automatically a codimension 1 singularity for $k=1$.  
\end{rem}

\begin{rem}
The last result also follows from rigidity. If $K$ is smooth and compact, the monodromy group is trivial, hence (by uniqueness) the period map should be the constant one. 
\end{rem}

\subsubsection{Properties of non-smooth CSG}\label{quasiregularnice}

In the previous sub-section we considered the case in which the geometry of the special cone is totally regular away from its vertex. We got only the (known) $k=1$ geometries and free theories for $k>1$. In all other cases there are singularities in $M\setminus\{0\}$
of the kind consistent with \textbf{R1}-\textbf{R4}. 

Since $B^\sharp$ is only quasi-regular, we have a locus $N\subset B^\sharp$ on which $U(1)_R$ does not act freely.
We write $\mathring{B}=B^\sharp\setminus N$. For $k>1$, 
$\mathring{K}=\mathring{B}/U(1)$ is a open dense K\"ahler sub-manifold of the singular space $K$, in facts $\mathring{K}=K\setminus D$, for some divisor $D$.  
Since
the period matrix $\boldsymbol{\tau}_{ij}$ is homogeneous of degree zero,
the period map $\boldsymbol{\tau}$ factors through $K$. We consider its restriction to the regular subspace, $\boldsymbol{\tau}\colon\mathring{K}\to Sp(2k,\Z)\backslash \mathfrak{H}_k$.
All the considerations in \S.\,\ref{jjzqwe} apply to this map; in particular, we have a monodromy representation
$\mathring{m}\colon\pi_1(\mathring{K})\to Sp(2k,\Z)$. Along an irreducible component $D_i$ of $D$ such that the monodromy is semi-simple, we may extend $\boldsymbol{\tau}$ holomorphically and $K$ has along $D_i$ only a $\Z_r$ cyclic orbifold singularity:
indeed, the geometry becomes smooth after the local base change $U\to U_r$, cfr.\! \S.\,\ref{jjzqwe}. Otherwise, the monodromy along $D_i$
satisfies $(m^r-1)^2=0$; the base change $U\to U_r$ sets the local geometry in the form discussed around eqn.\eqref{kkkaqznn}. Locally $U_r$ has the form $\Delta^*\times \Delta^{k-2}$ and we may choose local coordinates so that the K\"ahler form take the form (here $|q|<1$)
\be
i\partial\overline{\partial}\log\Phi\approx i\partial\overline{\partial}\log\!\left( |q|^2\big(1-\log|q|\big)+\sum_{i=2}^{k-1}|x_i|^2\right).
\ee
We may modify this metric to
\be\label{newkak}
 i\partial\overline{\partial}\log\!\left( |q|^2\Big(1-\log|q|\cdot f_C(-\log|q/\epsilon|^2)\Big)+\sum_{i=2}^{k-1}|x_i|^2\right)
\ee
where $f_C(x)\colon \R\to\R$ is a smooth function which equals 1 for  $x\leq0$ and
has all derivatives vanishing as $x\to+\infty$.
The new K\"ahler form \eqref{newkak} is smooth in $U_r$ and agrees with the original one for $|q|\geq \epsilon$;
one checks that one may choose the local deformation so that the metric and the curvatures remain positive. Of course, it is no longer a special K\"ahler metric. The point we wish to argue is that $K$ admits a non-special \emph{orbifold} K\"ahler metric, with only $\Z_r$ orbifold singularities, whose Ricci form satisfies a bound of the form $R_{i\bar\jmath}\geq k g_{i\bar\jmath}$. The same conclusion applies to $M$ (the Ricci tensor being non-negative in this case); then the basis $B$ is also regular except for cyclic orbifolds singularities. We may apply to $B$  the orbifold  version  of Meyers \cite{mey1,mey2}): a metrically complete Riemannian orbifold $X$, whose Ricci curvature satisfies $R\geq \lambda^2 g$, is compact with a diameter $d(X)\leq \pi/\lambda$. 
Once we are assured that $B$, albeit singular, is compact we conclude that $M$ is Stein.
Moreover, from the Synge theorem for orbifolds
\cite{mey2}, we see that $K$, albeit no smooth, is still compact and simply-connected.

$\cl$ is now a line orbi-bundle which is still positive by eqn.\eqref{hhhasz},  Kodaira-Baily embedding theorem (\textbf{Theorem \ref{kkkxza}})
guarantees that $K$ is a normal projective algebraic variety with at most cyclic orbifold singularities. The anticanonical divisor $-K_K$ is now a Weyl divisor which is \emph{ample} as a $\mathbb{Q}$-Cartier divisor.
A normal projective variety with ample anti-canonical $\mathbb{Q}$-divisor having only cyclic singularities is a \emph{log}-Fano variety, to be defined momentarily. 
The Picard number $\varrho$ is $1$ as in the smooth case. Indeed, $B$ is still diffeomorphic to a generalized Lens space $S^{2k-1}/G$ (see \textbf{Proposition \ref{mmmzaqrt}}) and the finite group $G$ centralizes the Reeb $U(1)$ action, so that the orbifold $K$ is homeomorphic to a finite quotient of $\mathbb{P}^{k-1}$, and hence $\mathrm{rank}\, H^2(K,\Z)=1$.
The index $\iota(F)$ of a \emph{log}-Fano variety $F$ is the greatest positive \emph{rational} such that
$-K_F=\iota(F)\,H$ for some Cartier divisor $H$ (called the \emph{fundamental divisor}). 
By a theorem of Shokurov \cite{fano}, the index of a \emph{log}-Fano $F$ is at most $\dim F+1$. This is not necessarily in contradiction with eqn.\eqref{hhhasz} since now $\cl$ is just a $\mathbb{Q}$-divisor. Let $\sigma$ be the smaller positive rational number so that $\cl^\sigma$ is a \emph{bona fide} line bundle. Then
\be\label{uuuza}
\iota(K)= \frac{\dim K+1+\delta}{\sigma}\geq \frac{\dim K+1}{\sigma}.
\ee
The theorem of Shokurov then yields
$1\leq \sigma$ with equality if and only if our special cone $M$ is $\C^k$ with the flat metric, i.e.\! we are talking about the free SCFT. Shokurov inequality is one of the many geometrical properties underlying the physical fact that saturating the unitarity bound means to be free.
We see that the main difference between the smooth case of the previous subsection and the general is that in the second case $G$ acts properly discontinuously on $S^{2k-1}$ but not freely.  
The previous argument giving $\varrho=\mathrm{rank}\,\mathrm{Pic}(K)=1$ extends to this (slightly) more general case
(indeed, $K$ is simply-connected and topologically a finite quotient of $\mathbb{P}^{k-1}$).

\subsubsection{ \emph{log}-Fano varieties}  
A singular Fano variety whose only singularities are cyclic orbifold ones is 
a \emph{log}-Fano variety.
We recall the relevant definitions \cite{fano}:

\begin{defn} \textbf{1)} A Weil divisor $D$ is called \emph{$\mathbb{Q}$-Cartier} if there exists $m\in\mathbb{N}$ such that $mD$ is a Cartier divisor. \textbf{2)} A normal variety $X$ is \emph{$\mathbb{Q}$-factorial} if all Weil divisors are $\mathbb{Q}$-Cartier. 
(Note that if $X$ is normal, the canonical divisor $K_X$ is well-defined in the Weil sense\footnote{\ Indeed, $X$ is non-singular in codimension 1, so a canonical divisor over the smooth open set $X_\text{smooth}\subset X$ can be extended as a Weil divisor to $X$.}).
\textbf{3)} A normal variety $X$ 
is said to have
\emph{terminal,} \emph{canonical,} \emph{log terminal,} or \emph{log canonical singularities} iff its
canonical divisor $K_X$ is $\mathbb{Q}$-Cartier and there exists a projective birational morphism (whose exceptional locus has normal crossing\footnote{\ Such resolutions exists by Hironaka theorem \cite{hiro}.}) $f\colon V\to X$ from a smooth variety $V$ such that: 
\be
K_V=f^*K_X+\sum_i a_i\, E_i,\qquad\quad \left|\text{
\begin{minipage}{190pt}
\begin{scriptsize}where $E_i$ are the \emph{prime exceptional divisors,}
i.e.\!\! the irriducibile components of the exceptional locus of $f$
of codimension 1\end{scriptsize}\end{minipage}}\right.
\ee
and, respectively:\vskip-16pt
\be
\begin{array}{llcll}
a_i>0 &\text{terminal,}
&\phantom{mmm}&a_i\geq 0 &\text{canonical,}\\
a_i>-1 &\text{log-terminal,} & &
a_i\geq -1 &\text{log-canonical.}\end{array}
\ee
\textbf{4)} A normal projective variety $X$ with only log-terminal singularities whose anti-canonical divisor $-K_X$ is an ample $\mathbb{Q}$-Cartier divisor is called a \emph{log-Fano variety.} \textbf{5)} The greatest rational $\iota(X)>0$ such that $-K_X=\iota(X)\,H$ for some (ample) Cartier divisor $H$ is called the \emph{index} of $X$ and $H$ is called a \textit{fundamental divisor.} \textbf{6)} The \textit{degree} of a Fano variety $X$ is
the self-intersection index $d(X)=H^{\dim X}$.   
\end{defn}

Comparing our discussion in \S.\ref{quasiregularnice} with the above definitions, we conclude (cfr.\!  \textbf{Proposition 7.5.33} of \cite{sas}):
\begin{corl}
The symplectic quotient $K=M/\!\!/U(1)$ of a CSG is a \emph{log}-Fano variety
with Picard number one and  Hodge diamond $h^{p,q}=\delta^{p,q}$. Moreover, in the smooth sense, $K\cong \mathbb{P}^{k-1}/G$ for some finite group $G$.
\end{corl}

%
%

Thus a CSG $M$ is a complex cone over a normal projective variety of a very restricted kind; in particular, $M$ is affine and a \emph{quasi-cone} in the sense of \S.3.1.4 of \cite{dolga}.

\begin{exe} A large class of examples of log-Fano varieties with the properties in the \textbf{Corollary} is given by the weighted projective spaces (WPS) \cite{dolga}. Il $\boldsymbol{w}=(w_1,w_2,\cdots, w_k)\in \mathbb{N}^k$ is a system of weights, $\mathbb{P}(\boldsymbol{w})$ is
\be
\big(\C^k\setminus\{0\}\big)\big/\sim\quad \text{where}\quad(z_1,\cdots,z_k)\sim (\lambda^{w_1}z_1,\cdots, \lambda^{w_k}z_k)\ \forall\;\lambda\in\C^\times.
\ee
All such spaces are log-Fano with just cyclic orbifold singularities, simply-connected, have Hodge numbers $h^{p,q}=\delta^{p,q}$, and are isomorphic to $\mathbb{P}^{k-1}/G$ for some finite Abelian group $G$.
More generally, all quasi-smooth complete intersections in weigthed projective space of degree $d<\sum_i w_i$ and dimension at least 3 are log-Fano, simply-connected, and have $\varrho=1$ \cite{dolga} but typically their Hodge numbers $h^{p,q}\not\equiv \delta^{p,q}$. 
A slightly more general class of examples of projective varieties satisfying all conditions above is given by the \emph{fake} weighted projective spaces \cite{toric}. A fake WPS is canonically the quotient of a WPS by a finite Abelian group. One shows that a log-Fano with Picard number one which is also \emph{toric} is automatically a fake WPS.  
\end{exe}

\begin{exe} Let $G\subset PGL(k,\C)$ be a finite subgroup; then $\mathbb{P}^{k-1}/G$ is a ($\mathbb{Q}$-factorial) \emph{log}-Fano variety with $\varrho=1$ \cite{fano}. By \S.\,\ref{quasiregularnice}, $K$ is always \emph{diffeomorphic} to such a variety.
\end{exe}

\section{The chiral ring $\mathscr{R}$ of a $\cn=2$ SCFT}

\subsection{General considerations}

In the previous section we reviewed the general properties of the CSG describing the Coulomb branch $M$ of a 4d $\cn=2$ SCFT in which we have switched off all mass and relevant deformations so that 
the conformal symmetry  is only spontaneously broken  by the non-zero expectation values of some chiral operators $\langle \phi_i\rangle_x\neq0$
in the susy vacuum $x\in M$.
Although the K\"ahler metric on $M$ has singularities in codimension 1, due to additional degrees of freedom becoming massless or symmetries getting restored, these singularities are assumed to be mild and the underlying complex space $M$ is regular\footnote{\ Regularity of the underlying analytic space $M$ in codimension 1 already follows from the fact that we are free to assume it to be normal.}. The analysis of local special geometry on the open, everywhere dense, regular set $M^\sharp= M\setminus \cd$ only determines the special geometry up to ``birational equivalence'' since the singular fibers over the discriminant locus $\cd$ may be resolved in different ways;
different consistent resolutions give non-isomorphic CSG which correspond to distinct physical models. For instance, $SU(2)$ SQCD with $N_f=4$ and $SU(2)$ $\cn=2^*$ have ``birational equivalent'' CGS whose total fiber spaces $X$ are certainly different in codimension 1; other examples of such ``birational'' pairs in rank one are given by Minahan-Nemeshansky $E_r$ SCFTs and Argyres-Wittig models \cite{Argyres:2007tq} with the same Coulomb branch dimension but smaller flavor groups.
The basic invariant of an equivalence class of SCFTs is the chiral ring $\mathscr{R}=\Gamma(M,\co_M)$.

The  common lore is that $\mathscr{R}$ is (the Fr\'echet completion of) a free graded polynomial ring $\mathscr{R}=\C[u_1,\cdots, u_k]$ whose grading is given by the action of the holomorphic Euler field,
$\pounds_\ce u_i =\Delta_i u_i$. The set of rational number $\{\Delta_1,\cdots, \Delta_k\}$ are the Coulomb dimensions.
This lore may be equivalently stated by saying that the Hamiltonian reduction $K=M/\!\!/U(1)_R$ is (birational to) the weighted projective space (WPS) $\mathbb{P}(\Delta_1,\cdots,\Delta_k)$.
In the previous section we deduced from special geometry several detailed properties of $M$ and $K$ which are consistent with this idea: $M$ is affine while $K$ is a normal projective log-Fano variety, which is simply-connected and has Hodge numbers $h^{p,q}=\delta^{p,q}$. These are quite restrictive requirements on an algebraic variety, and they hold automatically for all WPS; indeed they almost characterize such spaces. For instance, if one could argue that $K$ (or $M$) is a toric variety then these properties would imply that $K$ is a fake WPS \cite{toric}, and in particular a finite Abelian quotient of a WPS.

We sketch some further arguments providing additional evidence  that $K$ is (birational to) a finite quotient of a WPS (possibly non-Abelian).
The monodromy and $SL_2$-orbit theorems describe the asymptotical behavior of the special K\"ahler metric on $M$ as we approach the discriminant $\cd$ (at least for $\cd$ \textsc{snc}). As argued in \S.\,\ref{quasiregularnice}, we can modify the metric to a K\"ahler metric which is smooth in $M\setminus \{0\}$, agrees with the original one outside a tubular neighborhood of $\cd$ of size $\epsilon$, is conical and K\"ahler with non-negative curvatures. Of course the new metric is no longer special K\"ahler and cannot be written in terms of a holomorphic period map; nevertheless it is a nice regular metric on the complex manifold $M\setminus\{0\}$ with all the good properties. Working in the smooth category, and using the sphere theorems (cfr.\! \S.\ref{uuzaqf}), we conclude that the Riemannian base $B$ of the cone is
diffeomorphic to a space form $S^{2k-1}/G$ for some finite $G$ acting freely, hence to $S^{2k-1}$ up to a finite cover. The finite cover $M^\prime$ of $M$ is a Riemannian cone over $S^{2k-1}$ is diffeomorphic to $\R^{2k}\cong\C^k$. Moreover we know that $M^\prime$ is Stein, in fact a normal affine variety. If we could say that $M^\prime$, being an affine variety diffeomorphic to $\C^k$ is biholomorphically equivalent to $\C^k$, we would be done.  Unfortunately, for $k\geq3$ there are several examples of \textit{exotic $\C^k$}, that is, affine algebraic spaces which are diffeomorphic to $\C^k$ but not biholomorphic to $\C^k$ \cite{exotic}. However, no example is known of an affine variety over $\C$ which is  diffeomorphic to $\C^k$ but not birational to $\C^k$; the conjecture which states that there are none being still open \cite{exotic}. Hence, assuming the conjecture to be true, we conclude that $M^\prime$ is \emph{birational} to $\C^k$. On the other hand, $M^\prime$ is not just a complex space diffeomorphic to $\C^k$, it has additional properties as (e.g.) a holomorphic Euler vector $\ce$ whose spectrum is 
an additive semigroup
$\mathsf{spec}(\ce)\subset \{0\}\cup \mathbb{Q}_{\geq 1}$, and the condition of non-negative bisectional curvatures.
Were not for the singularity at the tip of the cone, this last property by itself would guarantee\footnote{\ The following theorem (a special case of Yau's conjecture \cite{yaucon1,yaucon2}) holds:\begin{thms}[Chau-Tam \cite{chau}]
$N$ a complete non-compact K\"ahler manifold with non-negative and bounded bi-sectional curvature and maximal volume growth. $N$ is biholomorphic to $\C^n$.\end{thms}
\noindent $M$ with the modified metric satisfy non-negativity of bisectional curvatures and maximal volume growth, but the curvatures (if non zero) blow up as $r\to0$, and geodesic completeness fails at the vertex.} that $M^\prime$ is analytically isomorphic to $\C^k$.

There is still the question of the relation between the Coulomb branch $M$ and its finite cover $M^\prime$. Since $M\cong M^\prime/G$, the statement $M=M^\prime$ is equivalent to $G=1$, that is (for $k>1$)
that after smoothing the discriminant locus by local surgery $M$ becomes simply-connected. This holds in the sense that the monodromy along a cycle not associated to the discriminant is trivial (since the period map factors through $K$).

All the above results and considerations, while short of a full mathematical proof, provide convincing evidence for the general expectation that $M$ is birational to $\C^k$, so for the purposes of our ``birational classification'' we may take $M$ to be just $\C^k$.
$\C^\times$ acts on $M$ through the exponential action of the Euler field $\ce$.
Again, an action of $\C^\times$ on $\C^k$ is guaranteed to be linear only for $k\leq 3$ \cite{exotic}; in larger dimensions exotic actions do exist. However, in the present case the action is linear in the smooth sense; under the assumption that the complex structure is the standard one for $\C^k$, we have a linear action, that is, the chiral ring has the form
\be\label{rrring}
\mathscr{R}= \C[u_1,\cdots, u_k].
\ee
where the $u_i$ can be chosen to be eigenfunctions of $\pounds_\ce$, i.e.\! homogeneous of a certain degree $\Delta_i$. Then $K$ is the WPS $\mathbb{P}(\Delta_i)$.
Thus, under some mild regularity conditions and modulo the plausible conjecture that a smooth affine variety $M$ whose underlying $C^\infty$-manifold is diffeomorphic to $\C^n$ is actually bimeromorphic to $\C^n$, we conclude that the common lore is correct up to birational equivalence (and, possibly, finite covers).
For the rest of this paper we shall assume this to be the case without further discussion.

\begin{rem} It is interesting to compare the above discussion with the one in ref.\cite{Argyres:2017tmj}. The previous argument is based on the idea that the unitary bound $\Delta(u)>1$ implies that we may smooth out the metric while preserving the curvature inequalities. In ref.\cite{Argyres:2017tmj} they find that if the chiral ring is not free there must be a local parameter $u$ with $\Delta(u)<1$, so if there exists such an $u$ the argument leading to
\eqref{rrring} cannot be invoked. Clearly $u$ should not be part of the physical chiral ring $\mathscr{R}_\text{ph}$. However, in a Stein space the local parameter may always be chosen to be a global holomorphic function, so $u\in \Gamma(M,\co_M)$, which is our definition of the ``chiral ring''.
Going through the examples of ref.\cite{Argyres:2017tmj} one sees that their Coulomb branches are \textit{non-normal analytic spaces,} that is, they have the same underlying topological space $M_\text{top}$
by a different structure sheaf $\co_\text{ph}\neq\co_M$ whose stalks are not integrally closed.
In this case the physical chiral ring
$\mathscr{R}_\text{ph}=\Gamma(M,\co_\text{ph})$ may be identified with a subring of the geometrical one $\mathscr{R}=\Gamma(M,\co_M)$. (See also the third \emph{caveat} in \S.\ref{kkkaz1204} and \S.\ref{lllxcv3} below).
\end{rem}

To fully determine the graded  rings $\mathscr{R}$ which may arise as Coulomb chiral rings of a CSG, it remains to determine the allowed dimension $k$-tuples $\{\Delta_1,\dots, \Delta_k\}$.
We shall address this question in section \ref{jjjazsqw3}. Before going to that, we consider the simplest possible CSG just to increase the list of explicit examples on which we may
test the general ideas. 
     
\subsection{The simplest CSG:
constant period maps}\label{connstantmaoa}
For $k=1$ all CSG have constant period map $\boldsymbol{\tau}=\mathrm{const}$. For $k>1$ constant period maps are still a possibility, although very special. For instance, this happens in a Lagrangian SCFT in the limit of extreme weak coupling (classical limit). This is good enough to compute the Coulomb dimensions $\Delta_i$, since they are protected by a non-renormalization theorem (or just by the fact that they are $U(1)_R$ characters and $U(1)_R$ is not anomalous).

The K\"ahler metric is locally flat, i.e.\! has the local form \be\label{jjjasqwe}
ds^2=\mathrm{Im}\,\boldsymbol{\tau}_{ij}\, da^i d\bar a^j.
\ee
 Then $M=\C^k/\cg$
where $\cg\subset \mathrm{Fix}(\boldsymbol{\tau})\subset Sp(2k,\Z)$ is a subgroup of the isotropy group $\mathrm{Fix}(\boldsymbol{\tau})$ of $\boldsymbol{\tau}$. $\mathrm{Fix}(\boldsymbol{\tau})$, being discrete and compact, is finite. Then 
\be
\C[a_1,\cdots, a_k]^\cg\subset \mathscr{R}.
\ee
 We stress that (in general) we have just an \textit{inclusion} not an equality: in \textbf{Example \ref{reqasx}} the three rank-1 Argyres-Douglas SCFT (respectively of types $A_2$, $A_3$ and $D_4$) have 
\be
\C[a]^{\Z_{m+1}}=\C[u^m]\subsetneq \C[u]\equiv\mathscr{R},
\ee with $m=5,4,2$ respectively. Note, however, that we have still 
\be
K=\mathsf{Proj}\,\mathscr{R}\cong \mathsf{Proj}\,\C[a]^\cg
\ee 
in these cases. The point is that $K$ does not fix uniquely the chiral ring, unless we specify its orbifold structure, see the discussion in \cite{reid}.
We shall study the orbifold behavior  in more generality in the next section.
Here we limit ourselves to the very simplest case in which 
\be
\mathscr{R}\cong \C[a_1,\cdots,a_k]^\cg
\ee
 as graded $\C$-algebras.
The Shephard-Todd-Chevalley theorem \cite{ST1,ST2} states that  $\C[a_1,\cdots,a_k]^\cg$ is a graded polynomial ring if and only if $\cg$ is a finite (complex) reflection group\footnote{\ A degree-$k$ \textit{reflection group} $\cg$ is a concrete group of $k\times k$ matrices generated by reflections, i.e.\! by matrices $g\in \cg$ such that $\dim\mathrm{coker}(g-1)=1$. In particular, a reflection group comes with a defining representation $V$, whose dimension $k$ is called the \emph{degree} of the reflection group.}.
In such a case the Coulomb branch dimensions $\Delta_i$ coincide with the degrees of the fundamental invariants $u_i$ of the group $\cg$
such that $\C[a_1,\cdots,a_k]^\cg=\C[u_1,\cdots, u_k]$. 
However, not all finite reflection groups can appear in special geometry, since only subgroups of the Siegel modular group, $\cg\subset Sp(2k,\Z)$, are consistent with Dirac quantization.\footnote{\ More generally, subgroups of duality-frame groups $S(\Omega)_\Z$.} In \S.\ref{cyclicee} we review the well known:

\begin{fact} Let $\cg\subset Sp(2k,\Z)$ be a \emph{finite} subgroup of the Siegel modular group. The subset $\mathrm{Fix}(\cg)\subset\mathfrak{H}_k$ of points in the Siegel upper half-space fixed by $\cg$ is a \emph{non empty,} \emph{connected} complex submanifold. 
\end{fact}

Thus, if $\cg$ is a finite reflection subgroup of $Sp(2k,\Z)$ there is at least one period $\boldsymbol{\tau}$ with $\boldsymbol{\tau}\subset\mathrm{Fix}(\cg)$ and the quotient by $\cg$  of $\C^k$ with the flat metric \eqref{jjjasqwe} makes sense (as a orbifold). All such metrics are obtained by continuous deformations of a reference one, and so they belong to a unique deformation-type. The dimension of the fixed locus, $\mathsf{d}= \dim_\C \mathrm{Fix}(\cg)$, is given by eqn.\eqref{kkkasqw}.
\medskip

The finite complex reflection groups (and their invariants $u_i$) have been fully classified by Shephard-Todd \cite{ST1,ST3,refl1}. They are direct products of irreducible finite reflection groups. 
The list of irreducible finite complex reflection groups is\footnote{\ In the classification one does not distinguish a group and its complex conjugate since the two are conjugate in $GL(V)$.} \cite{ST1,ST3,refl1}:
\begin{itemize}
\item[1)] The cyclic groups $\Z_m$ in degree 1;
\item[1)] the symmetric group $S_{k+1}$ in the degree $k$ representation (i.e.\! the Weyl group of $A_k$);
\item[2)] the groups $G(m,d,k)$ where $k>1$ is the degree, $m>1$ and $d\mid m$;
\item[3)] 34 sporadic groups denoted as $G_4, G_5,\cdots, G_{37}$ of degrees $\leq 8$.
\end{itemize}
For our purposes, we need to classify the embeddings
of the Sphephard-Todd (ST) groups of rank-$k$ into the Siegel modular group $Sp(2k,\Z)$ modulo conjugacy. We shall say that a ST group is \emph{modular} iff it has at least one such embedding.
A degree-$k$ ST group $\cg$ which preserves a lattice $L\subset \C^k$ 
is called \emph{crystallographic} \cite{refl1,refl2}; clearly a modular ST group is in particular crystallographic. 

Let $\BK$ be the character field of an irreducible reflection group\footnote{\ Note that if $\BK$ is the character field of an irreducible finite complex reflection group $\cg$, then $\mathsf{Gal}(\BK/\BQ)$ is Abelian. Moreover, if $\cg$ is crystallographic, its class number is  $h(\BK)=1$.} and $\fO$ its ring of integers. One shows that $\cg\subset GL(k,\fO)$ \cite{ST3}, so $\cg$ is crystallographic iff $\BK$ is either $\BQ$ or an \emph{imaginary} quadratic field $\BQ[\sqrt{-d}]$. The crystallographic Shephard-Todd groups are listed in table \ref{fieldch} together with their character field $\BK$. The groups with $\BK=\mathbb{Q}$ are just the irreducible real crystallographic groups, namely the Weyl groups of simple Lie algebras. Rank-$k$ Weyl groups are obviously subgroups of $Sp(2k,\Z)$.
Let us consider the case 
$\BK=\BQ[\sqrt{-d}]$; we choose an embedding $\cg\hookrightarrow GL(k,\fO)$; 
since $\cg$ is finite and absolutely irreducible, it preserves an Hermitian form with coefficients in $\fO$,
of the form $H_{ij} \psi^i\bar \psi^j$ where\footnote{\ $\zeta\in\fO$ is an integer of $\BQ[\sqrt{-d}]$ such that $\{1,\zeta\}$ is a basis of the ring $\fO$ as a $\Z$-module. Clearly $\bar\zeta\neq\zeta$.} $\psi^i\equiv x^i+\zeta y^i\in \fO^k$, $(x^i,y^j)\in\Z^{2k}$, and 
$\bar H_{ij}=H_{ji}$. Hence it preserves the skew-symmetric form with rational coefficients
\be
\frac{1}{\zeta-\bar\zeta}\,H_{ij}\,\psi^i\wedge \bar \psi^j.
\ee
Clearing denominators, we get a non-degenerate integral skew-symmetric form $\Omega$ on $\Z^{2k}$ which is preserved by $\cg$ whose entries have no non-trivial common factor. The given embedding $\cg\hookrightarrow GL(k,\fO)$
induces an embedding $\cg\hookrightarrow Sp(2k,\Z)$
iff $\det \Omega=1$; in particular, we have the necessary condition
\be
\det H\ \text{is a $k$-th power in $\BQ$}.
\ee  
 
\begin{exe}\label{e6} In the \textbf{Introduction} we claimed that there is no rank-2 CSG with principal polarization and Coulomb dimensions $\{\Delta_1,\Delta_2\}=\{12,8\}$. However, the crystallographic ST group $G_8$ has degrees\footnote{\ In fact $\{12,8\}$ are even regular degrees in the sense of Springer theory.} $\{12,8\}$ so $\C^2/G_8$ is potentially a candidate counter-example to our claim. The only way out is to show that $G_8\not\subset Sp(4,\Z)$, so that this particular CSG is consistent only for suitable non-principal polarizations. Up to conjugacy, there is a unique embedding $G_8\hookrightarrow GL(2,\Z[i])$ \cite{feit1}, 
generated by two reflections
$r_1$, $r_2$, with invariant 
Hermitian form $H$:
\be
r_1=\begin{bmatrix} i & i\\
0 & 1\end{bmatrix},\qquad
r_2=\begin{bmatrix} 1 & 0\\
-1 & i\end{bmatrix},\qquad
H=\begin{bmatrix} 2 & 1-i\\
1+i & 2\end{bmatrix}
\ee
$\det H=2$ is not a square in $\mathbb{Q}$,
and hence the polarization is non-principal with \textit{charge multipliers} $(e_1,e_2)=(1,2)$ (cfr.\! eqn.\eqref{chargemult}).
Thus $G_8$ is no contradiction to our claim.
We shall return to this in section 5.
\end{exe}

For the groups $G(m,d,k)$ with $m=3,4,6$ the usual monomial basis is both defined over $\fO$ and orthonormal, so $\det H=1$
and they are all subgroups of  $Sp(2k,\Z)$. However, for special values of $(m,d,k)$ we may have more than one inequivalent embeddings in $Sp(2k,\Z)$:

\begin{lem}[\!\!\cite{feit1}] Consider $G(m,d,k)$ with $m=3,4,6$ and $(m,d,k)\neq (m,m,2)$; there is a single conjugacy class of embeddings $G(m,d,k)\hookrightarrow Sp(2k,\Z)$ except for $m= 3,4\equiv p^s$ ($p=3,2$ is a prime), $1<d\mid p^s$, and $p\mid k$ in which case we have $p+1$ inequivalent embeddings.
\end{lem}

\begin{table}
\begin{scriptsize}
$$
\begin{array}{c||c|cccccccccc}\hline\hline
\mathbb{Q} & \text{ST names} &\Z_2 & G(2,1,k) & G(2,2,k) & G(6,6,2) & S_{k+1} & G_{28} & G_{35} & G_{36} & G_{37}\\
& \text{Coxeter} & A_1 & B_k & D_k & G_2 & A_k & F_4 & E_6 & E_7 & E_8\\\hline
\mathbb{Q}(\zeta_3) & \text{ST names} & \Z_3 & \Z_6 & G(m,d,k) & G_4 & G_5 &  G_{25} & G_{26} & G_{32} & G_{33} & G_{34}\\
& \text{notes} & & & m=3,6;\ d\mid m &  &\\
&&&& (m,d,k)\neq (m,m,2) &&
 \\\hline
\mathbb{Q}(i) & \text{ST names} & \Z_4 & 
G(4,d,k) & G_8 & G_{29} & G_{31}\\
& \text{notes} && d\mid 4&\\
&&& (4,d,k)\neq (4,4,2) & \\\hline
\mathbb{Q}(\sqrt{-2}) & \text{ST names} &G_{12}\\
&&\\\hline
\mathbb{Q}(\sqrt{-7}) & \text{ST names} & G_{24}\\
&&\\\hline\hline
\end{array}
$$
\end{scriptsize}
\caption{\label{fieldch}The Shephard-Todd irriducible complex cristallographic groups; $\zeta_3$ stands for a primitive third root of 1.
The excluded cases are $G(3,3,2)\equiv \mathrm{Weyl}(A_2)$, $G(6,6,2)\equiv \mathrm{Weyl}(G_2)$, and $G(4,4,2)$, which is conjugate in $U(2)$ to $\mathrm{Weyl}(B_2)$.}
\end{table}

\begin{rem}
This statement is the complex analogue of the usual GSO projection of string theory \cite{pol}. Indeed, consider the \emph{real} reflection group $G(2,2,k)\equiv\mathrm{Weyl}(D_k)$ and count the number of inequivalent embeddings\footnote{\ In addition there is an embedding $G(2,2,k)\hookrightarrow Sp(2k,\Z)$ which does not factor through $SO(k,\Z)$.} $G(2,2,k)\hookrightarrow SO(k,\Z)\subset Sp(2k,\Z)$, or equivalently the number of maximal local subalgebras of the $Spin(2k)_k$ chiral current algebra. We have just one, generated by the free fermion fields $\psi^a(z)$ unless $p\equiv 2\mid k$ in which case we have $3\equiv 2+1$ of them, the additional ones being the two GSO projections of opposite chirality. In the complex case, chirality is replaced by a $\Z_p$ symmetry. The proof is essentially the same as in the GSO case \cite{feit1}.   
\end{rem}

\begin{exe}
A Lagrangian SCFT with gauge group $\prod_i G_i$
at extreme weak coupling has a CSG which asymptotically takes this constant period form with 
$\cg=\prod_i \mathrm{Weyl}(G_i)$ in the standard reflection representation. 
\end{exe}

\begin{exe}\label{kkkaszzzz}
The model associated to the group $G(m,1,k)$ ($m=3,4,6$) has the simple physical interpretation of representing (the birational class of) the rank-$k$ MN $E_r$ SCFT for $r=6,7,8$, respectively. Indeed, geometrically the quotient of $\C^k$ by $G(m,1,k)$ is the same as taking the $k$-fold symmetric product of quotient of $\C$ by $\Z_m$. Correspondingly, the isotropy group of the diagonal period matrix, $\mathrm{Fix}(e^{2\pi i/m}\,\boldsymbol{1}_k)\subset Sp(2k,\Z)$, is $G(m,1,k)$ (see \cite{eie} for a discussion in the $k=3$ case). The dimensions are
$\{\Delta, 2\Delta, 3\Delta, \cdots, k\Delta\}$ with $\Delta=m$.
\end{exe} 

\begin{exe} The CSG geometries associated to the groups $G(m,d,k)$ are just $\Z_d$ covers of the previous one. We have a new operator $\co$ of dimension $km/d$ such that $\co^d$ is the MN operator of maximal dimension.
\end{exe}

\textbf{Example \ref{kkkaszzzz}} may be generalized. One has a CSG $M$ and takes the $n$-th symmetric power. This works well if $M$ has dimension 1, but in general the resulting  geometry may be more singular than permitted.
\medskip 

In general, when we have a CSG of the form $\C^k/\cg$, where $\cg=\prod_a \cg_a$ with $\cg_a$ irreducible reflection groups, and, in addition, $\mathscr{R}=\C[a_1,\cdots, a_k]^\cg$, the Coulomb dimensions are equal to the degrees of $\cg$.
The period $\boldsymbol{\tau}$, being symmetric, transforms in the symmetric square of the defining representation $V$ of the reflection group $\cg\subset Sp(2k,\Z)$
(see \S.\ref{cyclicee} below).
Hence the dimension of the space of allowed deformations of $\boldsymbol{\tau}$, that is, the dimension $\mathsf{d}$ of the conformal manifold of a constant-period CSG is given by the multiplicity of the trivial representation in $\odot^2 V$. By Schur-Frobenius,     
\be\label{kkkasqw}
\mathsf{d}= \frac{1}{2\,|\cg|}\sum_{g\in\cg}\Big(\chi_V(g)^2+\chi_V(g^2)\Big)= \left[\begin{array}{l}\text{multiplicity of 2}\\ 
\text{as a weight of $\cg$}\end{array}\right.= \left[\begin{array}{l}\#\;\text{irreducible factors $\cg_a$}\\
\text{which are Weyl groups},
\end{array} \right.
\ee 
since a reflection group has a degree 2 invariant iff it is defined over the reals (and hence, if crystallographic,  should be a Weyl group). 
The physical interpretation of this result is that such a CSG represents a gauge theory with gauge group the product of all simple Lie groups whose Weyl groups are factors of $\cg$ coupled to some other intrinsically strongly interacting SCFT associated to the complex factor groups $\cg_a$ (as well as  hypermultiplets in suitable representations of the gauge group). The $\mathsf{d}$ marginal deformations of the geometry are precisely the $\mathsf{d}$ Yang-Mills couplings which are associated to $\mathsf{d}$ chiral operators with $\Delta=2$.
Note that the Yang-Mills couplings may be taken as weak as we please.

Thus the Lagrangian models and higher MN SCFTs already account for all CSG with constant period map up (at most) to finitely many exceptional ones.

\section{The Coulomb dimensions $\Delta_i$}\label{jjjazsqw3}

Now we come to the main focus of the paper, namely to get geometric restrictions on the spectra of Coulomb branch dimensions $\Delta_i$.

We may address two different problems:

\begin{prob} For $k\in\mathbb{N}$ specify the set
\be
\Xi(k)=\left\{\Delta\in \BQ_{\geq1}\ \begin{array}{l}\text{such that: there is a CSG $M$}\\
\text{with }\dim M\leq k\ \text{and a generator}\\ \text{$u$ of }\mathscr{R}\equiv \C[M]\ \text{with }\pounds_\ce u=\Delta\,u\end{array}\right\}
\ee
\end{prob}

\begin{prob}
For $k\in\mathbb{N}$ determine the set
\be
\Lambda(k)=
\left\{\Big(\Delta_1,\Delta_2,\cdots,\Delta_k)\in \big(\BQ_{\geq1}\big)^k \ \begin{array}{l}\text{such that: there is a CSG $M$ with}\\
\mathscr{R}\cong \C[u_1,u_2,\cdots, u_k],\ \pounds_\ce u_i=\Delta_i\,u_i\end{array}\right\}
\ee
\end{prob}

The solution to the \textbf{Problem 2} contains vastly more information than the answer to \textbf{Problem 1}, since there are strong correlations between the dimensions $\Delta_i$ of a given CSG and $\Lambda(k)$ is a rather small subset of $\Xi(k)^k$. However \textbf{Problem 1} is much simpler, and its analysis is a first step in answering \textbf{Problem 2}. We give a solution in the form of a necessary condition:  $\Xi(k)\subset\widehat{\Xi}(k)$, where $\widehat{\Xi}(k)$ is a simple explicit set of rational numbers. There are reasons to believe that the discrepancy between the two sets $\Xi(k)$ and $\widehat{\Xi}(k)$ is small and vanishes as $k\to\infty$.
In facts $\Xi(k)=\widehat{\Xi}(k)$ for $k=1,2$,
and the equality may hold for all $k$.
\medskip 

To orient our ideas, we discuss a few special instances as a warm-up for the general case.

\subsection{Warm-up: revisiting some well-understood cases}

\subsubsection{Rank-one  again}\label{rrabskone}

We know that
\be\label{pxqj}
\Xi(1)\equiv \Lambda(1)= \Big\{1,2,3,3/2,4,4/3,6,6/5\Big\}.
\ee
where 1 corresponds to the free theory and 2 to a $SU(2)$ (Lagrangian) SCFT.
We have already obtained this result from several points of view. In rank-1 the period map is automatically constant and we have $M\cong\C/\cg$, with $\cg$ a rank-1 modular ST group. However, for a complex $\cg$, there are two possible values of $\Delta$, as we stressed already several times.
This correspond to the fact that, in this case, we have two distinct embeddings $m\colon\cg\hookrightarrow SL(2,\Z)$ modulo conjugacy. 

A rank-1 ST group is cyclic. Let $\sigma$ be a generator and consider 
\be\label{zzxx39}
m(\sigma)=\begin{bmatrix} a & b\\ c& d\end{bmatrix}\in SL(2,\Z)\ \text{elliptic,}\quad\tau\in\mathfrak{h}\ \text{a fixed point } a\tau+b=\tau(c\tau+d).
\ee
The explicit matrices $m(\sigma)$ corresponding to the 8 elliptic conjugacy classes in $SL(2,\Z)$ are written in the first column of \textsc{Table I} of Kodaira \cite{koda} (omitting the non semi-simple ones). Under the action of $\sigma$, the period $a$ transforms with the modular factor
\be\label{kkkaz12b}
\sigma\colon a\to a^\prime=(c\tau+d)a\quad \text{while}\quad \sigma\colon u\to e^{2\pi i} u,
\ee
since $u$ is an univalued function. $(c\tau+d)$ is a character of the cyclic group $\cg$, hence a root of unity.
Let $\alpha$ be the unique real number $0<\alpha\leq 1$ such that 
$e^{2\pi i \alpha}=c\tau+d$. Clearly the only functional dependence $a=a(u)$ consistent with  \eqref{kkkaz12b} such that $a\to 0$ as $u\to 0$ (the tip of the cone) is
\be\label{mmmzxq}
a=u^{\alpha}\quad\Longrightarrow\quad \Delta(u)=\frac{1}{\alpha}\geq 1. 
\ee
Using the matrices in \textsc{Table I} in \cite{koda} one recovers the well-known result, see table \ref{kkkaoda}. 

\begin{table}
$$
\begin{array}{c|c|c||c|c|c}\hline\hline
\text{fiber type} & \Delta & \text{modular factor}&
\text{fiber type} & \Delta & \text{modular factor} \\\hline
\text{regular} & 1 & 1 &
I_0^* & 2 & e^{\pi i} \\
II & 6/5 & e^{5\pi i/3} &
II^* & 6 & e^{\pi i/3} \\
III & 4/3 & e^{3\pi i/2} &
III^* & 4 & e^{\pi i/2} \\
IV & 3/2 &e^{4\pi i/3} &
IV^* & 3 & e^{2\pi i/3} \\\hline\hline
\end{array}
$$
\caption{\label{kkkaoda} Modular factors for the Kodaira fibers with semi-simple holonomy invariant.}
\end{table}

\medskip

Our basic strategy is to mimic this
analysis of the $k=1$ for general $k$. Before doing that, we discuss another special case in two different ways: first we review the conventional approach and then recover the same results by reducing the analyis to the $k=1$ case. This will give the first concrete example of the basic strategy of the present paper.

\subsubsection{Hypersurface singularities in $F$-theory and 4d/2d correspondence}\label{kkaqw}

There is a special class of simple (typically non-Lagrangian) 4d $\cn=2$ SCFTs which are engineered in $F$-theory out of an isolated quasi-homogeneous dimension-3 hypersurface singularity with $\hat c<2$ (see eqn.\eqref{kkkz120})
\cite{GVW,Cecotti:2010fi}. 
Their SW geometry (in absence of mass deformations) is given by the hypersurface $\cf\subset\C^4$ of equation
\be\label{jjjbza}
F(x_1,x_2,x_3,x_4)\equiv W(x_1,x_2,x_3,x_4)+\sum_{a=0}^{k-1} u_a\, \phi_a(x_i)=0,
\ee
where $W(x_i)$ is a quasi-homogeneous polynomial
\be\label{kkkz120}
\begin{aligned}
&W(\lambda^{q_i}x_i)=\lambda\, W(x_i),\quad \forall\,\lambda\in \C,\\
&\text{with }q_i\equiv \deg x_i\in\mathbb{Q}\ \text{and } \hat c\equiv \sum_{i=1}^4(1-2q_i)<2,
\end{aligned}
\ee
while $\phi_a(x_i)$ are the elements of the (2,2) chiral ring $\mathcal{R}\equiv\C[x_i]/(\partial_i W)$ having degree
$<\hat c/2$. In particular $\phi_0\equiv 1$ is the identity operator. The number of $\phi_a$'s with $\deg\phi_a<\hat c/2$ is equal to the rank $k$ of the corresponding 4d $\cn=2$ SCFT and the complex parameters $u_a$ in eqn.\eqref{jjjbza} are the Coulomb branch (global) homogeneous coordinates. One has
\be
\deg u_a= 1-\deg \phi_a\qquad\text{where } 0\leq \deg\phi_a < \hat c/2 <1.
\ee
The regularity condition of ref.\!\!\cite{GVW}, $\hat c<2$, ensures that the $u_a$'s have positive degree. 
The SW 3-form is the obvious one
\be\label{kkkaqw}
\Omega=PR\left(\frac{ dx_1\wedge\cdots \wedge dx_4}{F}\right),
\ee
where $PR$ stands for ``Poincar\'e residue'' \cite{GH}. At the conformal point, $u_a=0$,
$\Omega$ has degree
\be
\sum_{i=1}^4 q_i-1=1-\frac{\hat c}{2}>0.
\ee
Since $\Omega$ (by definition) has $U(1)_R$ charge $q$ equal $1$, for all chiral object $\phi$ with a definite degree, we have
\be
\Delta(\phi)\equiv q(\phi)=\frac{\deg \phi}{1-\hat c/2},
\ee
in particular,
\be\label{plcase}
\Delta(u_a)=\frac{1-\deg\phi_a}{1-\hat c/2}\equiv 1+\frac{\hat c/2-\deg \phi_a}{1-\hat c/2}.
\ee
 To get the Coulomb branch dimensions $\Delta(u_a)$, it remains to determine the $k$ rational numbers
\be
t_a\equiv\frac{\hat c}{2} -\deg\phi_a,\quad a=0,1,\dots, k-1,\quad\text{in particular }
t_0=\frac{\hat c}{2}.
\ee
There are many ways of doing this, including pretty trivial ones. Here we shall compute the $t_a$'s in a way that seems unnaturally complicate:  but recall that we are doing this computation as a warm-up, meaning that we wish to perform this elementary computation in a way which extends straightforwardly to the general case where simple minded methods fail.

\subparagraph{Picard-Lefschetz analysis \cite{unpubb}.}  One convenient viewpoint is the 4d/2d correspondence of ref.\!\!\cite{Cecotti:2010fi}. One considers the 2d (2,2) Landau-Ginzburg model with superpotential $F(x_i)$ and uses the techniques of $tt^*$ geometry \cite{Cecotti:1991me,Cecotti:1992rm} to compute 2d quantities which are then reinterpreted in the 4d language. In 2d $\hat c$ is one-third the Virasoro central charges and $\deg\phi$ is the $R$-charge  (in the 2d sense) of the chiral object $\phi$.

In the 2d approach, the $t_a$'s are computed using the Picard-Lefschetz theory \cite{singularities} (see \cite{Cecotti:1992rm} for a survey in the present language). Consider the family of hypersurfaces $\cf_z=\{ F(x_i)=z\}$ parametrized by $z\in\C$; we have $H_3(\cf_z,\Z)\cong\Z^\mu$, where $\mu$ is the Milnor number of the singularity $W=0$ (i.e.\! the dimension of the (2,2) chiral ring $\mathcal{R}$).
The \emph{classical monodromy} $H$ of the quasi-homogeneous singularity $W$  is given by the lift on the homology of the fiber, $H_3(\cf_z,\Z)$, along the closed loop in the base $z=\rho\, e^{2\pi i t}$ ($t\in[0,1]$ and $\rho\gg1$).  Concretely, $H$ is a $\mu\times \mu$ integral matrix acting on the lattice $\Z^\mu$ whose action on $\C^\mu\equiv \Z^\mu\otimes \C$ is semi-simple of spectral radius 1 \cite{Cecotti:1992rm}. 
Let $\Phi\subset\Z^\mu$ be the sublattice fixed (element-wise) by $H$, and consider the quotient lattice $\Gamma=\Z^m/\Phi$
which has rank $2k$. $H$ induces an automorphism $\overline{H}$ of $\Gamma$.
The intersection form in the homology of the hypersurface $\cf_z$, induces a non-degenerate, integral, skew-symmetric pairing $\langle -,-\rangle$ on $\Gamma$, preserved by $\overline{H}$. 
Then
\be
\overline{H}\in Sp(2k,\mathbb{Q})
\ee
In simple examples the induced polarization $\langle-,-\rangle$ is principal, and one has $\overline{H}\in Sp(2k,\Z)$; in the general case we reduce to this situation by a suitable isogeny in the intermediate Jacobian of $\cf_z$.
Moreover, $\overline{H}$ is semi-simple of spectral radius 1 so (by Kronecker's theorem) it 
has a finite order $\ell$.  Its eigenvalues
are of the form $\{\exp(2\pi i \alpha_a), \exp(2\pi i(1-\alpha_a))\}$ for some $0<\alpha_a\leq 1$,
$a=0,1,\dots, k-1$ with $\ell \alpha_a\in\mathbb{N}$. These eigenvalues are related to the $t_a$ by the 2d spectral-flow relation \cite{Cecotti:1992rm}
\be
\mathsf{Spectrum}\,\overline{H}=\Big\{e^{2\pi i \alpha_a},e^{2\pi i (1-\alpha_a)}\Big\}\equiv\Big\{e^{\pm 2\pi i t_a}\Big\}.
\ee
Since $0 < t_a <\hat c/2 <1$,
for each index $a=0,1,\dots, k-1$ we have two possibilities
\be\label{ambamb}
t_a =\alpha_a\ \ \text{or}\ \  
1-\alpha_a,
\ee
where we relabel the indices of $t_a$ so that $t_0=\max t_a$. Thus knowing the spectrum of $\overline{H}$ fixes the $t_a$'s up to a 
 $2^k$-fold ambiguity corresponding to choosing for each $a$ one of the two possible values \eqref{ambamb}. Let us explain the origin of this ambiguity. For simplicity of illustration we assume the characteristic polynomial $P(z)$ of $\overline{H}$ to be irreducible over $\BQ$; in this case the spectrum uniquely fixes $\overline{H}$ up to conjugacy in $GL(2k,\mathbb{Q})$. However two physical systems described by monodromy matrices $\overline{H}_1,\overline{H}_2\in Sp(2k,\mathbb{Z})$  are physically equivalent iff they are related by a change of duality frame i.e.\! iff the corresponding reduced monodromies $\overline{H}_1,\overline{H}_2$ are conjugate in the smaller group $Sp(2k,\mathbb{Z})$: the unique $GL(2k,\mathbb{Q})$-conjugacy class of $\overline{H}$ decomposes\footnote{\ Since $P(z)$ is irreducible, $\BQ[\overline{H}]$ is an Abelian number field $\BK$ of degree $2k$. Let $\Bbbk$ be its maximal real subfield (of degree $k$); the elements $\xi\in \Bbbk$ may be written as the rational matrices $\xi(\overline{H})\equiv\sum_s a_s (\overline{H}+\overline{H}^{-1})^s$ with $a_s\in\BQ$. Then the map
 $$\Omega\to \Omega^\prime= \Omega\,\xi(\overline{H})$$
 identifies the set of inequivalent $\BQ$-symplectic structures compatible with $H$ with the multiplicative group $\Bbbk^\times/\Bbbk^\times_\text{tot.pos.}\cong \Z_2^k$. See section 5 for more details.} into $2^k$ distinct $Sp(2k,\BQ)$-conjugacy classes in one-to-one correspondence with the inequivalent choices of the $t_a$'s. (The conjugacy classes over the integral group, $Sp(2k,\Z)$, are trickier, and will be discussed in section 5).
The Picard-Lefschetz theory has a canonical symplectic structure (i.e.\! the intersection form in homology) and hence a canonical choice of the $\{t_a\}$.
The spectrum of Coulomb branch dimension for the SCFT engineered by the singularity is given by plugging these canonical $\{t_a\}$ in the expressions
\be\label{mmmzq1}
\Big\{\Delta(u_a)\Big\}=\left\{1+\frac{t_a}{1-t_0}\right\},\quad\text{where }t_0=\max \{t_a\}.
\ee

\subparagraph{The ray analysis.} Let us rephrase the above Picard-Lefschetz analysis in a different language. We return to eqn.\eqref{jjjbza} and consider the Coulomb branch $M$ of the associated $\cn=2$ SCFT;  we see $M$ as the affine cone over the WPS with homogeneous  coordinates 
$(u_0,u_1,\cdots, u_{k-1})$ the coordinate $u_a$ having grade $\Delta(u_a)$.
We focus on the closed one-dimensional sub-cone $M_0$ parametrized by the \textsc{vev} $u_0$  of the chiral operator of largest dimension $\Delta(u_0)$,
\be
M_0=\Big\{u_1=u_2=\cdots=u_{k-1}\Big\}\subset M,
\ee
which is preserved by the $\C^\times$-action generated by the Euler field $\ce$.
Metrically, $M_0$ is a flat cone and the restricted period map $\boldsymbol{\tau}|_{M_0}$ is constant. The only difference with respect to the rank-one case of \S.\,\ref{rrabskone}  is that the monodromy $\overline{H}$ around the tip of the cone $M_0$ is now valued in $Sp(2k,\Z)$ instead of $Sp(2,\Z)$.

A Coulomb vacuum $x\in M_0\subset M$ preserves a discrete subgroup of $U(1)_R$ of the form $\Z_n$ where $n$ is the order of
$1/\Delta(u_0)$ in $\mathbb{Q}/\Z$. Since in an interacting unitary theory $\Delta(u_0)>1$, this subgroup is never trivial, and $n>2$ (i.e.\! the unbroken subgroup is complex) if
$\Delta(u_0)\neq 2$ (i.e.\! unless $u_0$ is superficially marginal). The  $\Z_n$  symmetry unbroken along the locus $M_0\subset M$ is generated by
\be
\exp\!\big(2\pi i R/\Delta(u_0)\big),
\ee 
and the action of this operator on the cohomology of a (regular) fiber of the special geometry is given by $\overline{H}\,e^{2\pi i (1-\hat c/2)}$, the extra factor $e^{2\pi i (1-\hat c/2)}$ corresponding to the spectral flow in the (2,2) language \cite{Cecotti:1992rm}. Thus
\be
\begin{aligned}
&\mathsf{Spectrum}\,\overline{H}= \Big\{
\exp\!\Big[\pm2\pi i\Big(\hat c/2-1+\Delta(u_a)/\Delta(u_0)\Big)\Big]\Big\}\\
&\Rightarrow\quad t_a=\hat c/2-1+\Delta(u_a)/\Delta(u_0),\qquad\quad t_0=\hat c/2,\\
&\text{and } t_a=(1-t_0)(\Delta(u_a)-1),
\end{aligned}
\ee
which gives back \eqref{mmmzq1}.
Thus, in the context of 4d $\cn=2$ SCFTs engineered by hypersurface singularities, the classical monodromy computation of the dimensions $\Delta(u_a)$  may be rephrased as a local analysis along the sub-locus $M_0$ in the Coulomb branch $M$ of vacua which leave unbroken the largest possible discrete subgroup of $U(1)_R$. The local analysis on $M_0$ is essentially identical to the $k=1$ case.

The classical monodromy approach is unsatisfactory in two ways:
\begin{itemize}
\item[a)] it works for a particular class of SCFTs;
\item[b)] is not \emph{democratic} in the following sense: we have many sub-loci in the Coulomb branch $M$ over which some discrete $R$-symmetry is restored, but the classical monodromy (when applicable) applies to just one of them (the one with the largest unbroken symmetry). 
\end{itemize}

In order to solve \textbf{Problems 1,\,2} we need a  generalization of the classical monodromy construction which may be applied uniformly to all loci in the Coulomb branch with some unbroken discrete $R$-symmetry, and to all SCFT, while reducing to the classical  Picard-Lefschetz theory when we consider the locus of largest unbroken $U(1)_R$ symmetry of a SCFT engineered by a $F$-theory singularity.

Suppose such a generalization exists.
Along the Coulomb branch of a typical SCFT we have several loci with enhanced (discrete) $R$-symmetry; each such locus produces a list of $\Delta_a$. Then we get the highly non-trivial constraint that the dimension set $\{\Delta_a\}$ should be the same independently of which special locus we use to compute it. On the other hand, the agreement of the dimensions computed along different loci in $M$ is convincing evidence of the correctness of the method.
\medskip

In the remaining part of this note we describe the generalized method, and  check its consistency is a variety of examples.
Although we could present the algorithm already at this stage in the form of an educated guess inspired by the classical Picard-Lefschetz formulae,
we prefer to deduce it mathematically from scratch. Before doing that, we need some elementary preparation.

\subsection{Cyclic subgroups of Siegel modular groups I}\label{cyclicee}

We saw already in the $k=1$ case that an important ingredient in the classification of all possible dimension sets $\{\Delta_a\}$ is the list of all   embeddings of the cyclic group $\Z_n$ (or more generally  of a finite group $G$) into the Siegel modular group $Sp(2k,\Z)$ modulo symplectic conjugacy in $Sp(2k,\Z)$ (cfr.\! discussion around eqn.\eqref{ambamb}).

We start by establishing a fact, already mentioned in \S.\ref{connstantmaoa}, which applies to all subgroups of the Siegel modular group, cyclic or otherwise.

\begin{lem}\label{kkkazzz1v}
Let $\cg\subset Sp(2k,\Z)$ be a finite subgroup. The fixed locus in the Siegel upper half-space $\mathfrak{H}_k\equiv\{\boldsymbol{\tau}\in \C(k)\,|\, 
\boldsymbol{\tau}=\boldsymbol{\tau}^t,\ \Im\,\boldsymbol{\tau}>0\}$
\be
\mathrm{Fix}(\cg)=\Big\{\boldsymbol{\tau}\in \mathfrak{H}_k\cong Sp(2k,\R)/U(k)\ \Big|\ g\cdot \boldsymbol{\tau}=\boldsymbol{\tau},\ \;\forall\; g\in\cg\Big\}
\ee
 is non-empty and connected.
\end{lem}

\begin{rem} The proof shows that the \textbf{Lemma} holds for all duality-frame groups $S(\Omega)_\Z$.
\end{rem}

\begin{proof} Being finite, $\cg$ is compact. Hence $\cg\subset K$ for some maximal compact subgroup $K\subset Sp(2k,\R)$. All maximal compact subgroups in $Sp(2k,\R)$ are conjugate to the standard one, the isotropy group of $i\boldsymbol{1}_k\in \mathfrak{H}_k$, that is, there is $R\in Sp(2k,\R)$
\be
K= \left\{R^{-1} \!\begin{bmatrix}A &B \\ C & D\end{bmatrix}\!R\ \ \Bigg|\ \ Ai+B=i(Ci+D)\in U(k)\right\}.
\ee 
The Cayley transformation  $C$ maps biholomorphically the Siegel upper half-space $\mathfrak{H}_k$
into the Siegel disk \cite{bump,bbb}
\be
\mathfrak{D}_k=\big\{\boldsymbol{w}\in\C(k)\;\big|\; \boldsymbol{w}^t=\boldsymbol{w}, \ 1-\boldsymbol{w}\boldsymbol{w}^*>0\big\},
\ee
taking $\boldsymbol{\tau}=i\boldsymbol{1}_k$ to the origin $\boldsymbol{w}=0$ and conjugating  
 the standard maximal compact subgroup into the diagonal subgroup. Then 
\be\label{ooo90}
C R\, \cg R^{-1} C^{-1}\subset\left\{\begin{bmatrix} U & 0 \\ 0 & U^*\end{bmatrix},\quad U\in U(k)\right\}.
\ee
Consider the embedding $\cu\colon\cg\hookrightarrow U(k)$ sending $g$ into the upper-left block of $CR g R^{-1}C^{-1}$; we write $V$ for the corresponding degree-$k$ unitary representation. The action of $\cg$ on the Siegel disk $\mathfrak{D}_k$ is linear
\be
\boldsymbol{w}\mapsto \cu(g)\,\boldsymbol{w}\; \cu(g)^t,\quad\  g\in\cg,
\ee
i.e.\! the Cayley-rotated period $\boldsymbol{w}$ transforms in the symmetric square representation $\odot^2V$.
The fixed locus of $\cg$ is the intersection of $\mathfrak{D}_k\subset \odot^2V$ with the linear subspace 
\be
(\odot^2 V)^{\cu(\cg)}\subset \odot^2 V
\ee
 of trivial representations whose dimension $\mathsf{d}$ is as in eqn.\eqref{kkkasqw}; in particular $\mathrm{Fix}(\cg)$ is non-empty and connected.  In the special case that 
$\odot^2V$ does not contain the trivial representation, the fixed locus reduces to the origin in $\mathfrak{D}_k$, and hence is an isolated point. Mapping back to the Siegel upper half-space $\mathfrak{H}_k$, the fixed point is
\be
\boldsymbol{\tau}=R^{-1}(i\,\boldsymbol{1}_k) \equiv (i A_R+B_R)(i C_R+D_R)^{-1},\quad\  R^{-1}\equiv\begin{bmatrix}A_R & B_R\\
C_R & D_R\end{bmatrix}\in Sp(2k,\R).
\ee
In the general case
\be
\mathrm{Fix}(\cg)= R^{-1}C^{-1}\big((\odot^2 V)^{\cu(\cg)}\big)\cap \mathfrak{H}_k.
\ee
\end{proof}

If $\boldsymbol{\tau}\in \mathrm{Fix}(\cg)$ we have
\be
\begin{bmatrix} A & B\\
C & D\end{bmatrix}\begin{bmatrix}\boldsymbol{\tau} & \boldsymbol{\bar\tau}\\
1 & 1\end{bmatrix}=\begin{bmatrix}\boldsymbol{\tau} & \boldsymbol{\bar\tau}\\
1 & 1\end{bmatrix}\begin{bmatrix}C\boldsymbol{\tau}+D & 0\\
0 & C\boldsymbol{\bar\tau}+D
\end{bmatrix},\quad \forall\begin{bmatrix} A & B\\
C & D\end{bmatrix}\in \cg\subset Sp(2k,\Z),
\ee
and the embedding $\cu\colon \cg\hookrightarrow U(k)$ is given by the modular factor $C\boldsymbol{\tau}+D\in U(k)$,
cfr.\! the discussion of the $k=1$ case around eqn.\eqref{zzxx39}. This is the same factor appearing in the transformation of the $a$-periods, $a\to (C\boldsymbol{\tau}+D)a$, and is the one which controls the Coulomb dimensions, as we saw in the $k=1$ case. Implicitly we have already used these facts in the discussion of the CSG with constant period map, \S.\,\ref{connstantmaoa}.
\medskip

Let us specialize to the case in which
$\cg$ is a cyclic group $\Z_n$ generated by a matrix $m\in Sp(2k,\Z)$. $m$ is called \emph{regular} iff $\mathrm{Fix}(m)$ is an isolated point, i.e.\! iff $\odot^2V$ does not contain the trivial representation. Note that this is a weaker notion than $m$ being regular as an element of the Lie group $Sp(2k,\R)$ which requires the characteristic polynomial of $m$ to be square-free; we shall refer to the last  situation as \emph{strongly regular}. 
The spectrum of the unitary matrix $\cu(m)=C\boldsymbol{\tau}+D$ is a set of $k$
$n$-th roots of unity, $\{\zeta_1,\dots,\zeta_k\}$, and $m$ is regular iff 
\be\label{kkkkzzznncx}
\zeta_i\zeta_j\neq 1\quad\text{for all}\ \ 1\leq i,j\leq k.
\ee 
The eigenvalues of the $2k\times 2k$ matrix $m$ are $\{\zeta_1,\dots,\zeta_k\}\cup \{\zeta_1^{-1},\dots,\zeta_k^{-1}\}$ the union being disjoint iff $m$ is regular. 
Let $\psi_\zeta$ be the normalized eigenvector of $m$ associated to the eigenvalue $\zeta$.
The symplectic matrix $\Omega$ of $Sp(2k,\Z)$ corresponds to the 2-form
\be
\Omega= i\sum_i \psi_{\zeta_i}\wedge \psi_{\zeta^{-1}_i} 
\ee
Thus, for $m$ regular, the splitting of the spectrum of $m$ in the two disjoint sets $\mathsf{Spectrum}\;\cu(m)$ and $\mathsf{Spectrum}\;\overline{\cu(m)}$  may be read directly from $\Omega$: the eigenvalue $\zeta$ belongs to the spectrum of $\cu(m)$ iff
the term $\psi_\zeta\wedge\psi_{\zeta^{-1}}$ appears in $-i\Omega$ with the $+$ sign, otherwise
$\zeta^{-1}\in\mathsf{Spectrum}\;\cu(m)$. 

We shall present a much more detailed discussion of the cyclic subgroups of the Siegel modular group in section 5 below.

 \subsection{The Universal Dimension Formula} 
  \subsubsection{Normal complex rays in CSG}
\label{lllxcv3}

As anticipated at the end of \S.\ref{rrabskone}, our goal is to reduce the determination of the $\Delta_a$'s to the analysis of \emph{one-dimensional} conic complex geometries.
In this subsection we introduce the basic construction.  
We start with a definition:
  
\begin{defn} Let $M$ be the K\"ahler cone of a CSG with holomorphic Euler vector $\ce=(E-iR)/2$. A \textit{complex ray} $M_\ast\subset M$ is a closed one-dimensional complex subspace preserved by the action of $\ce$. 
A complex ray $M_\ast\subset M$ is called \emph{normal} iff it is normal as an analytic subspace of $M$.
\end{defn}  

Thus, a complex ray $M_\ast$ is the orbit under the Lie group generated by $E$,$R$ of a point $x\neq0$ in $M$. Equipped with the induced metric, $M_\ast$ is a K\"ahler cone, hence locally flat. Restricting our considerations to a ray makes physical sense; indeed
  
\begin{pro} Let $M$ be a K\"ahler cone and $M_\ast\subset M$ a complex ray. Then $M_\ast$ is \emph{totally geodesic in $M$.} 
\end{pro} 
\begin{proof} Since $M$ is a K\"ahler cone,  we have the holomorphic field $\ce\equiv E+i R$ such that $ \nabla_{\bar i} \ce^j =0$ while $\nabla_i \ce^j=\delta_i^j$. With respect to the induced K\"ahler metric, $M_\ast$ is also a K\"ahlerian cone with a holomorphic Euler vector $\ce_\ast=\ce|_{M_\ast}$. $\ce_\ast$, $\bar\ce_\ast$ span the real tangent bundle $TM_\ast$. Let $\nabla^\ast$ be the Levi-Civita connection of the induced metric on $M_\ast$. The second fundamental form of $M_\ast$ in $M$ is
\be
\begin{aligned}
&I\!I(\ce_\ast,\ce_\ast)=\nabla_\ce\ce\big|_{M_\ast}-\nabla_{\ce_\ast}^\ast\ce_\ast=\ce\big|_{M_\ast}-\ce_\ast=0\\
&\text{then also}\quad I\!I(\ce_\ast,\bar\ce_\ast)=I\!I(\bar\ce_\ast,\bar\ce_\ast)=0.
\end{aligned}
\ee
\end{proof}

The chiral ring of $M$ is the ring of global homomorphic functions, $\mathscr{R}=\Gamma(M,\co_M)$. The linear differential operator $\ce$ acts on $\mathscr{R}$ and we write $S$ for its spectrum, i.e.\! the set of dimensions of chiral operators. $\mathscr{R}$ is the Frech\'et completion of a finitely generated graded ring. If $\phi\in \mathscr{R}$ is a homogeneous element of degree $\Delta(\phi)$, we have $\ce\phi=\Delta(\phi)\phi$.
In the same way we write $\mathscr{R}_\ast$ for the ring of holomorphic functions on the ray $M_\ast$ and $S_\ast$ for the spectrum of $\ce_\ast$ in $\mathscr{R}_\ast$.

\begin{pro}\label{kkazq1} Let $M_\ast\subset M$ be a \underline{normal} complex ray of the special cone $M$. Then $S_\ast\subset S$. Moreover, $\mathscr{R}_\ast$ is the Frech\'et completion of the graded ring $\C[u_\ast]$
where the single generator $u_\ast$ has dimension $1\leq\Delta_\ast\in S_\ast$, and there is a generator $u$ of $\mathscr{R}$ having the same dimension $\Delta_\ast$.  
\end{pro}
\begin{proof}
$M$ is Stein, and $M_\ast\subset M$ is a closed analytic subspace. Hence, by Cartan's extension theorem \cite{stein1,stein2} we have an epimorphism of chiral rings
\be
\mathscr{R}\xrightarrow{\;i^\ast\;} \mathscr{R}_\ast\to 0,
\ee
given by restriction, which preserves the $\ce$ action. Hence $S_\ast\subset S$. 
In this argument we do not use the fact that $M_\ast\subset M$ is normal. However, if $M_\ast$ is not normal, there is a subtlety in the above statement. Let $i\colon M_\ast\to M$ be the closed inclusion. Cartan's theorem states
\be
\Gamma(M,\co_M)\xrightarrow{\;i^\ast\;} \Gamma(M_\ast, i^\ast\co_M)\to H^1(M, \mathrm{ker}\,i^\ast)\equiv 0,
\ee
where $i^\ast\co_M$ is the structure sheaf of $M_\ast$ seen as an analytic subspace of $M$.
Iff $M_\ast$ is normal, $i^\ast\co_M$ coincides with the structure sheaf of its normalization $\co_{M_\ast}$ (i.e.\! with the integral closure of $i^\ast\co_M$ stalk-wise);
its global sections $\Gamma(M_\ast,\co_{M_\ast})$ correspond to the ``intrinsic'' notion of holomorphic functions on $M_\ast$,
while $\Gamma(M_\ast, i^\ast\co_M)$ is the subset of holomorphic functions on $M_\ast$ seen as a concrete sub-space of $M$. We illustrate this subtlety in \textbf{Example \ref{jjasqwe2}} below. Then, although the result is valid in general, to apply it for $M_\ast\subset M$ not normal we need to be able to distinguish holomorphic functions in the two different senses. If $M_\ast$ is normal, it coincides with its own normalization, and the two notions coincide.

Since $M_\ast$ is one-dimensional and normal, it is smooth.\footnote{\ As a complex space; the K\"ahler metric has a conical singularity at the tip.} On $M_\ast$ we have a $\C^\times$ action $\zeta\mapsto\zeta^\ce$, acting transitively on $M_\ast\setminus\{0\}$.  Then $M_\ast$ is analytically (hence algebraically) a copy of $\C$ on which $\zeta^\ce$ acts by automorphisms fixing the origin. Let $u_\ast$ be a standard coordinate on this $\C$; the polynomial ring $\C[u_\ast]$ is dense in $\Gamma(M_\ast,\co_{M_\ast})$ and graded by $\deg u_\ast=\ce u_\ast>0$. 

If $w\in\mathscr{R}$ is a homogeneous holomorphic function  with $\Delta(w)\not\in S_\ast$ we must have
$w|_{M_\ast}\equiv 0$. Let $\{v_i\}_{i=1}^k$ be a set of homogeneous generators of (a dense subring of) the global chiral ring $\mathscr{R}$. Not all restrictions $\{v_i|_{M_\ast}\}_{i=1}^k$ may vanish identically since in a Stein manifold the ring of homolorphic functions separates points \cite{stein1}. Let $u\in \{v_i\}$ be a generator of $\mathscr{R}$ with $u|_{M_\ast}\not\equiv0$ of smallest degree $d_\ast$.
All (non constant) homogenous elements $f\in\mathscr{R}$ either restrict to zero $f|_{M_\ast}\equiv0$ or have a degree $\geq d_\ast$.  Since $\mathscr{R}_\ast\cong \C[u_\ast]$, for some $u_\ast$ having minimal positive degree in $\mathscr{R}_\ast$, and $u_\ast$ is the restriction of a function in $\mathscr{R}$, we conclude that we may choose $u_\ast=u|_{M_\ast}$.
Therefore the dimension $d_\ast\equiv \Delta_\ast$ of the generator $u_\ast$ of $\mathscr{R}_\ast$,  is equal to the dimension $\Delta(u)$ of the generator $u$ of the full chiral ring $\mathscr{R}$.
\end{proof}

\begin{exe}\label{jjasqwe2} We illustrate the subtle point in the proof of the above \textbf{Proposition}.
We consider $SU(3)$ SQCD with $N_f=6$
at weak coupling and zero quark masses. In this case the Coulomb branch ring is $\C[u,v]$ where the generators $u$ and $v$ have dimensions $2$ and $3$.
The (reduced) complex rays are
\be
M_u=\{v=0\},\quad M_v=\{u=0\},\quad
M_{(\alpha)}=\{v^2=\alpha\, u^3\},\ \alpha\in\C^\times. 
\ee
$M_u$ and $M_v$ are normal with $\mathscr{R}_u=\C[u]$, $\mathscr{R}_v=\C[v]$; we see that the generators of the two normal ray rings, $u$ and $v$, are generators of the full chiral ring $\mathscr{R}$.
Instead, $M_{(\alpha)}$ is a plane cubic with a cusp, which is the simplest example of a non-normal variety (see e.g.\! \S.4.3 of \cite{Eisebund}). The ring of holomorphic functions in the subspace sense is not integrally closed. The integral closure of the ray ring contains the function $v/u$ (indeed, $(v/u)^2 =\alpha u$) which is holomorphic in the normalization and has dimension 1, a dimension $\not\in S$.
\end{exe}

\begin{rem} Note that the non-normal rays $M_{(\alpha)}$ in the above example correspond to the ``non-free'' geometry discussed in ref.\!\cite{Argyres:2017tmj}.
\end{rem}

\subsubsection{Unbroken $R$-symmetry along a ray} \label{jjjjjjz}

The statements in \S.\ref{lllxcv3} reduce the computation of the Coulomb dimensions of the generators of $\mathscr{R}$ to local computations at normal complex rays in the conical Coulomb branch $M$. This leads to the following two questions: 1) are there \emph{enough} normal complex rays to compute all Coulomb dimensions? 2) how we characterize the \emph{normal} complex rays?

In \textbf{Example \ref{jjasqwe2}} we see that the two normal complex cones in the Coulomb branch $M$ of $SU(3)$ SQCD are precisely the loci of Coulomb vacua with an unbroken discrete subgroup of $U(1)_R$, namely
$\Z_3$ for $M_v$ and $\Z_2$ for $M_u$. All non-normal rays consist of vacua which completely break $U(1)_R$ (except at the tip where the full $U(1)_R$ is restored). This characterization of normal rays by unbroken $R$-symmetry holds in general.

To a complex ray $M_\ast$ there is associated a rational positive number $\alpha_\ast$, namely the smallest positive number such that
\be
\exp(2\pi \alpha_\ast R)=\mathrm{Id}_{M_\ast}.
\ee
$\alpha_\ast\in\mathbb{Q}_{>0}$ since the $R$-symmetry group $U(1)_R$ is compact.
In the physical language this means that a subgroup of $U(1)_R$ is unbroken:

\begin{fact} A Coulomb vacuum $x\in M_\ast\subset M$ preserves a discrete $\Z_n\subset U(1)_R$ $R$-symmetry where $n$ is the order of $\alpha_\ast$ in
$\mathbb{Q}/\Z$. We say that $M_\ast$ is an \emph{elliptic ray} iff $n>2$.
\end{fact}

In section 3 we presented some evidence that the chiral ring is a free polynomial ring (or simply related to such a ring).
 In this case $M$ is parametrized by weighted homogeneous coordinates $u_i$ having weights $\Delta_i\in\mathbb{Q}_{\geq1}$. Then the complex ray along the $i$-th axis
\be\label{hhhhq}
M_i\equiv \Big\{(u_1,u_2,\cdots, u_k)= (0,0,\cdots,\overset{i\text{-th}}{u},0,\cdots,0),\ u\in\C\Big\}\subset M
\ee
is \emph{normal} and 
has $\alpha_i=1/\Delta_i$. If $\Delta_i=r_i/s_i>1$ with $(r_i,s_i)=1$, the order of the residual $R$-symmetry is $n_i=r_i\geq 2$ with equality iff $\Delta_i=2$, i.e.\! iff $u_i$ is the vev of a (superficially) marginal operator. 
In particular, if $\mathscr{R}$ is a free polynomial ring we do have enough normal rays.
\medskip

Let $M_\ast$ be a normal ray and $x\in M_\ast\setminus \{0\}$. We consider the closed $R$-orbit
\be
x(t)=\exp(2\pi \alpha_\ast t R)\cdot x \in M_\ast,\qquad t\in [0,1].
\ee
The monodromy of this path is an element $m_\ast$ of the modular group $Sp(2k,\Z)$ which is independent of $x$ modulo conjugacy. We have
\be\label{kkkasqw}
\det[z-m_\ast]= \prod_{\ell\mid n}\Phi_\ell(z)^{s(\ell)},\qquad s(\ell)\in \{0,1,2,\cdots\},\quad \sum_{\ell\mid n} s(\ell)\,\phi(n)=2k,
\ee
since $\exp(2\pi n \alpha_\ast R)$ acts trivially on the periods. We say that $M_\ast$ is regular if its monodromy $m_\ast$ is \emph{strongly} regular. Regularity is equivalent to  $s(\ell)\in\{0,1\}$
(so $s(1)=s(2)=0$ in the regular case). Regularity implies that the fixed period $\boldsymbol{\tau}$ is unique. More crucially, it means that the ray is not part of the ``bad'' discriminant locus, $M_\ast\not\subset \cd_\text{bad}$. Here $ \cd_\text{bad}$ is the union of the irreducible components of $\cd$ with a non-semi-simple monodromy; along $\cd_\text{bad}$ the period matrix $\boldsymbol{\tau}$ degenerates in agreement with the $SL_2$-orbit theorem, see \S.\,\ref{jjzqwe}. Indeed, strong regularity implies that all the eigenvalues of $m_\ast$ are distinct, and no non-trivial Jordan block may be present.
More generally, we split the product in \eqref{kkkasqw} in the square-free factor $\prod_{s(\ell)=1}\Phi_\ell(z)$ and the complementary factor $\prod_{s(\ell)>1}\Phi_\ell(z)^{s(\ell)}$. $m$
is conjugate in $Sp(2k,\BQ)$ to a block-diagonal matrix of the form $m_\text{reg}\oplus m_\text{comp}$
with $m_\text{reg}\in Sp(2k_\text{reg},\Z)$ strongly regular.  Then, up to isogeny,\footnote{\ These statements follow from the Poincar\'e total reducibility theorem.}
$\boldsymbol{\tau}=\boldsymbol{\tau}_\text{reg}\oplus \boldsymbol{\tau}_\text{comp}$ for a unique $\boldsymbol{\tau}_\text{reg}$. 
The regular rank $k_\text{reg}$ of $M_\ast$ is
\be
k_\text{reg}=\frac{1}{2}\sum_{\ell\colon s(\ell)=1}\phi(\ell).
\ee
Thus, locally at the ray, the family of Abelian varieties $X\to M_\ast$ splits (modulo isogeny) in a product $X_1\times X_2\to M_\ast$, and $M_\ast$ is not in the ``bad'' discriminant of the first factor (but it may be for the second one). 

\medskip

Let $M_\ast$ be a normal ray which is also regular. The real function $r^2=\mathrm{Im}\,\boldsymbol{\tau}_{ij}\,a^i\bar a^j$ is smooth and non-zero on $M_\ast$; since $\mathrm{Im}\,\boldsymbol{\tau}_{ij}$ is the unique fixed period, which is non singular
on $M_\ast$ (in particular, bounded), it means that there exists a $\C$-linear combination $a_\ast$ of the periods $a^i$
which does not vanish on $M_\ast$ (more precisely, $a_\ast$ is well defined on a finite cyclic cover of $M_\ast$ branched at the tip). 
Applying $m_\ast$ to $a_\ast$, we see that $e^{2\pi i \alpha_\ast}\equiv e^{2\pi i/\Delta(u_\ast)}$ belongs to the spectrum of $m_\ast$. Since $m_\ast$ is the lift of the  generator of the unbroken subgroup $\Z_n\subset U(1)_R$, we conclude that $e^{2\pi i \alpha_\ast}$
 is a $n$-root of unity. In facts, it should be a primitive root, otherwise the unbroken symmetry would be smaller.
Comparing with the $k=1$ case, and using  \textbf{Proposition \ref{kkazq1}} we learn that $1/\alpha_\ast$ is the dimension of a generator of the full chiral ring $\mathscr{R}$. Now suppose that $M_\ast$ is not regular.
Taking into account only the ``good'' block,
that is, focusing on the first local family,
$X_1\to M_\ast$, we reduce to the regular situation. The argument applies to the irregular block too, as long as $\mathrm{Im}\,\boldsymbol{\tau}_{ij}$ is not singular on $M_\ast$. In general we may decompose $\mathrm{Im}\,\boldsymbol{\tau}_{ij}$ in a regular block and one which is in the modular orbit of $i\infty$. The argument works as long as the regular block is non-trivial, that is as long as along $M_\ast$ not all photons decouple.

Then 

\begin{fact} If in $M$ there is a \emph{normal} ray with residual $R$-symmetry $\Z_n$, there should be a generator of $\mathscr{R}$ with $\Delta=n/s$, $s\in (\Z/n\Z)^\times$\, i.e.\! $e^{2\pi i/\Delta}$ is a primitive $n$-th root of unity.
$e^{2\pi i/\Delta}$ is an eigenvalue of a quasi-unipotent element $m_\ast\in Sp(2k,\Z)$, and hence $\phi(n)\leq 2k$.
\end{fact}

\begin{rem} Above we assumed the polarization $\Omega$ to be principal.
In the general case, we replace the Siegel modular group with the relevant 
duality-frame group $S(\Omega)_\Z$.
The conclusion $\phi(n)\leq 2k$ being still valid.
\end{rem}

\subsubsection{The Stein tubular neighborhood of a complex ray}
A ray $M_\ast\subset M$ is a closed analytic subset, hence a Stein submanifold of the Stein manifold $M$, and the restriction
map $\mathscr{R}\to \mathscr{R}_\ast$ is essentially surjective (the image is dense in the Frech\'et sense).
The Docquier-Grauert theorem \cite{stein2} guarantees that we can find a Stein tubular (open) neighborhood $M_\circ$ of $M_\ast$ in $M$. There is a holomorphic retraction of $M_\circ$ onto $M_\ast$, and $M_\circ$ is biholomorphic to a neighborhood of the zero section in the normal bundle of $M_\ast$ in $M$ \cite{stein2}. The fact that $M_\circ$ retracts holomorphically to $M_\ast$, means that the monodromy group of the special geometry restricted to $M_\circ$ is just the cyclic group generated by $m_\ast$. In particular, the $a$-periods are well-defined on an unbranched cover $\widetilde{M}_\circ\setminus\{0\}\to M_\circ\setminus\{0\}$ with Galois (deck) group $\Z_n$ generated by $m_\ast$.

Since $(M, \co_M)$ is a normal analytic space, so is $(M_\circ, \co_M|_{M_\circ})$. We have the restriction morphisms
\be
\mathscr{R}\xrightarrow{\;\text{mono}\;}\mathscr{R}_\circ\xrightarrow{\;\text{epi}\;} \mathscr{R}_\ast.
\ee
\begin{lem} The map $\mathscr{R}\to \mathscr{R}_\circ$ preserves the spectrum of $\ce$.
\end{lem}
\begin{proof} Suppose on the contrary that there is an eigenvalue $\lambda$ of $\ce_\circ$ which is not in the spectrum $S$ of $\ce$, and let $f\in \mathscr{R}_\circ$ be a non-zero eigenfunction. Let $x\in M_\circ$ be such that $f(x)\neq0$, and extend $f$ by homogeneity on the closed complex ray
generated by $x$, $M_x\equiv \overline{\zeta^\ce\cdot x}$. Since $i_x\colon M_x\to M$ is a closed embedding and $f\in \Gamma(M_x, i_x^\ast\co_M)$, by Cartan's theorem $f$ extends to a function in $\Gamma(M,\co_M)$ and then $\lambda\in S$.\end{proof}

\subsubsection{Tubular neighborhoods and
the Universal Dimension Formula}
\label{steintube}

Identifing the tubular neighborhood $M_\circ$ of the  normal ray $M_\ast$ with a neighborhood of the zero-section in the normal bundle, we may introduce
homogeneous complex coordinates 
\be
(u,v_1,v_2,\cdots, v_{k-1})
\ee
 such that 
$M_\ast\subset M_\circ$ is given by the analytic set 
$v_1=v_2=\cdots= v_{k-1}=0$, while $u=u_\ast$ is the coordinate along the normal ray $M_\ast$. Indeed, the additional coordinates $v_i$ are just linear coordinates along the fibers of the holomorphic normal bundle. The $v_i$ are globally defined in $M_\circ$ since the holomorphic normal bundle of $M_\ast$ is holomorphically trivial. This follows from a result of Grauert (cfr.\! \textbf{Theorem 5.3.1}(iii) of \cite{stein2}) since $\dim_\C M_\ast=1$.

 We saw in the previous subsection that $u$ is homogeneous of degree $1/\alpha_\ast$. Along the ray $M_\ast$ only a complex-linear combination of the  $a$-periods, $a_\ast$, does not vanish.
In the tubular neighborhood $M_\circ$ all $k$ linear combinations of the $a$-periods are not (identically) zero.
In a \emph{conical} special geometry the $a$ periods transform through the modular factors 
\be
a^\prime = (C\boldsymbol{\tau}+D)a,
\ee
where $\boldsymbol{\tau}$ is the (constant) period matrix on $M_\ast$. Hence, if $m_\ast$ is a regular elliptic element of $Sp(2k,\Z)$ and
$(e^{2\pi i\alpha_\ast}, e^{2\pi i \beta_1},\cdots, e^{2\pi i \beta_{k-1}})$ is the spectrum of $C\boldsymbol{\tau}+D$, with $e^{2\pi i\alpha_\ast}$ a primitive $n$-th root of unity, we may find  complex-linear combinations of $a$-periods $a_\ast,a_i$ in $M_\circ$ which diagonalize the action of $m_\ast$
\be\label{jza19s}
\begin{aligned}
a_\ast^\prime&= e^{2\pi i \alpha_\ast}\,a_\ast,&&0<\alpha_\ast\leq 1\\ 
a_s^\prime &= e^{-2\pi i \beta_s} a_s&& s=1,\dots, k-1,
\ \ 0\leq \beta_s <1,
\end{aligned}
\ee
where $a_\ast$ is the  linear combination non-zero along the ray $M_\ast$. 

\begin{fact}
The generalization of the $k=1$ equation \eqref{mmmzxq} to the tubular neighborhood $M_\circ$ of a \emph{regular normal} ray $M_\ast$ is 
\begin{align}
a_\ast&\propto u^{\alpha_\ast},\label{first}\\ 
a_s&\propto v_s\, u^{-\beta_s}\quad\text{for } s=1,\dots, k-1.\label{second}
\end{align}
\end{fact}
\begin{proof} Eqn.\eqref{first} is just the previous result along the normal ray $M_\ast$. Let us consider the $\C$-periods $a_s$ vanishing along $M_\ast$.
The $a$-periods should transform with the correct monodromy $m_\ast$ along the path \be
u\to e^{2\pi i t}u,\ \ t\in[0,1],\quad v_s=\mathrm{const}
\ee
 cfr.\! eqn.\eqref{jza19s}. Thus we must have
\be
a_s=f_s(v)\,u^{\alpha_s}\qquad \text{with } f_s(0)=0\ \text{and } 2\pi i\alpha_s=\log\!\big(e^{-2\pi i\beta_s}\big).
\ee
It remains to fix the functions $f_s(v)$ and the branch of the logarithm giving the correct value of $\alpha_s$. The holomorphic symplectic form $\Omega=da^i\wedge(dx_i-\boldsymbol{\tau}_{ij} \,dy^i)$ should have maximal rank when restricted to $M_\ast$, so that $df_1\wedge\cdots\wedge df_{k-1}|_{M_\ast}\neq0$ and hence
\be
f_s= A_{st}\,v_t+\text{higher order}\qquad \det A\neq0.
\ee
We may set $A_{st}$ to the identity matrix by a linear redefinition of the $v_s$. Since the $f_s$ are homogeneous functions, the higher order corrections vanish. Finally the branch of the logarithm is fixed to $\alpha_s\equiv -\beta_s$ by the requirement that the dimensions satisfy the unitary bound $\Delta[v_s]\geq 1$.
\end{proof}

Then the $k$ dimensions of the generators of $\mathscr{R}$ are
\be\label{kkkazqw}
\Delta_i\equiv \Delta(v_i)= \begin{cases} 1+\beta_i/\alpha_\ast & i=1,\dots, k-1\\
1/\alpha_\ast & i=k,
\end{cases}= 1+\frac{\beta_i}{1-\beta_k}
\ee
where we set $v_k\equiv u_\ast$ and $\beta_k=1-\alpha_\ast$.
\medskip 

If $M_\ast$ is normal but non regular, we cannot determine all dimensions from an analysis in the neighborhood $M_\circ$ of $M_\ast$ but only as many as its regular rank $k_\text{reg}$.

\begin{rem} These formulae have the following natural property.
Let $m_\ast$ be weakly regular, that is, some eigenvalue $\zeta$ of $\cu(m_\ast)$ have multiplicity $s>1$. Assume $a_\ast$ is an eigenperiod associated to $\zeta$; then
$\alpha_\ast=\alpha$ while $(s-1)$ $\beta$'s are equal $1-\alpha$. The dimensions of the $s$ operators associated to the eigenvalue $\zeta$ are all equal to $1/\alpha$, without distinction between the operator parametrizing the ray and the operators parametrizing its normal bundle. This property guarantees that we get the correct dimension spectrum for SCFT whose Coulomb branch is birational to a product of identical cones, so that the largest dimension $\Delta_\text{max}$ is degenerate. This happen e.g.\! in class $\cs[A_1]$ SCFTs where $\{\Delta_i\}=\{2,2,2,\cdots,2\}$.  
\end{rem}

The (universal) dimension formula \eqref{kkkazqw}, if correct, should pass three crucial consistency checks:
\begin{itemize}
\item[a)] it should reproduce the well-known formulae for constant period maps, in particular for all weakly-coupled Lagrangian SCFTs.
\item[b)] it should reproduce the Picard-Lefschetz results for models engineered by hypersurface singularities in $F$-theory;
\item[c)] it should produce the same spectrum of dimensions independently of which normal regular ray $M_\ast\subset M$ we consider;
\end{itemize}
The third requirement is quite strong, and it seems \emph{a priori} quite unlike that such a strong property may be acutally true. We perform the three checks in turn.

\subsubsection{Relation to Springer Theory}\label{spirpri}

We have to check that the ``abstract'' dimension formula \eqref{kkkazqw} reproduces the obvious dimensions for a weakly-coupled Lagrangian SCFT and more generally for all CSG with constant period maps of the form $M=\C^k/\cg$ for a degree-$k$ ST group $\cg$
whose chiral ring $\mathscr{R}$ coincides with the ring of polynomial invariants (see \S.\ref{connstantmaoa}). 

That eqn.\eqref{kkkazqw} correctly reproduces the ST degrees $d_i$ as dimensions $\Delta_i$ of the generators of $\mathscr{R}$ is a deep result in the Springer Theory of regular elements in finite reflection groups 
\cite{springer,cohen,springertheory}.

We recall the definitions: let the finite group $\cg$ act as a reflection group on the $\C$-space $V$. 
A vector $v\in V$ is said to be \emph{regular}
iff it does not lay in a reflection hyperplane.
An element $g\in \cg$ is said to be \emph{regular} if it has a regular eigenvector $v$. The \emph{regular degrees} of $\cg$ are a (minimal) set of integer numbers such that the order of all regular elements of $\cg$ is a divisor of an element of the set and conversely all divisors of these numbers are the order of a regular element. Then

\begin{thm}[see \cite{springer,cohen}]\label{sprr}
Let $\zeta$ be a primitive $d$-root of unity. Let $g\in\cg$ be regular with regular eigenvector $v\in V$ and related eigenvalue $\zeta$. Denote by $W$ the $\zeta$-eigenspace 
\be
W=\{x\in V\;|\; gx=\zeta x\}.
\ee
 Then:
\begin{itemize}
\item[(i)] $d$ is the order of $g$, and $g$
has eigenvalues $\zeta^{1-d_1}$, $\zeta^{1-d_2},\cdots, \zeta^{1-d_k}$, where $d_i$ are the degrees of $\cg$;
\item[(ii)] $\dim W=\#\{i\;|\; d\ \textit{is a divisor of }d_i\}$;
\item[(iii)] the restriction to $W$ of the centralizer of $g$ in $\cg$ defines an isomorphism onto a reflection group in $W$ whose degrees are the $d_i$ divisible by $d$ and whose order is $\prod_{d\mid d_i}d_i$;
\item[(iv)] the conjugacy class of $g$ consists of all elements of $\cg$ having $\dim W$ eigenvalues $\zeta$.  
\end{itemize} 
\end{thm}
One can show that an integer is regular iff it divides as many degrees as co-degrees \cite{refl1}.

\begin{rem}\label{kkcc4443} All irreducible ST groups have at least one regular degree. In facts, they have either 1 or 2 regular degrees, except for the Weyl groups of $E_6$ and $E_8$ which have \textit{three} regular degrees each, respectively $\{8,9,12\}$ and $\{20,24,30\}$.
If $\cg$ is a Weyl group, the Coxeter number $h$ is one of the regular degrees. See refs.\!\cite{cohen,refl1} for tables of regular degrees. 
\end{rem}
\begin{corl}\label{dim1} Going through the the tables {\rm\cite{cohen,refl1}}, one sees that
 part (ii) of the {\rm\textbf{Theorem}} implies that for an irreducible \emph{crystallographic} Shephard-Todd group $\dim W=1$ for all regular degrees $d$. In other words, in the crystallographic case, the only degree which is an integral multiple of the regular degree $d$ is $d$ itself.
\end{corl}
\medskip

Let us re-interpret \textbf{Theorem \ref{sprr}} in the context of the constant-period class of CSG's discussed in
\S.\,\ref{connstantmaoa}, with Coulomb branch $M=\C^k/\cg$,
$\cg$ a degree-$k$ crystallographic ST group, and chiral ring 
$\mathscr{R}=\C[a^1,\cdots,a^k]^\cg$. For simplicity, we take $\cg$ irreducible.

  By definition, $v\in \C^k$ is a regular vector iff it does not lay on a reflection hyperplane, i.e.\! if its projection $\tilde v$ in the Coulomb branch $M=\C^k/\cg$ does not belong to the discriminant locus $\cd\subset M$, i.e.\! iff $\tilde v\in M^\sharp$, the smooth locus.
  
 Let $d$ be a regular weight of $\cg$ (so $d\equiv d_{i_0}$ for some $i_0$). 
 By definition, there is an element $m_{i_0}\in \cg$ of order precisely $d_{i_0}$. Let $v$ be a regular eigenvector of $m_{i_0}$ corresponding to the primitive $d_{i_0}$-th root\footnote{\ This is a special choice for complex $\cg$;  in the case of a Lagrangian SCFT, i.e.\! if $\cg$ is a Weyl group, this choice does not imply any loss of generality.} $\zeta=e^{2\pi i/d_{i_0}}$. The (closure of the) $\C^\times$-orbit of $\tilde v\in M^\sharp$,
$M_v\subset M$, is a complex ray not lying in the discriminant $\cd$ (more precisely, intersecting $\cd$ only at the tip). 

We claim that $M_v$ is also normal. Recall that 
$\mathscr{R}=\C[a^1,\cdots,a^k]^\cg\equiv \C[u_1,\cdots, u_k]$ by the Shephard-Todd-Chevalley theorem.
Homogeneity implies
\be
u_i\big|_{M_v}= c_i\, (\lambda v)^{d_i} \quad \lambda\in \C, 
\ee
for some constants $c_i$. Applying $m_{i_0}$ on both sides, and using \textbf{Corollary \ref{dim1}} i.e.
\be
\frac{d_i}{d_{i_0}}\in \mathbb{N}\quad\Longrightarrow\quad i\equiv i_0,
\ee
we conclude that
\be
M_v\equiv \big\{ u_i=0\ \text{for }i\neq i_0\big\}\equiv M_{i_0},
\ee 
is automatically a normal right of the form in eqn.\eqref{hhhhq}. This also shows that $m_{i_0}$ is the monodromy along the normal ray $M_v$, which is then regular. 
This is exactly the set up in which we deduced the universal dimension formula
\eqref{kkkazqw}.  
From item \textit{(i)} in the \textbf{Theorem \ref{sprr}} we see that 
\be
\beta_i=(d_i-1)/d\quad \text{and}\quad \alpha_\ast=1-\beta_{i_0}=1/d.
\ee
The universal formula \eqref{kkkazqw} then yields
\be
\Delta_i= 1+d\, \beta_i= d_i,
\ee  
which is the correct result. 

Thus the formula \eqref{kkkazqw} is nothing else than the plain extension to the non-Lagrangian SCFT of the usual formula, valid for all Lagrangian SCFT, following from the standard supersymmetric non-renormalization theorems.

\subsubsection{Recovering Picard-Lefschetz
for hypersurface singularities}

We consider the hypersurface
\be
F\equiv\left\{ W(x_1,x_2,x_3,x_4)+\sum_{a=0}^{k-1} u_a\, \phi_a(x_i)=0\right\}\subset \C^4,
\ee
where $u_a$ has
dimension $\Delta_0(1-\deg \phi_a)$,
$\Delta_0$ being the relative normalization of the $R$-charges in the 4d and 2d sense under the 4d/2d correspondence \cite{Cecotti:2010fi}. We consider the ray parametrized by the coupling $u_0$ of the 2d identity operator
\be
M_0=\{u_1=\cdots=u_{k-1}=0\}\subset M.
\ee
The corresponding monodromy $m_0$,
along the path $u_0\to e^{2\pi it}u_0$, $t\in[0,1]$ is, by construction, the one induced on $H_1(F,\Z)/\mathrm{rad}\,\langle-,-\rangle$ by the classical monodromy $H$ of the  hypersurface singularity, that is, ($f\equiv \mathrm{rank}\,\mathrm{rad}\,\langle-,-\rangle$)
\be
\det[z-H]=(z-1)^f\; \det[z-m_0],\qquad 
\det[z-m_0]=\prod_{\ell\geq 2}\Phi_\ell(z)^{s(\ell)}.
\ee
Thus the spectrum of $m_0$ is
\be
\mathrm{spec}(m_0)=\Big\{e^{2\pi i(q_a-\hat c/2)}\colon q_a\ U(1)_R\ \text{charge of 2d chiral primaries with }q_a\neq \hat c/2\Big\}.
\ee
This way of writing implicitly selects a special embedding of $m_0$ into $Sp(2k,\Z)$ as well as the eingenvalue which corresponds to the non trivial period $a_0$ on $M_0$; these are the canonical choices dictated by Picard-Lefschetz theory. This choice yields
\be
\alpha_0 = 1-\frac{\hat c}{2},\qquad \beta_a =\frac{\hat c}{2}-q_a,\ \quad 0<q_a<\hat c/2,
\ee 
so that,
\be
\Delta_0=\frac{1}{1-\hat c/2},\qquad
\Delta_a=1+\frac{\hat c/2-q_a}{1-\hat c/2}\equiv 
\frac{1-q_a}{1-\hat c/2},\quad 0<q_a<\hat c/2,
\ee
which precisely yields back the Picard-Lefschetz formula
\eqref{plcase} ($q_a\equiv \deg \phi_a$ by definition).

\begin{exe}[Complete intersections of singularities] The previous argument applies to all SCFT which have a 4d/2d correspondent in the sense of \cite{Cecotti:2010fi}; in the general case the classical monodromy of the singularity $H$ should be  replaced by the (2,2) quantum monodromy as defined in \cite{Cecotti:1992rm}. The eigenvalues of the (2,2) quantum mondromy have the form $e^{2\pi i (q_a-\hat c/2)}$ where $q_a$ are the $U(1)_R$ charges of the 2d chiral operators \cite{Cecotti:1992rm}. For instance, this result applies to the models engineered by the complete intersection of two hypersurface singularities in $\C^5$ \cite{Wang:2016yha}
\be
W_1(x_i)+\sum_\alpha u_{1,\alpha}\,\phi_{1,\alpha}(x_i)=W_2(x_i)+\sum_\alpha u_{2,\alpha}\,\phi_{2,\alpha}(x_i)=0,
\ee 
where $W_a(x_i)$ are quasi-homogeneous of degree $d_1\equiv 1$ and $d_2$ (we assume $d_1\leq d_2$ with no loss)
$W_1(\lambda^{w_i}x_i)=\lambda W_1(x_i)$,
$W_2(\lambda^{w_i}x_i)=\lambda^d W_2(x_i)$ and $\phi_{a,\alpha}(x_i)$ is a basis of the admissible deformations of the equations.\footnote{\ That is, a basis of the Jacobian ring.} 
This correspond to a 2d model with superpotential \cite{Wang:2016yha}
\be
\cw(x_i)=
W_1(x_i)+\sum_\alpha u_{1,\alpha}\,\phi_{1,\alpha}(x_i)+\Lambda\Big(W_2(x_i)+\sum_\alpha u_{2,\alpha}\,\phi_{2,\alpha}(x_i)
\Big)
\ee
where $\Lambda$ is an extra chiral field of $U(1)_R$ charge $q_\Lambda=1-d_2$.  
The scaling dimension of the parameters $u_{a,\alpha}$ is $1-q_{a,\alpha}$ where
$q_{\alpha,\alpha}$ is the charge of the corresponding operator perturbing $\cw$, i.e.\! $\phi_{1,\alpha}$ and $\Lambda\phi_{2,\alpha}$, respectively. One has $1-\hat c/2=\sum_i w_i-d_1-d_2$ and so \cite{Wang:2016yha}
\be
\Delta(u_{a,\alpha}) = \frac{d_a-\deg \phi_{a,\alpha}(x_i)}{\sum_i w_i -\sum_a d_a}.
\ee
\end{exe}

\begin{exe}[The DZVX models]
A third class of 4d SCFT with a nice 2d correspondent is the one constructed in \cite{DelZotto:2015rca} parametrized by a pair (affine star, simply-laced Lie algebra).
Our formulae yield the correct dimension spectrum by construction. 
\end{exe}

\subsubsection{Consistency between different normal rays}
\label{conray}

Let us consider a simple (but instructive) class of examples, the Argyres-Douglas (AD) models of type $A_{N-1}$ with $N$ odd. The special geometry of the $A_{N-1}$ AD model corresponds to the intermediate Jacobian of the following family of hypersurfaces in $\C^4$
\be\label{kkaz10e}
y^2+w^2+z^2= x^N+\sum_{a=0}^{(N-3)/2}u_a\, x^a,
\ee
the period matrix $\boldsymbol{\tau}$ being degenerate on the locus where the discriminant of the polynomial in the 
\textsc{rhs} vanishes.
The coefficient $u_a$ ($a=0,1,\dots,[(N-3)/2]$) have degree $2(N-a)$, so
we have a map
\be
M\setminus\{0\}\to \mathbb{P}\!\big(2N,2N-2,\cdots, N+3\big),
\ee
through which the period map $\boldsymbol{\tau}$ factorizes.
We consider the normal rays
\be
M_a=\{u_b=0,\ b\neq a\}\subset M\qquad a=0,1,\cdots, (N-3)/2.
\ee 
Along the ray $M_0$ we have an unbroken $R$-symmetry $\Z_{2N}$ with a (strongly) regular monodromy
\be
\det[z-m_0]= \frac{z^N+1}{z+1}.
\ee
Instead along the ray $M_1$ the unbroken  $R$-symmetry is $\Z_{2(N-1)}$ again with regular monodromy
\be
\det[z-m_1]= z^{N-1}+1.
\ee
Along a ray $M_a$ with $a\geq 2$ the discriminant of the polynomial in the \textsc{rhs} of \eqref{kkaz10e} vanishes, and  our considerations apply only to the regular factor.
The unbroken $R$-symmetry is $\Z_{2(N-a)}$ and 
we have
\be
\Big(\det[z-m_a]\Big)_\text{regular factor}=
\begin{cases} (z^{N-a}+1)/(z+1) & a\ \text{even}\\
z^{N-a}+1 & a\ \text{odd}.
\end{cases}
\ee
For each $a$ we have an embedding
\be
m_a\in Sp(2k_{a,\text{reg}},\Z)\equiv Sp(2[(N-a)/2],\Z)
\ee  
specified by the spectrum of $(C\boldsymbol{\tau}+D)_\text{reg.block}$ which consists of $k_{a,\text{reg}}$ out of the $2k_{a,\text{reg}}$ roots of the characteristic polynomial of $m_a$: it is a set of
$2(N-a)$-th roots of unity satisfying the condition in equation \eqref{kkkkzzznncx}. We also need  to specify which one of the $k_{a,\text{reg}}$ roots corresponds to the non-trivial period $a_a$ on the ray $M_a$. There is a ``canonical'' Picard-Lefschetz choice. We have
(here $0\leq a\leq (N-1)/2$)
$$
\begin{array}{c|lllll}\hline\hline
a & \zeta& \text{embeding} & \alpha & \beta_\ell & \Delta_\ell\\\hline 
0 & e^{\pi i/N} & \zeta^{2N-2\ell+1} & (N+2)/2N & (2\ell-1)/2N & (N+2\ell+1)/(N+2)\\
&&& &\text{\begin{scriptsize}$\ell\neq(N-3)/2$\end{scriptsize}}& \text{\begin{scriptsize}$1\leq \ell\leq (N-1)/2$\end{scriptsize}}\\\hline
1& e^{\pi i/(N-1)} && (N+2)/(2N-2) & (2\ell-1)/(2N-2) & (N+2\ell+1)/(N+2) \\
&&&& \text{\begin{scriptsize}$\ell\neq (N-3)/2$\end{scriptsize}}&
\text{\begin{scriptsize}$1\leq \ell\leq (N-1)/2$\end{scriptsize}}\\\hline
a & e^{\pi i/(N-a)} & \zeta^{2(N-a-\ell)+1} & (N+2)/(2N-2a) & (2\ell-1)/(2N-2a) & (N+2\ell+1)/(N+2)\\
&& &&\text{\begin{scriptsize}$\ell\neq (N-2a-1)/2$\end{scriptsize}} &
\text{\begin{scriptsize}$1\leq \ell\leq [(N-a)/2]$\end{scriptsize}}\\\hline\hline
\end{array}
$$
we see that the several rays $M_a$ yield
mutually consistent results for the spectra of dimension. $M_0$ and $M_1$ yield the full set of dimensions (which agrees between the two rays), while $M_{2\ell}$ and $M_{2\ell+1}$ yield a partial list of $k-\ell$ out of the $k$ dimensions. (But the formal ``analytic continuation'' gives the full correct set of dimensions at all normal rays).

\begin{rem} Note that the list of the $k_\text{reg}$ dimensions computed from a $M_a$ is a $k_\text{reg}$-tuple of dimensions which is allowed for a rank $k_\text{reg}$ SCFT. In particular, the dimension of the operator parametrizing a ray with
$k_\text{reg}=1$ should be in the one-dimensional list
$\{1,2,3,4,6,3/2,4/3,6/5\}$.
\end{rem}

\begin{exe} Consider the Argyres-Douglas models of type $D_5$ and $D_6$ which have $k=2$. The ray parametrized by the operator of the largest dimension corresponds to Picard-Lefschetz theory and is regular, while the one associated to the operator of lesser dimension is non-regular. We deduce by the previous \textbf{Remark} that for these models the smaller of the two Coulomb dimensions should belong to the $k=1$ list. Indeed it is $6/5$ for type $D_5$, and $4/3$ for type $D_6$. This statement may be generalized
in the form of inter-SCFT consistency conditions relating the spectrum of dimensions in different SCFT. For $\mathfrak{g}\in ADE$
we write $\{\Delta\}_\mathfrak{g}$ for the set of Coulomb branch dimensions of the Argyres-Douglas model of type $\mathfrak{g}$ and we write $\{\Delta\}_\mathfrak{g}^{(s)}$
for the subset obtained from $\{\Delta\}_\mathfrak{g}$ by omitting the $s$ largest dimensions. Then we have the relations
\be
\{\Delta\}_{A_n}= \{\Delta\}_{D_{n+3}}^{(1)}.
\ee
\end{exe}

We present a more complicated example of such consistency condition between Coulomb dimensions in different ranks.

\begin{exe}[Argyres-Douglas of type $E_8$]
This SCFT has $k=4$ with Coulomb dimensions $\Delta_1=15/8$, $\Delta_2=3/2$
$\Delta_3=5/4$ and $\Delta_4=9/8$.
Since $\tfrac{1}{2}\phi(15)=4$ from the ray $M_1$ we should be able to compute all 4 dimensions (this is classical Picard-Lefschetz for the $E_8$ minimal singularity). For the other normal rays we have consistency requirements saying that a subset of the dimension set $\{15/8,3/2,5/4,9/8\}$ should be the dimension $k^\prime$ tuple for the appropriate regular rank $k^\prime<k$.
Since $\tfrac{1}{2}\phi(3)=1$, from $M_2$ we compute the single dimension $3/2$ which is in the $k=1$ list.
$\tfrac{1}{2}\phi(5)=2$ so from $M_3$ we get a pair of dimensions in the $k=2$ list, namely $\{5/4,3/2\}$.
Finally $\tfrac{1}{2}\phi(9)=3$ and 3 dimensions out of 4 should form a dimension triple for rank-3. This is consistent since both $\{9/8,15/8,3/2\}$ and $\{9/8,3/2, 5/4\}$ are in the $k=3$ list.
\end{exe}

\subsubsection{A comment on the conformal manifold}

As further evidence of the correctness of the universal dimension formula \eqref{kkkazqw}, let us consider the conformal manifold, that is, the space of moduli (deformations) of the CSG.
We have already discussed the special case in which the period map is constant while 
the chiral ring is the invariant subring $\C[a_1,\dots,a_k]^\cg$ (see \S.\,\ref{connstantmaoa}). In that case the dimension of the conformal manifold was given by the number of generators of the chiral ring of dimension 2.
In the general case, the allowed continuous deformations of the CSG are described by the rigidity principle (\textbf{Proposition \ref{rigidity}}). Since the monodromy representation is a discrete datum, a continuous deformation of the CSG is uniquely determined by the deformation it induces on a single Abelian fiber $X_u$. Suppose that the CSG under considerations has a normal ray with semi-simple monodromy $m_\ast$. We focus on a fiber over a point in the ray. If the monodromy is regular, by definition there is no deformation of the fiber, that is, the conformal manifold reduces to an isolated point. On the other hand, $m_\ast$ regular implies that no chiral operator has dimension 2. Indeed, in an interacting theory a chiral operator of dimension 2 is necessarily a generator of $\mathscr{R}$.
Suppose the dimension of the $i$-th generator is 2. Then from eqn.\eqref{kkkazqw} 
\be
1+\frac{\beta_i}{1-\beta_0}=2\quad\Rightarrow\quad e^{-2\pi i\beta_i}\,e^{-2\pi i\beta_0}=\zeta_i\zeta_0=1,
\ee 
contradicting the assumption that $m_\ast$ is regular. If $m_\ast$ is semi-simple, but not regular, the same argument shows that the dimension $\mathsf{d}$ of the conformal manifold is $\leq$ the number of generators
of $\mathscr{R}$ having dimension 2. This is the physically expected result: the number $\mathsf{d}$ of exactly marginal operators is not greater than the number of chiral operators of dimension 2.
Note that this is just an inequality since being $m_\ast$-invariant is only a necessary condition for a deformation to be allowed. 
Some of the $m_\ast$-invariant deformations may be obstructed by other elements of the monodromy group; comparison with the constant period map case shows that the obstruction comes from the non-semisimple part of the monodromy group. The non-semisimple part of the monodromy measures the effective one-loop beta-function of the QFT (cfr.\! discussion around eqn.\eqref{kkkaqznn}). Thus the dimension formula \eqref{kkkazqw} is consistent with the physical expectations on the conformal manifold.

\subsection{The set $\Xi(k)$ of allowed dimensions}

We write 
$\Xi(k)\subset (\mathbb{Q}_{\geq1})$ for the set of all rational numbers which appear as dimension $\Delta$ of a generator of the chiral ring $\mathscr{R}$ in a CSG of rank $\leq k$.
Clearly $\Xi(k)$ is monotonic in $k$: 
 $\Xi(k-1)\subset \Xi(k)$. If $M$ is a rank $k_2$ CSG, its symmetric power $M^{[k_1]}$ -- if not too singular -- should also be a CSG.
 This rather sloppy argument would suggest that
 the set $\Xi(k)$ also satisfies the following requirement: 
 for all $k_1,k_2\in \mathbb{N}$
\be\label{vvvcz}
\bigcup_{s=1}^{k_2} s\cdot \Xi(k_1)\subset \Xi(k_1k_2),
\ee
where $s\cdot\Xi(k)$ stands for the set of rationals obtained by multiplying all rationals in $\Xi(k)$ by the integer $s$.

To determine $\Xi(k)$ it suffices to give the difference of the sets for two successive ranks $k$; we already know that
$\Xi(1)=\{1,2,3,4,6,3/2,4/3,6/5\}$. 
The \textbf{Fact} stated at then end of \S.\ref{jjjjjjz} implies
\be\label{kkkazxcv}
\Xi(k)\subseteq \widehat{\Xi}(k)\equiv \left\{\frac{\ell}{s}\in\mathbb{Q}_{\geq1}\ \bigg|\ \phi(\ell)\leq 2k,\ (\ell,s)=1\right\}.
\ee
\textit{A priori} this is just an inclusion, that is, the conditions we got insofar are just necessary conditions. However, experience with the first few $k$'s suggests that the two sets $\Xi(k)$ and $\widehat{\Xi}(k)$ are pretty close, and likely equal. The discrepancy (if there is any) is expected to vanish as $k$ increases. Note that 
\be
\BQ_{\geq 1}=\bigcup_{k=1}^\infty\widehat{\Xi}(k),
\ee
that is, all rational numbers $\geq 1$ appear in the list for some (large enough)
rank $k$. For instance if $\Delta$ is a \emph{large integer,} the Coulomb dimension $\Delta$ will first appear in rank $k_\text{min}$ \cite{barkley}
\be
k_\text{min}= \frac{1}{2}\,\phi(\Delta)> \frac{1}{2}\,\frac{\Delta}{e^\gamma\,\log\log\Delta+\frac{A}{\log\log\Delta}},\quad\text{where }A=2.50637,
\ee
while (for comparison) the minimal rank for a Lagrangian SCFT is 
\be
k_\text{min,Lag.}=\Delta/2 > \phi(\Delta)/2.
\ee
\smallskip

Eqn.\eqref{kkkazxcv} is shown by recursion in $k$.
For $k=1$ we know that it is true with
$\subseteq$ replaced by $=$.
 We consider the rays $M_\ast\subset M$ generated by the vev of a single chiral field of  dimension $\Delta$; along $M_\ast$  a $R$-symmetry $\Z_\ell$ is preserved,  $\ell$ being the order of $1/\Delta$ in $\mathbb{Q}/\Z$,
so that $1/\Delta= s/\ell$ with $1\leq s\leq \ell$ and $(s,\ell)=1$. We have $\ell=1,2$ iff $\Delta=1,2$. Now let $m_\ast$ be the corresponding element of the monodromy.
$m_\ast^\ell=1$. $m_\ast$ acts on the $\C$-period $a_\ast$ non-zero on $M_\ast$ by a primitive $\ell$-th root of unity, hence 
the cyclotomic polynomial $\Phi_\ell(z)$
divides $\det[z-m_\ast]$, so that
$\phi(\ell)=\deg \Phi_\ell(z)\leq \deg\det[z-\mu_\ast]=2k$. The monotonicity of $\widehat{\Xi}(k)$ is obvious.  Eqn.\eqref{vvvcz} for $\widehat{\Xi}(k)$ is equivalent to the inequality
$\phi(\ell k_2)\leq \phi(\ell)\,k_2$ which holds by Euler's formula
\be
\frac{\phi(\ell k_2)}{\phi(\ell)\, k_2}=\prod_{p\mid k_2\atop p\nmid \ell}\left(1-\frac{1}{p}\right)\leq 1.
\ee
For bookkeeping, it is convenient to list the 
(candidate) ``new-dimensions''
at rank $k$
\be
\mathfrak{N}(k)\equiv \widehat{\Xi}(k)\setminus \widehat{\Xi}(k-1).
\ee

\begin{exe} For instance, for $k=2$
\be\label{aswq}
\mathfrak{N}(2)=\left\{ 5,\ \frac{5}{2},\ \frac{5}{3},\
\frac{5}{4},\ 8,\ \frac{8}{3},\ \frac{8}{5},\ \frac{8}{7},\ 10,\ \frac{10}{3},\ \frac{10}{7},\ \frac{10}{9},\ 12,\ \frac{12}{5},\ \frac{12}{7},\ \frac{12}{11}\right\}.
\ee
\end{exe}

\begin{rem} After the completion of this paper, the paper \cite{Argyres:2018zay} appeared on the arXiv in which
$\Xi(2)$ is also determined. The agreement with  \eqref{aswq} is perfect.
\end{rem}

The new-dimension sets $\mathfrak{N}(k)$ up to $k=13$ are listed in table \ref{new1}.

\subsubsection{The number of allowed  dimensions at given rank $k$} 
The number of elements of $\mathfrak{N}(k)$ (resp.\! the number of \emph{integers}  in $\mathfrak{N}(k)$) is 
\be\label{xxyy1}
|\mathfrak{N}(k)|=2k\cdot\nu(2k)\qquad\text{resp.}\qquad 
|\mathfrak{N}(k)|_\text{int.}=\nu(2k),
\ee
where the Number-Theoretic function $\nu(d)$ is the \emph{totient multiplicity} of $d\in\mathbb{N}$ \cite{handbb}, that is, the number of solutions to $\phi(x)=d$. 
A positive integer $d$ is called a \emph{totient} iff it belongs to the range of $\phi$, i.e.\! if $\nu(d)>0$; an integer $d$ is called a \emph{nontotient} if it not a totient, i.e.\! if $\nu(d)=0$.  Thus, if $2k$ is a \emph{nontotient} there are no new-dimensions in rank $k$, $\mathfrak{N}(k)=\emptyset$, and 
 $\widehat{\Xi}(k)=\widehat{\Xi}(k-1)$. The first few even nontotients are (see sequence A005277 in OEIS \cite{oeis})
 \be
 14,\;26,\; 34,\; 38,\; 50,\; 62,\; 68,\; 74,\;76,\; 86,\; 90,\; 94,\; 98,\; 114,\, \cdots
 \ee
hence 
\begin{equation}
\begin{split}
\mathfrak{N}(7)&=\mathfrak{N}(13)=\mathfrak{N}(17)=\mathfrak{N}(25)=
\mathfrak{N}(31)=\mathfrak{N}(34)=
\mathfrak{N}(37)=\\
&=\mathfrak{N}(38)=
\mathfrak{N}(43)=\mathfrak{N}(45)=
\mathfrak{N}(47)=\mathfrak{N}(49)=\mathfrak{N}(57)=\cdots=\emptyset.
\end{split}
\end{equation}
The first few valued of the totient multiplicity 
$\nu(2k)$ (for $k\in\mathbb{N}$) are
(see sequences A014197 or A032446 in OEIS \cite{oeis})
\be\label{xxyy2}
\nu(2k)=3,\,4,\,4,\,5,\,2,\,6,\,0,\,6,\,4,\,5,\,2,\,10,\,0,\,2,\,2,\,7\,,0,\,8,\,0,\,9,\,4,\,3,\,2,\,11,\,0,\,\cdots
\ee 

 \medskip

\subparagraph{Properties of $\nu(d)$.} We list some useful properties of the function $\nu(d)$  \cite{handbb}:

\noindent\textbf{1)} Nontotients have density 1 in $\mathbb{N}$, i.e.\! totients are ``sparse'' (density zero) so for ``most'' ranks there are \emph{no} new-dimensions.
More precisely, let $N(x)$ be the number of totients less or equal $x$, 
\be
N(x)=\#\big\{m\;\big|\; \exists\, n\in\mathbb{N}\colon \phi(n)=m,\ \text{and }m\leq x\big\},
\ee
 then 
for all $\varepsilon>0$ there exists $x(\varepsilon)\in\mathbb{N}$ such that
\be
N(x)< \frac{x}{(\log x)^{1-\epsilon}}\quad\text{for }x>x(\varepsilon).
\ee
\textbf{2)} Every integer has a multiple which is a nontotient. 

\noindent\textbf{3)} The function $\nu(d)$ takes all integral values $\geq 2$ infinitely times. 

\noindent\textbf{4)} The Carmichael conjecture (still open) states that $\nu(d)\neq1$ for all $d$; the conjecture is known to be true for
$d<10^{10^{10}}$. 

\noindent\textbf{5)} One has
\be
k\ \text{odd}\quad\Rightarrow\quad |\mathfrak{N}(k)|\leq 8k.
\ee

\subparagraph{The cardinality $\boldsymbol{N}(k)$ of the set $\widehat{\Xi}(k)$.} $\boldsymbol{N}(k)$
may be written in the form of a Stieltjes integral 
\be
\boldsymbol{N}(k)\equiv|\widehat{\Xi}(k)|=2+\sum_{\ell=1}^k 2\ell\,\nu(2\ell)=2+\int\limits_{2-\epsilon}^{2k+\epsilon} x\, dV(x).
\ee
where $V(x)$ is the Erd\"os-Bateman Number-Theoretic function \cite{handbb,bateman}
\be
V(x)=\sum_{n\leq x} \nu(n),
\ee
whose values for $x\in \mathbb{N}$ are given by the sequence A070243 in OEIS \cite{oeis}. Note that $V(2k)$ is also the number $\boldsymbol{N}(k)_\text{int}$ of \emph{integral} elements of $\widehat{\Xi}(k)$
\be
\boldsymbol{N}(k)_\text{int}\equiv \big|\,\widehat{\Xi}(k)\cap \mathbb{N}\,\big|=V(2k).
\ee
The first few values are
\begin{equation}
\begin{split}
\boldsymbol{N}(k)_\text{int}=&\;5,\,9,\,13,\,18,\,20,\,26,\,26,\,32,\,36,\,41,\,43,\,53,\,53,\,54,\,57,\,64,\,64,\\
&72,\,72,\,81,\,85,\,88,\,90,\,101,\,101,\,103,\,105,\,108,\,110,\,119,\,119,\,127,\,\cdots
\end{split}
\end{equation}
For large $x$ \cite{handbb}\footnote{\ $\zeta(s)$ is the Riemann zeta-function.}
\begin{gather}
V(x)= \frac{\zeta(2)\,\zeta(3)}{\zeta(6)}\,x+ o\!\left(\frac{x}{(\log x)^C}\right)\qquad \forall\;C>0,\\
\frac{\zeta(2)\,\zeta(3)}{\zeta(6)}=1.9435964\dots
\end{gather}
so that as $k\to\infty$
\be
\boldsymbol{N}(k) = \frac{2\,\zeta(2)\,\zeta(3)}{\zeta(6)}\,k^2+o(k^2).
\ee
Various expressions for the error term may be found in ref.\!\cite{handbb,bateman}.
\medskip

To construct the putative new-dimension sets $\mathfrak{N}(k)$ we only need to solve the equation $\phi(x)=2k$ for all $k\geq 2$. There is an explicit algorithm to solve recursively this equation for a given $k$ once we know the solutions for all $k^\prime <k$ \cite{phik}.

\subsubsection{Analytic expressions of
$\boldsymbol{N}(k)$ and $\boldsymbol{N}(k)_\text{int}$}\label{anex}

We define $\boldsymbol{N}(x)_\text{int}$ for \emph{real} $x$ as
\be
\boldsymbol{N}(x)_\text{int}=V(2x)\equiv \sum_{n\leq 2x}\nu(n).
\ee
Note that $\boldsymbol{N}(1/2)_\text{int}=2$ corresponding to the rank-1 Lagrangian SCFT. With this convention, eqn.\eqref{veryfancy} holds\footnote{\ The $+\epsilon$ prescription is needed since $\boldsymbol{N}(x)$ is not a continuous function but rather a function of bounded total variation on compact subsets of $\R$. }
 \be\label{veryfancy2}
 \boldsymbol{N}(k)=2\int\limits_0^{k+\epsilon} x\;d\boldsymbol{N}\big(x\big)_\text{int}.
 \ee
 To obtain an analytic formula for 
 $\boldsymbol{N}(x)_\text{int}$ one starts from the Dirichlet generating function
 $\mathscr{N}(s)$ (chap.\;1 of \cite{ANT}) for the totient multiplicities $\nu(m)$. $\mathscr{N}(s)$ has an Euler product 
 \be\label{kkkkaqss}
 \mathscr{N}(s):=\sum_{m=1}^\infty \frac{\nu(m)}{m^s}=\sum_{n=1}^\infty\frac{1}{\phi(n)^s}=\prod_{p\colon\text{prime}} \left(1+\frac{p^s}{(p-1)^s\,(p^s-1)}\right).
 \ee
 The function $\mathscr{N}(s)/\zeta(s)$ is analytical in the half-plane $\mathrm{Re}\,s>0$, and takes the value $\zeta(2)\zeta(3)/\zeta(6)$ at $s=1$ \cite{bateman}. 
 Then the subtracted generating function
 \be
 \widetilde{\mathscr{N}}(s):= \mathscr{N}(s)-\frac{\zeta(2)\,\zeta(3)}{\zeta(6)}\;\frac{s}{s-1},
 \ee
is analytic in the full half-plane $\mathrm{Re}\,s>0$. The totient multiplicities $\nu(m)$
are obtained from $\widetilde{\mathscr{N}}(s)$ by taking the inverse Mellin transform
of eqn.\eqref{kkkkaqss}. Thus (cfr.\! \cite{bateman} \S.\,6)
\be\label{kkkkzzzz876}
\boldsymbol{N}(x)_\text{int}= \frac{2\,\zeta(2)\,\zeta(3)}{\zeta(6)}\,x+\lim_{t\to\infty} \int\limits_{2-it}^{2+it} 
\frac{(2x+\epsilon)^s}{s}\; \widetilde{\mathscr{N}}(s)\;\frac{ds}{2\pi i}.
\ee
(note the prescription for this convergent but not absolutely convergent integral).
From eqn.\eqref{kkkkzzzz876} the asymptotic formula eqn.\eqref{asssyym} follows by deforming the integration contour to a suitable path along which $\mathrm{Re}\,s<1$ \cite{bateman}.

\subsection{Dimension $k$-tuples $\{\Delta_1,\dots,\Delta_k\}$}
\label{kkktuple}

As discussed in \S.\,\ref{steintube}, from the local analysis at a regular (normal) ray we get the full $k$-tuple of dimensions $\{\Delta_1,\dots,\Delta_k\}$, a prototypical case being the Picard-Lefschetz theory of a SCFT engineered by $F$-theory on a singularity. More generally, from local considerations on a ray of regular rank $k_\text{reg}$ we get $k_\text{reg}$ out of the $k$ dimensions. The number $k$-tuples of allowed dimensions in presence of rays of regularity $k_\text{reg}\geq 2$ is much less than $|\Xi(k)|^k$ due to correlations between the dimensions of the  various chiral operators of a given SCFT.

\medskip

At a regular ray we have
\be\label{kkaq12}
\det[z-m_\ast]=\Phi_{d_1}(z)\Phi_{d_2}(z)\cdots\Phi_{d_t}(z)
\ee
where the factors are all distinct.
Two embeddings $m_\ast\hookrightarrow Sp(2k,\Z)$ are equivalent if they are conjugate in $Sp(2k,\Z)$. A weaker notion of equivalence is conjugacy in $Sp(2k,\R)$.
It follows from the considerations in \S.\,\ref{cyclicee} that a $Sp(2k,\R)$ conjugacy class is characterized by a subset of $k$ out of the $2k$ roots of the characteristic polynomial \eqref{kkaq12}
(namely the spectrum of $\cu(m_\ast)$),
$\{\zeta_1,\cdots, \zeta_k\}$, having the property that
$\zeta_i\zeta_j\neq1$ for all $1\leq i, j\leq k$
(cfr.\! eqn.\eqref{kkkkzzznncx}).
Once given the spectrum $\{\zeta_i,\cdots,\zeta_k\}$ of $\cu(m_\ast)$, we have to select which one of the $k$ roots is the eigenvalue associated to the non-zero (multivalued) period $a_\ast$ on $M_\ast$. Renumbering the roots so that this is the first one, we have
\be
e^{2\pi i\alpha_\ast}=\zeta_1,\quad e^{-2\pi i \beta_j}=\zeta_j,\ \ j\geq 2,
\ee
the dimension $k$-tuple $\{\Delta_i\}$ then being given by eqn.\eqref{kkkazqw}. It may happen that some of the dimensions in $\{\Delta_i\}$ so found do not belong to $\widehat{\Xi}(k)$; such a $k$-tuple should be discarded.

\begin{table}
$$
\begin{array}{c||c|cc||c|cc}\hline\hline
\Phi_3\Phi_6 & \boldsymbol{e^{2\pi i/3}}, e^{10\pi i/6}& \{3,3/2\} &
& \boldsymbol{e^{2\pi i/6}}, e^{4\pi i/3} & \{6,3\} \\
& \boldsymbol{e^{4\pi i/3}}, e^{10\pi i/6}& \{3/2,5/4\} &
& \boldsymbol{e^{2\pi i/6}}, e^{2\pi i/3} & \{6,5\} \\\hline
\Phi_4\Phi_6 & \boldsymbol{e^{2\pi i/4}}, e^{10\pi i/6}& \{4,5/3\} &
& \boldsymbol{e^{2\pi i/6}}, e^{6\pi i/4} & \{6,5/2\} \\\hline
\Phi_5 & \boldsymbol{\zeta}, \zeta^3& \{5,3\} &
& \boldsymbol{\zeta^2},\zeta^4 & \{5/2,3/2\} \\
& \boldsymbol{\zeta}, \zeta^2& \{5,4\} &
& \boldsymbol{\zeta^3}, \zeta^4  & \{5/3,4/3\} \\
& \boldsymbol{\zeta^2}, \zeta& \{5/2,3\} &
& \boldsymbol{\zeta^4}, \zeta^3 & \{5/4,3/2\} \\\hline
\Phi_8 & \boldsymbol{\zeta}, \zeta^5& \{8,4\} &
& \boldsymbol{\zeta^3}, \zeta^7 & \{8/3,4/3\} \\
& \boldsymbol{\zeta}, \zeta^3& \{8,6\} &
& \boldsymbol{\zeta^5}, \zeta^7 & \{8/5,6/5\} \\
& \boldsymbol{\zeta^3}, \zeta& \{8/3,10/3\} &
& \boldsymbol{\zeta^7}, \zeta^5 & \{8/7,10/7\}\\
& \boldsymbol{\zeta^5}, \zeta& \{8/5,12/5\} &
& \boldsymbol{\zeta^7}, \zeta^3 & \{8/7,12/7\}  \\\hline
\Phi_{10} & \boldsymbol{\zeta}, \zeta^7& \{10,4\} &
& \boldsymbol{\zeta^3}, \zeta^9 & \{10/3,4/3\} 
\\
& \boldsymbol{\zeta}, \zeta^3& \{10,8\} &
& \boldsymbol{\zeta^7}, \zeta^9 & \{10/7,8/7\} \\
& \boldsymbol{\zeta^3}, \zeta& \{10/3,4\} &
& \boldsymbol{\zeta^9}, \zeta^7 & \{10/9,4/3\}  \\\hline
\Phi_{12} & \boldsymbol{\zeta}, \zeta^7& \{12,6\} &
& \boldsymbol{\zeta^5}, \zeta^{11} & \{12/5,6/5\}\\\hline\hline
\end{array}
$$
$$
\begin{gathered}
\text{NON PRINCIPAL POLARIZATION}\\
\begin{array}{c||c|c||c|c}\hline\hline
\Phi_{12} & \phantom{mm}\boldsymbol{\zeta}, \zeta^5\phantom{mm} &
\phantom{m} \{12,8\}\phantom{m} &
\phantom{mm}\boldsymbol{\zeta^7}, \zeta^{11}\phantom{mm} &\phantom{m} \{12/7,8/7\}\\\hline\hline
\end{array}
\end{gathered}
$$
\caption{\label{tttr}List of inequivalent embeddings of regular finite cyclic subgroups of $Sp(4,\Z)$ that lead to dimensions $\{\Delta_1,\Delta_2\}\in \widehat{\Xi}(2)^2$. 
For the characteristic polynomial $\Phi_d$, $\zeta$ stands for the standard primitive $d$-root, $\zeta=e^{2\pi i/d}$. The boldface root is the one associated with the non-zero period $a_\ast$ along the regular ray.  The dimension of the operator spanning the ray $M_\ast$ is always the \emph{first} one in the ordered pair $\{\Delta_1,\Delta_2\}$. }
\end{table}

 \textit{A priori,} there are $2^k$ ways of splitting the spectrum of $m_\ast$ into two non-overlapping sets satisfying eqn.\eqref{kkkkzzznncx}. Taking into account the $k$ choices of the root we call $\zeta_1$, this yield  $k\cdot 2^k$ possibilities for $\{\alpha_\ast,\beta_i\}$ for a given  characteristic polynomial of the form \eqref{kkaq12}. However, some of these possibilities come with a restriction: it is not true in general that we may find arithmetical embeddings $m_\ast\hookrightarrow Sp(2k,\Z)$ which realize as spectrum of $\cu(m_\ast)$ all the subsets
 of roots consistent with eqn.\eqref{kkkkzzznncx}. The simplest example 
 is provided by the regular embeddings 
 \be
 \Z_{12}\hookrightarrow Sp(4,\Z),\qquad
\det[z-m_\ast]=\Phi_{12}(z).
\ee
 In this case the sets $\{e^{2\pi i/12},e^{14\pi i/12}\}$ and
 $\{e^{10\pi i/12},e^{22\pi i/12}\}$ are realized
 as $\mathsf{Spectrum}\;\cu(m_\ast)$, while  $\{e^{2\pi i/12},e^{10\pi i/12}\}$
 and $\{e^{14\pi i/12},e^{22\pi i/12}\}$ are \emph{not} realized. Thus the allowed dimension sets depends on subtle Number-Theoretical aspects of the  classification of all inequivalent embedding $\Z_n\hookrightarrow Sp(2k,\Z)$; this topic will be addressed in section 5. There we shall justify the above claim on the embeddings $\Z_{12}\hookrightarrow Sp(4,\Z)$. In section 5 we shall also see that the two spectra $\{e^{2\pi i/12},e^{10\pi i/12}\}$
 and $\{e^{14\pi i/12},e^{22\pi i/12}\}$ may be realized by an embedding in $\Z_{12}\hookrightarrow S(\Omega)_\Z$
 where $\Omega$ is a polarization with charge multiplies $(1,e_2)$ with $e_2\geq 2$; this result agrees with \textbf{Example \ref{e6}}.

\subparagraph{Summarizing:} a conjugacy class of regular embeddings $\Z_n\hookrightarrow Sp(2k,\Z)$ is a candidate for the monodromy $m_\ast$ at a normal (regular) ray $M_\ast\subset M$ along which the discrete subgroup $\Z_n\subset U(1)_R$ is unbroken. The datum of the conjugacy class, together with a choice of $\zeta_1$, produces a candidate dimension $k$-tuple
$\{\Delta_1,\cdots,\Delta_k\}$ by eqn.\eqref{kkkazqw}.  
 It is not guaranteed that $\{\Delta_1,\cdots,\Delta_k\}\in \widehat{\Xi}(k)^k$, and $k$-tuples which do not belong to $\widehat{\Xi}(k)^k$ should be discarded. At this point we must also impose consistency between the various normal rays $M_\ast$. This reduce the list of allowed $k$-tuple even further. While we have no proof that all the surviving $k$-tuples are actually realized by some CSG, the experience suggests that this algorithm produces few ``spurious'' $k$-tuples, \emph{if any.} As an illustration, we now run the algorithm in detail for $k=2$.
Ranks $k=3$ and $k=4$ may be found in the tables of section 6.

\subsubsection{$\{\Delta_1,\Delta_2\}$ for rank-2 SCFTs}\label{kkkxxx3}

The possible regular characteristic polynomial are $\Phi_3\Phi_4$, $\Phi_3\Phi_6$, $\Phi_4\Phi_6$, $\Phi_5$,
$\Phi_8$, $\Phi_{10}$, $\Phi_{12}$.
In table \ref{tttr} we write only the embeddings of the corresponding regular cyclic groups in $Sp(4,\Z)$ which lead to dimensions $\{\Delta_1,\Delta_2\}\in \widehat{\Xi}(2)^2$.
For instance, $\Phi_3\Phi_4$ does not have  any embedding with this property.
We have written separately the dimensions pairs associated to embeddings in groups $S(\Omega)_\Z$ with $\Omega$ a non-principal polarization. The fine points on the conjugacy classes of elliptic elements in the Siegel modular group $Sp(2k,\Z)$ will be discussed in the next section; table \ref{tttr} 
summarizes the results of that analysis in the special case $k=2$. 
\medskip

Consistency between the two rays $u_2=0$ and $u_1=0$ leads us to consider three  situations:
\begin{itemize}
\item[\textbf{RR}] pairs of dimensions which appear twice in table \ref{tttr}, once in the form $\{\Delta_1,\Delta_2\}$ and once in the form $\{\Delta_2,\Delta_1\}$ corresponding to the case of both rays being strongly regular;
\item[\textbf{RN}] a pair of dimensions $\{\Delta_1,\Delta_2\}$ where the second one $\Delta_2\in \Xi(1)$ corresponding to one strongly regular and one weakly regular of irregular ray;
\item[\textbf{NN}] $\{\Delta_1,\Delta_2\}\in \Xi(1)^2$ two weakly regular/irregular rays.
\end{itemize}

All rank-2 CSG with a conformal manifold of positive dimension must be of type \textbf{NN}. This holds, in particular, for the weakly coupled Lagrangian models which have 
\be
\{\Delta_1,\Delta_2\}=\{2,2\},\ \{2,3\},\ \{2,4\},\ \text{or }\{2,6\}.
\ee

The list of dimension pairs satisfying
\textbf{RR} is rather short:
$\{3/2,5/4\}$ and $\{10/7,8/7\}$ which correspond, respectively, to
AD of types $A_5$ and $A_4$. We have 17 dimension pairs of type \textbf{RN} which may be read from table \ref{tttr};  all of them but $\{5,4\}$, $\{8,6\}$ and $\{10/3,4\}$
have already appeared  in the literature as Coulomb dimensions of some CSG.
There are three \textbf{RN} dimension pairs
with $\Delta_i<2$:
$\{8/5,6/5\}$ (AD of type $D_5$),
$\{5/3,4/3\}$ (AD of type $D_6$) and
$\{10/9,4/3\}$ on which we shall comment in the next subsection.

\subsubsection{Comparing with Argyres \emph{et al.}\! refs.\!\cite{ACSW,ACS2}}

The authors of refs.\!\cite{ACSW,ACS2} have given a (possibly partial) classification of the dimension pairs $\{\Delta_1,\Delta_2\}$ which may appear in a rank 2 SCFT using quite different ideas. It is interesting to compare their list with the present arguments. To perform the comparison, we need to keep in mind the two \emph{caveat} in \S.\ref{kkkaz1204}.

\begin{table}
\begin{equation*}
 \begin{array}{cc|c||c|c}\hline\hline
 \Delta[v] &\Delta[u]& y^2=\cdots & M_v & M_u \\\hline
 10/7 & 8/7 & x^5+ux+v & RI\ \Z_{10} & RI\ \Z_8 \\ 
 8/5 & 6/5 & x^5+ux^2+vx & RI\ \Z_8 & Ir\ \Z_6 \\
 5/2 & 3/2 & x^5+ux^2+2uv x+v^2 & RI\ \Z_{5} & Ir\ \Z_3\\
 4 & 2 & \text{Lagrangian SCFT} &&\\
 10 & 4 & x^5+(ux+v)^3 & RI\ \Z_{10} & Ir\, \Z_4\\
 3/2 & 5/4 & x^6+ux+v & RR\ \Z_{6} & RI\ \Z_{5}\\
 5/3 & 4/3 & x^6+ux^2+vx & RI\ \Z_{5} & Ir\ \Z_4\\
 3 & 2 & \text{Lagrangian SCFT} &&\\
 5 & 3 & x^6+x(ux+v)^2 & RI\ \Z_{5} & Ir\ \Z_3\\
  10/3 & 4/3 & v^{-1}[x^6+v^2x(ux+2v)] & RI\ \Z_{10}\\ 
 6 & 2 & \text{Lagrangian SCFT}&\\\hline\hline 
 \end{array}
\end{equation*}
\caption{Geometries in refs.\!\cite{ACSW,ACS2} with univalued symplectic structure. First two columns give the Coulomb dimensions, third the family of hyperelliptic curves, third and fourth the regularity/irregularity of the rays along the axes together with the corresponding  unbroken $R$-symmetry.}\label{univalued}
\end{table}

Let us recall the framework of \cite{ACSW,ACS2}. They start from the fact that all families of rank 2 \emph{principally polarized} Abelian varieties are families of Jacobians of genus 2 hyperelliptic curves which they write in two ways
\be
y^2=v^{-r/s}\Big(x^5+\cdots\Big)\quad\text{and}\quad y^2=v^{-r/s}\Big(x^6+\cdots\Big),
\ee
 where $\cdots$ stand for certain polynomials in $x$, $u$, $v$ depending on the particular CSG which are listed in \cite{ACSW,ACS2}.  $u$, $v$ are the ``global'' (in their sense) coordinates in the conical Coulomb branch, with $v$ the operator of larger dimension. $0\leq r/s\leq 1$ is a rational number written in minimal terms, i.e.\! $(r,s)=1$. Their SW differential has the form
\be\label{kkkazbbb}
\lambda\equiv v\, \frac{dx}{y}+ u\, \frac{x\, dx}{y}+\text{$u$, $v$ independent}.
\ee
 The two rays $M_v=\{u=0\}$ and $M_u=\{v=0\}$ preserve a discrete $R$-symmetry which may be read for each CSG from the explicit polynomials in the large parenthesis. From the same expressions we may read if these rays are regular irreducible (RI), regular reducible (RR), or irregular (Ir).
 
We may distinguish their geometries in two classes. The first one is when $r/s\in \Z$, that is, the global pre-factor in the \textsc{rhs} of eqn.\eqref{kkkazbbb} is a univalued function of $v$. The geometries with this property listed in refs.\!\cite{ACSW,ACS2} are recalled in table \ref{univalued}
(we do not bother to discuss the Lagrangian models since the agreement with our results is obvious in this case). The first two columns are the dimensions of $v$ and $u$ as listed in refs.\!\cite{ACSW,ACS2}. We see that in all cases these dimensions belong to the intersection of the sets of dimensions associated with the cones $M_v$, $M_u$, yielding perfect agreement with our approach. Note that only dimensions consistent with a principal polarization appear, since this is an assumption in \cite{ACSW,ACS2}.
 
 The second class of geometries is  when $r/s\not\in\Z$, that is, a multi-valued prefactor. 
 As discussed in \cite{ACSW} these geometries lead to a susy central charge $Z$ which is well-defined up to a (locally constant) unobservable phase. Here the first remark of \S.\ref{kkkaz1204} applies: to get a univalued SW differential $\lambda$ we need to go to a finite cover where a suitable fractional power of $v$ becomes uni-valued. 
Then we consider as global coordinates on the cover the functions $(v^{(1-r/s)},u)$ and compare their dimensions 
with the ones in our table. 
This leads to the dimension list in table \ref{multiva}; we see that all dimension pairs agree with our table on the nose.

\begin{table} 
\begin{equation*}
 \begin{array}{cc|c}\hline\hline
 \Delta[v^{|1-r/s|}] &\Delta[u]& y^2=\cdots \\\hline
 4/3 & 10/9 & v^{-2/5}[x^5+v(5ux^2-15vx-6uv)]  \\ 
 8/5 & 6/5 &v^{-1/3}[x^5+vx(2ux+3v)] \ \Z_8  \\
 8/5 & 6/5 &v^{-1/3}[x^2-4v][x^3-2v(3x+2u)] \\
 5/2 & 3/2 & v^{-1/3}[x^5+v(2ux+3v)^2]\ \Z_{5} \\
 8/5 & 6/5 & v^{-2/3}[x^5+v^2x(ux+3v)]\\
5/2 & 3/2 & v^{-2/3}[x^5+v^2(ux+v)^2]  \\
 4 & 2 & v^{-1/2}[x^5+vx(ux+2x)^2]\\
10 & 4 & v^{-1/2}[x^5+v(ux+2v)^3]\\
10/9 & 4/3 & v^{-1/2}[x^6+vx(3ux+4v)]\\
 2 & 2 & v^{-1/2}[x^6+v(3ux+4v)^2]\\
 10/3 & 4/3 & v^{-3/2}[x^6+v^3x(ux+4v)]\\
 5/2 & 3 & v^{-2/3}[x^6+vx(2ux+3v)^2]\\
 6 & 2 & v^{-3/2}[x^6+v^3(ux+4v)^2]\\
 5 & 3 & v^{-4/3}[x^6+v^2x(ux+3v)^2]\\
 \hline\hline 
 \end{array}
\end{equation*}
\caption{The geometries in refs.\!\cite{ACSW,ACS2} with multivalued symplectic structure. First two columns contain the cover Coulomb dimensions; third one  the  hyperelliptic curves.} \label{multiva}
\end{table}

Of course, if the physically correct Coulomb branch is the geometrically natural covering which has a well-defined holomorphic symplectic structure $\Omega$ or its quotient considered by Argyres et al.\!  it is a question of physics not of geometry. There is one aspect that suggests that quotient of \cite{ACSW,ACS2} is the physical Coulomb branch: the dimension pair $\{10/9,4/3\}$ enters in their list (twice) only through quotient CSG. The dimensions of the two quotient geometries are, respectively, $\{20/9, 10/9\}$ and $\{20/9,4/3\}$. Now, the both covering dimensions are $<2$, and there is evidence that a consistent SCFT with all $\Delta_i<2$ should be an Argyres-Douglas model of type $ADE$; since 
$\{10/9,4/3\}$ does not correspond to such a model, we are inclined to think that physics requires to take a discrete quotient of the geometrically natural geometry, as the authors of \cite{ACSW,ACS2} do.

\section{Elliptic conjugacy classes in Siegel modular groups}

Listing the dimension $k$-tuples $\{\Delta_1,\cdots,\Delta_k\}$ has been reduced to understand the conjugacy classes of finite order elements inside the Siegel modular group $Sp(2k,\Z)$ or, in case of more general polarizations (non-trivial charge multipliers $e_i$, see eqn.\eqref{chargemult}) in the commensurable arithmetic group $S(\Omega)_\Z$.
In this section we give an explicit description of such classes. Readers not interested in Number Theoretic subtleties may skip the section.

\subsection{Preliminaries}

We write $\Omega$ for the $2k\times 2k$ symplectic matrix in normal form and $\langle -,-\rangle$ for the corresponding skew-symmetric bilinear pairing.

\subsubsection{Elements of $Sp(2k,\Z)$ with spectral radius 1}

$m\in Sp(2k,\Z)$ has (spectral) radius 1 iff its characteristic polynomial is a product of cyclotomic ones
\be\label{whichcharpoly}
\det[z-m]=\prod_{d\in I}\Phi_d(z)^{s_d}\,\qquad \begin{aligned}&I=\{d_1,\cdots, d_{|I|}\}\subset\mathbb{N},\\ &s_d\in\Z_{\geq1},\quad\sum_{d\in I} s_d\,\phi(d)=2k,\end{aligned}
\ee
that is, if all its eigenvalues are roots of unit. An element $m$ of spectral radius 1 is \emph{semi-simple} (over $\C$) iff its minimal polynomial is \textit{square-free}, i.e. 
\be
\prod_{d\in I}\Phi_d(m)=0.
\ee
A semisimple element $m$ of  radius 1   has finite order,
$m^N=1$ with $N=\mathrm{lcm}\{d\in I\}$. Conversely, all elements of finite order are semi-simple of radius 1.

\begin{lem}\label{llkas} Let $m\in Sp(2k,\Z)$ be of finite order.\\
{\bf 1)}  There exists $R\in Sp(2k,\mathbb{Q})$ which sets $m$ in a  block-diagonal form over $\mathbb{Q}$ 
\be\label{blochdia}R\, m\, R^{-1}=\mathrm{diag}(m_1,\cdots, m_{|I|}),\quad \Phi_d(m_d)=0,\quad m_d\in Sp(s_d\phi(d),\Z).\ee
{\bf 2)} Suppose that no ratio $d_i/d_j$ ($i\neq j$) is a prime power. Then
in \eqref{blochdia} we may choose $R\in Sp(2k,\Z)$. More generally, if $\ell$ is a prime such that $\ell^r\neq d_i/d_j$ for all $i,j$ and $r\neq 0$, we may choose $R\in Sp(2k,\Z_\ell)$.
\end{lem}

\begin{proof} For each $d\in I$ we define
the integral $2k\times 2k$ matrices
\be
\Pi_d= \prod_{a\in (\Z/d\Z)^\times}\prod_{e\in I\atop e\neq d} \Big[m^{-a \phi(e) t(e)/2}\,\Phi_e(m^a)^{t(e)}\Big],\qquad\text{where } t(e)=\begin{cases}2 & e=1,2\\
1 &\text{otherwise.}
\end{cases}.
\ee
Since $m\Omega=\Omega m^{-t}$,
and
\be (1/z)^{\phi(e) t(e)/2}\, \Phi_e(z)^{t(e)}= z^{\phi(e) t(e)/2}\,\Phi_e(1/z)^{t(e)},
\ee
 we have
$\Pi_d\Omega=\Omega\Pi_d^t$.
Up to a rational multiple, the 
$\Pi_d$ form a complete set of orthogonal idempotents over $\BQ$
\be
\boldsymbol{1}=\sum_{d\in I}\frac{\Pi_d}{\varrho_d},\qquad \frac{\Pi_d}{\varrho_d}\cdot\frac{\Pi_e}{\varrho_e}=\delta_{d,e}\,\frac{\Pi_d}{\varrho_d}
\ee
compatible with the skew-symmetric pairing $\Omega$. Then they split over $\BQ$ the representation in the block diagonal form of item \textbf{1)}. The splitting is over $\Z$ iff the $\varrho_d$ are $\pm1$. We have
\be
\varrho_d=\pm \prod_{e\in I\atop e\neq d} \mathrm{R}(\Phi_d,\Phi_e)^{t(e)}
\ee
where $\mathrm{R}(P,Q)$ stands for the resultant of the two polynomials $P$ and $Q$. Under the assumption that $d/e$, $e/d$ are not prime powers, $\varrho_d=\pm1$ 
\cite{apo}.
In facts, $\varrho_d$ is divisible only by the primes $p$ such that there is $e\in I$ with
$d/e=p^r$, $0\neq r\in\Z$ \cite{apo}.
\end{proof}

In other words, if no ratio $d_i/d_j$ is a non-trivial prime power, all embeddings $m\hookrightarrow Sp(2k,\Z)$ are block-diagonal up to equivalence. If some $d_i/d_j$ is a prime power, in addition to the block-diagonal ones, we may have other inequivalent embeddings. 
We shall return to this aspect after studying the case that the minimal polynomial is irreducible over $\BQ$.

\subsubsection{Regular elliptic elements of the Siegel modular group}

We recall that a finite-order element $m\in Sp(2k,\Z)$ is regular iff  the eigenvalues $\{\zeta_1,\cdots,\zeta_k\}$ of $\cu(m)\equiv C\boldsymbol{\tau}+D$ satisfy $\zeta_i\zeta_j\neq1$.
The spectrum $\mathfrak{S}\equiv \{\zeta_i\}$ of $\cu(m)$ will be called the \emph{spectral invariant} of the regular elliptic element $m\in Sp(2k,\Z)$. It is a subset of $k$ roots, $\{\zeta_i\}$, out of the $2k$ ones of $\det[z-m]$ with the property
\be\label{kkkaqwz}
\zeta_i\zeta_j\neq1\qquad i,j=1,\dots, k.
\ee

\begin{rem} Two regular elliptic elements $m,m^\prime\in Sp(2k,\Z)$ which are conjugate in $Sp(2k,\R)$ but not in $Sp(2k,\Z)$ have the same spectral invariant, $\mathfrak{S}=\mathfrak{S}^\prime$ but their fixed points $\boldsymbol{\tau}$ and $\boldsymbol{\tau}^\prime$ are inequivalent periods. Two elements are fully equivalent (and should be identified) iff they are conjugate in $Sp(2k,\Z)$. However the dimension spectrum $\{\Delta_i\}$ depends only on the spectral invariant $\mathfrak{S}$ of the monodromy and hence only on its $Sp(2k,\R)$-conjugacy class.  
\end{rem}

\subsubsection{The spectral invariant as a sign function}\label{mumumqw}

We focus on a $m\in 
Sp(2k,\Z)$ whose minimal polynomial is $\BQ$-irreducible,
\be\label{7lllaq}
\det[z-m]=\Phi_d(z)^{s},\quad s\,\phi(d)=2k, \quad s\in\mathbb{N}.
\ee
 We fix a primitive $d$-root, $\zeta$, and write
$\BK\equiv \BQ[\zeta]$ for the corresponding cyclotomic field and $\Bbbk=\BQ[\zeta+\zeta^{-1}]$ for its maximal totally real subfield, $\mathsf{Gal}(\BK/\Bbbk)=\{\pm1\}$.

Let 
$\psi^\alpha_1\in \BK$ ($\alpha=1,\dots,s$) be a basis of the $\zeta$-eigenspace of the matrix $m$. Let $\sigma\in \mathsf{Gal}(\BK/\mathbb{Q})\cong (\Z/d\Z)^\times$; then $\psi^\alpha_\sigma\equiv \sigma(\psi_1^\alpha)$ form a basis of the 
$\sigma(\zeta)$-eigenspace. 
We write $\langle \psi_1^\alpha,\psi^\beta_{-1}\rangle=t^{\alpha\beta}\in \BK(s)$. Without loss of generality, we may assume $t^{\alpha\beta}$ to be diagonal
$t^{\alpha\beta}=t_\alpha\delta^{\alpha\beta}$ with $t_\alpha\in \BK$.
$t_\alpha$ is odd (i.e.\! purely imaginary) $\bar t_\alpha=-t_\alpha$. The symplectic structure is given by a 2-form 
\be
2\,\Omega=i\sum_\alpha\sum_{\sigma\in \mathsf{Gal}(\BK/\BQ)} \sigma(t_\alpha)\,\psi^\alpha_{-\sigma}\wedge \psi^\alpha_{\sigma}. 
\ee
Thus $m$ is the direct sum of $s$ (possibly inequivalent) embeddings $\Z_d\to Sp(\phi(d),\R)$. For each summand 
we define the \underline{odd} \emph{sign (function)}
\be
\mathsf{sign}_\alpha\colon \mathsf{Gal}(\BK/\BQ)\to \{\pm 1\},\qquad \sigma\mapsto \frac{\sigma(t_\alpha)}{i|\sigma(t_\alpha)|}.
\ee
The spectral invariant of the $\alpha$-th summand is
\be
\mathfrak{S}_\alpha\equiv \Big\{\;\zeta^\sigma\ \Big|\ \sigma\in \mathsf{Gal}(\BK/\BQ)\ \text{such that }\mathsf{sign}_\alpha(\sigma)=+1\Big\}.
\ee
It is obvious that $\mathfrak{S}_\alpha$ satisfies condition \eqref{kkkaqwz}. We shall denote by the same symbol, $\mathfrak{S}_\alpha$, both the spectral invariant and the corresponding sing function.
A semi-simple element $m$ satisfying \eqref{7lllaq} is regular iff the spectral invariant is the same for all its direct summands, i.e.\! $\mathfrak{S}_\alpha=\mathfrak{S}_\beta$.

\begin{rem} For a single regular block with characteristic polynomial $\Phi_d(z)$ the number of sub-sets $\mathfrak{S}\subset\{\zeta^a\colon a\in(\Z/d\Z)^\times\}$ satisfying condition \eqref{kkkaqwz} is $2^{\phi(d)/2}$. However it is not true (in general) that all such sub-sets (i.e.\! sign functions) are realized as spectral invariant of some embedding $\Z_d\to Sp(\phi(d),\Z)$.
For instance, for $d=12$ we have $2^{\phi(12)/2}=4$, but only 2 sign functions are produced by actual embeddings. The set $\{\mathfrak{S}\}$
of sign functions which do are realized satisfies the obvious condition
\be
\sigma\mapsto\mathsf{sign}(\sigma)\in \{\mathfrak{S}\}\ \Rightarrow\ \sigma\mapsto \mathsf{sign}(\tau\sigma)\in \{\mathfrak{S}\}\qquad \forall\; \tau\in\mathsf{Gal}(\BK/\BQ).
\ee
In particular if $\mathsf{sign}$ is realized $-\mathsf{sign}$ is also realized.\end{rem} 
\medskip

In \S.\ref{kkktuple} that the list of possible Coulomb branch dimensions $\{\Delta_1,\cdots,\Delta_k\}$ is determined from the $Sp(2k,\R)$-conjugacy classes of regular elements of $Sp(2k,\Z)$ (or the corresponding arithmetic group for non principal polarizations) through their sign function invariant $\mathfrak{S}$. Then our main problem at this point is to understand the set $\{\mathfrak{S}\}$ of  signs which do are realized for a given polarization. This is the next task.

 \subsection{Cyclic subgroups of integral matrix groups}
 
In this section we review the theory of the embedding of cyclic groups into groups of matrices having integral coefficients in a language convenient for our purposes (also providing explicit expressions for the matrices). See also \cite{newman,japonese}. Our basic goal is to describe the set $\{\mathfrak{S}\}$ of signs which do appear and more generally the regular elliptic elements of the Siegel modular group.

 \subsubsection{Embeddings $\Z_n\hookrightarrow GL(2k,\Z)$ vs.\! fractional ideals}\label{ssssq}
 
We focus on a single block, that is, we consider a matrix $m\in GL(2k,\Z)$ with minimal polynomial $\Phi_n(z)$. Then $2k\equiv s\,\phi(n)$ for some $s\in\mathbb{N}$.

\paragraph{Notations.} We fix once and for all a primitive $n$-root of unity $\zeta\in \C$, and write
$\BK\equiv \mathbb{Q}[\zeta]$ for the $n$-th cyclotomic field, $\mathfrak{O}\equiv \Z[\zeta]$ for its rings of integers,
$\Bbbk\equiv \mathbb{Q}[\zeta+\zeta^{-1}]$ for its maximal real subfield, and $\mathfrak{o}\equiv\Z[\zeta+\zeta^{-1}]$ for the ring of algebraic integers in $\Bbbk$. $\mathsf{Gal}(\BK/\Bbbk)\cong\Z_2$, the non-trivial element $\iota$  being complex conjugation, $\iota(x)=\bar x$. $\mathsf{Gal}(\Bbbk/\mathbb{Q})\cong (\Z/n\Z)^\times/\{\pm 1\}$. We write $\breve{n}$ for the \emph{conductor} of the field $\BK$:
\be
\bn=\begin{cases} n &\text{if }n\neq2\mod4\\ n/2 &\text{otherwise.}\end{cases}
\ee
We write $C_\BK$ ($C_\Bbbk$) for the group of ideal classes in $\BK$ (resp.\! in $\Bbbk$). $N$ will denote the relative norm
$\BK\to\Bbbk$ extended to the groups of fractional ideals $\mathfrak{I}_\BK\xrightarrow{N} \mathfrak{\;I\;}_\Bbbk$ in the usual way \cite{algn1,algn2}.
\bigskip

The embedding $\Z_n\hookrightarrow GL(2k,\Z)$ makes $\Z^{2k}$ into a finitely-generated torsion-less 
$\mathfrak{O}$-module $\cm$, multiplication by $\zeta$ being given by $m$. 
Conversely, any finitely-generated torsion-less $\mathfrak{O}$-module $\cm$ defines an embedding $\Z_n\hookrightarrow GL(2k,\Z)$ where $2k$ is the rank of $\cm$ seen as a (free) $\Z$-module.
The ring of cyclotomic integers $\mathfrak{O}$ is a Dedekind domain. The following statement holds for all such domains:

\begin{pro}[see \cite{algn1}]\label{kkkaqweX} A finitely-generated torsion-less  module $\cm$ over the Dedekind domain $\fO$ has the form $\mathfrak{a}_1\oplus\mathfrak{a}_2\oplus\cdots\oplus\mathfrak{a}_s$, where $\mathfrak{a}_i$ are fractional ideals in $\BK$. Two modules $\mathfrak{a}_1\oplus\mathfrak{a}_2\oplus\cdots\oplus\mathfrak{a}_s$
and $\mathfrak{b}_1\oplus\mathfrak{b}_2\oplus\cdots\oplus\mathfrak{b}_t$ are isomorphic if and only if  $s=t$
and the ideal class $\prod_i\mathfrak{a}_i\mathfrak{b}_i^{-1}$ is trivial. In particular we may always set $\cm=\fO\oplus\cdots\oplus\fO\oplus \mathfrak{a}\equiv (1)^{\oplus (s-1)}\oplus \mathfrak{a}$. Then the class of $\mathfrak{a}$ yields a one-to-one correspondence
\be\label{corrrr}
\big\{\text{$GL(2k,\Z)$-conjugacy classes of embeddings }\Z_n\hookrightarrow GL(2k,\Z)\big\}\ \overset{\text{1-1}}{\longleftrightarrow}\ C_\BK.
\ee
\end{pro}
In order to describe the  explicit embeddings we sketch the proof in the special case $s=1$.
 \begin{proof} Let $\mathfrak{a}\subset \BK$ be a fractional ideal. In particular $\mathfrak{a}$ is a torsion-free finitely generated $\Z$-module, hence a lattice isogeneous to $\mathfrak{O}$, and thus of rank $2k$.
 Choosing generators, we may write
 $\mathfrak{a}=\bigoplus_{a=1}^{2k}\Z \,\omega_a$, with $\omega_a\in \BK$.
Now $\zeta\,\omega_a\in \mathfrak{a}$, and hence there is an \emph{integral} $2k\times 2k$ matrix $m$ such that
\be
\zeta\omega_a= m_{ab}\,\omega_b. 
\ee 
The minimal polynomial of $m$ is the $n$-th cyclotomic polynomial, $\Phi_n(m)=0$. Thus the matrix $m$ yields an explicit embedding of $\Z_n$ into $GL(2k,\Z)$. Had we chosen a different set of generators for $\mathfrak{a}$, $\omega^\prime_a$, we would had gotten an integral matrix $m^\prime$ which differs from $m$ by conjugacy in $GL(2k,\Z)$. Indeed, 
$\omega^\prime_a= A_{ab}\,\omega_b$, for some $A\in GL(2k,\Z)$. Thus the map $\mathfrak{a}\mapsto\text{(conjugacy class of }m)$ is independent of all choices.
By construction, the vector $\boldsymbol{\omega}\equiv (\omega_1,\cdots, \omega_{2k})\in \BK^{2k}$ is the eigenvector of $m$ associated to the eigenvalue $\zeta$. The eigenvector associated to the eigenvalue $\sigma(\zeta)$, $\sigma\in \mathsf{Gal}(\BK/\mathbb{Q})$, is the $\sigma(\boldsymbol{\omega})$.

Conversely, if $m\in GL(2k,\Z)$ with minimal polynomial $\Phi_n$, consider an eigenvector $\boldsymbol{\omega}\equiv (\omega_1,\cdots, \omega_{2k})\in \BK^{2k}$ associated to the eigenvalue $\zeta$, and set $\mathfrak{a}=\bigoplus_{a=1}^{2k} \Z \omega_a$. Clearly, if $\boldsymbol{\omega}$ is such an eigenvector so is $\mu\,\boldsymbol{\omega}$ for all $\mu\in\BK^\times$. Hence $\mathfrak{a}$ and $(\mu)\mathfrak{a}$ ($\mu\in \BK^\times$) describe the same conjugacy class of integral matrices $m$, that is, the conjugacy class of $m$ depends only on the class of the fractional ideal $\mathfrak{a}$ in $\C_\BK=\mathfrak{I}_\BK/(\BK^\times)$.
 \end{proof}
 
 The action of $\zeta$ on the module
 $\bigoplus_{i=1}^s\mathfrak{a}_i$ is unitary for the natural Hermitian form
\be\label{kkkaqw21j}
 \langle a_i, b_i\rangle=\sum_{i=1}^s \mathrm{Tr}_{\BK/\BQ}(\bar a_i b_i).
\ee
 
 \begin{rem}
A $\fO$-module $\cm$ gives a unitary  representation of $\Z_n$ on the associated $\C$-space $V_\cm=\cm\otimes_\Z\C$ which corresponds to the natural embedding
$GL(2k,\Z)\subset GL(2k,\C)$.
 \end{rem}
 
 \subsubsection{The dual embedding}
 
 Given an embedding of $\Z_n\hookrightarrow GL(2k,\Z)$, generated by the integral matrix $m$, we have a second embedding, the \emph{dual one,} 
where the generator is represented by the integral matrix $(m^t)^{-1}$. If we write the matrices in an unitary basis of $V_\cm$ (instead of an integral one), the two representations of $\Z_n$ are related by complex conjugation.

Let $\cm=\bigoplus_i\mathfrak{a}_i$ be an $\fO$-module associated to the embedding $m$ as in \textbf{Proposition \ref{kkkaqweX}}; then the dual $\fO$-module $\cm^\vee$ is associated to the dual embedding $(m^t)^{-1}$. $\cm^\vee$ is uniquely determined up to isomorphism.

\begin{lem}
$\cm^\vee=\bigoplus_i \mathfrak{a}_i^\ast$ where $\mathfrak{a}_i^\ast$ is the  dual (\emph{a.k.a.}\! complementary) fractional  ideal of $\mathfrak{a}_i$ (with respect to \eqref{kkkaqw21j}). One has {\rm\cite{langalge,hasse}}
\be\label{jjzaq1}
\mathfrak{a}_i^*= \frac{1}{(\,\overline{\Phi^\prime_n(\zeta)}\,)\,\bar{\mathfrak{a}}_i}
\ee  
where $\Phi^\prime_n(z)$ is the derivative of $\Phi_n(z)$ and $\bar{\mathfrak{a}}_i$ is the complex conjugate ideal of $\mathfrak{a}_i$.
\end{lem}
 Note that
\be
\mathfrak{a}^{\ast\ast} = \left(\frac{\Phi^\prime_n(\zeta)}{\overline{\Phi^\prime_n(\zeta)}}\right)\mathfrak{a}=\mathfrak{a},\qquad \text{since }\ 
\frac{\Phi^\prime_n(\zeta)}{\overline{\Phi^\prime_n(\zeta)}}\ \text{is a unit in }\mathfrak{O}.
\ee

In \textsc{appendix \ref{kkkkza3b2}}
we show a some properties of the map
$\mathfrak{a}\to \mathfrak{a}^*$ which greatly simplify the computations.
In particular:

\begin{lem}\label{kkkka12}
For all fractional ideal $\mathfrak{a}$ of $\BK$ we have 
\be\label{jj1234}\mathfrak{a}^\ast=\varrho/\bar{\mathfrak{a}},\ 
\text{for a certain }\varrho\in\BK^\times\ \text{with }\iota(\varrho)=-\varrho.
\ee
If $\bn$ is not the power of an odd prime, we may alternatively choose $\varrho$ to be real by multiplying it by a purely imaginary unit,
e.g.\! $(\zeta-\zeta^{-1})$ for $\bn\neq 2^r$
or $i$ for $\bn=2^r$ .
\end{lem}

\subsubsection{Complex, real, quaternionic}

At the level of underlying $\C$-linear representations, $V_\cm\cong V_\cm^\vee\equiv V_{\cm^\vee}$. An \emph{anti}-linear morphism
$R\colon V_\cm\to V_{\cm^\vee}$ is said to be a \emph{real} (resp.\! \textit{quaternionic})  structure iff it squares to $+1$ (resp.\! $-1$) \cite{brrook}. A real (resp.\! quaternionic) structure embeds the matrix $m$ in the orthogonal (resp.\! symplectic) group. To get an embedding
$\Z_n\hookrightarrow Sp(2k,\Z)$ we need a quaternionic structure defined over $\Z$.  
First of all, this requires $\cm$ and $\cm^\vee$ to be isomorphic as $\fO$-modules. Writing $\cm=(1)^{\oplus(s-1)}\oplus\mathfrak{a}$, we must have 
\be\label{kkkaqwj}
(1)^{\oplus(s-1)}\oplus \mathfrak{a}\cong (\varrho)^{\oplus(s-1)}\oplus \varrho\, \bar{\mathfrak{a}}^{-1}.
\ee
which implies that 
\be\label{rrrtzqm}
N\mathfrak{a}\cdot \mathfrak{O}= \mathfrak{a}\,\bar{\mathfrak{a}}= \big(\eta).
\ee
Since the natural map $C_\Bbbk\to C_\BK$, $[\mathfrak{b}]\mapsto [\mathfrak{b}\cdot\mathfrak{O}]$ is injective \cite{lang1}, the fractional ideal
$N\mathfrak{a}$ is principal in $\Bbbk$, that is, $N\mathfrak{a}=(\eta)$ for some $\eta\in\Bbbk^\times$.
\medskip 

From eqn.\eqref{kkkaqwj} we see that the construction of the embedding is essentially reduced to the case $s=1$.
From now on we specialize to this case,
so that $\cm\equiv\mathfrak{a}$, $\cm^\vee=\mathfrak{a}^*$.
Fix a $\Z$-basis\footnote{\ That is, a set of generators of $\mathfrak{a}$ seen as a free $\Z$-module.} $\{\omega_a\}$ of $\mathfrak{a}$ and let $\phi^a$ be the dual basis of $\mathfrak{a}^*$, i.e.\!
$\langle \omega_a,\phi^b\rangle={\delta_a}^b$. If $\mathfrak{a}$ satisfies condition \eqref{rrrtzqm},
$\mathfrak{a}^*= \varrho/\bar{\mathfrak{a}}=\varrho\eta^{-1}\mathfrak{a}$, and we write  
\be\label{hhhhaq}
\mathfrak{a}^*=\lambda_v\,\mathfrak{a}\quad\text{where }\lambda_v=v \varrho/\eta^{-1}\quad \text{with }v\ \text{a unit of }\fO.
\ee 
Then $\{\lambda_v\,\omega_a\}$ is also a $\Z$-basis of the dual fractional ideal $\mathfrak{a}^*$ and there is a
a matrix $J_{ab}\in GL(2k,\Z)$ (depending on the unit $v$) such that
\be\label{kkkaz192}
\lambda_v\,\omega_a= J_{ab}\,\phi^b
.\ee
One has
\be\label{symppp}
\bar\lambda_v^{-1} J_{ab}=\bar\lambda^{-1}_v \langle J_{ac}\phi^c,\omega_b\rangle 
=\langle\omega_a,\omega_b\rangle = \langle \omega_a,J_{bc}\phi^c\rangle\lambda^{-1}= J_{ba}\,\lambda^{-1}_v
\ee
i.e.\!  the integral unimodular matrix $J_{ab}$ is skew-symmetric (resp.\! symmetric) if the unit $v$ is such that $\lambda_v$ is purely imaginary (resp.\! real). In the first case $J_{ab}$ is a principal integral symplectic structure, hence similar over $\Z$ to the standard one $\Omega$,
i.e.\! $J=h^t\Omega h$ for some $h\in GL(2k,\Z)$. In the second case $J$ is 
a unimodular symmetric quadratic form.
Thus for a fixed fractional ideal $\mathfrak{a}$, we find a symplectic structure (i.e.\! an embedding in $Sp(2k,\Z)$) per each choice of the unit $v$ such that $\lambda_v$ is purely imaginary.
We shall count the inequivalent ones in the next subsection.

Since $\eta\in\Bbbk^\times$ is always real, and $\varrho$ was chosen to be purely imaginary (cfr.\! \textbf{Lemma \ref{kkkka12}}) $J_{ab}$ is skew-symmetric iff $v$ is real, and symmetric iff it is purely imaginary. In particular, the two obvious choices $v=\pm1$, always produce embeddings $\Z_n\hookrightarrow Sp(2k,\Z)$ ($2k\equiv\phi(n)$).

\subsubsection{Conjugacy classes of embeddings $\Z_n\hookrightarrow Sp(\phi(n),\Z)$}

Regular embeddings $\Z_n\hookrightarrow Sp(\phi(n),\Z)$ exist for all $n\geq3$. Indeed, the condition $N\mathfrak{a}$ principal is trivially satisfied if $\mathfrak{a}$  itself is principal. Thus the trivial ideal class $(1)$ yields regular embeddings $\Z_n\hookrightarrow Sp(\phi(n),\Z)$ for all $n\geq3$. We proceed as follows: we fix an embedding $(m,J)$ associated to the ideal (1) and call it the \emph{reference} embedding.
All inequivalent embeddings are obtained by acting on the reference one $(m,J)$ with a certain Abelian group $H$ defined in the next \textbf{Proposition}. A subgroup of $H$ is easy to describe: 
in  eqn.\eqref{hhhhaq} we may choose a different \emph{real} unit $v\in\fO$ and still get an invariant quaternionic structure; this is the same as multiplying $\eta\in\Bbbk^\times$ by a unit of $\fo$. In this way we  get new (inequivalent) embeddings $(m,J^\prime)$ for a given fractional ideal $\mathfrak{a}$:
they correspond to embeddings which are conjugate over $GL(\phi(n),\Z)$ but not over the subgroup $Sp(\phi(n),\Z)$. 
On the other hand, under $\mathfrak{a}\to \mu \mathfrak{a}$ with $\mu\in \BK^\times$ we have $\eta\to \eta\, N\!\mu$, hence the image of $\eta$
in the group $\Bbbk^\times/N\BK^\times$
is independent of the choice of the representative ideal $\mathfrak{a}$ in the ideal class. To describe also the embeddings belonging to different $GL(\phi(n),\Z)$ conjugacy classes, it is convenient to consider the group
\be
L=\ker\!\Big(\mathfrak{I}_\BK\xrightarrow{N} C_\Bbbk,\ \mathfrak{a}\mapsto [N\mathfrak{a}]\Big)
\ee
of fractional ideal classes in $\BK$ whose relative norm is principal in $\Bbbk$.
Then we have group
\be\label{whatK}
K=\big\{(\mathfrak{a},\eta)\in L\times \Bbbk^\times\;:\; N\mathfrak{a}=(\eta)\big\} 
\ee
and the group homomorphism
\be\label{whatPi}
\pi\colon K\to C_\BK\times \Bbbk^\times/N\BK^\times,\qquad (\mathfrak{a},\eta)\mapsto ([\mathfrak{a}],[\eta]).
\ee

The above discussion shows the

\begin{pro}[see \cite{japonese}] Let $n\geq3$.
The (Abelian) group 
\be\label{whatH}
H\equiv \Im\pi
\subset C_\BK\times \Bbbk^\times/N\BK^\times
\ee
acts freely and transitively on the set of the $Sp(\phi(n),\Z)$-conjugacy classes of 
embeddings $\Z_n\hookrightarrow Sp(\phi(n),\Z)$. In particular, the number of $Sp(\phi(n),\Z)$-conjugacy classes is 
\be\label{jjjjja}\big|H\big|\equiv
\big|\ker\!\big(C_\BK\xrightarrow{N} C_\Bbbk\big)\big|\times \big|\boldsymbol{u}/N\boldsymbol{U}\big|,
\ee
where $\boldsymbol{U}$ (resp.\! $\boldsymbol{u}$) is the group of unities of $\mathfrak{O}$ (resp.\! $\mathfrak{o}$).
\end{pro}

Let $h=|C_\BK|$ and $h^+=|C_\Bbbk|$
be the class numbers of the fields $\BK$ and $\Bbbk$, respectively. The map
$N\colon C_\BK\to C_\Bbbk$ is surjective [cite], and the ratio 
\be
h^-=h/h^+=\big|\ker(C_\BK\xrightarrow{\,N\,}C_\Bbbk)\big|
\ee
 is called the \emph{relative class number} of $\BK$. $h^-$ is much easier to compute that either $h$ or $h^+$
(it has an explicit expression in terms of generalized Bernoulli numbers \cite{lang1,apo2}).
It turns out that $h^-=1\Leftrightarrow h=h^+=1$ \cite{apo2}.
For $n\leq 22$ the relative class number is $h^-=1$,
while for large $n$ we have the asymptotic behavior \cite{apo2}
\be 
\log h^-\sim \frac{1}{4}\,\phi(n)\,\log \bn\qquad n\to\infty.
\ee
$h^-=1$ iff the conductor $\bn$ is one of the numbers \cite{apo2}
\be
\begin{split}
\bn=& 1,3,4,5,7,8,9,11,12,13,15,16,17,19,20,
21,\\ &24,25,27,28,32,33,35,36,40,44,45,48,60,84.\end{split}
\ee
When $h^-=1$, all inequivalent embeddings $\Z_n\hookrightarrow Sp(\phi(n),\Z)$ arise from the same fractional ideal, and hence are all conjugate in the larger group $GL(\phi(n),\Z)$.

To compute the second factor in \eqref{jjjjja} we consider the group of units in the relevant fields.

\paragraph{Group of units.}
We write $\boldsymbol{\mu}$ for the group generated by the roots of unity in $\BK$
\be
\boldsymbol{\mu}=\big\{\pm \zeta^k\big\}.
\ee
The \emph{Hasse unit index} $Q$ of the cyclotomic field $\BK$ is
\be
Q\equiv [\boldsymbol{U}:\boldsymbol{\mu}\,\boldsymbol{u}].
\ee

\begin{pro}[{\rm see \cite{hasse}}] Let $\BK$ be a cyclotomic field of conductor $\bn$. One has
\be
Q= \begin{cases} 1 & \bn\ \text{is a prime power}\\
2 &\text{otherwise}.
\end{cases}
\ee
Moreover, $\boldsymbol{U}/(\boldsymbol{\mu}\, \boldsymbol{u})\cong N\boldsymbol{U}/\boldsymbol{u}^2$ and then
\be
[N\boldsymbol{U}:\boldsymbol{u}^2]=Q.
\ee
In other words, if $Q=1$ all $\{N\varepsilon\,:\, \varepsilon\in\boldsymbol{U}\}$ are squares in $\boldsymbol{u}$, while for $Q=2$ only half of them are squares.
\end{pro}

\begin{rem} Let the conductor $\bar n$ be divisible by two distinct primes. Let us describe explicitly a generator of the group $N\boldsymbol{U}/\boldsymbol{u}^2\cong \Z_2$. $(1-\zeta)$ is a unit in $\BK$, and
\be
N(1-\zeta)=2-\zeta-\zeta^{-1}= 2(1-\cos(2\pi/n))=4\sin^2(\pi/n)\in \Bbbk, 
\ee
while its square root $2\sin(\pi/n)$ is not in $\Bbbk$.
Hence, for $n$ divisible by two distinct primes, 
\be
N\boldsymbol{U}= N(1-\zeta)^a\,\boldsymbol{u}^2,\quad a=0,1.
\ee
In particular $\varepsilon_1\equiv N(1-\zeta)$ is a fundamental unit of $\Bbbk$.
\end{rem}

For $n\geq 3$, let $k=\phi(n)/2$. From Dirichlet unit theorem \cite{algn1} we know that
\be
\boldsymbol{u}=\Big\{\pm \varepsilon_1^{s_1}\varepsilon_2^{s_2}\cdots
\varepsilon_{k-1}^{s_{k-1}}, \ s_a\in\Z\Big\}
\ee
where $\varepsilon_a$, $a=1,\dots, k-1$ are the real positive fundamental units.
For $Q=1$ we have
\be
\boldsymbol{u}/N\boldsymbol{U}\cong \boldsymbol{u}/\boldsymbol{u}^2= \Big\{\pm
\varepsilon_1^{s_1}\varepsilon_2^{s_2}\cdots
\varepsilon_{k-1}^{s_{k-1}}, \ s_a\in\Z/2\Z\Big\}\cong \Z_2^k. 
\ee
while for $Q=2$
\be
\boldsymbol{u}/N\boldsymbol{U}\cong \boldsymbol{u}/\big(\Z_2\times \boldsymbol{u}^2\big)= \Big\{\pm
\varepsilon_2^{s_2}\cdots
\varepsilon_{k-1}^{s_{k-1}}, \ s_a\in\Z/2\Z\Big\}\cong \Z_2^{k-1}.
\ee

In conclusion,
\begin{corl} Let $n\geq3$. The number of $Sp(\phi(n),\Z)$-conjugacy classes of embeddings $\Z_n\hookrightarrow Sp(\phi(n),\Z)$ is
\be
\frac{2^{\phi(n)/2}}{Q}\,h^-\equiv |\boldsymbol{\mu}|\prod_{\chi\;\text{odd}}\big(-B_{1,\chi}),
\ee
where $Q=1,2$ is the Hasse unit index,  $h^-$ the relative class number of the cyclotomic field $\BK$, and $B_{1,\chi}$ the first Bernoulli number of the odd Dirichlet character $\chi$.
\end{corl}

However, to fully solve our problem we need  to know also when two distinct conjugacy classes are conjugate in the larger group $Sp(\phi(n),\R)$ (or, equivalently, in $Sp(\phi(n),\BQ)$).

\subsubsection{Embeddings
$\Z_n\hookrightarrow S(\tilde\Omega)_\Z$ for $\tilde\Omega$ non-principal}\label{kkzaqww}
 
The symplectic matrix $J$ defined in eqn.\eqref{kkkaz192}, for $\lambda_v$ as in \eqref{hhhhaq} (with $v$ a unit of $\mathfrak{o}$), corresponds to a principal polarization, i.e.\! $J$ is an integral skew-symmetric matrix with $\det J=1$.
Let $0\neq\kappa\in \mathfrak{o}$
and consider the matrix $J^\kappa$ defined by
\be\label{kkkxxxz10}
J_{ab}^\kappa\,\phi^b=\kappa\lambda\,\omega_a. 
\ee
If $\kappa$ is a unit, $J^\kappa$ is a principal-polarization. For $\kappa$ just integer in $\Bbbk$, $J^\kappa$ is an integral skew-symmetric matrix with determinat
\be
\det J^\kappa= \big(N_{\Bbbk/\BQ}\, \kappa\big)^2.
\ee

\subsubsection{$Sp(\phi(n),\BQ)$-conjugacy classes}

As we saw in \S.\,\ref{mumumqw}, the $Sp(\phi(n),\R)$-conjugacy classes are distinguished by the sign  of the corresponding (integral) symplectic structure $\mathsf{sign}_\sigma$. Then we need to understand the action of the group of $Sp(\phi(n),\Z)$-conjugacy classes of embeddings, $H$, on the sign function which we see as a map 
$\mathsf{Gal}(\Bbbk/\BQ)\to \Z_2$. The set of such maps form a group isomorphic to $\Z_2^{\phi(n)/2}$. 
Then we have a well-defined homomorphism of Abelian groups
\be
\mathfrak{s}\colon H\longrightarrow \Z_2^{\phi(n)/2},\qquad 
([\mathfrak{a}],[\eta]) \longmapsto \left\{
\mathsf{sign}_{([\mathfrak{a}],[\eta])}\colon \sigma\mapsto \frac{\sigma(\eta)}{|\sigma(\eta)|}\right\}
\ee
An element $([\mathfrak{a}],[\eta])\in H$ changes the $Sp(\phi(n),\Z)$ conjugacy class of $m$ without changing its $Sp(\phi(n),\R)$-conjugacy class iff 
it belongs to the kernel of $\mathfrak{s}$, that is,

\begin{corl}[Midorikawa \cite{japonese}] The group $H_\R\equiv H/\ker\mathfrak{s}$ acts freely and transitively on the $Sp(\phi(n),\R)$-conjugacy classes of embeddings $\Z_n\hookrightarrow Sp(\phi(n),\Z)$.
\end{corl}

An element $\eta\in \Bbbk^\times$ is said to be \emph{totally positive} iff
$\sigma(\eta)>0$ for all $\sigma\in\mathsf{Gal}(\Bbbk/\BQ)$; the set of all
totally positive elements $\Bbbk^\times_+\subset \Bbbk^\times$ form a subgroup while (for $n\geq 3$) \cite{algn1}
\be\label{kkkwerxx}
\Bbbk^\times_+/\Bbbk^\times\cong \Z_2^{\phi(n)/2}.
\ee

Comparing with eqn.\eqref{whatH}, we see that
\be
\ker\mathfrak{s}= H\bigcap \Big(C_\BK \times \Bbbk^\times_+/N\BK^\times\Big).
\ee

The group of  principal fractional ideals $(\eta)$
with $\eta\in \Bbbk^\times_+$  is a subgroup
of the group of all principal fractional ideals.
The quotient $\mathfrak{I}_\Bbbk/(\Bbbk_+^\times)$ is called the \emph{narrow-ideal class}, $C_\Bbbk^\text{nar}$. Likewise we have the subgroup of totally positive units $\boldsymbol{u}_+\subset \boldsymbol{u}$; from the ray class exact  sequence \cite{algn1,milnecft}
\be\label{kkkkz9e}
1\to \boldsymbol{u}/\boldsymbol{u}_+\to \Z_2^{\phi(n)/2}\to C_\Bbbk^\text{nar}/C_\Bbbk\to1,
\ee 
we get
$\boldsymbol{u}/\boldsymbol{u}_+\cong \Z_2^{\phi(n)/2-a}$, $C_\Bbbk^\text{nar}/C_\Bbbk\cong\Z_2^a$ for some $0\leq a\leq \phi(n)/2-1$ ($a\geq 1$ when $Q=2$). Then
\be
\big|\ker \mathfrak{s}\big|=\big|\ker\!\big(C_\BK\xrightarrow{\,N\,}C_\Bbbk^\text{nar}\big)\big|\times \big|\boldsymbol{u}_+/N\boldsymbol{U}\big|,
\ee 
and the number of $Sp(\phi(n),\R)$-conjugacy classes of embeddings $\Z_n\hookrightarrow Sp(\phi(n),\Z)$ (i.e.\! the number of possible sign assignments in the integral symplectic structure is
\be
\big|H_\R\big|= \frac{|\ker(C_\BK\xrightarrow{N} C_\Bbbk)|}{|\ker(C_\BK\xrightarrow{N} C^\text{nar}_\Bbbk)|}\cdot\big|\boldsymbol{u}/\boldsymbol{u}_+\big|.
\ee
Hence $\ker(C_\BK\xrightarrow{N} C_\Bbbk)/\ker(C_\BK\xrightarrow{N} C^\text{nar}_\Bbbk)\cong\Z_2^b$, with   $b\leq a$ 
while $\boldsymbol{u}/\boldsymbol{u}_+
\cong \Z_2^{\phi(n)/2-a}$ and
\be\label{hhhowman}
H_\R\cong \Z_2^{\phi(n)/2+(b-a)}
\ee
so that the number of $Sp(\phi(n),\R)$ inequivalent embeddings is $2^{\phi(n)/2+(b-a)}\leq 2^{\phi(n)/2}$. 

\begin{corl} Let the class number of $\BK$, $h_\BK$, be \emph{odd}. Then
\be
H_\R\cong \boldsymbol{u}/\boldsymbol{u}_+\cong \Z_2^{\phi(n)/2-a},\qquad b=0.
\ee
\end{corl}

From eqn.\eqref{hhhowman} we see that 
if $b<a$ not all signatures of the symplectic structure may be realized. Indeed, from eqn.\eqref{kkkkz9e} we see that all signatures are realized iff the kernel of the natural map $C_\Bbbk^\text{nar}\to C_\Bbbk$ is contained in the image of $N$.
We mention a few known facts on $a$:
 \begin{itemize}
 \item[a)] {\rm(Weber)} if $n=2^r$, $\boldsymbol{u}_+= \boldsymbol{u}^2$, that is, $a=0$;
\item[b)] {\rm(Kummer, Shimura \cite{kum1,kum2})} if $\bn$ is a prime, $a=0$ if and only if the class number of $\BK$ is odd;
\item[c)] of course $a>0$ if $\bn$ is divisible by two distinct primes.
\end{itemize}  

Thus, for instance, if $h_\BK$ is odd and $\bn$ composite $\neq 2^r$, not all signs of the symplectic form are realizable.


\begin{rem} If $h_\BK=2$ we have
$H\cong \Z_2\times \boldsymbol{u}/N\boldsymbol{U}$, since we must have\footnote{\ Indeed, $2\mid h_\Bbbk\Rightarrow 2\mid h^-$ [cite] so that $2\mid h_\Bbbk\Rightarrow 4\mid h_\BK$.} $h_\Bbbk=1$; if (in addition) $\bn$ is not a prime power, $H\cong \Z_2^{1+\phi(n)/2-1}\cong \Z_2^{\phi(n)/2}$. This happens e.g.\! for $n=39$, $56$, $78$. For these three instances the conductor is divisible by just two distinct primes, and hence (by a result of Sinnott \cite{sinnot}) $\boldsymbol{u}$ coincides with the group of cyclotomic units. 
\end{rem}

\subsubsection{The sign function}
From \eqref{symppp} we see that the sign function is
\be
\mathsf{sign}\colon \sigma\mapsto \frac{\sigma(\lambda_v)}{i|\sigma(\lambda_v)|}
=\frac{\sigma(\varrho)}{i|\sigma(\varrho)|}\cdot \frac{\sigma(\eta)}{|\sigma(\eta)|}.
\ee

Comparing eqn.\eqref{kkkxxxz10} with 
eqn.\eqref{kkkwerxx}, we conclude

\begin{corl}\label{mmmxz32} All sign functions (i.e.\! all spectral invariants) are realized for some
arithmetic embedding $\Z_n\hookrightarrow Sp(\tilde\Omega,\Z)$ with $\tilde\Omega$ a non-necessarily principal polarization. 
\end{corl}

A few examples are in order.
We have seen above that if
\be
\boldsymbol{(\ast\ast)}\qquad \text{$h_\BK$ is odd and the conductor $\bn$ is either a power of 2 or an odd prime}
\ee
then \emph{all} $2^{\phi(\bn)/2}$ signs are realized with $\Omega$ principal. Let us consider the first few  $n$'s which do not satisfy these condition
$\boldsymbol{(\ast\ast)}$. The first one is $9$.

\begin{exe} \fbox{\;$n=9$\;} In this case $h_\BK=Q=1$ so there are $2^{\phi(n)/2}=8$ distinct $Sp(6,\Z)$-conjugacy classes of order $9$ elements. Since 9 is a prime power, $\boldsymbol{u}$ is the group of the the real cyclotomic units, that is, 
\be
\boldsymbol{u}=\pm (\zeta+\zeta^8)^\Z\,(\zeta^4+\zeta^5)\equiv \pm\, u^\Z\, v^\Z.
\ee
The sign table for the three elements $\sigma_a\in \mathsf{Gal}(\Bbbk/\BQ)$ are
\be
\begin{array}{c|ccc}
& \sigma_1 & \sigma_2 & \sigma_4\\\hline
u & +& + & -\\
v & -& + &+
\end{array}
\ee
Thus $\boldsymbol{u}/\boldsymbol{u}_+\cong \Z_2^3\cong \boldsymbol{u}/N\boldsymbol{U}$; hence all $Sp(6,\Z)$-conjugacy classes are distinct as $Sp(6,\R)$-conjugacy classes and all 8 signs are realized (cfr.\! also \cite{eie}).
\end{exe}

\begin{exe} \fbox{\;$n=12$\;} In this case $h_\BK=1$ and $Q=2$,
so we have only $2^{\phi(12)/2}/Q=2$ inequivalent embeddings over $Sp(4,\Z)$.
They correspond to $\eta=\pm1$.
One has
\be
\overline{\Phi_{12}^\prime(e^{2\pi i/12})}= 2\sqrt{3}\,e^{-2\pi i/3},
\ee
so as $\varrho$ we may choose 
\be
-i\, e^{-2\pi i/3}/(2\sqrt{3}\,e^{-2\pi i/3})\equiv -i/(2\sqrt{3})=\frac{1}{2(\zeta^4-\zeta^{-4})}.
\ee
Thus for $n=12$ the sign function is
\be
\begin{array}{l}
(\Z/12\Z)^\times\to \{\pm1\},\\ a\mapsto i^{a-1},\end{array}\quad \text{i.e. }\quad  \begin{array}{ll}1\mapsto-1, &5\mapsto+1,\\ 7\mapsto-1, &11\mapsto+1.\end{array}
\ee
The group of units of $\Bbbk\equiv\BQ[\sqrt{3}]$ is $\boldsymbol{u}=\pm(2-\sqrt{3})^\Z$, and $\boldsymbol{u}/\boldsymbol{u}_+\cong\{\pm1\}$. Thus only two spectral invariants out of four are realized by embeddings $\Z_{12}\to Sp(4,\Z)$ (as expected) namely $\{\zeta,\zeta^7\}$ and
$\{\zeta^5,\zeta^{11}\}$ (where $\zeta=e^{2\pi i/12}$). 

\begin{rem} Since the ST group $G_8$ has degrees $\{12,8\}$, according to the dimension formulae of section 4, its elements of order 12 should have an embedding with spectral invariant $\{\zeta,\zeta^5\}$. We saw in \S.\,\ref{connstantmaoa} (\textbf{Example \ref{e6}}) that the $G_8$-invariant polarization has $\det\tilde\Omega=2^2$. This polarization has the form in \S.\,\ref{kkzaqww} with $\kappa=1+\sqrt{3}\in \mathfrak{o}$ which is not totally positive (its norm is negative)
$N_{\Bbbk/\BQ}(1+\sqrt{3})=-2$.
This illustrates \textbf{Corollary \ref{mmmxz32}}.
\end{rem}
\end{exe}

\begin{exe} \fbox{\;$n=15$\;}
Again $h_\BK=1$ and we have $2^{\phi(15)/2}/Q=8$ different $Sp(8,\Z)$-conjugacy classes. We have ($\zeta\equiv e^{2\pi i/15}$)
\be
\Phi^\prime_{15}(\zeta)=15\,\frac{\zeta^{-1}(\zeta-1)}{(\zeta^5-1)(\zeta^3-1)}
\quad\text{we choose}\quad
\varrho=-\frac{1}{15}\, \zeta^3(\zeta^2-1)(\zeta^{10}-1)(\zeta^{12}-1),
\ee
and then the signs of the reference embedding are
\be
\frac{\sigma_a(\varrho)}{i\,|\sigma_a(\varrho)|}=\begin{cases} +1 & a=1,2\\
-1 & a=4,7
\end{cases}\quad \xrightarrow{\ \text{spectral inv.}\ }\quad \{\zeta,\zeta^2,\zeta^8,\zeta^{11}\}.
\ee
Writing $\xi=\zeta+\zeta^{-1}$, we have
(according to \textsc{Mathematica})
\be
\boldsymbol{u}=\pm (-1+3\xi-\xi^3)^\Z\, (2+3\xi-\xi^2-\xi^3)^\Z\,(-1+4\xi-\xi^3)^\Z=
\pm u_1^\Z\,u_2^\Z\,u_3^\Z,
\ee
whose signs are 
\be
\begin{array}{c|cccc}
& \sigma_1 & \sigma_2 & \sigma_4 & \sigma_7\\\hline
u_1 & -& + & -&+\\
u_2& -& + &+& -\\
u_3 & + & + &- &-
\end{array}
\ee
so that $\boldsymbol{u}/\boldsymbol{u}_+\cong\Z_2^3$, i.e.\! all $Sp(8,\Z)$  classes correspond to $Sp(8,\R)$ classes, with signs:
$$\{++--\},\ \{-++-\},\ \{-+-+\},\ \{++++\},\ 
\{--++\},\ \{+--+\},\ \{+-+-\},\ \{----\},$$
that is, the spectral invariants
\be
\begin{gathered}
\{\zeta,\zeta^2,\zeta^{11},\zeta^8\},\quad \{\zeta^{14},\zeta^2,\zeta^4,\zeta^8\},\quad 
\{\zeta^{14},\zeta^2,\zeta^{11},\zeta^7\},\quad
\{\zeta,\zeta^2,\zeta^4,\zeta^7\},\\
\{\zeta^{14},\zeta^{13},\zeta^{4},\zeta^7\},\quad\{\zeta,\zeta^{13},\zeta^{11},\zeta^7\},\quad \{\zeta,\zeta^{13},\zeta^{4},\zeta^8\},\quad \{\zeta^{14},\zeta^{13},\zeta^{11},\zeta^8\}.
\end{gathered}   
\ee 
\end{exe}

\begin{exe} \fbox{\;$n=20$\;} Again $h_\BK=1$ and we have $2^{\phi(20)/2}/Q=8$ different $Sp(8,\Z)$-conjugacy classes. We have ($\zeta\equiv e^{2\pi i/20}$)
\be
\Phi^\prime_{20}=\frac{10\,\zeta^8}{\zeta+\zeta^{-1}}\quad\text{we choose}\quad
\varrho= \frac{1}{10}(\zeta+\zeta^{-1})\zeta^5.
\ee
Then
\be
\frac{\sigma_a(\varrho)}{i\,|\sigma_a(\varrho)|}=\begin{cases} +1 & a=1,7\\
-1 & a=3,9
\end{cases}\quad \xrightarrow{\ \text{spectral inv.}\ }\quad \{\zeta,\zeta^{17},\zeta^7,\zeta^{11}\}.
\ee
Again, with $\xi=\zeta+\zeta^{-1}$
\be
\boldsymbol{u}=\pm (1+\xi)^\Z\, (2-\xi^2)^\Z\,(1-3\xi+\xi^3)^\Z=\pm u_1^\Z\,u_2^\Z\,u_3^\Z,
\ee
with signs
\be
\begin{array}{c|cccc}
& \sigma_1 & \sigma_3 & \sigma_7 & \sigma_9\\\hline
u_1 & +& + & -&-\\
u_2& -& + &+& -\\
u_3 & + & - &+ &-
\end{array}
\ee
so, again $\boldsymbol{u}/\boldsymbol{u}_+\cong\Z_2^3$ and $Sp(8,\Z)$ and $Sp(8,\R)$ conjugacy classes coincide.
\end{exe}

\begin{exe} \fbox{\;$n=21$\;}
Again $h_\BK=1$ and we have $2^{\phi(21)/2}/Q=32$ different $Sp(12,\Z)$-conjugacy classes. We have ($\zeta\equiv e^{2\pi i/21}$)
\be
\Phi^\prime_{21}(\zeta)=21\,\frac{\zeta^{-1}(\zeta-1)}{(\zeta^7-1)(\zeta^3-1)}
\quad\text{we choose}\quad
\varrho=-\frac{1}{21}\, \zeta^9(\zeta^2-1)(\zeta^{14}-1)(\zeta^{18}-1),
\ee
and then the signs of the reference embedding are
\be
\frac{\sigma_a(\varrho)}{i\,|\sigma_a(\varrho)|}=\begin{cases} +1 & a=1\\
-1 & a=2,4,5,8,10
\end{cases}\quad \xrightarrow{\ \text{spectral inv.}\ }\quad \{\zeta,\zeta^{19},\zeta^{17},\zeta^{16},\zeta^{13},\zeta^{11}\}.
\ee
Writing $\xi=\zeta+\zeta^{-1}$, we have
(according  to \textsc{Mathematica})
\be
\begin{split}
\boldsymbol{u}&=\pm \xi^\Z\, (2-\xi^2)^\Z\, (2-4\xi^2+\xi^4)^\Z\,(3-8\xi-\xi^2+6\xi^3-\xi^5)^\Z\,(2-5\xi-\xi^2+5\xi^3-\xi^5)^\Z\\
&=\pm u_1^\Z\,u_2^\Z\,u_3^\Z\,u_4^\Z\,u_5^\Z,
\end{split}
\ee
whose signs are 
\be
\begin{array}{c|cccccc}
& \sigma_1 & \sigma_2 & \sigma_4 & \sigma_5 & \sigma_8 & \sigma_{10}\\\hline
u_1 & +& + & +&+&-&-\\
u_2& -& - &+& + & - &-\\
u_3 & + & - &+ &+ &- &+\\
u_4&  + &+ &-&+&+&-\\
u_5& - & + &- &+&-&-
\end{array}
\ee
so that $u_1u_2u_3u_4u_5\in \boldsymbol{u}_+$ and $\boldsymbol{u}/\boldsymbol{u}_+\cong \Z_2^{\phi(21)/2-1}\equiv\Z_2^5$, and only 32 out of the possible 64 signs are actually realized.
The allowed spectral invariants may be read from the above tables.
\end{exe}

\subsection{Explicit matrices}\label{exmatr}

We now write explicitly the integral  matrices yielding a reference embedding 
$\Z_n\hookrightarrow Sp(\phi(n),\Z)$
on which we act with the groups $H$ or  
$H_\R$ to get the inequivalent embeddings over $Sp(\phi(n),\Z)$ and $Sp(\phi(n),\R)$, respectively.

Let $\mathfrak{a}$ be a fractional ideal of $\BK$ such that $N\mathfrak{a}=(\eta)$, $\eta\in \Bbbk^\times$. We write $k=\phi(n)/2$ and
 choose generators of the free $\Z$-module $\mathfrak{a}$, $\mathfrak{a}=\bigoplus_{a=1}^{2k}\Z \omega_a$. Define
the dual vector $(\phi^a)\in \BK^{2k}$ by the condition
\be\label{kkkaq0}
\mathrm{Tr}_{\BK/\BQ}\big(\omega_a\,\bar\phi^b\big)={\delta_a}^b.
\ee
By definition, the dual ideal is
$
\mathfrak{a}^\ast=\bigoplus_{a=1}^{2k} \Z \phi^a$.
Since $\mathfrak{a}^\ast=\lambda\mathfrak{a}$ with $\lambda$ purely imaginary (cfr.\! eqn.\eqref{hhhhaq}), there exists
$\Lambda\in GL(2k,\Z)$ such that
\be\label{nnnza2}
\Lambda^{ab}\omega_b= \lambda^{-1} \phi^b,\qquad \Lambda_{ba}\phi^b=\bar \lambda\omega_a,
\ee
where $\Lambda_{ab}$ is the inverse of $\Lambda^{ab}$; the second equation being a consequence of the first in view of \eqref{kkkaq0}). Then
\be
{((\Lambda^t)^{-1}\Lambda)_a}^b\, \omega_b =\lambda^{-1} \bar\lambda\,\omega_a=-\omega_a,
\ee
and the integral matrix $\Lambda$ is antisymmetric with determinant 1, hence similar over the integers to the standard symplectic matrix $\Omega$.

Each ideal class $[\mathfrak{a}]\in C^-_\BK=\ker(C_\BK\xrightarrow{\;N\;}\C_\Bbbk)$
yields an embedding $\Z_n\to GL(2k,\Z)$ which is quaternionic with respect to $2^k/Q$ inequivalent (over $\Z$) symplectic structures. To get the reference embedding, 
let us consider the trivial class in $C^-_\BK$; as a representative ideal we take
$\mathfrak{O}\equiv (1)$ itself. 

As a $\Z$-basis of $\mathfrak{O}$ we take $\omega_x=\zeta^{x-1}$ with $x=1,\dots,2k$. 
It is convenient to re-label the elements of this basis. Let $n=p^{r_1}_1p^{r_2}_2\cdots p_s^{r_s}$ be the prime decomposition of $n$. Choose   primitive
$p_i^{r_i}$-th roots of 1, $\zeta_i$. Then $\zeta=\prod_i\zeta_i$ is a primitive $n$-th root. By the Chinese remainder theorem, there exist
integers $e_i$, $i=1,\dots,s$, such that
\be
e_i = \delta_{ij}\!\!\!\mod p_j^{r_j}\quad \forall\; i,j=1,\dots,s.  
\ee
We write the index $x=1,2,\dots,2k$ uniquely as 
\be
\begin{aligned}
&x-1=\sum_i  e_i\big((a_i-1)p^{r_i-1}+(\alpha_i-1)\big)\\ 
&\text{with }a_i=1,\dots,p_i-1\ \text{and }\alpha_i\in \Z/p^{r_i-1}_i\Z.
\end{aligned}
\ee
Then the basis $\zeta^{x-1}$ is re-written as a tensor product over the primes $p_i\mid n$
\be\label{jjjaz}
\omega_{a_i\alpha_i}=\prod_i \zeta_i^{p_i^{r_i-1}(a_i-1)+(\alpha_i-1)},\qquad \omega_{\boldsymbol{a}\,\boldsymbol{\alpha}}=\bigotimes_{i=1}^s (\omega_i)_{a_i\alpha_i} 
\ee
It is convenient to realize the action of $m$ as multiplication by a different primitive $n$-root
\be
\zeta^\prime=\prod_i \zeta_i^{p_i^{r_i-1}+1}.
\ee
With these conventions, the action of $\Z_n\cong \prod_i \Z_{p_i^{r_i}}$ explicitly factorizes in the product of the action of the factor groups $\Z_{p_i^{r_i}}$
\be
\zeta^\prime\,\omega_{\boldsymbol{a}\,\boldsymbol{\alpha}}=\bigotimes_{i=1}^s {(m_i)_{a_i\alpha_i}}^{b_j\beta_j}\,\omega_{b_j\beta_j}\quad\text{that is}\quad m=\bigotimes_{i=1}^s m_i.
\ee 
We have to discuss separately the matrices $m_i$ associated to an odd prime and the one associated to 2 (if present). For $p_i$ odd,
 each $m_i$ factorizes in the matrix $m_{(p)}$ yielding the reference embedding $\Z_p\to GL(p-1,\Z)$
times the $p_i^{r_i-1}$-circulant
\be
m_i= m_{(p_i)}\otimes C_{p_i,r_i},
\ee
where $m_{(p)}$ (resp.\! $C_{p,r}$) is the $(p-1)\times (p-1)$ matrix (resp.\! $p^{r-1}\times p^{r-1}$)
\be\label{mattrr}
m_{(p)}=\left[\begin{array}{c|ccc}
0 & 1 &&\\
\vdots & &\ddots &\\
0 &&& 1\\\hline
-1 & -1 &\cdots & -1\end{array}\right],\qquad C_{p,r}=\left[\begin{array}{c|ccc}
0 & 1 &&\\
\vdots & &\ddots &\\
0 &&& 1\\\hline
1 & 0 &\cdots & 0\end{array}\right]
\ee
For $p_1=2$, $m_1$ is just the scalar $-1$ for $r_1=1$. For $r_1\geq 2$, $m_1$ is the tensor product of the $2\times 2$ matrix
$m_{(2)}$ yielding the embedding $\Z_4\to GL(2,\Z)$ (the $2\times 2$ matrix of the form on the left of \eqref{mattrr}) 
with the $2^{r_1-2}$-circulant. 

Likewise, the Hermitian structure factorizes
\be
\mathrm{Tr}_{\BK/\BQ}\big(\omega_{a_i\alpha_i}\,\bar\omega_{b_i,\beta_i}\big)=\prod_{i} \Big(p_i^{r_i-1}\,T_{i,a_ib_i} \,\delta^{(p_i^{r_i-1})}_{\alpha_i\beta_i}\Big).
\ee
where  
\be
T_{i,ab}= p_i\,\delta_{ab}-v_a v_b,\quad v_a=(1,1,\dots,1),\qquad \delta^{(\ell)}_{\alpha\beta}=\begin{cases} 1 & \alpha\equiv\beta\mod \ell\\
0 & \text{otherwise,}
\end{cases}
\ee
and we used \cite{carcar}
\be\label{kkka}
\mathrm{Tr}_{\BK/\mathbb{Q}}\!\big(\zeta^t\big)\equiv\sum_{\ell\in (\Z/n\Z)^\times} \zeta^{\ell t}=\frac{\phi(n)}{\phi(n/(n,t))}\,\mu(n/(n,t)),
\ee
where $\mu(x)$ is the M\"obius function.
Let $\zeta_i$ be the primitive $p^{r_i}$ we have chosen and set $\xi_i=\zeta_i^{p_i^{r_i-1}}$ (a primitive $p$-root). We write 
\be
\omega_i^{a_i\,\alpha_i}= (\xi_i^{a_i-1}-\xi_i^{-1})\zeta_i^{\alpha_i-1}
\ee
The dual basis of $\omega_{\boldsymbol{a}\,\boldsymbol{\alpha}}$ (cfr.\! \eqref{jjjaz})
is
\be
\phi^{\boldsymbol{a}\,\boldsymbol{\alpha}}=\frac{1}{n}\bigotimes_{i=1}^s (\omega_i)^{a_i\,\alpha_i}.
\ee
Using the reference $\varrho$ described in \textsc{appendix \ref{kkkkza3b2}}, for $s$ \emph{odd} the symplectic matrix $\Lambda$ of our reference embedding,
 is simply the tensor product of the 
reference symplectic matrices $\Lambda_i$ for each prime  $p_i|n$, $\Lambda=\bigotimes_i\Lambda_i$; for $p_i$ odd
\be
\Lambda_i = U_i \otimes \boldsymbol{1}_{p_i^{r_i-1}},\quad U_i=\left[\begin{array}{c|c} 0 & \ -J_i\ \\\hline
J_i^t & 0\end{array}\right]\quad \begin{array}{l}\text{where $J_i$ is the $(p_i-1)/2\times(p_i-1)/2$}\\
\text{Jordan block of eigenvalue $-1$,}\end{array}
\ee 
while for $p_1=2$ we have $\Lambda_1=1$ if $r_1=1$ and otherwise
\be
\Lambda_1=\begin{bmatrix}0 & 1\\ -1 &0\end{bmatrix}\otimes \boldsymbol{1}^{r_1-2}.
\ee
To set this matrix in the standard form $\Omega$, it suffices to replace the above $\Z$-basis with the $\Z$-basis
\be
\bigotimes_{i=1}^s(\tilde\omega_i)_{a_i,\alpha_i}\quad \text{where }\ \tilde\omega_{a_i,\alpha_i}= \begin{cases}
\omega_{a_i,\alpha_i} & 1\leq a_i < p_i/2\\
\phi^{a_i,\alpha_i} & p_i/2< a_i < p_i.\end{cases}
\ee
Therefore, for $s$ odd the reference embedding
$\Z_n\cong\prod_{i=1}^s \Z_{p_i^{r_i}}$
in $Sp(\phi(n),\Z)$ is simply the tensor product of the embeddings of the factor groups $\Z_{p_i^{r_i}}\to Sp(\phi(p_i^{r_i},\Z)$.
 
For $s$ even the above tensor product produces an orthogonal rather than a symplectic embedding since $\bigotimes_i\Lambda_i$ is symmetric; to get an antisymmetric pairing, we
multiply it by the reference imaginary unit of \textbf{Lemma \ref{kkkdh}}. We get
\be
\Lambda= \begin{cases}m^{n/4} \bigotimes_{i=1}^s \Lambda_i & 4\mid n\\
(m - m^{-1})\bigotimes_{i=1}^s\Lambda_i
&\text{otherwise.}
\end{cases}
\ee
One checks that
\be
m \Lambda m^t =\Lambda.
\ee
This completes the explicit description of the reference embedding. Now we act with the group $H$ on it to get all other inequivalent embeddings.

\paragraph{Multiplying by an element of
$\boldsymbol{u}/N\boldsymbol{U}$.}
This subgroup of $H$ does not change the matrix $m$ but only the symplectic matrix $\Lambda$. Since $\fo=\Z[\zeta+\zeta^{-1}]$, each element of $v\in\boldsymbol{u}/N\boldsymbol{U}$ may be represented by a polynomial 
\be
p_v(\zeta+\zeta^{-1})=\sum_{i=0}^{\phi(n)/2-1} c(v)_i(\zeta+\zeta^{-1})^i\quad \text{with } c(v)_i\in\Z.
\ee
Then the change in the matrix $\Lambda$ produced by multiplication by $v$ is simply
\be
\Lambda\to \Lambda_v \equiv p_v(m+m^{-1})\,\Lambda= \Lambda\, p_v(m+m^{-1})^t\quad\Longrightarrow\quad m\Lambda_v m^t=\Lambda_v.
\ee

\paragraph{Replacing $(1)$ with a 
non-principal fractional ideal.}
In a Dedekind domain $\fO$, an ideal
$\mathfrak{a}$ which is not principal may be generated by two elements, that is,
has the form
\be
\mathfrak{a}= x\,\fO+y\,\fO,\qquad x,y \in \fO.
\ee
Let $\omega_a$ a $\Z$-basis of $\fO$
and $\varpi_a$ a $\Z$=basis of $\mathfrak{a}$. There are integral matrices $X$, $Y$ such that
\be
x \,\omega_a= {X_a}^b\,\varpi_b,\qquad
y\, \omega_a= {Y_a}^b\,\varpi_b.
\ee
The matrix $M$ yielding the action of $\zeta^\prime$ in the basis $\varpi_a$ is
$M=X^{-1}mX=Y^{-1}mY$. The condition that $N\mathfrak{a}$ is principal then implies that the induced $\Lambda$ is defined over $\Z$.

\subsection{Reducible minimal polynomial}

For completeness, we give some additional details on the case that the minimal polynomial of the elliptic element $m\in Sp(2k,\Z)$
is reducible over $\BQ$
\be
M(z)=\Phi_{d}(z)\,\Phi_n(z),\qquad n>d.
\ee
If $n/d$ is not a prime power, all embeddings are conjugate to a block-diagonal one
$m_d\oplus m_n$, see \textbf{Lemma \ref{llkas}}. Suppose that $n/d=p^r$ with $p$ prime while $(d,p)=1$. We still have the block-diagonal embeddings, and all embeddings are conjugate to block-diagonal ones over $\BQ$. Thus there is an element $R\in Sp(2k,\BQ)$ which splits the $\Z[m]$-module $V$ and the symplectic structure
\be
\begin{aligned}
&R^{-1}V=\mathfrak{a}\oplus \mathfrak{b}\quad\text{with}\quad \mathfrak{a}\in \mathfrak{I}_{\BQ[\zeta_n]},\ \ \mathfrak{b}\in \mathfrak{I}_{\BQ[\zeta_d]},\\
&R^{-1} mR\ \text{acts as as multiplication by $\zeta_n\times \zeta_d$},\\
&R^t\,\Omega\, R= J_n\oplus J_d.
\end{aligned}
\ee
$R$ must have the form
\be
R= \sum_{s=0}^{(\phi(n)+\phi(d))/2-1} a_s\, \big(m+m^{-1}\big)^s,\qquad a_s\in\Z
\ee
for certain coefficients $a_s$.

\section{Tables of dimensions for small $k$}

In this section we present some sample tables of both dimensions and dimension $k$-tuples for small values of the rank $k$. 

\subsection{New-dimension lists for $k\leq 13$}

In table \ref{new1} we list the new-dimension sets $\mathfrak{N}(k)$, for ranks $1\leq k\leq 13$. They have been computed using the defining formula:
\be
\mathfrak{N}(k):= \left\{ \frac{l}{s} \in \mathbb Q_{\geq 1}\;:\; \phi(l)= 2k,\ (l,s)=1\right\}.
\ee 
The set of dimensions allowed in rank $k$ is contained in the set
\be
\widehat{\Xi}(k)=\bigcup_{\ell=1}^k \mathfrak{N}(k).
\ee
From eqns.\eqref{xxyy1}\eqref{xxyy2},
 the cardinalities of the new-dimension sets $\mathfrak{N}(k)$ are:
\be
\begin{tabular}{c|ccccccccccccc}\hline\hline
$k$ & 1 & 2 & 3 & 4 &5& 6 & 7 & 8 & 9 & 10 & 11 & 12 & 13\\
$|\mathfrak{N}(k)|$ & 2+6 & 16 & 24 & 40 &20& 72& 0 & 96 & 72 & 100 & 44 &240 & 0\\
$|\mathfrak{N}(k)|_\text{int}$ & 2+3 &4 & 4&5 
&2& 6 & 0& 6 &4 &5 &2 &10 & 0\\
\hline\hline
\end{tabular}
\ee

\subsection{Dimension $k$-tuples: USE OF THE TABLES}

Tables of all ALLOWED dimension $k$-tuples become quite long pretty soon as we increase $k$.
For conciseness we list only the \textsc{strongly regular} dimension $k$-tuples from which one can infer all allowed $k$-tuples. The tables of \textsc{strongly regular} $k$-tuples contain the basic informations needed to check whether a proposed dimension $k$-tuple $\{\Delta_1,\cdots,\Delta_k\}$ is consistent or not with the arguments of the present paper.
By definition, a \textsc{strongly regular}
dimension $k$-tuples is a set of dimensions as computed using eqn.\eqref{kkkazqw} along a normal ray $M_\ast\subset M$ with strongly regular monodromy $m_\ast$ (i.e.\! such that the  characteristic polynomial of $m_\ast$ is square-free). In  turn, the \textsc{strongly regular} monodromies may be distinguished in two kinds: the ones consistent with a principal polarization, $m_\ast\in Sp(2k,\Z)$, and those associated to a suitable non principal polarization, $m_\ast\in S(\Omega)_\Z$ ($\det\Omega\neq1$). 
The complete list of all \textsc{allowed} dimension $k$-tuples is then recovered from the tables of the \textsc{strongly regular} ones by the algorithm described in section 4.5.1 which we review in the next subsection.

In tables \ref{3tuples}, \ref{3btuples}, and \ref{4tuples} we present the list of the strongly regular dimension $k$-tuples for $k=3$ and $k=4$. For $k=3$ we list both the principal (table \ref{3tuples}) and non-principal 3-tuples (table \ref{3btuples}), while for $k=4$ we limit ourselves to the principal ones (table \ref{4tuples}).

\subsubsection{The algorithm to check
admissibility of a given dimension $k$-tuple}

Suppose we are given a would-be dimension $k$-tuple, $\{\Delta_1,\cdots,\Delta_k\}$, written in non-increasing order $\Delta_i\geq \Delta_{i+1}$, and we wish to determine whether it is consistent with the geometric conditions discussed in the present paper.
In order to answer the question, we focus on the $k$ normal rays
in the Coulomb branch 
\be
M_i=\big\{u_j=0\ \text{for }i\neq j\big\}\subset M,\qquad i=1,2,\dots,k.
\ee
The monodromy along $M_i$, $m_i$,
may be either regular or irregular.
From a regular monodromy $m_i$ we may read all $k$ dimensions using the Universal Formula  \eqref{kkkazqw}. Experience with explicit examples (e.g.\! the ones having constant period map or those engineered in $F$-theory) suggests that the ray $M_1$ associated with the chiral operator $u_1$ of largest dimension $\Delta_1=\max_i\Delta_i$ always has a regular monodromy.\footnote{\ If there are $\ell>1$ chiral generators of largest dimension, we have
a $\mathbb{P}^{\ell-1}$ family of normal rays associated to this dimension. In this case the expectation is that the \emph{generic} ray in the family has regular monodromy.} However $m_1$ may be just \emph{weakly} regular; in this case eqn.\eqref{kkkazqw} still applies but the corresponding $k$-tuple is not listed in the tables, and we need to follow the procedure described below. 
Recall that we have defined the \emph{regular rank} $k_{\text{reg},\,1}\leq k$ of the monodromy matrix $m_1$ to be one-half the degree of the square-free part of its characteristic polynomial.  

To run the algorithm, one begins writing the rational number $\Delta_1$ in minimal form, $\Delta_1\equiv n_1/r_1$ with $(n_1,r_1)=1$; 
we may assume $n_1>2$ by the argument in \S.\ref{jjjjjjz}. Let $\ell_1$ be the largest integer such that $\Delta_{\ell_1}=\Delta_1$, that is, the multiplicity of the largest dimension.
$\Delta_1$ is a new-dimension in some rank $k_1=\tfrac{1}{2}\phi(n_1)$ and $k_1\ell_1\leq k$. 
If $k_1=k$, the monodromy $m_1$ is automatically strongly regular, and hence the full $k$-tuple should appear in the tables of strongly regular $k$-tuples under the characteristic polynomial (C.P.) $\Phi_{n_1}$. More generally, if $k_1\ell_1=k$, the dimension $k$-tuple is the union of $\ell_1$  strongly regular dimension $k_1$-tuples for $\Phi_{n_1}$.
If $\ell_1 k_1< k$, the $k$-tuple is the union of $\ell_1$ strongly regular $k_1$-tuples and a residual $(k-\ell_1k_1)$-tuple $\{\Delta_i\}_{i\in A_1}$ ($A_1\subset\{1,\dots,k\}$).
Under the assumption that $m_1$ is (weakly) regular
we have
\be
\beta_i= \frac{\Delta_i-1}{\Delta_1},\qquad i\in A_1,
\ee
Let $s_i\in \mathbb{N}$ be the order of $\beta_i$ in $\BQ/\Z$.
The multiplicity $\ell(s)$ of each integer $s\in\mathbb{N}$ in $\{s_i\}_{i\in A_1}$ should be an integral multiple of $\phi(s)/2$, and $\{\beta_i\}_{i\in A_1}\cup \{1-\beta_i\}_{i\in A_1}$ should consist in the union of $2\ell(s_i)/\phi(s_i)$ copies of each set $B(s_i)=\{s_i/r,\ (s_i,r)=1,\ 1<r<s_i\}$. This corresponds to a characteristic polynomial
\be
\det[z-m_1]=\Phi_{n_1}(z)^{\ell_1}\prod_{s\in \mathbb{N}} \Phi_{s}(z)^{2\ell(s)/\phi(s)}.
\ee
If our candidate $k$-tuple satisfies all these requirements at the ray $M_1$, we next consider consistency conditions at the rays $M_2$, $M_3$, etc.\! along the lines of \S.\ref{conray}. The arguments are parallel to the ones for $M_1$ except that now we do not expect the monodromy to be fully regular (not even in the weak sense) so that only a
sub $k_\text{reg,\,$i$}$-tuple
of dimensions is fixed at each ray. This still yields non trivial consistency conditions as in the examples of \S.\ref{conray}.

The algorithm is longer to explain than to run. We illustrate the method in a typical example.  

\begin{exe} In rank 4 the following (non strongly regular) $4$-tuple exists\footnote{\ We thank Jacques Distler for suggesting this example.}
\be\label{vvvz}
\big\{14,10,8,4\big\}.
\ee
Let us apply the procedure to it. The largest dimension, $\Delta_1=14$ has multiplicity 1 and is a rank 3 new-dimension (see table \ref{new1}); then
3 out of the 4 dimensions
\eqref{vvvz} should form a strongly regular $3$-tuple to be found in table \ref{3tuples}  under $\Phi_{14}$. 
Indeed, there we find $\{14,10,4\}$.
The set of residual dimensions is $\{8\}$ (i.e.\! $A_1=\{3\}$) and 
\be
\beta_3=\frac{7}{14}\equiv \frac{1}{2},\qquad s_3=2,\qquad
2\ell(s_3)/\phi(2)\equiv 2, 
\ee
so that the 4-tuple \eqref{vvvz} is consistent with $m_1$ being weakly regular with
\be
\det[z-m_1]=\Phi_{14}(z)\,\Phi_2(z)^2.
\ee
Next we consider the second dimension $\Delta_2=10$.
From table \ref{new1} we see that it is a rank 2 new-dimension, hence we expect that two out of the four dimensions \eqref{vvvz}
 form a regular 2-tuple of the form $\{10,\ast\}$. Indeed in table \ref{tttr} we find both 
$\{10,8\}$ and $\{10,4\}$.
$\Delta_3=8$ is a rank 2 new-dimension, and $\{8,4\}$ is a regular $2$-tuple. Finally $\Delta_4=4$ is a rank 1 new-dimension (and hence a regular $1$-tuple). 
Thus the 4-tuple \eqref{vvvz} satisfies all the requirements.
\end{exe}

\subsection{Constructions of the lists}
The procedure to determine the lists is the one explained in sections 4, 5 which we sum up here.
We start by computing $\rho$ in the case of a cyclic group with an indecomposable characteristic polynomial $\Phi(z)$. We start the algorithm with 
\be
\rho_{temp}:=\frac{1}{\bar \Phi(\xi)'}.
\ee 
If $\rho_{temp}$ is purely imaginary, then $\rho=\rho_{temp}$, otherwise we find the unit $u$ such that $u\rho_{temp}$ is purely imaginary.\footnote{\ This algorithm exploits the fact that the group of units has a finite number of generators (called \emph{fundamental units}).} Once $\rho$ is computed, we get the initial signs for each element $\sigma_i$ of the Galois group associated to $\Phi(z)$. From the positive signs we compute the dimension tuple. All the other signs can be computed exploiting the fundamental units of the cyclotomic field (using \textsc{Mathematica} or PARI \cite{pari} when the former software fails): to each generator of the unit group we associate a sign tuple (the signs of the Galois elements). From the signs of $\rho$, it is easy to compute all possible signs by repeatedly multiplying the sings amongst one another. This is the algorithm to get all dimensions tuples.

The embedding may be obtained by a direct sum of lower order cyclic elements. In this case, we get the product of cyclotomic polynomials 
\be
\Phi_{d_1,d_2,\cdots,d_s}(z):=\Phi_{d_1}(z)\,\cdots\,\Phi_{d_s}(z).
\ee
 The procedure is similar to the above: we first compute all the signs separately for each factor -- using the above algorithm -- and then we put them all together to compute the full list of dimensions.

Particular attention must be payed to those products in which the ratio of the conductors of the cyclotomic factors is a (power of a) prime number, e.g. $\Phi_{12}\Phi_{4}$ in the rank 3 case. In this situation, the cyclic group representation is no longer irreducible, and thus the theory of Dedekind domains of rank 1 cannot be applied: $\rho$ is no longer a number but rather a matrix. Since this branch of number theory is not well developed, we preferred the explicit construction of the symplectic matrices $\Lambda$'s. We only consider the action of the group $H=\boldsymbol{u}/N\boldsymbol{U}$ on the initial embedding, whose signs are still defined by those of $\rho$ (the corresponding matrix shall be called $m$ and is given in subsection \ref{exmatr}). The action of $H$ only modifies the symplectic matrix: the signs of the characteristic polynomial of the new symplectic matrix, evaluated at the cyclotomic roots, give the sign changes to be applied to the original signs of $\rho$. Hence, the problem is to find the group $H$. It is a very had task to find this group: fortunately we know how it acts on the symplectic matrices:
\be
\Lambda_v:=p_v(m+m^{-1})\Lambda,\qquad \forall\; v \in H,
\ee
where $p_v$ is a polynomial with integer coefficients of maximal rank $\phi(n)-1$.
Thus, we can write an algorithm that looks for as many $p_v$'s as possible: every time we find one we check whether $\Lambda_v$ is principal (i.e. of unit determinant) and symplectic; once these two conditions are matched, we add the sign tuple to our results. In the end, we compute all the dimensions starting from the signs of $\rho$ and we multiply the signs with those explicitly found by our algorithm. If we find all possible signs, then the final result is definitive. In general, if we do not find all signs, we can only state our results with high confidence.

In table \ref{3tuples} we list the fully regular $3$-tuples for $k=3$. The first column yields the characteristic polynomial of the embedding and the second column the corresponding dimension 3-tuples.

Table \ref{4tuples} contains the fully regular $4$-tuples for $k=4$.

\section*{Acknowledgments}

We thank Michele Del Zotto for several interesting discussions about the classification of SCFT which prompted this project. We thank Ugo Bruzzo, Barbara Fantechi,  and Alessandro Tanzini for useful discussions.

\bgroup
\def\arraystretch{1.5}

\newpage

\begin{longtable}{c|lr}
\caption{New-dimension sets for ranks $1\leq k\leq 13$.}\label{new1} 
\\
\hline\hline
Rank $k$ & $\mathfrak{N}(k)$ \\ 
 \hline 
\endfirsthead

\multicolumn{3}{c}%
{{\bfseries \tablename\ \thetable{} -- continued from previous page}} \\
\hline 
Rank $k$ & $\mathfrak{N}(k)$ \\ 
 \hline 
\endhead

 \multicolumn{3}{r}{{Continued on next page}} \\ 
\endfoot

\hline \hline
\endlastfoot

1 & $1,2,3,\frac{3}{2},4,\frac{4}{3},6,\frac{6}{5}$\\
\hline
2 & $\frac{12}{11},\frac{10}{9},\frac{8}{7},\frac{5}{4},\frac{10}{7},\frac{8}{5},\frac{5}{3},\frac{12}{7},\frac{12}{5},\frac{5}{2},\frac{8}{3},\frac{10}{3},5,8,10,12$\\
\hline
3 & $\frac{18}{17},\frac{14}{13},\frac{9}{8},\frac{7}{6},\frac{14}{11},\frac{9}{7},\frac{18}{13},\frac{7}{5},\frac{14}{9},\frac{18}{11},\frac{7}{4},\frac{9}{5},\frac{9}{4},\frac{7}{3},\frac{18}{7},\frac{14}{5},\frac{7}{2},\frac{18}{5},\frac{9}{2},\frac{14}{3},7,9,14,18$\\
\hline
4 & $\frac{30}{29},\frac{24}{23},\frac{20}{19},\frac{16}{15},\frac{15}{14},\frac{15}{13},\frac{20}{17},\frac{16}{13},\frac{24}{19},\frac{30}{23},\frac{15}{11},\frac{24}{17},\frac{16}{11},\frac{20}{13},\frac{30}{19},\frac{30}{17},\frac{16}{9},\frac{20}{11},\frac{24}{13},\frac{15}{8},\frac{15}{7},\frac{24}{11},$\\
 & $\frac{20}{9},\frac{16}{7},\frac{30}{13},\frac{30}{11},\frac{20}{7},\frac{16}{5},\frac{24}{7},\frac{15}{4},\frac{30}{7},\frac{24}{5},\frac{16}{3},\frac{20}{3},\frac{15}{2},15,16,20,24,30$\\
\hline
5 & $\frac{22}{21},\frac{11}{10},\frac{22}{19},\frac{11}{9},\frac{22}{17},\frac{11}{8},\frac{22}{15},\frac{11}{7},\frac{22}{13},\frac{11}{6},\frac{11}{5},\frac{22}{9},\frac{11}{4},\frac{22}{7},\frac{11}{3},\frac{22}{5},\frac{11}{2},\frac{22}{3},11,22$\\
\hline
6 & $\frac{42}{41},\frac{36}{35},\frac{28}{27},\frac{26}{25},\frac{21}{20},\frac{13}{12},\frac{21}{19},\frac{28}{25},\frac{26}{23},\frac{42}{37},\frac{36}{31},\frac{13}{11},\frac{28}{23},\frac{21}{17},\frac{26}{21},\frac{36}{29},\frac{13}{10},\frac{21}{16},\frac{42}{31},\frac{26}{19},\frac{36}{25},\frac{13}{9},$\\
 & $\frac{42}{29},\frac{28}{19},\frac{26}{17},\frac{36}{23},\frac{21}{13},\frac{13}{8},\frac{28}{17},\frac{42}{25},\frac{26}{15},\frac{42}{23},\frac{13}{7},\frac{28}{15},\frac{36}{19},\frac{21}{11},\frac{21}{10},\frac{36}{17},\frac{28}{13},\frac{13}{6},\frac{42}{19},\frac{26}{11},\frac{42}{17},\frac{28}{11},$\\
 & $\frac{13}{5},\frac{21}{8},\frac{36}{13},\frac{26}{9},\frac{28}{9},\frac{42}{13},\frac{13}{4},\frac{36}{11},\frac{26}{7},\frac{42}{11},\frac{21}{5},\frac{13}{3},\frac{36}{7},\frac{26}{5},\frac{21}{4},\frac{28}{5},\frac{13}{2},\frac{36}{5},\frac{42}{5},\frac{26}{3},\frac{28}{3},\frac{21}{2},$\\
 & $13,21,26,28,36,42$\\
\hline
7 & None\\
\hline
8 & $\frac{60}{59},\frac{48}{47},\frac{40}{39},\frac{34}{33},\frac{32}{31},\frac{17}{16},\frac{40}{37},\frac{34}{31},\frac{32}{29},\frac{48}{43},\frac{60}{53},\frac{17}{15},\frac{48}{41},\frac{34}{29},\frac{32}{27},\frac{40}{33},\frac{17}{14},\frac{60}{49},\frac{34}{27},\frac{60}{47},\frac{32}{25},\frac{40}{31},$\\
 & $\frac{48}{37},\frac{17}{13},\frac{34}{25},\frac{48}{35},\frac{40}{29},\frac{32}{23},\frac{60}{43},\frac{17}{12},\frac{60}{41},\frac{34}{23},\frac{40}{27},\frac{32}{21},\frac{17}{11},\frac{48}{31},\frac{34}{21},\frac{60}{37},\frac{48}{29},\frac{32}{19},\frac{17}{10},\frac{40}{23},\frac{34}{19},\frac{32}{17},$\\
 & $\frac{17}{9},\frac{40}{21},\frac{48}{25},\frac{60}{31},\frac{60}{29},\frac{48}{23},\frac{40}{19},\frac{17}{8},\frac{32}{15},\frac{34}{15},\frac{40}{17},\frac{17}{7},\frac{32}{13},\frac{48}{19},\frac{60}{23},\frac{34}{13},\frac{48}{17},\frac{17}{6},\frac{32}{11},\frac{40}{13},\frac{34}{11},\frac{60}{19},$\\
 & $\frac{17}{5},\frac{60}{17},\frac{32}{9},\frac{40}{11},\frac{48}{13},\frac{34}{9},\frac{17}{4},\frac{48}{11},\frac{40}{9},\frac{32}{7},\frac{60}{13},\frac{34}{7},\frac{60}{11},\frac{17}{3},\frac{40}{7},\frac{32}{5},\frac{34}{5},\frac{48}{7},\frac{17}{2},\frac{60}{7},\frac{48}{5},\frac{32}{3},$\\
 & $\frac{34}{3},\frac{40}{3},17,32,34,40,48,60$\\
\hline
9 & $\frac{54}{53},\frac{38}{37},\frac{27}{26},\frac{19}{18},\frac{27}{25},\frac{38}{35},\frac{54}{49},\frac{19}{17},\frac{54}{47},\frac{38}{33},\frac{27}{23},\frac{19}{16},\frac{38}{31},\frac{27}{22},\frac{54}{43},\frac{19}{15},\frac{38}{29},\frac{54}{41},\frac{27}{20},\frac{19}{14},\frac{38}{27},\frac{27}{19},$\\
 & $\frac{54}{37},\frac{19}{13},\frac{38}{25},\frac{54}{35},\frac{19}{12},\frac{27}{17},\frac{38}{23},\frac{27}{16},\frac{19}{11},\frac{54}{31},\frac{38}{21},\frac{54}{29},\frac{19}{10},\frac{27}{14},\frac{27}{13},\frac{19}{9},\frac{54}{25},\frac{38}{17},\frac{54}{23},\frac{19}{8},\frac{27}{11},\frac{38}{15},$\\
 & $\frac{27}{10},\frac{19}{7},\frac{54}{19},\frac{38}{13},\frac{19}{6},\frac{54}{17},\frac{27}{8},\frac{38}{11},\frac{19}{5},\frac{27}{7},\frac{54}{13},\frac{38}{9},\frac{19}{4},\frac{54}{11},\frac{27}{5},\frac{38}{7},\frac{19}{3},\frac{27}{4},\frac{38}{5},\frac{54}{7},\frac{19}{2},\frac{54}{5},$\\
 & $\frac{38}{3},\frac{27}{2},19,27,38,54$\\
\hline
10 & $\frac{66}{65},\frac{50}{49},\frac{44}{43},\frac{33}{32},\frac{25}{24},\frac{50}{47},\frac{33}{31},\frac{44}{41},\frac{66}{61},\frac{25}{23},\frac{66}{59},\frac{44}{39},\frac{25}{22},\frac{33}{29},\frac{50}{43},\frac{33}{28},\frac{44}{37},\frac{25}{21},\frac{50}{41},\frac{66}{53},\frac{44}{35},\frac{33}{26},$\\
 & $\frac{50}{39},\frac{25}{19},\frac{33}{25},\frac{66}{49},\frac{50}{37},\frac{25}{18},\frac{66}{47},\frac{44}{31},\frac{33}{23},\frac{25}{17},\frac{50}{33},\frac{44}{29},\frac{66}{43},\frac{25}{16},\frac{66}{41},\frac{50}{31},\frac{44}{27},\frac{33}{20},\frac{50}{29},\frac{33}{19},\frac{44}{25},\frac{66}{37},$\\
 & $\frac{25}{14},\frac{50}{27},\frac{66}{35},\frac{44}{23},\frac{25}{13},\frac{33}{17},\frac{33}{16},\frac{25}{12},\frac{44}{21},\frac{66}{31},\frac{50}{23},\frac{25}{11},\frac{66}{29},\frac{44}{19},\frac{33}{14},\frac{50}{21},\frac{33}{13},\frac{44}{17},\frac{50}{19},\frac{66}{25},\frac{25}{9},\frac{66}{23},$\\
 & $\frac{44}{15},\frac{50}{17},\frac{25}{8},\frac{33}{10},\frac{44}{13},\frac{66}{19},\frac{25}{7},\frac{50}{13},\frac{66}{17},\frac{33}{8},\frac{25}{6},\frac{50}{11},\frac{33}{7},\frac{44}{9},\frac{66}{13},\frac{50}{9},\frac{25}{4},\frac{44}{7},\frac{33}{5},\frac{50}{7},\frac{33}{4},\frac{25}{3},$\\
 & $\frac{44}{5},\frac{66}{7},\frac{25}{2},\frac{66}{5},\frac{44}{3},\frac{33}{2},\frac{50}{3},25,33,44,50,66$\\
\hline
11 & $\frac{46}{45},\frac{23}{22},\frac{46}{43},\frac{23}{21},\frac{46}{41},\frac{23}{20},\frac{46}{39},\frac{23}{19},\frac{46}{37},\frac{23}{18},\frac{46}{35},\frac{23}{17},\frac{46}{33},\frac{23}{16},\frac{46}{31},\frac{23}{15},\frac{46}{29},\frac{23}{14},\frac{46}{27},\frac{23}{13},\frac{46}{25},\frac{23}{12},$\\
 & $\frac{23}{11},\frac{46}{21},\frac{23}{10},\frac{46}{19},\frac{23}{9},\frac{46}{17},\frac{23}{8},\frac{46}{15},\frac{23}{7},\frac{46}{13},\frac{23}{6},\frac{46}{11},\frac{23}{5},\frac{46}{9},\frac{23}{4},\frac{46}{7},\frac{23}{3},\frac{46}{5},\frac{23}{2},\frac{46}{3},23,46$\\
\hline
12 & $\frac{90}{89},\frac{84}{83},\frac{78}{77},\frac{72}{71},\frac{70}{69},\frac{56}{55},\frac{52}{51},\frac{45}{44},\frac{39}{38},\frac{35}{34},\frac{70}{67},\frac{45}{43},\frac{39}{37},\frac{56}{53},\frac{35}{33},\frac{52}{49},\frac{84}{79},\frac{78}{73},\frac{72}{67},\frac{90}{83},\frac{35}{32},\frac{45}{41},$\\
 & $\frac{56}{51},\frac{78}{71},\frac{52}{47},\frac{72}{65},\frac{39}{35},\frac{35}{31},\frac{90}{79},\frac{39}{34},\frac{70}{61},\frac{84}{73},\frac{52}{45},\frac{78}{67},\frac{90}{77},\frac{72}{61},\frac{84}{71},\frac{45}{38},\frac{70}{59},\frac{56}{47},\frac{35}{29},\frac{52}{43},\frac{45}{37},\frac{39}{32},$\\
 & $\frac{72}{59},\frac{70}{57},\frac{90}{73},\frac{56}{45},\frac{84}{67},\frac{39}{31},\frac{90}{71},\frac{52}{41},\frac{78}{61},\frac{84}{65},\frac{35}{27},\frac{56}{43},\frac{72}{55},\frac{70}{53},\frac{78}{59},\frac{45}{34},\frac{90}{67},\frac{39}{29},\frac{35}{26},\frac{72}{53},\frac{56}{41},\frac{70}{51},$\\
 & $\frac{84}{61},\frac{39}{28},\frac{52}{37},\frac{45}{32},\frac{78}{55},\frac{84}{59},\frac{56}{39},\frac{45}{31},\frac{35}{24},\frac{72}{49},\frac{78}{53},\frac{90}{61},\frac{52}{35},\frac{70}{47},\frac{56}{37},\frac{35}{23},\frac{90}{59},\frac{84}{55},\frac{72}{47},\frac{45}{29},\frac{39}{25},\frac{52}{33},$\\
 & $\frac{84}{53},\frac{35}{22},\frac{78}{49},\frac{45}{28},\frac{70}{43},\frac{78}{47},\frac{72}{43},\frac{52}{31},\frac{39}{23},\frac{56}{33},\frac{90}{53},\frac{70}{41},\frac{45}{26},\frac{72}{41},\frac{39}{22},\frac{84}{47},\frac{52}{29},\frac{70}{39},\frac{56}{31},\frac{78}{43},\frac{90}{49},\frac{35}{19},$\\
 & $\frac{70}{37},\frac{78}{41},\frac{90}{47},\frac{52}{27},\frac{56}{29},\frac{35}{18},\frac{72}{37},\frac{39}{20},\frac{84}{43},\frac{45}{23},\frac{45}{22},\frac{84}{41},\frac{39}{19},\frac{72}{35},\frac{35}{17},\frac{56}{27},\frac{52}{25},\frac{90}{43},\frac{78}{37},\frac{70}{33},\frac{35}{16},\frac{90}{41},$\\
 & $\frac{78}{35},\frac{56}{25},\frac{70}{31},\frac{52}{23},\frac{84}{37},\frac{39}{17},\frac{72}{31},\frac{45}{19},\frac{70}{29},\frac{90}{37},\frac{56}{23},\frac{39}{16},\frac{52}{21},\frac{72}{29},\frac{78}{31},\frac{70}{27},\frac{45}{17},\frac{78}{29},\frac{35}{13},\frac{84}{31},\frac{52}{19},\frac{39}{14},$\\
 & $\frac{45}{16},\frac{72}{25},\frac{84}{29},\frac{90}{31},\frac{35}{12},\frac{56}{19},\frac{70}{23},\frac{52}{17},\frac{90}{29},\frac{78}{25},\frac{72}{23},\frac{35}{11},\frac{45}{14},\frac{56}{17},\frac{84}{25},\frac{78}{23},\frac{45}{13},\frac{52}{15},\frac{39}{11},\frac{84}{23},\frac{70}{19},\frac{56}{15},$\\
 & $\frac{72}{19},\frac{35}{9},\frac{39}{10},\frac{90}{23},\frac{45}{11},\frac{78}{19},\frac{70}{17},\frac{72}{17},\frac{56}{13},\frac{35}{8},\frac{84}{19},\frac{78}{17},\frac{52}{11},\frac{90}{19},\frac{39}{8},\frac{84}{17},\frac{56}{11},\frac{90}{17},\frac{70}{13},\frac{72}{13},\frac{39}{7},\frac{45}{8},$\\
 & $\frac{52}{9},\frac{35}{6},\frac{56}{9},\frac{70}{11},\frac{45}{7},\frac{84}{13},\frac{72}{11},\frac{90}{13},\frac{78}{11},\frac{52}{7},\frac{84}{11},\frac{70}{9},\frac{39}{5},\frac{90}{11},\frac{35}{4},\frac{39}{4},\frac{72}{7},\frac{52}{5},\frac{78}{7},\frac{56}{5},\frac{45}{4},\frac{35}{3},$\\
 & $\frac{90}{7},\frac{72}{5},\frac{78}{5},\frac{84}{5},\frac{52}{3},\frac{35}{2},\frac{56}{3},\frac{39}{2},\frac{45}{2},\frac{70}{3},35,39,45,52,56,70,72,78,84,90$\\
\hline
13 & None
\end{longtable}



\begin{longtable}{c|lr}
\caption{Strongly regular principal 3-tuples for rank 3} \label{3tuples} \\
\hline\hline
 C.P. & strongly regular principal 3-tuples \\ 
 \hline 
\endfirsthead

\multicolumn{3}{c}%
{{\bfseries \tablename\ \thetable{} -- continued from previous page}} \\
\hline 
C.P. & strongly regular principal 3-tuples \\ 
 \hline 
\endhead

 \multicolumn{3}{r}{{Continued on next page}} \\ 
\endfoot

\hline \hline
\endlastfoot

$\Phi_{7}$ & $\left\{\frac{7}{6},\frac{4}{3},\frac{5}{3}\right\},\left\{\frac{7}{6},\frac{4}{3},\frac{3}{2}\right\},\{7,3,5\},\{7,3,4\},\{7,6,5\},\{7,6,4\},\left\{\frac{6}{5},\frac{7}{5},\frac{9}{5}\right\},$\\
& $\left\{\frac{6}{5},\frac{7}{5},\frac{8}{5}\right\},\left\{\frac{3}{2},\frac{7}{2},3\right\},\left\{\frac{3}{2},\frac{7}{2},\frac{5}{2}\right\},\left\{4,\frac{7}{2},3\right\},\left\{4,\frac{7}{2},\frac{5}{2}\right\},\left\{\frac{4}{3},\frac{5}{3},\frac{7}{3}\right\},\left\{\frac{5}{4},\frac{3}{2},\frac{7}{4}\right\},$\\
 & $\left\{\frac{4}{3},\frac{8}{3},\frac{7}{3}\right\},\left\{\frac{5}{4},\frac{9}{4},\frac{7}{4}\right\},\left\{3,\frac{5}{3},\frac{7}{3}\right\},\left\{\frac{5}{2},\frac{3}{2},\frac{7}{4}\right\},\left\{3,\frac{8}{3},\frac{7}{3}\right\},\left\{\frac{5}{2},\frac{9}{4},\frac{7}{4}\right\}$\\ 
\hline
$\Phi_{9}$ & $\left\{\frac{9}{8},\frac{3}{2},\frac{5}{4}\right\},\{9,5,8\},\{9,5,3\},\{9,6,8\},\{9,6,3\},\left\{\frac{6}{5},\frac{9}{5},\frac{12}{5}\right\},\left\{\frac{6}{5},\frac{9}{5},\frac{7}{5}\right\},$\\
& $\left\{\frac{5}{4},\frac{9}{4},\frac{3}{2}\right\},\left\{3,\frac{9}{4},\frac{3}{2}\right\},\left\{\frac{3}{2},3,\frac{9}{2}\right\},\left\{\frac{3}{2},\frac{7}{2},\frac{9}{2}\right\},\left\{\frac{8}{7},\frac{12}{7},\frac{9}{7}\right\},\left\{5,3,\frac{9}{2}\right\},\left\{5,\frac{7}{2},\frac{9}{2}\right\}$\\ 
\hline
$\Phi_{14}$ & $\{14,6,12\},\{14,6,4\},\{14,10,12\},\{14,10,4\},\left\{\frac{10}{9},\frac{14}{9},\frac{4}{3}\right\},\left\{\frac{6}{5},\frac{14}{5},\frac{8}{5}\right\},$ \\
 & $\left\{\frac{18}{5},\frac{14}{5},\frac{8}{5}\right\},\left\{\frac{4}{3},\frac{8}{3},\frac{14}{3}\right\},\left\{\frac{4}{3},4,\frac{14}{3}\right\}$\\ 
\hline
$\Phi_{18}$ & $\left\{\frac{18}{7},\frac{12}{7},\frac{8}{7}\right\},\left\{\frac{12}{5},\frac{18}{5},\frac{6}{5}\right\},\{12,14,18\},\{12,6,18\},\{8,14,18\},\{8,6,18\}$\\ 
\hline
$\Phi_{3,5}$ & $\left\{3,\frac{8}{5},\frac{14}{5}\right\},\left\{\frac{8}{3},5,4\right\},\left\{\frac{8}{3},5,3\right\},\left\{\frac{14}{9},\frac{4}{3},\frac{5}{3}\right\},\left\{\frac{14}{9},\frac{7}{3},\frac{5}{3}\right\},\left\{\frac{8}{3},\frac{3}{2},\frac{5}{2}\right\},\left\{\frac{8}{3},3,\frac{5}{2}\right\}$\\ 
\hline
$\Phi_{4,5}$ & $\left\{\frac{9}{4},5,4\right\},\left\{\frac{9}{4},5,3\right\},\left\{\frac{9}{4},\frac{4}{3},\frac{5}{3}\right\},\left\{\frac{9}{4},\frac{7}{3},\frac{5}{3}\right\}$\\ 
\hline
$\Phi_{5,6}$ & None\\ 
\hline
$\Phi_{3,8}$ & None\\ 
\hline
$\Phi_{4,8}$ & $\left\{\frac{4}{3},\frac{7}{6},\frac{3}{2}\right\},\left\{4,\frac{3}{2},\frac{5}{2}\right\},\left\{4,\frac{3}{2},\frac{7}{2}\right\},\left\{4,\frac{9}{2},\frac{5}{2}\right\},\left\{4,\frac{9}{2},\frac{7}{2}\right\},\left\{\frac{9}{7},\frac{8}{7},\frac{10}{7}\right\},\left\{\frac{9}{7},\frac{8}{7},\frac{12}{7}\right\},$\\
 & $\{3,8,4\},\{3,8,6\},\{7,8,4\},\{7,8,6\},\left\{\frac{7}{5},\frac{6}{5},\frac{8}{5}\right\},\left\{\frac{5}{3},\frac{4}{3},\frac{8}{3}\right\},\left\{\frac{7}{5},\frac{12}{5},\frac{8}{5}\right\},$ \\
 & $\left\{\frac{5}{3},\frac{10}{3},\frac{8}{3}\right\},\left\{3,\frac{4}{3},\frac{8}{3}\right\},\left\{3,\frac{10}{3},\frac{8}{3}\right\}$\\ 
\hline
$\Phi_{6,8}$ & $\left\{\frac{7}{3},8,4\right\},\left\{\frac{7}{3},8,6\right\},\left\{\frac{7}{3},\frac{6}{5},\frac{8}{5}\right\},\left\{\frac{7}{3},\frac{12}{5},\frac{8}{5}\right\}$\\ 
\hline
$\Phi_{3,10}$ & None\\ 
\hline
$\Phi_{4,10}$ & $\left\{\frac{7}{2},\frac{10}{3},4\right\},\left\{\frac{7}{2},\frac{10}{3},\frac{4}{3}\right\},\left\{\frac{7}{2},4,10\right\},\left\{\frac{7}{2},8,10\right\}$\\ 
\hline
$\Phi_{6,10}$ & $\left\{6,\frac{14}{5},\frac{8}{5}\right\},\left\{\frac{14}{9},\frac{10}{3},4\right\},\left\{\frac{14}{9},\frac{10}{3},\frac{4}{3}\right\},\left\{\frac{8}{3},4,10\right\},\left\{\frac{8}{3},8,10\right\}$\\ 
\hline
$\Phi_{3,12}$ & $\left\{\frac{9}{5},\frac{12}{5},\frac{6}{5}\right\},\{5,6,12\},\{9,6,12\}$\\ 
\hline
$\Phi_{4,12}$ & $\left\{4,\frac{8}{3},\frac{14}{3}\right\},\left\{4,\frac{8}{3},\frac{4}{3}\right\},\left\{4,\frac{14}{3},\frac{10}{3}\right\},\left\{4,\frac{10}{3},\frac{4}{3}\right\},\left\{\frac{10}{7},\frac{12}{7},\frac{18}{7}\right\},$\\
 & $\left\{\frac{10}{7},\frac{12}{7},\frac{8}{7}\right\},\left\{\frac{8}{5},\frac{12}{5},\frac{6}{5}\right\},\left\{\frac{14}{5},\frac{12}{5},\frac{6}{5}\right\},\{4,6,12\},\{4,12,8\},$\\
 & $\{10,12,8\},\left\{\frac{14}{11},\frac{18}{11},\frac{12}{11}\right\},\{10,6,12\},\left\{\frac{4}{3},\frac{14}{9},\frac{10}{9}\right\}$\\ 
\hline
$\Phi_{6,12}$ & $\left\{6,\frac{9}{2},\frac{3}{2}\right\},\left\{\frac{9}{7},\frac{12}{7},\frac{18}{7}\right\},\left\{\frac{7}{5},\frac{12}{5},\frac{6}{5}\right\},\left\{3,\frac{12}{5},\frac{6}{5}\right\},\{3,6,12\}$\\ 
\hline
$\Phi_{3,4,6}$ & $\left\{3,\frac{7}{4},\frac{3}{2}\right\},\left\{3,\frac{7}{4},\frac{7}{2}\right\},\left\{\frac{7}{3},4,\frac{5}{3}\right\},\left\{3,\frac{5}{2},6\right\},\left\{5,\frac{5}{2},6\right\}$
\end{longtable}

\newpage

\begin{longtable}{c|lr}
\caption{Strongly regular non-principal 3-tuples} \label{3btuples}\phantom{mmmmmmmmm} \\
\hline\hline
 C.P. & strongly regular non-principal 3-tuples \\ 
 \hline 
\endfirsthead

\multicolumn{3}{c}%
{{\bfseries \tablename\ \thetable{} -- continued from previous page}} \\
\hline 
C.P. & strongly regular non-principal 3-tuples \\ 
 \hline 
\endhead

 \multicolumn{3}{r}{{Continued on next page}} \\ 
\endfoot

\hline \hline
\endlastfoot

 $\Phi_{3,12}$ & $\left\{3,\frac{9}{4},\frac{5}{4}\right\},\{5,12,8\},\{9,12,8\}$ \phantom{mmmmmmmmmmmmmmmmm} \\ 
\hline
 $\Phi_{6,12}$ & $\left\{6,\frac{7}{2},\frac{3}{2}\right\},\{3,12,8\},\left\{\frac{9}{7},\frac{12}{7},\frac{8}{7}\right\}$
\end{longtable}
%

\begin{longtable}{c|lr}
\caption{Strongly regular principal 4-tuples for rank 4} \label{4tuples} \\
\hline\hline
 C.P. & strongly regular principal 4-tuples \\ 
 \hline 
\endfirsthead

\multicolumn{3}{c}%
{{\bfseries \tablename\ \thetable{} -- continued from previous page}} \\
\hline 
C.P. & strongly regular principal 4-tuples \\ 
 \hline 
\endhead

 \multicolumn{3}{r}{{Continued on next page}} \\ 
\endfoot

\hline \hline
\endlastfoot

$\Phi_{15}$ & $\left\{\frac{15}{2},\frac{9}{2},3,8\right\},\left\{\frac{15}{2},5,3,\frac{3}{2}\right\},\left\{\frac{5}{4},\frac{15}{8},\frac{3}{2},\frac{9}{8}\right\},\left\{\frac{9}{7},\frac{15}{7},\frac{18}{7},\frac{8}{7}\right\},\left\{\frac{20}{7},\frac{15}{7},\frac{18}{7},3\right\},$\\
 & $\left\{\frac{3}{2},3,\frac{15}{4},\frac{5}{4}\right\},\{3,8,12,15\},\left\{\frac{8}{7},\frac{3}{2},\frac{9}{7},\frac{15}{14}\right\},\{3,9,5,15\},$\\
 & $\{14,8,5,15\},\{14,9,12,15\}$\\ 
\hline
 $\Phi_{16}$ & $\left\{\frac{16}{15},\frac{6}{5},\frac{4}{3},\frac{8}{5}\right\},\{16,4,6,8\},\{16,4,6,10\},\{16,4,12,8\},\{16,4,12,10\},$\\
 & $\{16,14,6,8\},\{16,14,6,10\},\{16,14,12,8\},\{16,14,12,10\},\left\{\frac{14}{13},\frac{16}{13},\frac{18}{13},\frac{20}{13}\right\},$\\
 & $\left\{\frac{14}{13},\frac{16}{13},\frac{24}{13},\frac{20}{13}\right\},\left\{\frac{4}{3},\frac{16}{3},\frac{8}{3},\frac{10}{3}\right\},\left\{\frac{4}{3},\frac{16}{3},\frac{8}{3},4\right\},\left\{\frac{4}{3},\frac{16}{3},\frac{14}{3},\frac{10}{3}\right\},\left\{\frac{4}{3},\frac{16}{3},\frac{14}{3},4\right\},$\\
 & $\left\{6,\frac{16}{3},\frac{8}{3},\frac{10}{3}\right\},\left\{6,\frac{16}{3},\frac{8}{3},4\right\},\left\{6,\frac{16}{3},\frac{14}{3},\frac{10}{3}\right\},\left\{6,\frac{16}{3},\frac{14}{3},4\right\},\left\{\frac{12}{11},\frac{14}{11},\frac{16}{11},\frac{18}{11}\right\},$\\
 & $\left\{\frac{12}{11},\frac{14}{11},\frac{16}{11},\frac{20}{11}\right\},\left\{\frac{6}{5},\frac{8}{5},\frac{16}{5},\frac{12}{5}\right\},\left\{\frac{6}{5},\frac{8}{5},\frac{16}{5},\frac{14}{5}\right\},\left\{\frac{12}{11},\frac{24}{11},\frac{16}{11},\frac{18}{11}\right\},\left\{\frac{12}{11},\frac{24}{11},\frac{16}{11},\frac{20}{11}\right\},$\\
 & $\left\{\frac{6}{5},\frac{18}{5},\frac{16}{5},\frac{12}{5}\right\},\left\{\frac{6}{5},\frac{18}{5},\frac{16}{5},\frac{14}{5}\right\},\left\{4,\frac{8}{5},\frac{16}{5},\frac{12}{5}\right\},\left\{4,\frac{8}{5},\frac{16}{5},\frac{14}{5}\right\},\left\{4,\frac{18}{5},\frac{16}{5},\frac{12}{5}\right\},$ \\
 & $\left\{4,\frac{18}{5},\frac{16}{5},\frac{14}{5}\right\},\left\{\frac{10}{9},\frac{4}{3},\frac{14}{9},\frac{16}{9}\right\},\left\{\frac{8}{7},\frac{10}{7},\frac{12}{7},\frac{16}{7}\right\},\left\{\frac{10}{9},\frac{4}{3},\frac{20}{9},\frac{16}{9}\right\},\left\{\frac{8}{7},\frac{10}{7},\frac{18}{7},\frac{16}{7}\right\},$\\
 & $\left\{\frac{8}{7},\frac{20}{7},\frac{12}{7},\frac{16}{7}\right\},\left\{\frac{8}{7},\frac{20}{7},\frac{18}{7},\frac{16}{7}\right\},\left\{\frac{8}{3},\frac{4}{3},\frac{14}{9},\frac{16}{9}\right\},\left\{\frac{8}{3},\frac{4}{3},\frac{20}{9},\frac{16}{9}\right\}$ \\ 
 $\Phi_{20}$ & $\{20,4,12,8\},\{20,4,10,14\},\{20,18,12,14\},\{20,18,10,8\},\left\{\frac{4}{3},\frac{20}{3},\frac{14}{3},\frac{10}{3}\right\},$\\
  & $\left\{\frac{4}{3},\frac{20}{3},4,\frac{16}{3}\right\},\left\{\frac{12}{11},\frac{14}{11},\frac{20}{11},\frac{18}{11}\right\},\left\{\frac{30}{11},\frac{14}{11},\frac{20}{11},\frac{24}{11}\right\},\left\{\frac{8}{7},\frac{10}{7},\frac{18}{7},\frac{20}{7}\right\},$ \\
   & $\left\{\frac{14}{13},\frac{30}{13},\frac{24}{13},\frac{20}{13}\right\},\left\{\frac{8}{7},\frac{24}{7},\frac{16}{7},\frac{20}{7}\right\}$\\
\hline
 $\Phi_{24}$ & $\{24,8,14,6\},\{24,8,12,20\},\{24,18,14,20\},\{24,18,12,6\},\left\{\frac{8}{7},\frac{24}{7},\frac{20}{7},\frac{12}{7}\right\},$\\
  & $\left\{\frac{30}{7},\frac{24}{7},\frac{18}{7},\frac{12}{7}\right\},\left\{\frac{12}{11},\frac{18}{11},\frac{24}{11},\frac{30}{11}\right\},\left\{\frac{14}{13},\frac{20}{13},\frac{24}{13},\frac{18}{13}\right\},\left\{\frac{6}{5},\frac{12}{5},\frac{18}{5},\frac{24}{5}\right\}$\\ 
\hline
 $\Phi_{30}$ & $\left\{\frac{30}{7},\frac{24}{7},\frac{18}{7},\frac{8}{7}\right\},\{8,14,20,30\},\{8,18,12,30\},\{24,14,12,30\},\{24,18,20,30\}$\\ 
 $\Phi_{3, 7}$ & $\left\{\frac{10}{3},\frac{7}{2},\frac{5}{2},4\right\},\left\{\frac{10}{3},\frac{7}{2},\frac{5}{2},\frac{3}{2}\right\},\left\{\frac{10}{3},\frac{7}{2},3,4\right\},\left\{\frac{10}{3},\frac{7}{2},3,\frac{3}{2}\right\},\left\{\frac{16}{9},\frac{5}{3},\frac{7}{3},3\right\},$\\
 & $\left\{\frac{16}{9},\frac{5}{3},\frac{7}{3},\frac{4}{3}\right\},\left\{\frac{16}{9},\frac{8}{3},\frac{7}{3},3\right\},\left\{\frac{16}{9},\frac{8}{3},\frac{7}{3},\frac{4}{3}\right\},\left\{\frac{10}{3},3,4,7\right\},\left\{\frac{10}{3},3,5,7\right\},$\\
 & $\left\{\frac{10}{3},6,4,7\right\},\left\{\frac{10}{3},6,5,7\right\},\left\{\frac{16}{9},\frac{4}{3},\frac{3}{2},\frac{7}{6}\right\},\left\{\frac{16}{9},\frac{4}{3},\frac{5}{3},\frac{7}{6}\right\}$ \\
  $\Phi_{4, 7}$ & $\left\{\frac{15}{8},\frac{7}{6},\frac{3}{2},\frac{4}{3}\right\},\left\{\frac{15}{8},\frac{7}{6},\frac{5}{3},\frac{4}{3}\right\},\left\{\frac{15}{8},\frac{3}{2},\frac{5}{2},\frac{7}{2}\right\},\left\{\frac{15}{8},\frac{3}{2},3,\frac{7}{2}\right\},\left\{\frac{15}{8},4,\frac{5}{2},\frac{7}{2}\right\},\left\{\frac{15}{8},4,3,\frac{7}{2}\right\}$\\  
\hline
$\Phi_{6, 7}$ & None \\ 
\hline
$\Phi_{3, 9}$ & $\left\{\frac{3}{2},\frac{7}{6},\frac{5}{3},\frac{4}{3}\right\},\left\{\frac{3}{2},\frac{7}{3},\frac{5}{3},\frac{4}{3}\right\},\left\{3,\frac{4}{3},\frac{7}{3},\frac{10}{3}\right\},\left\{3,\frac{4}{3},\frac{7}{3},\frac{5}{3}\right\},\left\{3,\frac{4}{3},\frac{8}{3},\frac{10}{3}\right\},\left\{3,\frac{4}{3},\frac{8}{3},\frac{5}{3}\right\},$\\
 & $\{4,9,5,8\},\{4,9,5,3\},\{4,9,6,8\},\{4,9,6,3\},\left\{\frac{7}{4},\frac{9}{8},\frac{3}{2},\frac{15}{8}\right\},\left\{\frac{7}{4},\frac{9}{8},\frac{3}{2},\frac{5}{4}\right\},\{7,9,5,8\},$\\
 & $\{7,9,5,3\},\{7,9,6,8\},\{7,9,6,3\},\left\{\frac{8}{5},\frac{6}{5},\frac{9}{5},\frac{12}{5}\right\},\left\{\frac{8}{5},\frac{6}{5},\frac{9}{5},\frac{7}{5}\right\},\left\{\frac{7}{4},\frac{5}{4},\frac{9}{4},\frac{3}{2}\right\},$\\
 & $\left\{\frac{7}{4},3,\frac{9}{4},\frac{3}{2}\right\},\left\{\frac{5}{2},\frac{5}{4},\frac{9}{4},\frac{3}{2}\right\},\left\{\frac{5}{2},3,\frac{9}{4},\frac{3}{2}\right\},\left\{\frac{5}{2},\frac{3}{2},3,\frac{9}{2}\right\},\left\{\frac{5}{2},\frac{3}{2},\frac{7}{2},\frac{9}{2}\right\},\left\{\frac{10}{7},\frac{8}{7},\frac{12}{7},\frac{9}{7}\right\},$\\
 & $\left\{\frac{5}{2},5,3,\frac{9}{2}\right\},\left\{\frac{5}{2},5,\frac{7}{2},\frac{9}{2}\right\},\left\{\frac{10}{7},\frac{15}{7},\frac{12}{7},\frac{9}{7}\right\},\left\{4,\frac{3}{2},3,\frac{9}{2}\right\},\left\{4,\frac{3}{2},\frac{7}{2},\frac{9}{2}\right\},\left\{4,5,3,\frac{9}{2}\right\},$\\
 & $\left\{4,5,\frac{7}{2},\frac{9}{2}\right\}$ \\  
\hline
$\Phi_{4, 9}$ & None \\
\hline
$\Phi_{6, 9}$ & $\left\{\frac{5}{2},\frac{9}{5},\frac{12}{5},\frac{6}{5}\right\},\left\{\frac{5}{2},\frac{9}{5},\frac{7}{5},\frac{6}{5}\right\},\left\{\frac{7}{4},\frac{7}{2},\frac{9}{2},5\right\},\left\{\frac{7}{4},\frac{7}{2},\frac{9}{2},\frac{3}{2}\right\},\left\{\frac{7}{4},3,\frac{9}{2},5\right\},\left\{\frac{7}{4},3,\frac{9}{2},\frac{3}{2}\right\},$\\
 & $\left\{\frac{5}{2},6,8,9\right\},\left\{\frac{5}{2},6,3,9\right\},\left\{\frac{5}{2},5,8,9\right\},\left\{\frac{5}{2},5,3,9\right\}$\\  
$\Phi_{3, 14}$ & None \\  
$\Phi_{4, 14}$ & $\left\{\frac{9}{2},14,4,10\right\},\left\{\frac{9}{2},14,4,6\right\},\left\{\frac{9}{2},14,12,10\right\},\left\{\frac{9}{2},14,12,6\right\},\left\{\frac{9}{2},\frac{4}{3},\frac{14}{3},4\right\},$\\
 & $\left\{\frac{9}{2},\frac{4}{3},\frac{14}{3},\frac{8}{3}\right\},\left\{\frac{9}{2},\frac{16}{3},\frac{14}{3},4\right\},\left\{\frac{9}{2},\frac{16}{3},\frac{14}{3},\frac{8}{3}\right\}$ \\  
$\Phi_{6, 14}$ & $\left\{\frac{10}{3},14,6,12\right\},\left\{\frac{10}{3},14,6,4\right\},\left\{\frac{10}{3},14,10,12\right\},\left\{\frac{10}{3},14,10,4\right\},\left\{\frac{10}{3},\frac{6}{5},\frac{14}{5},\frac{16}{5}\right\},$ \\
 & $\left\{\frac{10}{3},\frac{6}{5},\frac{14}{5},\frac{8}{5}\right\},\left\{\frac{10}{3},\frac{18}{5},\frac{14}{5},\frac{16}{5}\right\},\left\{\frac{10}{3},\frac{18}{5},\frac{14}{5},\frac{8}{5}\right\},\left\{\frac{16}{9},\frac{4}{3},\frac{8}{3},\frac{14}{3}\right\},\left\{\frac{16}{9},\frac{4}{3},4,\frac{14}{3}\right\},$\\
 & $\left\{\frac{16}{9},\frac{16}{3},\frac{8}{3},\frac{14}{3}\right\},\left\{\frac{16}{9},\frac{16}{3},4,\frac{14}{3}\right\}$ \\  
\hline
$\Phi_{3, 18}$ & $\{7,18,12,14\},\{7,18,12,6\},\{7,18,8,14\},\{7,18,8,6\}$ \\  
\hline
$\Phi_{4, 18}$ & None\\  
\hline
$\Phi_{6, 18}$ & $\left\{6,\frac{4}{3},\frac{14}{3},\frac{16}{3}\right\},\left\{6,\frac{4}{3},\frac{14}{3},\frac{8}{3}\right\},\left\{6,\frac{4}{3},\frac{10}{3},\frac{16}{3}\right\},\left\{6,\frac{4}{3},\frac{10}{3},\frac{8}{3}\right\},\left\{6,\frac{20}{3},\frac{14}{3},\frac{16}{3}\right\},$\\
& $\left\{6,\frac{20}{3},\frac{14}{3},\frac{8}{3}\right\},\left\{6,\frac{20}{3},\frac{10}{3},\frac{16}{3}\right\},\left\{6,\frac{20}{3},\frac{10}{3},\frac{8}{3}\right\},\left\{\frac{20}{17},\frac{18}{17},\frac{24}{17},\frac{30}{17}\right\},\{4,18,12,14\},$\\
& $\{4,18,12,6\},\{4,18,8,14\},\{4,18,8,6\},\{16,18,12,14\},\{16,18,12,6\},$\\
& $\{16,18,8,14\},\{16,18,8,6\},\left\{\frac{10}{7},\frac{8}{7},\frac{18}{7},\frac{20}{7}\right\},\left\{\frac{10}{7},\frac{8}{7},\frac{18}{7},\frac{12}{7}\right\},\left\{\frac{14}{11},\frac{12}{11},\frac{18}{11},\frac{24}{11}\right\},$\\
& $\left\{\frac{14}{11},\frac{12}{11},\frac{18}{11},\frac{16}{11}\right\},\left\{\frac{10}{7},\frac{24}{7},\frac{18}{7},\frac{20}{7}\right\},\left\{\frac{10}{7},\frac{24}{7},\frac{18}{7},\frac{12}{7}\right\},\left\{\frac{8}{5},\frac{6}{5},\frac{16}{5},\frac{18}{5}\right\},\left\{\frac{16}{13},\frac{14}{13},\frac{24}{13},\frac{18}{13}\right\},$\\
& $\left\{\frac{8}{5},\frac{6}{5},\frac{12}{5},\frac{18}{5}\right\},\left\{\frac{16}{13},\frac{14}{13},\frac{20}{13},\frac{18}{13}\right\},\left\{\frac{16}{13},\frac{30}{13},\frac{24}{13},\frac{18}{13}\right\},\left\{\frac{16}{13},\frac{30}{13},\frac{20}{13},\frac{18}{13}\right\},\left\{4,\frac{6}{5},\frac{16}{5},\frac{18}{5}\right\},$\\
 & $\left\{4,\frac{6}{5},\frac{12}{5},\frac{18}{5}\right\}$ 
 \\
\hline
$\Phi_{5, 8}$ & None \\
\hline
$\Phi_{5, 10}$ & $\left\{\frac{5}{4},\frac{7}{4},\frac{15}{8},\frac{9}{8}\right\},\left\{\frac{5}{4},\frac{3}{2},\frac{15}{8},\frac{9}{8}\right\},\left\{5,4,\frac{5}{2},\frac{3}{2}\right\},\left\{5,4,\frac{9}{2},\frac{3}{2}\right\},\left\{5,3,\frac{5}{2},\frac{3}{2}\right\},\left\{5,3,\frac{9}{2},\frac{3}{2}\right\},$\\
 & $\left\{\frac{3}{2},\frac{5}{2},\frac{7}{4},\frac{5}{4}\right\},\left\{\frac{4}{3},\frac{5}{3},\frac{3}{2},\frac{5}{2}\right\},\left\{\frac{4}{3},\frac{5}{3},\frac{3}{2},\frac{7}{6}\right\},\left\{3,\frac{5}{2},\frac{7}{4},\frac{5}{4}\right\},\left\{\frac{7}{3},\frac{5}{3},\frac{3}{2},\frac{5}{2}\right\},\left\{\frac{7}{3},\frac{5}{3},\frac{3}{2},\frac{7}{6}\right\},$\\
 & $\left\{\frac{5}{3},3,\frac{10}{3},4\right\},\left\{\frac{5}{3},3,\frac{10}{3},\frac{4}{3}\right\},\left\{\frac{5}{3},\frac{7}{3},\frac{10}{3},4\right\},\left\{\frac{5}{3},\frac{7}{3},\frac{10}{3},\frac{4}{3}\right\},\{3,7,4,10\},\{3,7,8,10\},$\\
 & $\{3,5,4,10\},\{3,5,8,10\},\{9,7,4,10\},\{9,7,8,10\},\{9,5,4,10\},\{9,5,8,10\}$\\  
\hline
$\Phi_{5, 12}$ & None \\  
\hline
$\Phi_{8, 10}$ & None\\  
\hline
$\Phi_{8, 12}$ & None\\  
\hline
$\Phi_{10, 12}$ & None\\  
\hline
$\Phi_{3, 4, 5}$ & $\left\{3,\frac{7}{4},\frac{8}{5},\frac{14}{5}\right\},\left\{\frac{8}{3},\frac{9}{4},5,4\right\},\left\{\frac{8}{3},\frac{9}{4},5,3\right\},\left\{\frac{14}{9},\frac{9}{4},\frac{4}{3},\frac{5}{3}\right\},\left\{\frac{14}{9},\frac{9}{4},\frac{7}{3},\frac{5}{3}\right\}$ \\  
$\Phi_{3, 5, 6}$ & $\left\{3,\frac{8}{5},\frac{14}{5},\frac{3}{2}\right\},\left\{3,\frac{8}{5},\frac{14}{5},\frac{7}{2}\right\}$\\  
\hline
$\Phi_{4, 5, 6}$ & None\\  
\hline
$\Phi_{3, 4, 8}$ & $\left\{\frac{7}{3},4,\frac{3}{2},\frac{5}{2}\right\},\left\{\frac{7}{3},4,\frac{3}{2},\frac{7}{2}\right\},\left\{\frac{7}{3},4,\frac{9}{2},\frac{5}{2}\right\},\left\{\frac{7}{3},4,\frac{9}{2},\frac{7}{2}\right\}$ \\  
\hline
$\Phi_{3, 6, 8}$ & None\\
\hline
$\Phi_ {4, 6, 8}$ & $\left\{4,\frac{5}{3},\frac{3}{2},\frac{5}{2}\right\},\left\{4,\frac{5}{3},\frac{3}{2},\frac{7}{2}\right\},\left\{4,\frac{5}{3},\frac{9}{2},\frac{5}{2}\right\},\left\{4,\frac{5}{3},\frac{9}{2},\frac{7}{2}\right\},\left\{3,\frac{7}{3},8,4\right\},\left\{3,\frac{7}{3},8,6\right\},$\\
 & $\left\{7,\frac{7}{3},8,4\right\},\left\{7,\frac{7}{3},8,6\right\},\left\{\frac{7}{5},\frac{7}{3},\frac{6}{5},\frac{8}{5}\right\},\left\{\frac{7}{5},\frac{7}{3},\frac{12}{5},\frac{8}{5}\right\}$ \\  
\hline
$\Phi_{3, 4, 10}$ & None\\
\hline
$\Phi_{3, 6, 10}$ & $\left\{3,6,\frac{14}{5},\frac{8}{5}\right\},\left\{5,6,\frac{14}{5},\frac{8}{5}\right\}$ \\ 
\hline 
$\Phi_{4, 6, 10}$ & $\left\{\frac{5}{2},6,\frac{14}{5},\frac{8}{5}\right\},\left\{\frac{7}{2},\frac{14}{9},\frac{10}{3},4\right\},\left\{\frac{7}{2},\frac{14}{9},\frac{10}{3},\frac{4}{3}\right\},\left\{\frac{7}{2},\frac{8}{3},4,10\right\},\left\{\frac{7}{2},\frac{8}{3},8,10\right\}$
\\
\hline
$\Phi_ {3, 4, 12}$ & $\left\{3,\frac{7}{4},\frac{9}{4},\frac{15}{4}\right\},\left\{\frac{7}{3},4,\frac{8}{3},\frac{14}{3}\right\},\left\{\frac{7}{3},4,\frac{10}{3},\frac{4}{3}\right\},\left\{\frac{9}{5},\frac{8}{5},\frac{12}{5},\frac{6}{5}\right\},\left\{\frac{9}{5},\frac{14}{5},\frac{12}{5},\frac{6}{5}\right\},$ \\
 & $\left\{\frac{15}{7},\frac{10}{7},\frac{12}{7},\frac{18}{7}\right\},\left\{\frac{15}{7},\frac{16}{7},\frac{12}{7},\frac{18}{7}\right\},\{5,4,6,12\},\left\{\frac{15}{11},\frac{14}{11},\frac{18}{11},\frac{12}{11}\right\},\{5,10,6,12\},$ \\
 & $\left\{\frac{15}{11},\frac{20}{11},\frac{18}{11},\frac{12}{11}\right\},\{9,4,6,12\},\{9,10,6,12\}$ \\ 
\hline
$\Phi_ {3, 6, 12}$ & $\left\{\frac{3}{2},\frac{5}{4},\frac{15}{8},\frac{9}{8}\right\},\left\{\frac{3}{2},\frac{9}{4},\frac{15}{8},\frac{9}{8}\right\},\left\{3,\frac{3}{2},\frac{9}{4},\frac{15}{4}\right\},\left\{3,\frac{7}{2},\frac{9}{4},\frac{15}{4}\right\},\left\{3,6,\frac{9}{2},\frac{3}{2}\right\},$ \\
 & $\left\{5,6,\frac{9}{2},\frac{3}{2}\right\},\left\{\frac{9}{5},\frac{7}{5},\frac{12}{5},\frac{6}{5}\right\},\left\{\frac{9}{5},3,\frac{12}{5},\frac{6}{5}\right\},\left\{\frac{15}{7},\frac{9}{7},\frac{12}{7},\frac{18}{7}\right\},\{5,3,6,12\},$\\
& $\{9,3,6,12\}$ \\ 
\hline 
$\Phi_{4, 6, 12}$ & $\left\{4,\frac{5}{3},\frac{14}{3},\frac{10}{3}\right\},\left\{4,\frac{5}{3},\frac{14}{3},\frac{8}{3}\right\},\left\{4,\frac{5}{3},\frac{8}{3},\frac{4}{3}\right\},\left\{4,\frac{5}{3},\frac{10}{3},\frac{4}{3}\right\},\left\{\frac{5}{2},6,\frac{7}{2},\frac{3}{2}\right\},\left\{\frac{5}{2},6,\frac{9}{2},\frac{3}{2}\right\},$\\
 & $\{4,3,12,8\},\left\{\frac{10}{7},\frac{9}{7},\frac{12}{7},\frac{8}{7}\right\},\left\{\frac{14}{5},3,\frac{12}{5},\frac{6}{5}\right\},\{10,3,12,8\},\{10,3,12,6\},$\\
 & $\left\{\frac{8}{5},\frac{7}{5},\frac{12}{5},\frac{6}{5}\right\},\left\{\frac{16}{7},\frac{9}{7},\frac{12}{7},\frac{8}{7}\right\},\{4,3,12,6\},\left\{\frac{8}{5},3,\frac{12}{5},\frac{6}{5}\right\},\left\{\frac{14}{5},\frac{7}{5},\frac{12}{5},\frac{6}{5}\right\},$ \\
 & $\left\{\frac{8}{5},\frac{7}{5},\frac{16}{5},\frac{12}{5}\right\},\left\{\frac{14}{5},\frac{7}{5},\frac{16}{5},\frac{12}{5}\right\},\left\{\frac{16}{7},\frac{9}{7},\frac{18}{7},\frac{12}{7}\right\},\left\{\frac{14}{5},3,\frac{16}{5},\frac{12}{5}\right\},\left\{\frac{10}{7},\frac{9}{7},\frac{18}{7},\frac{12}{7}\right\},$\\
 & $\left\{\frac{8}{5},3,\frac{16}{5},\frac{12}{5}\right\}$
\end{longtable}

\appendix

\section{Deferred proofs}

\subsection{Properties of the dual fractional ideal}\label{kkkkza3b2}
\begin{lem}\label{kkasq} Let $n=p_1^{r_1}p_2^{r_2}\cdots p_s^{r_s}$ be the decomposition of $n$ in prime factors.
We write $\zeta_i$ for a primitive $p_i^{r_i}$-root of unity. Then
\be
u\equiv \frac{\Phi^\prime_n(\zeta)}{\prod_{i=1}^s \Phi^\prime_{p_i^{r_i}}(\zeta_i)}\ \ \text{is a unit in }\mathfrak{O}.
\ee
\end{lem}
The proof will be given in \S.\ref{kkkkza3b}  below.

\begin{lem}[\!\!\cite{japonese}] \label{kkkdh}
Let $3\leq n\neq p^r,\,2p^r$ with $p$ an odd prime. Then in $\mathfrak{O}$ there is a purely imaginary unit.
\end{lem}
\begin{proof} If $n=0\mod4$, $i\in\mathfrak{O}$. If $n\neq0\mod4$, we replace $n$ by the conductor $\bn$ which is an odd integer divisible by two distinct primes. Let $\zeta$ be a primitive $\bn$-th root of unity. $(\zeta-\zeta^{-1})=\zeta(1-\zeta^{-2})$
is an imaginary unit.
\end{proof}

\begin{rem} It is easy to check that if $\bn=p^r$ with $p$ an odd prime and the class number of $\BK$ is 1, there are no imaginary units.
\end{rem}

\begin{lem}[\!\!\cite{japonese}]\label{llemmeuni} There exists a unit $\varepsilon\in \mathfrak{O}$ such that
\be\label{pplza}
\frac{\varepsilon}{\overline{\Phi^\prime_n(\zeta)}}\equiv \varrho
\ee
is $\iota$-odd (i.e.\! purely imaginary).
\end{lem}
\begin{proof} If $n$ is a prime power $p^r\neq2$ we have
\be
\bar\varepsilon=\begin{cases}-i\,\zeta^{-1} & n=2^r,\ r\geq2\\
\zeta^{-1}\zeta^{-(p^{r-1}:2)} & n=p^r,\ p\ \text{odd prime}
\end{cases}
\ee
where $:$ means division in $(\Z/p^r\Z)^\times$. In the general case, $n=\prod_{i=1}^sp_i^{r_i}$ we take $\epsilon$ equal to $\bar u$ in \textbf{Lemma \ref{kkasq}} times the product of the $\varepsilon$ associated to each prime power in the product. If $s$ is odd $\varrho$ is purely imaginary and we are done. If $s$ is even, $\varrho$ is real, and we multiply it by the imaginary unity in \textbf{Lemma \ref{kkkdh}.}\end{proof}

\begin{corl}\label{kkkka12x}
For all fractional ideal $\mathfrak{a}$ of $\BK$ we have 
\be\label{jj1234}\mathfrak{a}^\ast=\varrho/\bar{\mathfrak{a}},\ 
\text{for a certain }\varrho\in\BK^\times\ \text{with }\iota(\varrho)=-\varrho.
\ee
If $\bn$ is not of a power of an odd prime, we may alternatively choose $\varrho$ to be real by multiplying it by the appropriate imaginary unit.
\end{corl}

\subsection{Proof of Lemma \ref{kkasq}}\label{kkkkza3b} 

\begin{lem}
Let $n=p_1^{r_1}\cdots p_s^{n_s}$ and $\zeta$ a primitive $n$-th root of unity. We write $\zeta_{p_i^{r_i}}$ for a primitive $p_i^{r_i}$ root of unity. Then
\be
\Phi^\prime_n(\zeta) = u\prod_i\Phi_{p^{r_i}}^\prime(\zeta_{p_i^{r_i}})\qquad u\ \text{is a unit of $\Z[\zeta]$}
\ee
Here $\Phi^\prime_n$ denotes the derivative of the polynomial $\Phi_n$.
\end{lem}

\begin{proof} We shall use the symbol $\sim$ to mean equality up to multiplication by a unity. 
If $n$ is a prime power there is nothing to show, so we assume $n$ is divisible by two distinct primes.

We consider first the case of $n=p_1p_2\cdots p_s$ square-free (and $s\geq 2$). Then
\be
\begin{split}
\Phi^\prime_n(\zeta)&= \prod_{(a,n)=1\atop a\neq 1}(\zeta-\zeta^a)=\zeta^{\phi(n)-1}
\prod_{(a,n)=1\atop a\neq 1}(1-\zeta^{a-1})\sim\\
&\sim
\prod_{(a,n)=1\atop
(a-1,n)=p_1\dots\widehat{p_i}\dots p_s}(1-\zeta_i)\sim \prod_i (1-\zeta_i)^{m_i}
\end{split}
\ee
where $\zeta_{i}$ is a primitive $p^i$-root of unity and $m_i$ are certain integers to be determined.
To determine the $m_i$ it is enough to compute the norm of the \text{lhs} which is the discriminant of the cyclotomic polynomial. Hence
\be
\pm \prod_i \frac{p_i^{\phi(n)}}{p_i^{\phi(n)/(p_i-1)}}\sim \prod_i p_i^{m_i \phi(n)/(p_i-1)}
\ee
so $m_i=p_i-2$. On the other hand $\Phi^\prime_{p_i}(\zeta_i)\sim (1-\zeta_i)^{p_i-2}$, so the \textbf{Lemma} is proven for $n$ square-free. Now let $n=p_1^{r_1}p_2^{r_2}\cdots p_s^{r_s}$ and $q=p_1p_2\cdots p_s$ its radical (which is square-free by definition). One has
\be
\Phi_n(x)=\Phi_q(x^{n/q}).
\ee
Hence $\Phi_n^\prime(x)= (n/q) x^{n/q-1} \,\Phi^\prime_q(x^{n/q})$ and
\be
\Phi^\prime_n(\zeta_n)\sim \frac{n}{q}\, \Phi_q^\prime(\zeta_q)\sim \big(\prod_i p_i^{r_i-1}\big)\,
\prod_i \Phi_{p_i}^\prime(\zeta_{p_i})\sim \prod_i \Phi^\prime_{p_i^{r_i}}(\zeta_{p_i^{r_i}}).
\ee
\end{proof}

\end{document}